    \documentclass[10pt,oneside,a4]{extarticle}
    \usepackage{etex}
    \usepackage{latexsym,graphicx,enumerate,enumitem}
    \usepackage[utf8]{inputenc}
    \usepackage[left=3.1cm,top=3.1cm,right=3.1cm,bottom=3.6cm]{geometry}

    \usepackage[T1]{fontenc}
    \usepackage{breakcites}
    \usepackage{changepage}
    \usepackage[english]{babel}
    \usepackage{verbatim,amsfonts,enumitem}
    \usepackage{multirow}
    \usepackage{verbatimbox,fancyvrb}
    \usepackage{amsmath,amssymb,amsthm,amscd,url,soul}
    \usepackage{float}
    \usepackage{hhline}
    \usepackage[margin=1cm]{caption}
    \usepackage{subcaption}
    \usepackage{placeins,lipsum}
    \usepackage{dsfont}
    \usepackage{bbm}
    \usepackage{rotating}
    \usepackage{longtable}
    \usepackage[usenames,dvipsnames]{color}
    \usepackage{colortbl,xcolor}
    \usepackage{listings,siunitx}
    
    \newcommand{\overbar}[1]{\mkern 1.5mu\overline{\mkern-1.5mu#1\mkern-1.5mu}\mkern 1.5mu}
    
    \usepackage{lmodern}
	\newlength\longest
    
    \usepackage{hyperref}
    
    \hypersetup{colorlinks,citecolor=blue,linkcolor=blue,urlcolor=blue}
    
    \usepackage{tocloft}
    
    \usepackage[splitrule]{footmisc} 

	\usepackage{fancyhdr}
	\usepackage{lipsum}   
    \fancyheadoffset{0in}

	\renewcommand{\subsectionmark}[1]{}

    \fancypagestyle{plain}{%
      \fancyhead{} 
    }
    
    \fancypagestyle{myplain}
	{
      \fancyhf{}

      \fancyfoot[C]{\thepage}
	}
    
    \usepackage[lined,boxed,commentsnumbered]{algorithm2e}
    
    \usepackage{pstricks,pst-plot} 
    \usepackage{pgfplots}
    \usetikzlibrary{arrows,intersections}
    \usetikzlibrary{arrows.meta}
    \usetikzlibrary{automata}
    
    \pgfplotsset{
      compat=newest,
      xlabel near ticks,
      ylabel near ticks
    }
    
    \usepackage{tikz}
	\usepackage{amsmath}

	\usetikzlibrary{shapes,arrows}
	\usetikzlibrary{positioning}
	\usetikzlibrary{arrows,topaths,calc}
	\usetikzlibrary{arrows,positioning, shapes.symbols,shapes.callouts,patterns}
    
    \DeclareCaptionLabelFormat{andtable}{#1~#2  \&  \tablename~\thetable}

    \usepackage{mathtools}
    
    \usepackage{chngcntr}
    \counterwithin{table}{section}
    \counterwithin{figure}{section}
    
    \definecolor{Gray}{gray}{0.9}

    \usepackage[T1]{fontenc}
    
    \usepackage{titlesec}

    \titleformat{\paragraph}
    {\normalfont\normalsize\bfseries}{\theparagraph}{1em}{}
    \titlespacing*{\paragraph}
    {0pt}{3.25ex plus 1ex minus .2ex}{1.5ex plus .2ex}

    \usepackage{afterpage}

    \title{\textbf{Pricing of Mexican Interest Rate Swaps in Presence of Multiple Collateral Currencies}}
    \author{Jorge Íñigo Martínez\thanks{Email: \texttt{inigo\_89@ciencias.unam.mx}}\\
            Maestría en Finanzas Cuantitativas\\
            Facultad de Ciencias Actuariales\\
            Universidad Anáhuac\thanks{Av. Universidad Anáhuac 46, Lomas Anáhuac, 52786 Huixquilucan, México.}\\
            \\
            Supervisor: Jos\'e Luis Manrique Medina, MSc\thanks{Universidad Anáhuac México Norte, email: \texttt{luis.manrique.medina@gmail.com}}}
    \date{March 2017}
    
\begin{document}

    	  \newdimen\origiwspc%
          \newdimen\origiwstr%
          \origiwspc=\fontdimen2\font
          \origiwstr=\fontdimen3\font
          
          \fontdimen2\font=\origiwspc
          \fontdimen3\font=0.1em
    	\pagenumbering{roman}

    	\thispagestyle{empty}
        \maketitle
        \thispagestyle{empty}

        \renewcommand{\contentsname}{Contents}
    	
    	\renewcommand{\figurename}{Figure}
    	\renewcommand{\tablename}{Table}
    	\renewcommand\bibname{Bibliography}
    	\newtheorem{teo}{Theorem}[section]
    	\newtheorem{prop}{Proposition}[section]
    	\newtheorem{cor}{Corollary}[section]
    	\newtheorem{lema}{Lemma}[section]

        \theoremstyle{definition}
        \newtheorem{defi}{Definition}[section]
    	\newtheorem{nota}{Notation}[section]
    	\newtheorem{rem}{Remark}[section]
    	\newtheorem{assum}{Assumption}[section]
    	\newtheorem{example}{Example}[section]
        
    	\numberwithin{equation}{section}

	
    
    \newpage
    \thispagestyle{empty}
    
    \null\vspace{\fill}
    \begin{abstract}
    The financial crisis of 2007/08 caused catastrophic consequences and brought a bunch of changes around the world. Interest rates that were known to follow or behave similarly of each other diverged. Furthermore, the regulation and in particular the counterparty credit risk began to to be considered and quantified. Consequently, pre-crisis models are no longer valid. Indeed, this work sets the basis to define a valid model that considers the post-crisis world assumptions for the Mexican swap market. The model used in this work was the proposed by Fujii, Shimada and Takahashi in \cite{fujii2010note}. This model allow us to value interest rate derivatives and future cash flows with the existence of a collateral agreement (with a collateral currency). In this document we build the discounting and projection curves for MXN interest rate derivatives considering the collateral currencies: USD, EUR and MXN. Also, we present the pricing when the derivative is uncollateralized. Finally, we show the effect of the cross-currency swaps when valuing through different collateral currencies.
    \end{abstract}
    \par \vspace{1.5cm}
    \renewcommand{\abstractname}{Resumen}
    \begin{abstract}
    La crisis financiera del 2007-2008 trajo consigo varias consecuencias en el mundo de las finanzas. En particular varios niveles de tasas de interés y \emph{spreads} dejaron de comportarse de la forma que solían hacerlo. Además desde este evento la regulación y en particular el riesgo crediticio tomó mayor énfasis a la hora de definirlo y cuantificarlo. En consecuencia, los modelos  de valuación de derivados, usados antes de la crisis, dejaron de ser válidos. Este trabajo tiene como objetivo definir un modelo de valuación coherente que considere los supuestos del mundo actual. El modelo usado para definir la valuación de productos denominados en pesos (MXN) es el presentado por Fujii, Shimada y Takahashi en \cite{fujii2010note}. De hecho, este modelo demuestra que la moneda del colateral define la forma de descontar flujos de efectivo futuros en un mundo realista. En este documento se construyen las curvas para descontar flujos en pesos (MXN) cuando la moneda de colateral es: dólar americano (USD), euro (EUR) y pesos (MXN). Además, se presenta el caso de valuación cuando los swaps o en particular cualquier flujo de efectivo no tiene colateral. Finalmente se presenta un análisis de los factores que afectan a la construcción de curvas como los son los swaps de divisas.
    \end{abstract}
    \noindent \textbf{Palabras clave:} swap de tasa de interés, swap de divisas, swap de tasa de interés a un día, colateral, curva de descuento, curva de proyección, TIIE, LIBOR, tasa de fondos federales 
    \vspace{\fill}

    \newpage
    \thispagestyle{myplain}
    \renewcommand{\cftpartleader}{\cftdotfill{\cftdotsep}} 
	\renewcommand{\cftsecleader}{\cftdotfill{\cftdotsep}}
    \tableofcontents
    \thispagestyle{myplain}
    
    \newpage
    \thispagestyle{myplain}
    \addcontentsline{toc}{section}{List of Figures}
    \listoffigures
    \thispagestyle{myplain}
    
    \newpage
    \thispagestyle{myplain}
    \addcontentsline{toc}{section}{List of Tables}
    \listoftables
    \thispagestyle{myplain}

    \newpage
    \thispagestyle{myplain}
    \addcontentsline{toc}{section}{List of Abbreviations}
    \section*{List of Abbreviations}
    \noindent BS: Black and Scholes\par \medskip
    \noindent EUR: Euro code (ISO 4217)\par \medskip
    \noindent USD: United States Dollar code (ISO 4217)\par \medskip
    \noindent MXN: Mexican Peso code (ISO 4217)\par \medskip
    \noindent JPY: Japanese Yen code (ISO 4217)\par \medskip
    \noindent FX: Foreign Exchange\par \medskip
    \noindent IRS: Interest Rate Swap\par \medskip
    \noindent TS: Tenor Swap\par \medskip
    \noindent OIS: Overnight Index Swap\par \medskip
    \noindent FFS: Federal Funds Swap\par \medskip
    \noindent XCS: Cross-Currency Swap\par \medskip
    \noindent cnXCS: Constant Notional Cross-Currency Swap\par \medskip
    \noindent mtmXCS: Mark-to-Market Cross-Currency Swap\par \medskip
    \noindent OTC: Over The Counter\par \medskip
    \noindent IBOR: Interbank Offered Rate\par \medskip
    \noindent LIBOR: London Interbank Offered Rate\par \medskip
    \noindent EURIBOR: European Interbank Offered Rate\par \medskip
    \noindent TIIE: Tasa de Interés Interbancaria de Equilibrio\par \medskip
    \noindent PV: Net Present Value\par \medskip
    \noindent ISDA: International Swaps and Derivatives Association\par \medskip
    \noindent CSA: Credit Support Annex\par \medskip
    \noindent CVA: Credit Valuation Adjustment\par \medskip
    \thispagestyle{myplain}
	
   
    \newpage
	\pagenumbering{arabic}
    \section{Introduction}
    Since the publication of the well-known and famous paper \cite{black1973pricing}, the theory regarding the pricing of derivative securities has been developed through the time using this seminal paper as a main basis. Among all the papers published by these authors, Fischer Black and Myron Scholes, between 1973 and 1977, there were many assumptions that \emph{simplified}\footnote{The word \emph{simplify} is used in an academic context, i.e. in a critical and not in a derogatory way.} their Black-Scholes (BS) model \cite{hens2010financial}, such as:
    \begin{itemize}[noitemsep]
    \item Trading in the assets is continuous in time
    \item The market is arbitrage free
    \item There is a constant \emph{risk-free} rate for which banks can borrow and lend money (no limit amount!)
    \item There are no short-selling constraints
    \item There are no frictions, like transaction costs and taxes
    \item There is no dividend payments
    \item Neither counterparty to the transaction is at risk of default
    \end{itemize}
    It is important to point out that even in the mid-1970s, Black and Scholes were aware that their assumptions did (and do) not reflect the financial markets reality. Since then, these assumptions have been weakened by researchers with post-BS model papers suggesting modifications of the BS formula. For instance, in 1973, Robert C. Merton \cite{merton1973theory} removed the restriction of constant interest rates; in 1977, Jonathan Ingersoll \cite{ingersoll1977contingent} relaxed the assumption of no taxes and transaction costs; in 1979, John Cox, Stephen Ross and Mark Rubinstein \cite{cox1979option} presented a model that incorporates the timing and size of dividend payments.\par\smallskip
    Aside from these papers with variations of the BS model, derivatives valuation theory has been an active and popular research topic in financial engineering for both industry and academia. Indeed, many banks and financial institutions have been investing large amounts of money in research and development of software used for numerical calculations and simulations for pricing and risk management of derivatives. However, another major event ---and perhaps the most important--- that marked a milestone in the history of derivatives valuation was the Lehman Brothers collapse in 2008.\par\smallskip
    The financial crisis in 2007-2008 arose problems in many latitudes such as public policy, monetary policy and regulation of financial markets, without mentioning the historical plunge of stock markets and the paralysis of credit markets. Financial engineering was not exempt of the crisis. In fact, derivatives valuation frameworks entered on a process of changes in methodology and assumptions. To illustrate some of these changes we will present how the typical variables, considered for pricing a plain vanilla interest rate swap\footnote{Interest rate swaps (IRS) can be divided in three main types; plain vanilla IRS where two parties exchange fixed and floating rate payments, the tenor swap (TS) where they exchange interest rate payments based on different floating reference rates of the same currency and the cross-currency swap (XCS) where they exchange floating interest rate payments in two different currencies. See section \ref{sec:IRproducts}.} (IRS) in a pre-crisis world, have changed notably in a post-crisis world. This comparison was mostly taken from \cite{green2015xva} among other references cited below.\par\smallskip
    \subsection{IRSs in a Pre-crisis World}
    Before we analyze the main characteristics of IRSs in a pre-crisis world, let us explain briefly why swaps have become popular instruments. Around the world, companies and institutions typically own debt linked to interest rates (fixed or floating). Since there exist a lot of uncertainty about future interest rates levels, swaps allow them to hedge their interest rate exposures by exchanging a fixed rate for a floating rate, typically an Ibor\footnote{By Ibor (Interbank offered rate) we mean any reference rate which is fixed in a way similar to LIBOR, EURIBOR, TIBOR, TIIE, etc. See section \ref{section:iborrates}} rate, or vice versa. To illustrate this we present the following example in which a company uses an IRS to hedge its position.
    \begin{example}\textbf{(Example of using IRSs to hedge a loan)} Suppose a car manufacturer in Mexico enters into a loan offered by the Bank B. This loan has to be fully paid in 5 years and every month the manufacturer has to pay interest of the remaining capital. These interests are referenced to TIIE 28d plus 50 basis points due to the credit risk of the manufacturer. Companies (in particular the car manufacturer) do not like uncertainty, and even less uncertainty that is outside their core business. As a result, the car manufacturer will be happy to hedge out this interest rate risk, so it enters into a payer swap with Bank C in which the car manufacturer pays a fixed interest rate, say 7\%, and receives the floating rate (TIIE 28d) plus the 50 basis points to fulfill its debt obligation (interests of the loan with Bank B) see Figure \ref{swapexample}.
    \begin{figure}[h]
        \centering
        \vspace{-2mm}
        \includegraphics[scale=0.82]{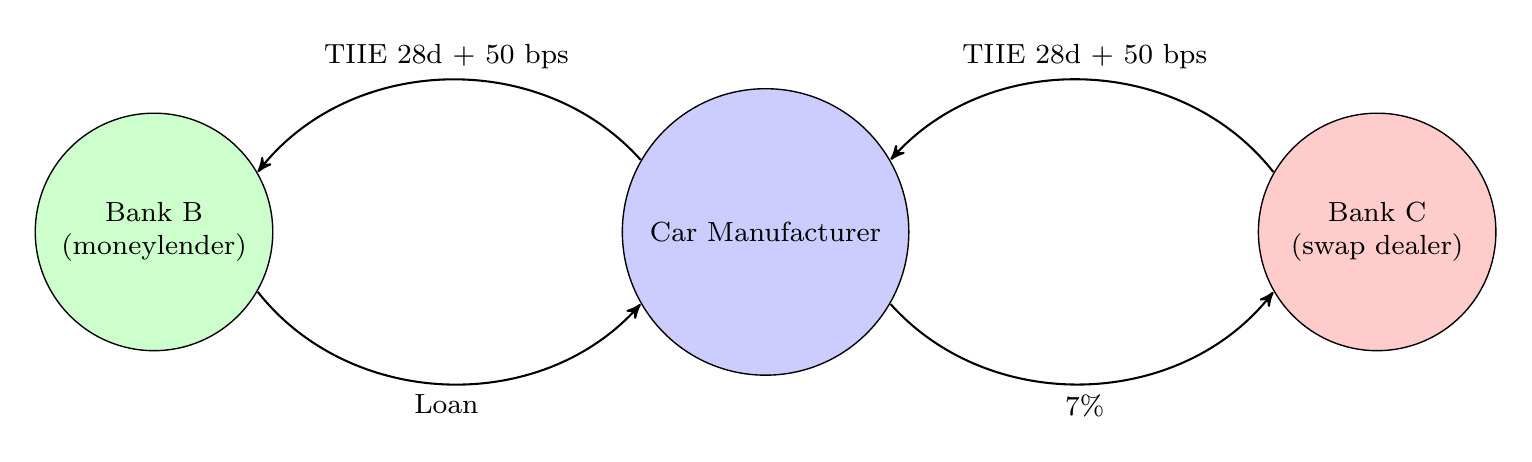}
        \caption[Example of an IRS]{Example of an IRS.}
        \label{swapexample}
        \vspace{-2mm}
     \end{figure}
     In summary, the car manufacturer will be happy because it knowns that it has to pay a 7\% annual rate regardless of what happens to the reference rate TIIE 28d during the term of the loan.
    \end{example}
    In a pre-crisis world, the traditional reference papers for pricing IRSs were \cite{bicksler1986economic} and \cite{miron1991pricing}, without mentioning the classical textbook \cite{hull1997options}. The general approach for pricing IRSs in mid-1990s was done as follows: every future cash flow was discounted using a discount curve and floating cash flows were projected using a forward curve. One motivation behind this approach was the assumption of the existence of a unique \emph{risk-free} rate at which one could borrow and lend any amount of money \cite{piterbarg2010funding}. In fact, in this framework the forward curve is totally defined from the discount curve which in turn is determined by a yield curve\footnote{In this work, when we use the term yield curve we are referring to a zero rate curve, i.e. the zero rate implied from the injective mapping $T \mapsto P(t,T)$ where $P(t,T)$ is the discount factor or zero coupon bond (see section \ref{sec:backtobasics}).}. The construction of the yield curve is done by using a simple bootstrap based on prices (quotes) of market instruments such as: cash deposits, interest rates futures, bonds and IRSs. In section \ref{pricingIRSprecrisis} we will present how a pre-crisis world simple bootstrapping is done. Using this yield curve, the valuation of any IRSs was relatively simple. Similarly, cross-currency swaps (XCSs) were valued with this \emph{single-curve} approach in each currency leg, which led into differences on mark-to-market. Also, swaps linked to the same interest reference rate but with different tenors (known as Tenor Swaps (TSs) see section \ref{sec:tenorswaps}), say three months Ibor versus six months Ibor, were typically priced with the same quote and sometimes a \emph{small} spread (premia) was charged on the shorter tenor leg to reflect operational costs. However these type of swaps should  trade flat (with no spread) if, and only if, they trade in a default-free market. Therefore they were mispriced since quotes did not consider liquidity and credit issues of the counterparty.\par\smallskip 
    Additionally, in this pre-crisis framework, counterparty risk was managed through a traditional credit limits set and, likewise funding costs, were not explicitly considered in the pricing of the swap. In 1996, the Basel I framework had already been introduced but capital management was considered a back office function \cite{green2015xva}.    
    \subsection{IRSs in a Post-crisis World}
    As the crisis hits in 2007, a deep evolution phase of the classical (pre-crisis world) framework was adopted. Indeed, market participants reacted rapidly and assumptions and negligible facts of financial models entered into a correction stage. One of the most important impacts of the financial crisis over the interest rate market dynamics was the explosion of the basis between Ibor rates and overnight indexed rates \cite{bianchetti2012markets}. These spreads widened rapidly since Ibor rates soared due to the credit and liquidity risk of the interbank market, whereas overnight reference rates deviated substantially from its target level since monetary policy decisions were adopted by international authorities in response to the financial turmoil. Indeed, the USD three month LIBOR fixing and the three month maturity overnight index swap (OIS) reached a peak of 365bps just after the collapse of Lehman Brothers in September 2008, having averaged 10bps prior to August 2007 \cite{sengupta2008libor} (see Figure \ref{plotLIBOROIS}). The reason of this widening was the cost of funding, since lending at a longer tenor (quarterly) is associated with a more counterparty risk than lending at a shorter tenor (daily). Academics and practitioners realized that Ibor rates were/are risky rates since the probability of default of leading banks cannot be neglected anymore.\par\smallskip
   \begin{figure}[h]
       \centering
       \vspace{-2mm}
       \includegraphics[scale=0.58]{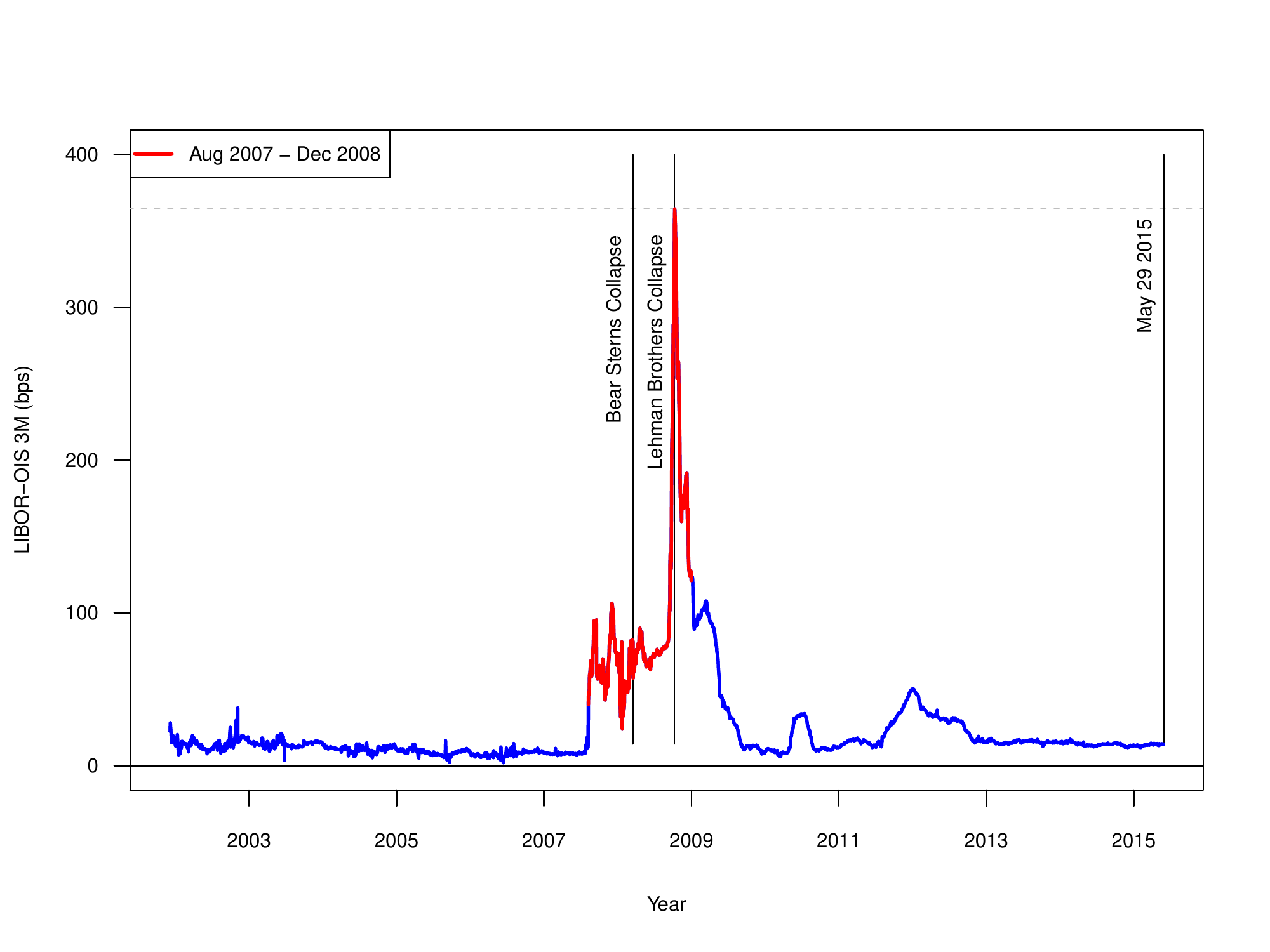}
       \caption[LIBOR-OIS 3m Spread (2001-2015)]{In this figure we can see the LIBOR-OIS 3m spread between Dec 05 2001 and May 29 2015. The spread is determined by the difference between the LIBOR 3m fixing (published rate) and the 3m swap indexed to Fed Fund Overnight Rate (end of day quote). The red line indicates the financial crisis period (August 2007 - December 2008). The Bloomberg Ticker for this spread is \texttt{LOIS Index}.}
       \label{plotLIBOROIS}
       \vspace{-2mm}
   \end{figure}
    Consequently, in the USD market, for example, due to the tenor spreads widening the unique projection/forward curve was replaced by three different projection curves, one for each LIBOR tenor: 1m, 3m and 6m. This new methodology leads to a tenor swaps par valuation (see \cite{henrard2010irony}). Moreover, traders realized that the discounting curve based on LIBOR 3m was not a good idea specially for trades under Credit Support Annex\footnote{A Credit Support Annex (CSA), is a legal document which regulates credit support (collateral) for derivative transactions. It is one of the four parts that make up an ISDA (International Swaps and Derivatives Association) Master Agreement but is not mandatory. CSAs are characterized by various clauses and parameters, such as margination frequency, margination rate, threshold, minimum transfer amount, eligible collateral, collateral currency, asymmetry, etc. See section \ref{pricingCollProd}.} (CSA) agreements. In fact, CSA agreements pay an interest rate on posted collateral typically equal to an overnight rate of the collateral currency. Therefore markets migrated rapidly to the usage of overnight rates for discounting collateralized cash flows of interest rate derivatives \cite{piterbarg2010funding}. The incorporation of tenor basis spreads and collateral rates forced to create a valuation framework that considers multiple rate curves known as \emph{multi-curve} framework. \par\smallskip
    Furthermore in a post-crisis world, credit and liquidity costs have started to be considered when swaps are priced. Indeed, CVA (Credit Valuation Adjustment) and FVA (Funding Valuation Adjustment) variables should be made for quantifying the costs of credit and funding of unsecured derivative transactions. In addition, capital management is no longer a back office function. Capital requirements are now an expensive resource that front office has to manage carefully as core activity. Hence every new transaction offered to the clients has to be priced considering the cost of capital, in order to determine whether it is profitable or not. This cost of capital is know as KVA (Capital Valuation Adjustment). \par\smallskip
    In summary, in a post-crisis world for the valuation of an IRS we need to consider a large number of interest rate instruments (cash deposits, futures, IRSs, XCSs, foreign exchange swaps, etc.) to fulfill multiple bootstrappings for projections and discounting curves and also we have to perform Monte-Carlo simulations for every counterparty to manage and calculate XVA variables: CVA (DVA)\footnote{In a bilateral model, DVA (Debt Valuation Adjustment) is considered a mirror of CVA since the valuation adjustment has to be symmetric among counterparties.}, FVA, KVA and MVA (Margin Valuation Adjustment). For further references in XVA variables we invite the reader to check \cite{gregory2012counterparty}, \cite{green2015xva}, \cite{ruiz2015xva} and \cite{lichters2015modern}.\par\smallskip
    In this thesis we will focus on the construction of the discount and forward curves in presence of a CSA agreement (collateral agreements) through different currencies. It is important to state that the transition away from risk-neutral valuation framework (pre-crisis world) to a more realistic valuation framework (post-crisis world) is not yet completed. Therefore the reader should know that the methodology presented here is not definitely and can be modified through the time once the rules of the game change or more assumptions start to weaken. 
    \subsection{The Scope of this Thesis}
    The aim of this work is to describe the construction under discount and projection curves used for the valuation of IRSs in MXN currency (Mexican Peso) for different collateral currencies such as: USD, MXN and EUR. Also we will explain how to price uncollateralized IRSs and how the \emph{multi-curve} framework affects IRSs in the absence of a collateral agreement. This work is a novelty since there are not known publications that take into account curve construction in a market where overnight index swaps do not exist. Furthermore, this could be considered a theoretical extension of the white paper \textsl{Análisis Comparativo de las Metodologías de Valuación de Swaps de TIIE} \cite{mexder2014} published by Mexican Derivatives Exchange (MexDer)\footnote{The Mexican Derivatives Exchange (MexDer) is an options and futures exchange in Mexico, located in the same building as the Mexican Stock Exchange (Bolsa Mexicana de Valores, BMV) and a subsidiary of the same owning group.} in 2014. Moreover we will justify mathematically every step and also to present explicit formulas for pricing IRSs. Finally this thesis distinguishes from \cite{mexder2014} since we detail all steps to compute a dual (or multiple) bootstrapping.\par \smallskip
   	This thesis is structured as follows: in section \ref{sec:literaturereview} we summarize a brief literature review of the main publications that manage topics such as: curve construction with collateral, bootstrapping and interpolation of curves. In section \ref{sec:backtobasics} we begin with the most basic concepts in mathematical finance i.e. discount factors, yield curves, forward rate, zero rates and how to use a rate curve. Additionally we introduce financial products such as IRSs, TSs, OISs and XCSs. In section \ref{pricingIRSprecrisis} we explain the pre-crisis methodology for pricing IRSs in the MXN market and describe formally every formula and the steps to perform a bootstrapping. Section \ref{sec:multicurveframework} is divided in three subsections. At the beginning of the section we summarize the general collateral valuation framework considered in \cite{fujii2010note} and \cite{piterbarg2010funding}. In the second subsection we present the curve construction of IRSs and OISs when the currency of the swap and the collateral currency are the same, detailing explicitly formulas for the calibration of the OIS-discount curve and LIBOR 1m forward curve in the USD market. In the third subsection we discuss the case when the payoffs currency of the swaps are different from the collateral currency, also we explain briefly the difference between pricing IRSs based on EUR and MXN when the collateral is posted in USD. Then we analyze why liquidity in the market and the non-existence of an overnight swap market affects the MXN interest rate curves construction. In section \ref{sec:MXNIRSUnderDifColl} we develop the main formulas and methodology for the valuation and pricing of MXN IRSs based on TIIE 28d with multiple collateral currencies, such as: in USD, MXN, EUR and non collateralized contracts. In section \ref{sec:results} we present the main results of the thesis; first we show the historical differences in swap rates considering the collateral currencies and we also analyze how the size of the spread of XCSs affects the differences of swap rates in each collateral currency. Finally in section \ref{sec:conclusions} we present our conclusions and further research to manage better the curve calibration in distinct currencies.
    
    \newpage
    \section{Literature Review}\label{sec:literaturereview}
    As we saw in the introduction, throughout this work we will entirely focus on the construction of interest rate curves in presence of collateral. Fortunately for the mathematical finance theory, the construction of interest rates curves has been an active area of research in both the industry and academia. A great variety of papers and research documents have been published treating rate curve construction through different currencies, mostly for EUR, USD and JPY (Japanese Yen) currencies. However, as in all mathematical finance research topic, most of the statistical and mathematical models that are applied in the industry are developed in-house and typically are not published for academic purposes. Nonetheless, across this section, we will survey some of the available literature for pricing collateralized interest rates swaps. This literature review has been organized into three main parts. First, the main papers treating curves construction in a \emph{multi-curve} framework with the incorporation of collateral are explored. Secondly we discuss the evolution of the papers published by Masaaki Fujii, Yasufumi Shimada and Akihiko Takahashi\footnote{University of Tokyo, Faculty of Economics} \cite{fujii2010collateral}, \cite{fujii2010note}, \cite{fujii2015choice}, followed by the applications of this framework through multiple currencies such as EUR, JPY and SEK (Swedish Krona). It is important to point out that \cite{fujii2010note} is our main reference and almost all the theory involved in this work is based on it. Finally we examine the basic literature for interpolation and bootstrapping methods.\par \smallskip
    The post-crisis world brought to financial theory a \emph{multi-curve} framework which is nowadays the standard pricing framework. In terminology, the multi-curve characteristics plus the existence of collaterals through CSA agreements, has often been reduced to the term OIS-discounting. However, we will not use this term since the property of using an overnight rate as collateral rate is just one of the characteristics of the multi-curve framework. One of the first papers that explores a multi-curve framework without collateral was \cite{henrard2007irony}. In this publication, the author questioned the way derivatives cash flows were discounted. He realized that, since counterparties have different credit ratings or default risks, each of them must have associated a unique discount curve for the pricing of its derivatives position. Nevertheless, in practice, for a counterparty who traded an OIS and a IRS based on LIBOR, different discounting curves were often used for the valuation of the mark-to-markets. Indeed, the IRS cash flows were valued using a discount curve implied from the LIBOR curve, whereas the OIS used an implied LIBOR $-$ 12.5bps ($\approx$ Fed Funds rate) curve.  In 2010 the author published a second part of this publication (\cite{henrard2010irony}), he wrote that \emph{ironically} his work \cite{henrard2007irony} was published in July 2007, just one month before the financial turmoil which leaded to the widening of LIBOR-OIS spread. Many other papers were published after the crisis to enhance the theory behind multi-curve frameworks. In \cite{bianchetti2008two}, the multi-curve framework is presented for pricing coherently IRS taking into account the forward basis spread taken from the TSs market. In \cite{pallavicini2010interest} the authors explore market evidences through swaptions and contant maturity swaps (CMSs) quoted in the market that suggest the existence of multiple yield curves that avoids arbitrage among products. In \cite{ametrano2013everything}, the authors discuss the multi-curve framework in a detailed way. They present the EUR market case, specifically which products have to be considered for the construction of multiple curves, how to perform bootstrappings and how to compute delta sensitivities.\par\smallskip 
    The model that we will discuss in this thesis is completely based on \cite{fujii2010note}. In this famous paper, Fujii, Shimada and Takahashi explain the method to construct multiple swap curves consistently with all the relevant swaps, say IRSs, TSs, XCSs, with and without a collateral agreement. They also consider the method to construct the term structures of collateralized swaps in the multi-currency setup. They present formulas that could be used for pricing swaps in USD and JPY currencies. In a later published article \cite{fujii2010collateral}, they show the importance of the \emph{choice} of collateral currency. They discuss the implications in market risk management when derivative contracts allow multiple currencies as eligible collateral and a free replacement among them. In their more recent papers \cite{fujii2015choice} and \cite{fujii2015general}, they extend their previous works. In particular they develop a formulation for the funding spread dynamics which is more suitable in the presence of non-zero correlation to the collateral rates. In \cite{gunnarsson2013curve} the author implements the \cite{fujii2010note} pricing framework for the EUR and USD markets. He also presents how to derive the discounting curve for EUR derivatives that are collateralized in USD. In \cite{essay63747}, the author presents how to bootstrap multiple discount curves using market quotes of collateralized interest rate products. He also develops how to compute the convexity adjustment between forward and future rates, while using the Eurodollar futures to bootstrap the three month EUR forward curve. Similarly, in \cite{lidholm2014implications}, the authors applied \cite{fujii2010note} collateralized pricing framework to the Swedish Krona (SEK) swap market. They also analyze the choice of collateral when SEK and USD are eligible.\par \smallskip
    In a \emph{multi-curve} framework the way to compute bootstrappings and interpolations require a robust and capable algorithms to perform the task. In fact, negative overnight rates have change the theory behind interpolation methods since forward negative rates are now allowed. The main references for interpolation algorithms are \cite{hagan2006interpolation} and \cite{hagan2008methods}. However, in \cite{du2011investigation} an extended analysis of a great variety of interpolation methods is presented. Throughout this work. the natural cubic splines algorithm is used.
    
    \newpage
    \section{Back to Basics}\label{sec:backtobasics}
    This work is entirely focused on valuation and risk management of interest rate derivatives. As we will see later, the pricing of an interest rate derivative reduces to the valuation of future cash flows, which are not necessarily known. Thus we require the following basic financial concepts:
    \begin{enumerate}
	\item Discount factors: allow us to calculate the present value of a cash flow received in the future.
    \item Forward rates: allow us to make assumptions of the future level of interest rates.
    \end{enumerate}
    Discount factors are also known as zero coupon bonds \cite{brigo2007interest}, recalling that these are the most simple product in the fixed income world, we defined them as follows:
    \begin{defi}\textnormal{\textbf{(Zero coupon bond)}} A $T$-maturity \emph{zero coupon bond} is a contract that guarantees its holder the payment of one unit of currency at time $T$, with no intermediate payments. The contract value at time $t<T$ is denoted by $P(t,T)$.
    \end{defi}
    To avoid arbitrage we need that $P(t,T)<1$ for all $t<T$ and $P(t,T)=1$ for all $t\geq T$. Note that if $C$ is a cash flow happening at time $T$, then $C \cdot P(t,T)$ gives the value at time $t$ (present value) of the cash flow $C$. Therefore, zero coupon bonds can be treated as discount factors. It is important to point out that the property of $P(t,T)<1$ in some markets, such as Europe and Japan, has been violated due to negative rates. For a more deep and focused discussion on the theme see \cite{hannoun2015ultra}, \cite{arteta2015negative} and \cite{jackson2015international}. Now we are able to define forward zero coupon bond.
    \begin{defi}\textnormal{\textbf{(Forward zero coupon bond)}} A $(T+\alpha)$-maturity \emph{forward zero coupon bond} is a contract observed at $t$ that pays $P(t,T,T+\alpha)$ to the issuer and guarantees its holder the payment of $1$ at time $T+\alpha$, with no intermediate payments.
    \end{defi}
    A forward zero coupon bond is the price at the date the contract is made for buying a zero coupon bond at a later date, but before its maturity. The next result defines the \emph{fair} price of a forward zero coupon bond.
    \begin{teo}
     The price of a forward zero coupon bond $P(t,T,T+\alpha)$ is given by,
        \begin{equation}
            P(t,T,T+\alpha)=\frac{P(t,T+\alpha)}{P(t,T)}.
        \end{equation}
    \end{teo}
    \begin{proof}
        To prove this formula we have to build a trading strategy that replicates the cash flows associated to the definition of forward zero coupon bond. Consider that at time $t$ we buy $1$ unit of a $(T+\alpha)$-maturity zero coupon bond and sell short $\frac{P(t,T+\alpha)}{P(t,T)}$ units of $T$-maturity zero coupon bond. The cost of this strategy is calculated as follows,
        $$-P(t,T+\alpha)+\frac{P(t,T+\alpha)}{P(t,T)}P(t,T)=-P(t,T+\alpha)+P(t,T+\alpha)=0.$$
        The cost of the strategy is equal to zero, thus we do not have cash flows at time $t$. Then at time $T$ the sell short transaction matures and we have to pay a cash flow of
        $$\frac{P(t,T+\alpha)}{P(t,T)}.$$
        Finally, at time $T+\alpha$ the $(T+\alpha)$-maturity zero coupon bond matures and we receive a cash flow of $1$. This bring us at time $t$ a strategy with the same cash flows for a long position in a forward zero coupon bond. Therefore, by no-arbitrage arguments,
        \begin{equation}\label{fwdzerocouponbond}
        P(t,T,T+\alpha)=\frac{P(t,T+\alpha)}{P(t,T)}.
        \end{equation}
    \end{proof}
        We will call these forward zero coupon bonds as forward discount factors, interchangeably. Discount factors can be expressed in terms of interest rates. This zero interest rate (associated to the zero coupon bond) could be simply-compounded or continuously-compounded.
        \begin{defi}\textnormal{\textbf{(Simply-compounded zero interest rate)}} The \emph{simply-compounded zero interest rate} prevailing at time $t$ for the maturity $T$ is denoted by $L(t,T)$ and is the constant rate at which an investment of $P(t,T)$ at time $t$ accrues proportionally to the investment time and yields to a unit at maturity $T$. In formula:
        \begin{equation}\label{DiscSimpleZeroRate}
            L(t,T):=\frac{1}{\tau(t,T)}\bigg( \frac{1}{P(t,T)}-1 \bigg),
        \end{equation}
    where $\tau(t,T)=T-t$ is the time difference expressed in years.
    \end{defi}
    Again, substituting \eqref{DiscSimpleZeroRate} in \eqref{fwdzerocouponbond} yields
    \begin{equation}\label{forwardsimply}
    \begin{aligned}
    L(t,T,T+\alpha)&=\frac{1}{\alpha}\Big[ \frac{1+L(t,T+\alpha)(T+\alpha-t)}{1+L(t,T)(T-t)} - 1\Big] \\
    			&=\frac{1}{\alpha}\Big[ \frac{P(t,T)}{P(t,T+\alpha)} - 1\Big],
    \end{aligned}
    \end{equation}
    where $L(t,T+\alpha)$ and $L(t,T)$ are simply compounded zero rates. We define $L(t,T,T+\alpha)$ as the simply compounded forward rate for the period $[T,T+\alpha]$ seen at time $t$.
    \begin{defi}\textnormal{\textbf{(Continuosly-compounded zero interest rate)}} The \emph{continuously-compounded zero interest rate} prevailing at time $t$ for the maturity $T$ is denoted by $R(t,T)$ and is the constant rate at which an investment of $P(t,T)$ at time $t$ accrues continuously to yield a unit at maturity $T$. In formula:
        \begin{equation}
            R(t,T):=-\frac{\ln P(t,T) }{\tau(t,T)},
        \end{equation}
    where $\tau(t,T)=T-t$ is the time difference expressed in years (according to a day count convention).
    \end{defi}
    It is easy to see that,
    \begin{equation}\label{DiscContinuousZeroRate}
        P(t,T)=e^{-R(t,T)(T-t)}.
    \end{equation}
    Note that if we substitute \eqref{DiscContinuousZeroRate} in \eqref{fwdzerocouponbond} we get
    $$R(t,T,T+\alpha)=\frac{R(t,T+\alpha)(T+\alpha-t) - R(t,T)(T-t)}{\alpha},$$
    where $R(t,T+\alpha)$ and $R(t,T)$ are continuously compounded zero rates. And $R(t,T,T+\alpha)$ is defined as the continuously compounded forward rate for the period $[T,T+\alpha]$ seen at time $t$. Let us now present the definition of each of these rates.
    
    \subsection{The Ibor Rates}\label{section:iborrates}
    The Ibor rates are daily reference rates based on average interest rates at which banks offer to lend unsecured funds to other banks, the name Ibor is the acronym for InterBank Offered Rate. Ibor rates are usually computed as the trimmed average between rates contributed by the participant banks. The lending period could be from one day to one year, the most common tenors are: 1 week, 1 month, 3 months and 6 months. The main usage of this rates, besides the depos (lend/borrow between financial institutions), is in swaps, caps, floors and other interest rate derivatives. Examples of Ibor rates are the following:
    \begin{itemize}
    \item LIBOR (London Interbank Offered Rate), which is determined by London banks and published by the British Banking Association at 11 a.m. GMT on each London business day. The tenors published are: 1m, 3m, 6m and 12m. There are 19 banks involved in setting the LIBOR rate: 3 US Banks and 16 non-US banks.
    \item EURIBOR (Euro Interbank Offered Rate) is determined by Eurozone banks and is published by the European Money Market Institute at 11 a.m. GMT$+2$ on each TARGET\footnote{TARGET (Trans-European Automated Real-Time Gross-Settlement Express Transfer) is an interbank payment system for the real-time processing of cross-border transfers throughout the European Union.} business day. The tenors published are: 1m, 3m and 6m.
     \item TIIE (Tasa de Interés Interbancaria de Equilibrio) is determined by Mexican banks and published by Banco de México at 12 p.m. GMT$-5$ on each Mexico business day. The tenors published are: 28d, 91d and 182d.
    \end{itemize}
    Other main currencies Ibor rates are: British Pound Sterling Ibor (GBP LIBOR), Swiss Franc Ibor (CHF LIBOR), Japanese Yen Ibor (TIBOR), Canadian Dollar Ibor (CDOR) and Hong Kong Dollar (HIBOR). Ibor rates have some conventions such as day count, spot lag and date rolling. The day count convention is the way for counting the days in a year, the most common are: ACT/360 and 30/360, see appendix \ref{App:AppendixA}. The spot lag is the number of days between the fixing date and the value or payment date. Finally, the date rolling convention determines the payment date (forward or backward) when the spot lag falls in a business day. For example, the LIBOR 3m has a spot lag of 2 days, an ACT/360 day count convention and a modified following business day convention; whereas TIIE 28d has a spot lag of 1 day, an ACT/360 day count convention and a following business day convention.\par\smallskip
    The fixings of Ibor rates tend to be constant (or with small changes) compared to IRSs or other interest rates derivatives since Ibor rates represent deposit rates with wider bid-offer spread than derivatives. However, Ibor rates adjust noticeably after central banks meeting; even when monetary policy changes are entirely expected, since banks participants generally want to see rate changes before they adjust rates.\par\smallskip
    Let us present some mathematical notations, specifically for Ibor rates, that will be used throughout this work.
    \begin{defi}\textnormal{\textbf{(Ibor rate)}}
    An \emph{Ibor rate} with fixing at time $\widetilde{S}$, accrual period $[S,T]$ and payment at time $\widetilde{T}$ is denoted by
    \begin{equation}
    \textnormal{\textbf{Ibor}}(S,T):=\textnormal{\textbf{Ibor}}_{\widetilde{S},\widetilde{T}}(S,T),
    \end{equation}
    where $\widetilde{S} \leq S < T \leq \widetilde{T}$.
    \end{defi}
    \begin{defi}\textnormal{\textbf{(Forward Ibor rate)}}
    A \emph{forward Ibor rate} at time $t$, with fixing at time $\widetilde{S}$, accrual period $[S,T]$ and payment $\widetilde{T}$ is denoted by
    \begin{equation}
    \mathbb{E}_t^{\mathbb{Q}_{\widetilde{T}}}(\textnormal{\textbf{Ibor}}_{\widetilde{S},\widetilde{T}}(S,T)):= \mathbb{E}_t^{\widetilde{T}} (\textnormal{\textbf{Ibor}}(S,T)),
    \end{equation}
    where $t < \widetilde{S} \leq S < T \leq \widetilde{T}$ and $\mathbb{Q}_{\widetilde{T}}$ is the forward measure associated with the \emph{numéraire} $P(t,\widetilde{T})$ ($\widetilde{T}$-maturity zero coupon bond).
    \end{defi}
    Note that the forward Ibor rate $\mathbb{E}_t^{\widetilde{T}} (\textnormal{\textbf{Ibor}}(S,T))$ is a zero rate for the period $[S,T]$. Hence, using the fact that Ibor rates are simply-compounded and the equation \eqref{forwardsimply} we get the following definition.
    \begin{assum}
    	Under the assumption that the probability of default of the banks on which Ibor rates are based can be negleted, we can express the forward Ibor rate in terms of zero coupon bonds, i.e.
        \begin{equation}\label{assumptionForwards}
        	\mathbb{E}_t^{\widetilde{T}} (\textnormal{\textbf{Ibor}}(S,T)):=\frac{1}{\tau(S,T)}\bigg( \frac{P(t,S)}{P(t,T)} - 1 \bigg), \hspace{5mm} t \leq \widetilde{S} \leq S < T \leq \widetilde{T},
        \end{equation}
        where $\tau(S,T)$ is the year fraction between $S$ and $T$.
    \end{assum}
    
    \subsection{Interest Rate Derivatives}\label{sec:IRproducts}
    In this subsection we will define some interest rate products mentioning the basic characteristics of them, also we provide term sheets samples of the plain vanilla derivatives. The products explained will be used throughout this work, therefore it is important to fully understand them. It is important to highlight that most of these products are traded over-the-counter (OTC). OTC trading is done directly between two parties, without any supervision of an exchange or a clearing house. On the contrary of OTC trading, exchanges have the following benefits: facilitate liquidity, mitigate credit risk concerning the default of one party and provide transparency. Therefore in an OTC market, contracts are less liquid compared to exchange trading, however they can be tailored for the clients. Moreover, prices of OTC contracts are not necessarily published in the market. The products that we present are types of interest rate swaps. We do not introduce products such as: FRAs, caps, floors, swaptions, digital caps, digital floors and other interest rate derivatives with volatility or an inflation index.
    
    \subsubsection{Interest Rate Swaps}
    An interest rate swap (IRS) is a derivative instrument in which both parties agree to make interest payments at fixed dates in the future. Normally, one party pays the other a fixed interest rate, while the other party makes interest payments in line with the future interest rate trend. Suppose counterparties A and B enter into an IRS contract in which A agrees to pay a fixed rate (receive a floating rate), whereas B agrees to pay floating rate (receive fixed rate). The attribute \emph{payer/receiver}, by convention, refers to the fixed leg of the swap, hence A enter into a payer swap while B enters into a receiver swap. The floating leg is typically based on Ibor rates.\par \medskip
    \noindent Let $\textnormal{PV}_{\textnormal{Payer}}(t)$ be the present value of a payer interest rate swap. Then
    \begin{equation}\label{payerIRS}
        \textnormal{PV}_{\textnormal{Payer}}(t)=\textnormal{\textbf{FloatLeg}}(t)-\textnormal{\textbf{FixedLeg}}(t),
    \end{equation}
    with
    \begin{equation}\label{legsIRS}
    	\begin{aligned}
        \textnormal{\textbf{FloatLeg}}(t) &=\sum_{i=1}^{M} \alpha(t_{i-1},t_{i}) \mathbb{E}_t^{\widetilde{t}_i}(\textnormal{\textbf{Ibor}}(t_{i-1},t_{i})) P(t,\widetilde{t}_i)\\
        \textnormal{\textbf{FixedLeg}}(t) &= k\sum_{j=1}^{N} \beta(s_{j-1},s_{j}) P(t,\widetilde{s}_j),
        \end{aligned}
    \end{equation}
    where:\par \smallskip
    \noindent \makebox[2.9cm][l]{$k$:} fixed rate of the interest rate swap \par
    \noindent \makebox[2.9cm][l]{$M,N$:} number of floating coupons (resp. fixed coupons) \par 
    \noindent \makebox[2.9cm][l]{$t_i,s_j$:} coupon periods of floating leg (resp. fixed leg) \par 
    \noindent \makebox[2.9cm][l]{$\widetilde{t}_i,\widetilde{s}_j$:} payment time of the $i$th coupon (resp. $j$th coupon) \par 
    \noindent \makebox[2.9cm][l]{$\alpha(t_{i-1},t_i))$:} accrual factor of the $i$th coupon \par
    \noindent \makebox[2.9cm][l]{$\beta(s_{j-1},s_j))$:} accrual factor of the $j$th coupon \par
    \noindent \makebox[2.9cm][l]{$\mathbb{E}_t^{\widetilde{t}_i}(\textnormal{\textbf{Ibor}}(t_{i-1},t_{i}))$:} the forward Ibor rate of the $i$th coupon. \par \medskip
   Now let us present a term sheet example of a plain vanilla IRS based on the MXN reference rate TIIE 28d:
    \FloatBarrier
        \begin{table}[H]
        \footnotesize
            \begin{center}
                \begin{tabular}{|c|c|}
                 \hline
                 \multicolumn{2}{|c|}{\textbf{MXN IRS Contract}} \\
                 \hline
                Trade Date & $t$ \\
                Spot Lag & 1 days\\
                Start Date & $t+1$ \\
                Tenor & 1820d \\
                Payer of Fixed Rate & Bank A \\
                Receiver of Fixed Rate & Bank B \\
                Fixed Rate & 4.87\% \\
                Index Rate & TIIE 28d\\
                Notional Value & MXN 10Mio\\
                Payment Frequency & 28 days\\
                Day Count Convention & ACT/360\\
                Business Days Calendar & Mexico City\\
                Date Roll Convention & Following\\
                \hline
                \end{tabular}
            \end{center}
          \caption{Term sheet sample of a plain vanilla IRS based on TIIE 28d with maturity of 5y.}
          \label{IRS-Contract}
        \end{table}
    \FloatBarrier
    Suppose that the trade date is 29-Jan-2015 (Thursday) thus the Start Date = Trade Date + Spot Lag = 29-Jan-2015 + 1 day (using following date rolling convention and Mexico City calendar) = 30-Jan-2015. Then, the end date is calculated as End Date = Start Date + 1820 days, using following date rolling convention and Mexico City calendar, we have that End Date = 24-Jan-2020. Finally we have that $1820/28=65$, which means that this IRS has 65 coupons, hence the present value of the payer swap is given by,
    \begin{equation}
        \textnormal{PV}_{\textnormal{Payer}}(t)=\textnormal{\textbf{FloatLeg}}(t)-\textnormal{\textbf{FixedLeg}}(t),
    \end{equation}
    with
    \begin{align}
        \textnormal{\textbf{FloatLeg}}(t) &=\num[group-separator={,}]{10000000} \sum_{i=1}^{65} \tau(t_{i-1},t_{i}) \mathbb{E}_t(L(t_{i-1},t_{i})) P(t,t_i)\label{exampleIRS1}\\
        \textnormal{\textbf{FixedLeg}}(t) &=\num[group-separator={,}]{10000000} \cdot 4.87\% \cdot \sum_{i=1}^{65} \tau(t_{i-1},t_{i}) P(t,s_j),\label{exampleIRS2}
    \end{align}
    where $$\tau(t_{i-1},t_{i})=\frac{\textnormal{Actual days between } t_{i-1} \textnormal{ and } t_i }{360}.$$
    Note that for all $i$ the accrual factor $\tau(t_{i-1},i_{j})$ is totally determined since the swap calendar does not depend on the market quotes. Indeed, swap calendars are built using their own characteristics i.e. business day calendars, rolling date conventions and payment frequencies.\par\smallskip
    From equations \eqref{exampleIRS1} and \eqref{exampleIRS2}, it is easy to see that the present value of the IRS we ``just'' need the discount factors and the forward rates. In a single curve framework the forward curve can be obtained by the discount curve and vice versa. However, in a multi curve framework the forward curve obtained implicitly from the discount curve typically is different for the forward rate. It important to point out that in the multi curve framwork the forward rates are preferably called as the index (or reference) rate, whereas the discount factors are not renamed but they have to be necessarily linked with the rates that will be used for managing the collaterals---previously agreed in a CSA agreement.
    
    \subsubsection{Tenor Swaps}\label{sec:tenorswaps}
    A tenor swap (TS) is a contract where the two parties exchange interest rates amounts based on floating (Ibor) reference rates of the same currency but with different tenors. Recall that Ibor rates are meant to mirror unsecured deposit rates, therefore a credit premium for long term lending versus shorter terms has to be payed. This premium is directly included on the spread or basis that is added on one of the TS legs. The general convention is to add the tenor basis spread to the leg with the shorter tenor \cite{fujii2011market}, while the payment frequency is determined by the longer tenor leg. For example, in the case of the USD market, one party agrees to pay LIBOR 3m quarterly and receive LIBOR 1m plus the tenor spread; in the latter leg we have to accumulate the monthly payments with compound interest and settle quarterly to match the 3m tenor leg.\par
    \begin{rem}
    In \textnormal{EUR} currency market, the tenor swaps are conventionally quoted as two swaps. Hence a quote for paying \textnormal{EURIBOR 3m} $+$ $12$bps versus receiving \textnormal{EURIBOR 6m} has the following meaning. In the first swap you pay \textnormal{EURIBOR 3m} versus receive a fixed rate (in an annual base). In the second swap you pay the same fixed rate plus the spread of $12$bps (in an annual base) and receive \textnormal{EURIBOR 6m} \textnormal{\cite{henrard2014book}}. Note that in this convention the tenor spread is paid on an annual basis, whereas in \textnormal{USD} market the tenor spread is paid quarterly.
    \end{rem}
    The TSs market could be consider as market indications of lending period preferences due to credit and liquidity risks. Another characteristic of TS is that the basis (or spread) typically is downward sloping i.e. the greater the maturity of the swap is, the smaller the spread is.\par \smallskip
    Let $\textnormal{PV}_{\textnormal{Payer}}(t)$ be the present value of a payer tenor swap. Then
    \begin{equation}\label{generalpayerTS}
        \textnormal{PV}_{\textnormal{Payer}}(t)=\textnormal{\textbf{Leg}}_{(A)}(t)-\textnormal{\textbf{Leg}}_{(B)}(t),
    \end{equation}
    with
    \begin{align}\label{generalTSlegs}
        \textnormal{\textbf{Leg}}_{(A)}(t) &=\sum_{i=1}^{M} \alpha(t_{i-1},t_{i}) \mathbb{E}_t^{\widetilde{t}_i}(\textnormal{\textbf{Ibor}}_{(A)}(t_{i-1},t_{i})) P(t,\widetilde{t}_i)\\
        \textnormal{\textbf{Leg}}_{(B)}(t) &= \sum_{j=1}^{N} \beta(s_{j-1},s_{j})\bigg[ \frac{1}{\sum_{k=1}^{N_j}\rho(u_{k-1},u_k)} \bigg( \prod_{k=1}^{N_j} \Big(1+\rho(u_{k-1},u_k)\mathbb{E}_t^{\widetilde{u}_k}(\textnormal{\textbf{Ibor}}(u_{k-1},u_{k}))\Big) \bigg)\notag\\
 & \hspace{35mm}+ B\bigg] P(t,\widetilde{s}_j)
    \end{align}
    where:\par \smallskip
    \noindent \makebox[3.75cm][l]{$B$:} fixed tenor spread \par
    \noindent \makebox[3.75cm][l]{$M,N$:} number of coupons of leg $A$ (resp. leg $B$) \par
    \noindent \makebox[3.75cm][l]{$N_j$:} number of fixings for the $j$th coupon of leg $B$ \par
    \noindent \makebox[3.75cm][l]{$t_i,s_j$:} coupon periods of leg $A$ (resp. leg $B$) \par
    \noindent \makebox[3.75cm][l]{$u_k$:} fixing periods of leg $B$ \par 
    \noindent \makebox[3.75cm][l]{$\widetilde{t}_i,\widetilde{s}_j$:} payment time of the $i$th coupon of leg $A$ (resp. $j$th coupon of leg $B$) \par 
    \noindent \makebox[3.75cm][l]{$\alpha(t_{i-1},t_i))$:} accrual factor of the $i$th coupon of leg $A$\par
    \noindent \makebox[3.75cm][l]{$\beta(s_{j-1},s_j))$:} accrual factor of the $j$th coupon of leg $B$\par 
    \noindent \makebox[3.75cm][l]{$\rho(u_{k-1},u_k))$:} accrual factor of the $k$th fixing of leg $B$ used to compute the compounded rate\par
    \noindent \makebox[3.75cm][l]{$\mathbb{E}_t^{\widetilde{t}_i}(\textnormal{\textbf{Ibor}}_{(A)}(t_{i-1},t_{i}))$:} the forward Ibor rates of leg $A$. \par 
    \noindent \makebox[3.75cm][l]{$\mathbb{E}_t^{\widetilde{u}_k}(\textnormal{\textbf{Ibor}}_{(B)}(u_{k-1},u_{k}))$:} the forward Ibor rates of leg $B$. \par \medskip
    \FloatBarrier
        \begin{table}[H]
        \footnotesize
            \begin{center}
                \begin{tabular}{|c|c|}
                 \hline
                 \multicolumn{2}{|c|}{\textbf{LIBOR 3m vs 1m 5y Contract}} \\
                 \hline
                Trade Date & $t$ \\
                Spot Lag & 2 days\\
                Start Date & $t+2$ \\
                Tenor & 5y \\
                Spread (Leg) & +0.10\% (LIBOR 1m) \\
                Index Rates & LIBOR 1m \& LIBOR 3m\\
                Notional Value & USD 10Mio\\
                Payment Frequency & 3 Months\\
                Day Count Convention & ACT/360\\
                Business Days Calendar & New York \& London\\
                Date Roll Convention & Modified Following\\
                \hline
                \end{tabular}
            \end{center}
          \caption{Example of a plain vanilla 5y tenor swap contract.}
          \label{Tenor-Contract}
        \end{table}
    \FloatBarrier
    
    \subsubsection{Overnight Index Swaps and Federal Funds Swaps} \label{OISsFFSs}
    As we will see throughout this work, overnight index swaps (OISs) play an important role in the construction of discount factors in a collateralized world. This kind of swap has two legs: floating rate leg and fixed rate leg, with coupon payments between the spot date and the maturity date. The main difference with a plain vanilla IRS is that the floating leg is linked to an overnight index instead of an Ibor index rate. \par 
    Overnight rates are published every business day as Ibor rates, however they are effective for only one day, this is why they are called \emph{overnight rates}. Recall that in a plain vanilla IRS, the calendar of the floating leg is scheduled with periods of the same length of the tenor of the Ibor index rate\footnote{In case the maturity of the swap is not a multiple of the Ibor rate tenor, then it is defined a stub period coupon (short or long) scheduled at the beginning (up front) or at the end (in arrears) of the swap. For example a swap which matures in 20 months based on LIBOR 3m could have either a 2 month coupon (short stub period) or a 4 month coupon (long stub period). Additionally, this stub period could be either scheduled at the beginning of the swap followed by 3-months coupons or at the end of the swap preceded by 3-months coupons.}. For instance, an IRS based on LIBOR 3m has payments every 3 months, hence it is natural to think that the floating leg of an OIS would have daily payments. Although it is not practical to have daily payments in swaps or any financial instrument, so floating leg payments in OISs are scheduled yearly or quarterly and the amount paid is computed by compounding or averaging the overnight rates. \par 
    Throughout this work, we only consider the USD OISs market. It is important to say that OISs are relatively liquid up to 30 years (see \cite{CMEProductScope}). However, for maturities longer than 10 years, prices are in the market quoted as Federal Funds Swaps (FFSs) which are a type of tenor swaps due to they exchange an Ibor payment for an overnight index based payment. In a Federal Funds Swap (FFS), the overnight indexed leg is computed different from OIS. Indeed, the payment is computed as the arithmetic mean of the overnight rates, while in OIS the payment is computed by compounding daily the overnight rates. In this work we used OISs for curve calibration and ignore quotes of FFSs. For more insight, regarding valuation and properties of FFSs, see \cite{takada2011valuation}. \par
    Likewise a payer IRS, a payer OIS is a contract in which we pay fixed payments and receive floating payments based on overnight rates. The spot lag, i.e. the difference in days between the trade date and the start date, is commonly two business days. \par \medskip
        \noindent Let $\textnormal{PV}_{\textnormal{Payer}}(t)$ be the present value of a payer overnight index swap. Then
    \begin{equation}
        \textnormal{PV}_{\textnormal{Payer}}(t)=\textnormal{\textbf{FloatLeg}}(t)-\textnormal{\textbf{FixedLeg}}(t),
    \end{equation}
    with
    \begin{align}
        \textnormal{\textbf{FloatLeg}}(t) &=\sum_{i=1}^{M} \alpha(t_{i-1},t_{i}) \mathbb{E}_t^{\widetilde{t}_i}(\textnormal{\textbf{DailyCompOI}}(t_{i-1},t_{i})) P(t,\widetilde{t}_i)\\
        \textnormal{\textbf{FixedLeg}}(t) &= k\sum_{j=1}^{N} \beta(s_{j-1},s_{j}) P(t,\widetilde{s}_j),
    \end{align}
    where:\par \smallskip
    \noindent \makebox[4.65cm][l]{$t_0-t$:} spot lag (difference between trade date and start daet)\par
    \noindent \makebox[4.65cm][l]{$k$:} fixed rate of the overnight index swap\par
    \noindent \makebox[4.65cm][l]{$M,N$:} number of floating coupons (resp. fixed coupons) \par 
    \noindent \makebox[4.65cm][l]{$t_i,s_j$:} coupon periods of floating leg (resp. fixed leg) \par 
    \noindent \makebox[4.65cm][l]{$\widetilde{t}_i,\widetilde{s}_j$:} payment time of the $i$th coupon (resp. $j$th coupon) \par 
    \noindent \makebox[4.55cm][l]{$\alpha(t_{i-1},t_i))$:} accrual factor of the $i$th coupon \par
    \noindent \makebox[4.65cm][l]{$\beta(s_{j-1},s_j))$:} accrual factor of the $j$th coupon \par
    \noindent \makebox[4.65cm][l]{$\mathbb{E}_t^{\widetilde{t}_i}(\textnormal{\textbf{DailyCompOI}}(t_{i-1},t_{i}))$:} the daily compounded overnight index rate. \par \medskip
    \noindent The daily compounded overnight index rate is defined as
    \begin{equation}
		\mathbb{E}_t^{\widetilde{t}_i}(\textnormal{\textbf{DailyCompOI}}(t_{i-1},t_{i})):=\frac{1}{\tau(t_{i-1},t_{i})} \Big( \prod_{h=0}^{K_i-1} \big( 1+\alpha(t_{i,h},t_{i,h+1}) \mathbb{E}_t^{\widetilde{t}_{i,h+1}} (\textnormal{\textbf{OI}} (t_{i,h},t_{i,h+1})) \big) -1 \Big)
    \end{equation}
    where:\par \smallskip
    \noindent \makebox[3.75cm][l]{$K_i$:} number of business days between $t_{i-1}$ and $t_i$ \par 
    \noindent \makebox[3.75cm][l]{$\{ t_{i,h} \}_{h=0}^{K_i}$:} is the collection of all business days in the accrual period $[t_{i-1},t_{i}]$\par
    \noindent \makebox[3.75cm][l]{$\alpha(t_{i,h},t_{i,h+1})$:} accrual factor of the $h$th business day \par 
    \noindent \makebox[3.75cm][l]{$\mathbb{E}_t^{\widetilde{t}_{i,h+1}} (\textnormal{\textbf{OI}} (t_{i,h},t_{i,h+1}))$:} overnight index rate effectively for the period $[t_{i,h},t_{i,h+1}]$ and paid at $\widetilde{t}_{i,h+1}$ \par \medskip
    \noindent Examples of OIS contracts are given in tables \ref{OIS-FF-Contract} and \ref{OIS-EONIA-Contract}.
    
        \begin{table}[H]
        \footnotesize
            \begin{center}
                \begin{tabular}{|c|c|}
                 \hline
                 \multicolumn{2}{|c|}{\textbf{OIS (Fed Fund) 5y Contract}} \\
                 \hline
                	Trade Date & $t=$ 02-Jun-2015 \\
                	Spot Lag & 2 days\\
                	Start Date & $t+2=$ 04-Jun-2015 \\
                	Tenor & \hspace{2mm} 5y $=$  04-Jun-2020\\
                	Fixed Rate & 1.5078\% \\
                	Overnight Index Rate & Fed Funds\\
                	Notional Value & USD 10Mio\\
                	Payment Frequency & Annual (both legs)\\
                	Day Count Convention & ACT/360 (both legs)\\
                	Business Days Calendar & New York\\
                	Date Roll Convention & Modified Following\\
                \hline
                \end{tabular}
            \end{center}
          \caption{Example of a plain vanilla OIS 5y contract linked to Fed Funds Overnight Rate.}
          \label{OIS-FF-Contract}
        \end{table}
    
        \begin{table}[H]
        \footnotesize
            \begin{center}
                \begin{tabular}{|c|c|}
                 \hline
                 \multicolumn{2}{|c|}{\textbf{OIS (Eonia) 2y Contract}} \\
                 \hline
                Trade Date & $t=$ 02-Jun-2015 \\
                Spot Lag & 2 days\\
                Start Date & $t+2=$ 02-Jun-2015 \\
                Tenor & \hspace{2mm} 2y $=$  05-Jun-2017\\
                Fixed Rate & -0.1020\% \\
                Overnight Index Rate & Eonia\\
                Notional Value & EUR 10Mio\\
                Payment Frequency & Annual (both legs)\\
                Day Count Convention & ACT/360 (both legs)\\
                Business Days Calendar & TARGET\\
                Date Roll Convention & Modified Following\\
                \hline
                \end{tabular}
            \end{center}
          \caption{Example of a plain vanilla OIS 2y contract linked to Eonia Overnight Rate.}
          \label{OIS-EONIA-Contract}
        \end{table}
    
    \subsubsection{Cross-Currency Swaps}\label{XCSs}
    A cross-currency swap (XCS) is a contract between two parties to exchange interest rate payments in two different currencies. Plain vanilla XCS exchanges floating payments linked to Ibor rates, in which one of the legs adds a fair basis spread that is traded in the market. For example an EURUSD XCS exchanges LIBOR 3m for EURIBOR 3m plus an additional basis spread. In contrast to the three previous interest rate contracts, the notional amounts switch hands at the initiation of the swap and then switch back at the maturity of the contract. There exist two types of plain vanilla XCS:
    \begin{itemize}
    	\item cnXCS (Constant Notional Cross-Currency Swap): In a cnXCS is a XCS where the notionals remain constant along the maturity of the swap.
        \item mtmXCS (Mark-to-Market Cross-Currency Swap): In a mtmXCS the notional of one leg is adjusted at each payment date by the current FX rate, and the interest to be paid is computed using the adjusted notional. For pricing a mtmXCS we must think a mtmXCS as a collection of cnXCS. The importance of this swap relies on reducing the credit exposure of both parties arose by fluctuations on the FX rate.
    \end{itemize}
    \noindent Throughout this work we focus only in cnXCSs, despite mtmXCSs are becoming more popular and liquid especially for G10 currencies. Let $\textnormal{PV}_{\textnormal{Payer}}^{A}(t)$ be the present value of a payer cnXCS swap in terms of currency $A$. Then
    \begin{equation}\label{xcs:eq1}
        \textnormal{PV}_{\textnormal{Payer}}^{A}(t)=\textnormal{\textbf{Leg}}_{A}(t)-f^{B\rightarrow A}(t) \textnormal{\textbf{Leg}}_{B}(t),
    \end{equation}
    with
    \begin{align}
        \textnormal{\textbf{Leg}}_{A}(t) &=\textnormal{N}_A \sum_{i=1}^{M} \alpha(t_{i-1},t_{i}) \mathbb{E}_t^{\widetilde{t}_i}(\textnormal{\textbf{Ibor}}_{A}(t_{i-1},t_{i})) P(t,\widetilde{t}_i) \label{xcs:eq2} \\
        \textnormal{\textbf{Leg}}_{B}(t) &=\textnormal{N}_B \sum_{j=1}^{N} \beta(s_{j-1},s_{j}) (\mathbb{E}_t^{\widetilde{s}_j}(\textnormal{\textbf{Ibor}}_{B}(s_{j-1},s_{j})) + S ) P(t,\widetilde{s}_j),\label{xcs:eq3}
    \end{align}
    where:\par \smallskip
    \noindent \makebox[3.75cm][l]{$\textnormal{N}_A$:} Notional of leg $A$ \par
    \noindent \makebox[3.75cm][l]{$\textnormal{N}_B$:} Notional of leg $B$ \par
    \noindent \makebox[3.75cm][l]{$f^{B\rightarrow A}(t)$:} spot FX exchange rate \par
    \noindent \makebox[3.75cm][l]{$S$:} basis spread \par
    \noindent \makebox[3.75cm][l]{$M,N$:} number of coupons  of leg $A$ (resp. leg $B$) \par 
    \noindent \makebox[3.75cm][l]{$t_i,s_j$:} coupon periods of leg $A$ (resp. leg $B$) \par 
    \noindent \makebox[3.75cm][l]{$\widetilde{t}_i,\widetilde{s}_j$:} payment time of the $i$th coupon of leg $A$ (resp. $j$th coupon of leg $B$) \par 
    \noindent \makebox[3.75cm][l]{$\alpha(t_{i-1},t_i))$:} accrual factor of the $i$th coupon of leg $A$\par
    \noindent \makebox[3.75cm][l]{$\beta(s_{j-1},s_j))$:} accrual factor of the $j$th coupon of leg $B$\par
    \noindent \makebox[3.75cm][l]{$\mathbb{E}_t^{\widetilde{t}_i}(\textnormal{\textbf{Ibor}}_{(A)}(t_{i-1},t_{i}))$:} the forward Ibor rate of leg $A$. \par 
    \noindent \makebox[3.75cm][l]{$\mathbb{E}_t^{\widetilde{s}_j}(\textnormal{\textbf{Ibor}}_{(B)}(s_{j-1},s_{j}))$:} the forward Ibor rate of leg $B$. \par \medskip
    
    \noindent For XCS quotes displayed in trading screens we have that $\textnormal{N}_A=\textnormal{N}_B f^{B\rightarrow A}(t) $.\medskip
    
    In the case of mtmXCS with constant notional in the leg $A$ and updates of notional in leg $B$, the present value of leg $A$ remains equal as in cnXCS, nevertheless leg $B$ present value has to capture the FX rate dynamics along the maturity of the swap. Thus the present value of a payer mtmXCS is given by
    \begin{equation}
        \textnormal{PV}_{\textnormal{Payer}}^{A}(t)=\textnormal{\textbf{Leg}}_{A}(t)-f^{B\rightarrow A}(t) \textnormal{\textbf{Leg}}_{B}(t),
    \end{equation}
    with
    \begin{align}
        \textnormal{\textbf{Leg}}_{A}(t) &= \sum_{i=1}^{M} \textnormal{N}_A \cdot \alpha(t_{i-1},t_{i}) \mathbb{E}_t^{\widetilde{t}_i}(\textnormal{\textbf{Ibor}}_{A}(t_{i-1},t_{i})) P(t,\widetilde{t}_i)\\
        \textnormal{\textbf{Leg}}_{B}(t) &=\sum_{j=1}^{N} \textnormal{N}_A \cdot f^{A\rightarrow B}(\widetilde{s}_j) \cdot \beta(s_{j-1},s_{j}) (\mathbb{E}_t^{\widetilde{s}_j}(\textnormal{\textbf{Ibor}}_{B}(s_{j-1},s_{j})) + S ) P(t,\widetilde{s}_j),
    \end{align}
    where:\par \smallskip
    \noindent \makebox[3.75cm][l]{$f^{A\rightarrow B}(\widetilde{s}_j)$:} FX forward rate with delivery at time $\widetilde{s}_j$ \par \medskip
    \noindent Examples of mtmXCS and cnXCS are given in tables \ref{mtmXCS-Contract} and \ref{cnXCS-Contract}, respectively.
    
    \FloatBarrier
        \begin{table}[H]
        \footnotesize
            \begin{center}
                \begin{tabular}{|c|c|}
                 \hline
                 \multicolumn{2}{|c|}{\textbf{Basis EURUSD 5y Contract}}\\
                 \hline
                Trade Date & $t$ \\
                Spot Lag & 2 days\\
                Start Date & $t+2$ \\
                Tenor & 5y \\
                Basis Spread (Leg) & 0.65\% (EUR)\\
                Index Rates & EURIBOR 3m \& LIBOR 3m\\
                Notional Value & EUR 10Mio\\
                Payment Frequency & Quarterly\\
                FX Reset Frequency & Quarterly\\
                Day Count Convention & ACT/360 (both legs) \\
                Business Days Calendar & New York, TARGET \& London\\
                Date Roll Convention & Modified Following\\
                Forex Rate Reset & Yes\\
                \hline
                \end{tabular}
            \end{center}
          \caption{Example of a plain vanilla mtmXCS EURUSD 5y Contract.}
          \label{mtmXCS-Contract}
        \end{table}
    \FloatBarrier
    
    \FloatBarrier
        \begin{table}[H]
        \footnotesize
            \begin{center}
                \begin{tabular}{|c|c|}
                 \hline
                 \multicolumn{2}{|c|}{\textbf{Basis USDMXN 10y Contract}} \\
                 \hline
                Trade Date & $t$ \\
                Spot Lag & 2 days\\
                Start Date & $t+2$ \\
                Tenor & 1820d ($\approx$5y) \\
                Basis Spread (Leg) & 0.65\% (USD)\\
                Index Rates & LIBOR 1m \& TIIE 28d\\
                Notional Value & USD 10Mio\\
                Payment Frequency & 28 days\\
                Day Count Convention & ACT/360\\
                Business Days Calendar & New York, London \& Mexico City\\
                Date Roll Convention & Following\\
                Forex Rate Reset & No\\
                \hline
                \end{tabular}
            \end{center}
          \caption{Example of a plain vanilla cnXCS 10y Contract.}
          \label{cnXCS-Contract}
        \end{table}
    \FloatBarrier
    
    \newpage
    \section{Pricing IRS in Single-curve Framework}\label{pricingIRSprecrisis}
    In this section we present the formulas for pricing an interest rate swap in a \emph{single-curve} framework, specifically in the MXN currency market. This framework is important from a historical point of view since they explain and serve as a basis of the \emph{multi-curve} framework. The general idea of the \emph{single-curve} framework is that all interest rate derivative, in the same currency, depend on only one curve, which is supposed to be the discount curve and Ibor index curve.
    \subsection{Case of MXN}
    This section presents the construction the MXN yield curve under the assumption that the plain vanilla swaps traded on the market do not have a collateral agreement. Therefore, we are assuming that the TIIE 28d rate (mexican interbank offered rate) is \emph{risk-free} and illiquidity or credit issues of participant banks are neglected, i.e., we do not need to incorporate a collateral rate. This framework is known as \emph{single-curve} framework, since one unique curve is used for extract discount factors and forward rates. Before the financial crisis in 2007, this framework was widely used and was considered the \emph{correct} method for pricing and valuation of interest rate derivatives. It is important to remind that this model assumes that a financial institution borrows and lends money with the same \emph{risk-free} rate, in this case TIIE 28d rate.\par\medskip
    \noindent From equations \eqref{payerIRS} and \eqref{legsIRS}, we have that the present value of a payer IRS of TIIE28d is given by  
    \begin{equation}\label{payerTIIE28D}
        \textnormal{PV}(t)= \sum_{i=1}^{M} \alpha(t_{i-1},t_{i}) \mathbb{E}_t^{\widetilde{t}_i}(\textnormal{\textbf{TIIE28D}}(t_{i-1},t_{i})) P(t,\widetilde{t}_i) - k\sum_{j=1}^{N} \beta(s_{j-1},s_{j}) P(t,\widetilde{s}_j),
    \end{equation}
    Since in a plain vanilla IRS linked to TIIE 28d the number of coupons in both legs are the same, the end date of each coupon period equals the payment date i.e. $\widetilde{t}_i=t_i$ for all $i$, and day count convention is the same for each leg, then equation \eqref{payerTIIE28D} reduces to
     \begin{equation}\label{reducedpayerTIIE28D}
        \textnormal{PV}(t)= \sum_{i=1}^{N} \alpha(t_{i-1},t_{i}) \mathbb{E}_t^{t_i}(\textnormal{\textbf{TIIE28D}}(t_{i-1},t_{i})) P(t,t_i) - k\sum_{i=1}^{N} \alpha(t_{i-1},t_{i}) P(t,t_i),
    \end{equation}
    where:\par \smallskip
    \noindent \makebox[3.8cm][l]{$k$:} fixed rate of the plain vanilla interest rate swap \par
    \noindent \makebox[3.8cm][l]{$N$:} number of coupons \par 
    \noindent \makebox[3.8cm][l]{$t_i$:} coupon periods (start date, end date and payment date)\par 
    \noindent \makebox[3.8cm][l]{$\alpha(t_{i-1},t_i))$:} accrual factor of the $i$th coupon \par
    \noindent \makebox[3.8cm][l]{$\mathbb{E}_t^{t_i}(\textnormal{\textbf{TIIE28D}}(t_{i-1},t_{i}))$:} the forward TIIE 28d rate of the $i$th coupon. \par \medskip
    \noindent Now, in a single-curve framework we have that
    \begin{equation}\label{TIIEfwdzerocoupon}
    	\mathbb{E}_t^{t_i}(\textnormal{\textbf{TIIE28D}}(t_{i-1},t_{i})) = 
        \bigg( \frac{1}{\tau(t_{i-1},t_i)} \bigg( \frac{P(t,t_{i-1})}{P(t,t_{i})} - 1 \bigg) \bigg).
    \end{equation}
    Note that in equation \eqref{TIIEfwdzerocoupon} we use $\tau$ as day count factor instead of $\alpha$ that is used in the swap. It is important to point out that, in general, $\alpha (t_{i-1},t_{i})) \neq \tau (t_{i-1},t_{i}))$, this is because $\alpha(t_{i-1},t_{i}))$ is the accrual factor (year fraction) used for the payment of the $i$th coupon, whereas $\tau(t_{i-1},t_{i}))$ is the day count (year fraction) used for the interpolation and construction of zero curve. In the swaps market, the accrual factors of payments usually are based on an ACT/360  or 30/360 convention, while interpolation methods for the construction of zero curves typically used an ACT/ACT or ACT/365 convention.\par \medskip
    \noindent Assume that the fixed rate $k$ is a mid market quote, hence the present value of the payer swap equals zero. If we substitute \eqref{TIIEfwdzerocoupon} in \eqref{reducedpayerTIIE28D} and then we equalize it to zero we obtain
    \begin{equation}
    	\sum_{i=1}^{N} \alpha_i \bigg( \frac{1}{\tau_i} \bigg( \frac{P(t,t_{i-1})}{P(t,t_{i})} - 1 \bigg) \bigg) P(t,t_i) - k\sum_{i=1}^{N} \alpha_i P(t,t_i)=0.
    \end{equation}
    Solving for $P(t,t_N)$ we get
    \begin{equation}
		P(t,t_N)=\frac{\sum_{i=1}^N \Big[ \frac{\alpha_i}{\tau_i} \big( P(t,t_{i-1}) - P(t,t_i) \big) - k \alpha_i P(t,t_i) \Big] + \frac{\alpha_N}{\tau_N} P(t,t_{N-1}) }{ \frac{\alpha_N}{\tau_N} + k \alpha_N}.
	\end{equation}
    If we assume that $\alpha_i=\tau_i$ for all $i$, then
    \begin{equation}\label{discBootstrapping}
    P(t,t_N)=\frac{P(t,t_0) - k \sum_{i=1}^N \alpha_i P(t,t_i)}{1 + k \alpha_N}.
    \end{equation}
    This equation is known as the bootstrapping equation associated to the quoted IRS contract of $N$-coupons with swap rate $k$. Note this equation has $N+1$ unknown variables, namely $$P(t,t_0),P(t,t_1),\dots,P(t,t_N).$$
    \noindent In table \ref{TIIEQuotesBloomberg} we present the maturities of the plain vanilla IRSs quoted in the market, the $k$ swap rate and the number of unknown variables in each bootstrapping equation. Note that from the IRS market we have in total 390 unknown variables and only 14 equations (since 14 swaps are quoted in the market).
    \begin{table}
    \scriptsize
    \centering
    \begin{tabular}{|r c c | c | c |}
      \hline
     \multicolumn{1}{|>{\centering\arraybackslash}m{8mm}}{\textbf{Tenor}} & 
     \multicolumn{1}{>{\centering\arraybackslash}m{10mm}}{\textbf{Rate (\%)}} &
     \multicolumn{1}{>{\centering\arraybackslash}m{10mm}|}{\textbf{Type}} & 
     \multicolumn{1}{>{\centering\arraybackslash}m{25mm}|}{\textbf{Number of Unknown Variables}} &
     \multicolumn{1}{>{\centering\arraybackslash}m{45mm}|}{\textbf{Unknown Variables}} \\ \hline
     
      ON & 3.0500 & Cash & 1 & $P(t,t+1)$ \\
      TN & 3.0500 & Cash & 1 & $P(t,t+2)$ \\ 
      28D & 3.2950 & Cash & 1 & $P(t,T_{\textnormal{28D}})$\\ \hline
      84D & 3.3200 & Swap & 3 & $P(t,t+1),P(t,T_{\textnormal{28D}}),\dots,P(t,T_{\textnormal{84D}})$\\ 
      168D & 3.4300 & Swap & 6 & $P(t,t+1),P(t,T_{\textnormal{28D}}),\dots,P(t,T_{\textnormal{168D}})$\\ 
      252D & 3.5620 & Swap & 9 & $P(t,t+1),P(t,T_{\textnormal{28D}}),\dots,P(t,T_{\textnormal{252D}})$\\ 
      364D & 3.7350 & Swap & 13 & $P(t,t+1),P(t,T_{\textnormal{28D}}),\dots,P(t,T_{\textnormal{364D}})$\\ 
      728D & 4.2360 & Swap & 26 & $P(t,t+1),P(t,T_{\textnormal{28D}}),\dots,P(t,T_{\textnormal{728D}})$\\ 
      1092D & 4.6710 & Swap & 39 & $P(t,t+1),P(t,T_{\textnormal{28D}}),\dots,P(t,T_{\textnormal{1092D}})$\\ 
      1456D & 5.0510 & Swap & 52 & $P(t,t+1),P(t,T_{\textnormal{28D}}),\dots,P(t,T_{\textnormal{1456D}})$\\ 
      1820D & 5.3610 & Swap & 65 & $P(t,t+1),P(t,T_{\textnormal{28D}}),\dots,P(t,T_{\textnormal{1820D}})$\\ 
      2548D & 5.8630 & Swap & 91 & $P(t,t+1),P(t,T_{\textnormal{28D}}),\dots,P(t,T_{\textnormal{2548D}})$\\ 
      3640D & 6.2380 & Swap & 130 & $P(t,t+1),P(t,T_{\textnormal{28D}}),\dots,P(t,T_{\textnormal{3640D}})$\\ 
      4368D & 6.4280 & Swap & 156 & $P(t,t+1),P(t,T_{\textnormal{28D}}),\dots,P(t,T_{\textnormal{4368D}})$\\ 
      5460D & 6.6320 & Swap & 195 & $P(t,t+1),P(t,T_{\textnormal{28D}}),\dots,P(t,T_{\textnormal{5260D}})$\\ 
      7280D & 6.8310 & Swap & 260 & $P(t,t+1),P(t,T_{\textnormal{28D}}),\dots,P(t,T_{\textnormal{7280D}})$\\ 
      10920D & 7.0210 & Swap & 390 & $P(t,t+1),P(t,T_{\textnormal{28D}}),\dots,P(t,T_{\textnormal{10920D}})$\\ 
       \hline
    \end{tabular}
    \caption{Quoted TIIE 28D Swaps (\%) on May 29 2015 (Source: Bloomberg). Every swap trades with the convention of 28 days Following with accrual period of ACT/360.}
    \label{TIIEQuotesBloomberg}
    \end{table}
    To solve this system of equations we will need an interpolation scheme. As has been mentioned previously, many interpolation methods for curve construction are available, see \cite{hagan2006interpolation}, \cite{hagan2008methods} and \cite{du2011investigation}. Besides the interpolation method, some short-term products must be included for the construction of the curve. For example, some zero coupon bonds might trade giving us an exact rate for the curve, however in some markets, where there is insufficient liquidity, some interbank money-market rates will be used. In table \ref{TIIEQuotesBloomberg} we consider cash instruments (money-market) for maturities: overnight, tomorrow-next and 28d. It is important to point out that these short-term rates are also helpful since conditions of smoothness and continuity in interpolation methods are required.\par 
    Throughout this work we will present results considering cubic splines in the zero rates as interpolation method. In appendix \ref{App:AppendixB} we explain in detail three different interpolation methods.\par
    Note that the bootstrapping equation \eqref{discBootstrapping} is expressed in terms of discount factors, however this equation can be expressed in term of zero rates
    \begin{equation}\label{bootstrappingEqSC}
    	\begin{aligned}
			R(t,t_N)&=-\frac{1}{\tau_N} \ln \bigg[ \frac{P(t,t_0) - k \sum_{i=1}^N \alpha_i P(t,t_i)}{1 + k \alpha_N} \bigg]\\
            &=\frac{1}{\tau_N} \ln \bigg[ \frac{1+k\alpha_N} { e^{-\tau_{t_0}R(t,t_0)} - k \sum_{i=1}^N \alpha_i e^{-\tau_{t_i}R(t,t_i)} } \bigg].
        \end{aligned}
	\end{equation}
    Hence the system of equations of table \ref{TIIEQuotesBloomberg} is given by
	\begin{equation}\label{systemTIIE28dSingleCurve}
    \left\{
    \begin{aligned}
        & R(t,t_{N_l})=\frac{1}{\tau_{N_l}} \ln \bigg[ \frac{1+k\alpha_{N_l}} { e^{-\tau_{t_0}R(t,t_0)} - k \sum_{i=1}^{N_l} \alpha_i e^{-\tau_{t_i}R(t,t_i)} } \bigg] \hspace{5mm} l=1,2,\dots,14 \hspace{2mm} \textnormal{(swaps)}\\
        & R(t,t_0)=R_{\textnormal{ON}}, \hspace{2mm} R(t_0,t_0+\textnormal{1D})=R_{\textnormal{TN}}, \hspace{2mm} R(t_0,t_0+\textnormal{28D})=R_{\textnormal{28D}}. \hspace{2mm} \textnormal{(cash)}
    \end{aligned}
    \right.
	\end{equation}
    With this equation we proceed to the bootstrapping algorithm which relies on an iterative solution algorithm. The idea is the following:
    \begin{enumerate}
    	\item Take the rates $R(t,x)$ already known from the money market
        \item Guess initial values for $\{R(t,t_N)\}$ where $t_N$ are the maturities of the 14th swaps
        \item With the interpolation method we calculate $R(t,t_i)$ for all $i$th coupon date
        \item Insert these rates into the right-hand side of equation \eqref{bootstrappingEqSC} and solve for $\{R(t,t_N)\}$
        \item We take these new guesses and again apply the interpolation algorithm
    \end{enumerate}
    Using this iterative algorithm we bootstrapped the TIIE 28d yield curve $R(t,T)$ using natural cubic splines on the yield curve. We used ACT/360 as day count convention to determine the year fractions used for discounting. 
	 \begin{figure}[H]
        \centering
        \vspace{-2mm}
        \includegraphics[scale=0.58]{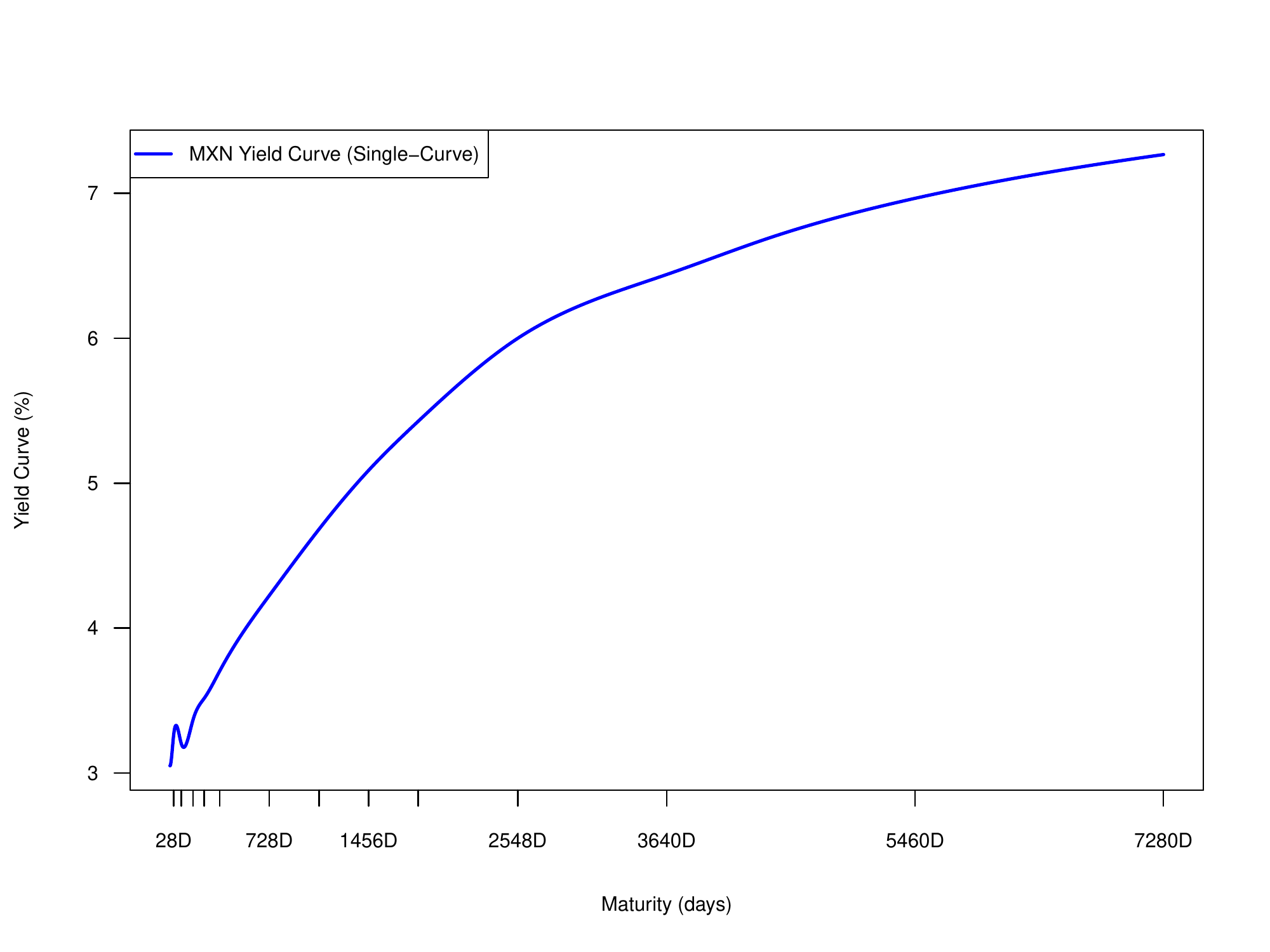}
        \caption[TIIE 28d yield curve in a \emph{single-curve} framework]{TIIE 28d yield curve $R(t,T)$ in a \emph{single-curve} framework using natural cubic splines interpolation in the yield rates. The swap rates used for the construction are in table \ref{TIIEQuotesBloomberg}.}
        \label{TIIE28D_YieldCurveSC}
        \vspace{-2mm}
     \end{figure}    

	\begin{figure}[H]
        \centering
        \vspace{-2mm}
        \includegraphics[scale=0.58]{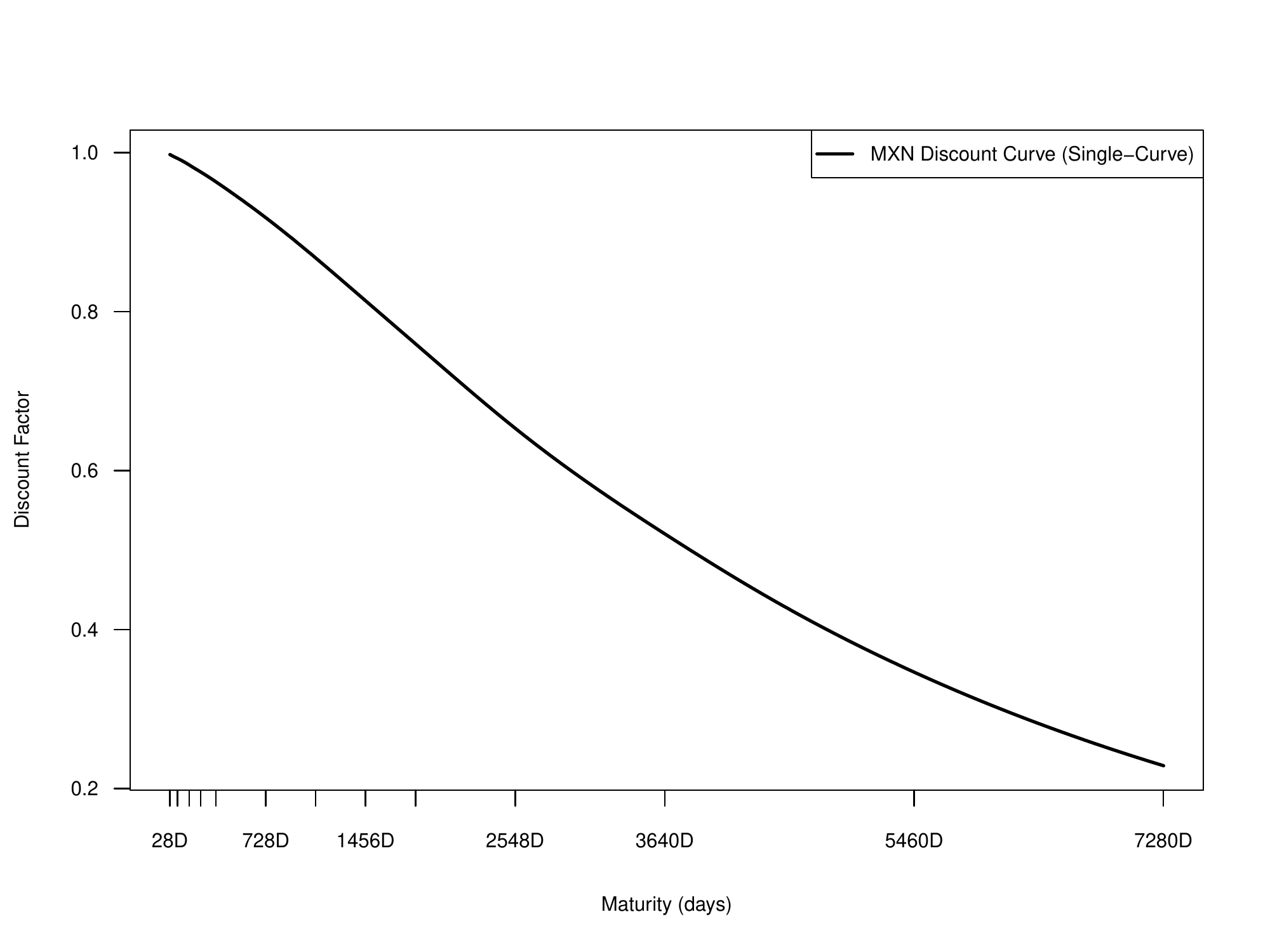}
        \caption[TIIE 28d discount curve in a \emph{single-curve} framework]{TIIE 28d discount curve $P(t,T)$ in a \emph{single-curve} framework using natural cubic splines interpolation in the yield rates. The swap rates used for the construction are in table \ref{TIIEQuotesBloomberg}.}
        \label{TIIE28D_DiscountCurveSC}
        \vspace{-2mm}
     \end{figure}
    
      \begin{figure}[H]
        \centering
        \vspace{-2mm}
        \includegraphics[scale=0.58]{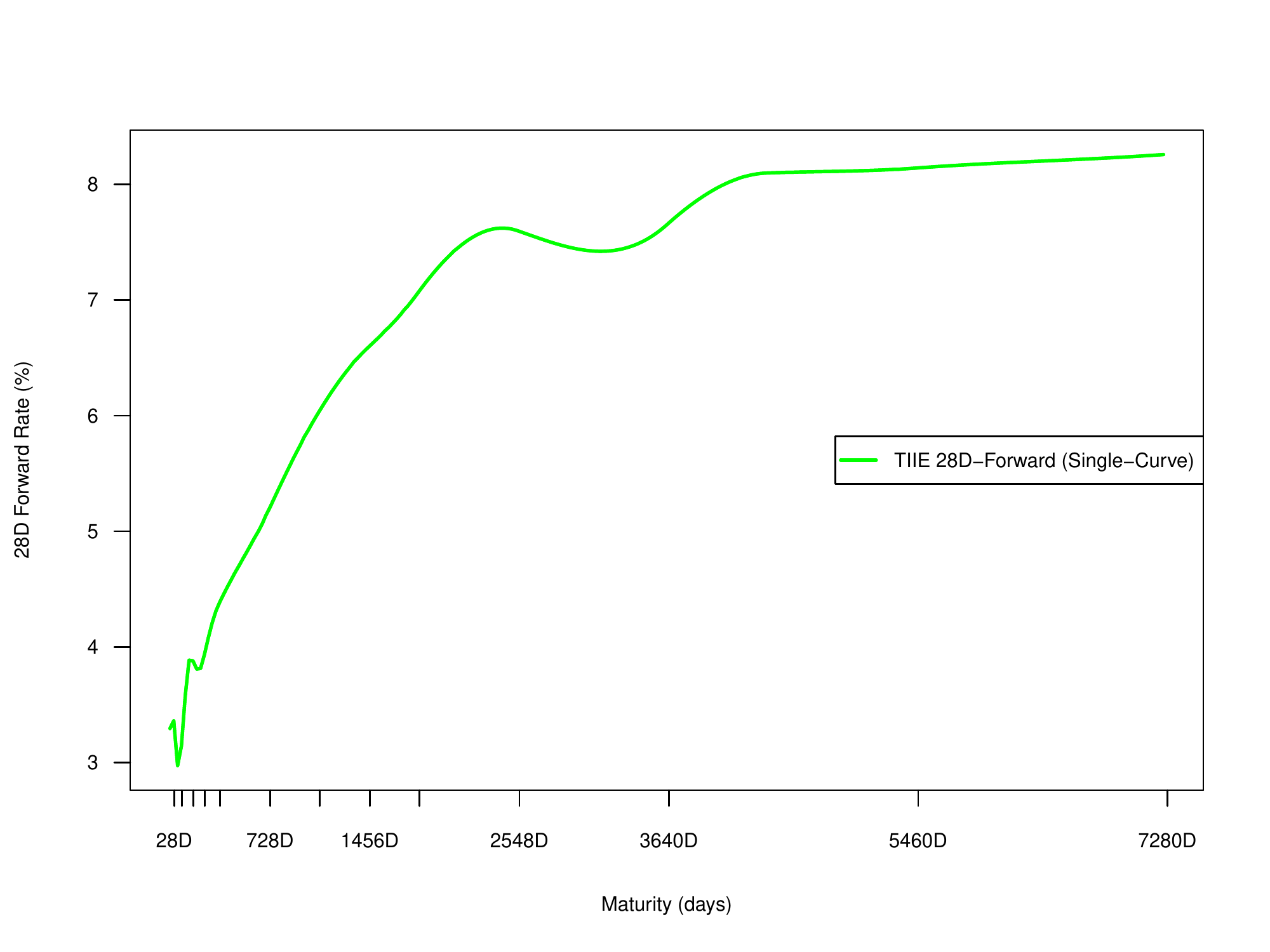}
        \caption[TIIE 28d forward curve in a \emph{single-curve} framework]{TIIE 28d forward curve $\mathbb{E}_t^{T_{i-1}}(\textnormal{\textbf{TIIE28D}}(T_{i-1},T_{i}))$ in a \emph{single-curve} framework using natural cubic splines interpolation in the yield rates. The swap rates used for the construction are in table \ref{TIIEQuotesBloomberg}.}
        \label{TIIE28D_ForwardCurveSC}
        \vspace{-2mm}
     \end{figure}

    \newpage
    \section{Pricing IRS in a Multi-curve Framework (in Presence of Collateral)}\label{sec:multicurveframework}
    In this section we present the valuation frameworks for pricing interest rate swaps (IRSs) in a \emph{multi-curve} framework with a collateral account associated to the derivative. The section is divided in three main subsections. In the first subsection we establish the general collateral framework, i.e. we explain how a a collateral account works, which are the advantages and disadvantages for having a collateral framework and what assumptions do we have to make for pricing collateralized IRSs. In the second subsection we focus on the easiest case: when the currency of payoffs of the derivative and the collateral currency coincide. Indeed we present the curve calibration of IRSs and OISs in USD currency. In the third and last subsection we give the multiple currencies collateral framework, i.e. when the payoffs' currency is different of the collateral currency. Furthermore we exemplify the differences between EUR and MXN pricing of IRSs when the collateral currency is USD. The material and results presented in this section were  mostly taken from \cite{fujii2010note}, \cite{fujii2010collateral}, \cite{fujii2011market}, \cite{piterbarg2010funding}, \cite{piterbarg2012cooking} and \cite{green2015xva}.
    \subsection{Pricing of Collateralized Products}\label{pricingCollProd}
    As we saw in the introduction there have been a lot of changes in the market since the financial turmoil in 2008. One of the most important questions that has to be answered since then is: what is the \emph{risk-free} rate? Before we try to answer this question let us present a quote taken from \cite{green2015xva}.
    \begin{center}
    Nothing in life and nothing that we do is \emph{risk-free}
    \end{center}
    \begin{flushright}
    Ken Salazar, US Politician\par
    \end{flushright}
    We have already seen that LIBOR rates in the USD market are not a good proxy anymore after the Lehman Brothers default. Recall that the world has entered into a new phase in which high-credit rating banks are able to default in matter of weeks. Once we accepted that LIBOR rate is not a good choice of \emph{risk-free} rate, it is normal to think on yield rates of government bonds. However governments also default, as the case of Greece in the Eurozone or Argentina in Latin America. Another alternative for \emph{risk-free} rate might be the repo rate. The repo rate is an interest rate that is paid on a collateralized loan and therefore should be very close to being \emph{risk-free}, unfortunately the repo market is only liquid for maturities up to one year, for our purpose for the valuation of long-term IRS we need a market with long-dated maturities. Hence the best candidate for this purpose is the OIS market. There are many valid reasons for using overnight rates as \emph{risk-free} rates, let us present some of these reasons:
    \begin{enumerate}[noitemsep]
    	\item Overnight Rates such as Fed Funds and Eonia are based on actual trades, indeed these rates are calculated by the average rate at which these transactions occur
        \item The OIS market is active and liquid in several currencies and have maturities of up to 30 years.
        \item Lending and borrowing money in this market has a low counterparty credit risk since transactions occur in a daily basis so the counterparty might change also daily
        \item ISDA contracts typically used this rates as the cash collateral rate
    \end{enumerate}
    So now we know that overnight rates are widely used as collateral rates. In this section we will prove that the collateral rate in the presence of a perfect CSA is the curve used for discounting. But before we proof this fact let us explain briefly how does the collateral works and why it is important in the valuation of interest rate derivatives.\par\medskip
    \noindent Suppose that a Bank B has a big and positive exposure $X$ (sum of all derivatives transactions) against Bank A. There is clearly a strong risk for the Bank B if the Bank A is to default. With a collateral agreement Bank B limits this exposure since Bank A has to post this mark-to-market (exposure) $X$ as a collateral. The collateral receiver, in this case Bank B due to having a positive mark-to-market, becomes the economic owner of the collateral if, and only if, the collateral giver (Bank A) defaults. While Bank A is away to enter into default, the collateral belongs to Bank A but it is in hands of Bank B. Hence, as a reward for posting a big amount of money as collateral, Bank A should receive from Bank B an interest rate $c$ (known as the collateral rate) periodically. In other words, Bank A receives $Xc\tau$ where $\tau$ is the year fraction of the periodicity of the collateral rate. Being rigorous the collateral giver is not Bank A, it is specifically the trading desk of Bank A who manage all the transactions against Bank B. Since the trading desks generally do not manage cash, then the trading desk of Bank A has to borrow $X$ from the funding desk of Bank A. For borrowing $X$ the trading desk has to pay an interest rate $r$ known as the funding rate. In figure \ref{collateralABC} we present in a diagram how the flows are exchanged between Bank B, Bank A's trading desk and Bank A's funding desk. To finish this example we may suppose that the Bank A has a positive mark-to-market against Bank C, but dealt without a collateral agreement. Therefore if Bank C defaults then Bank A (trading desk) would have potential losses.\par
	 \begin{figure}[h]
        \centering
        \vspace{-2mm}
        \includegraphics[scale=0.85]{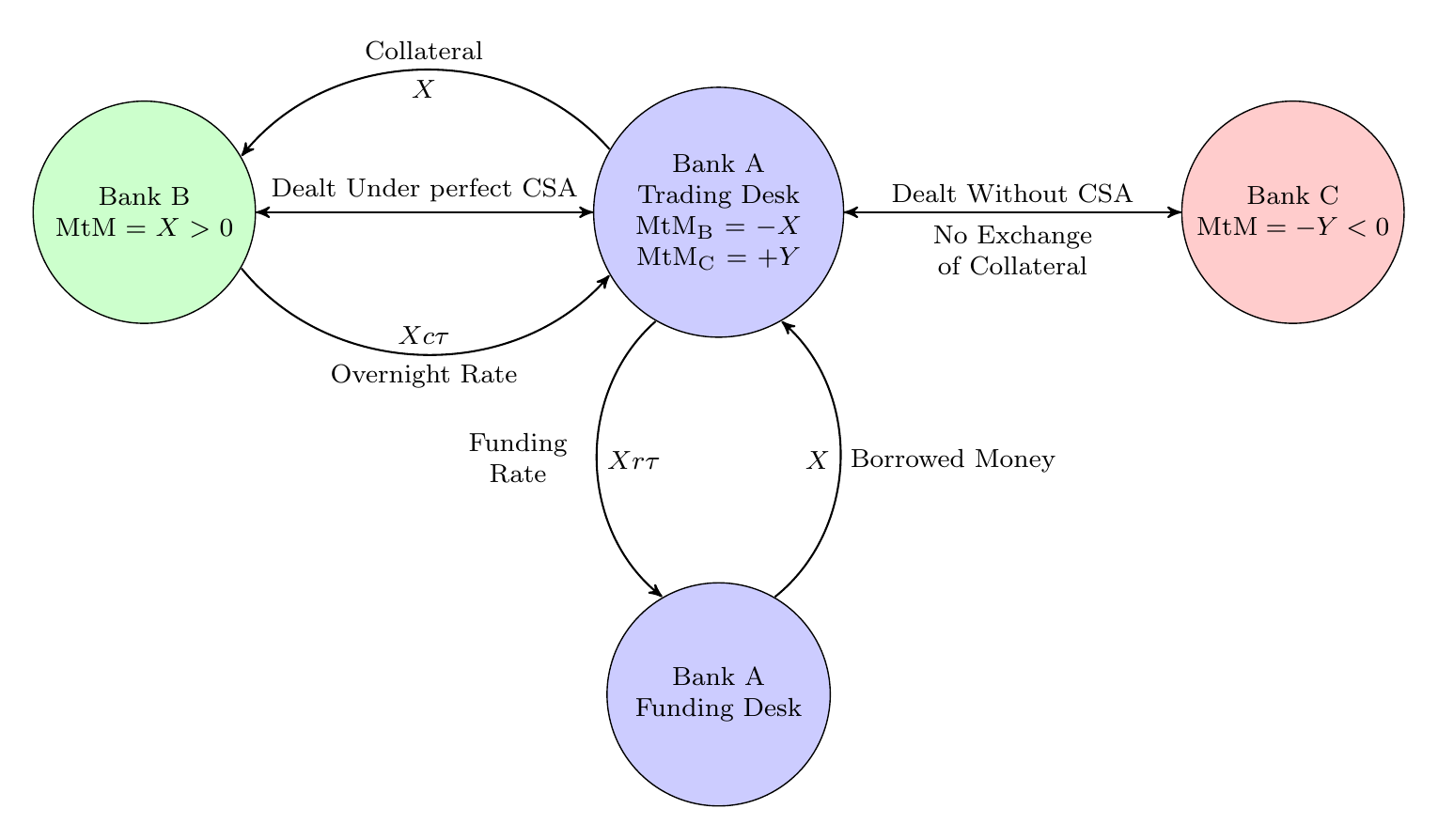}
        \caption[]{Illustration of the basic principle of collateralisation and non-collateralisation.}
        \label{collateralABC}
        \vspace{-2mm}
     \end{figure}    
    \noindent As there exist a large number of conditions within the CSA agreements, the collateralisation is critically defined by a large number of key parameters \cite{gregory2012counterparty}. In this work we will assume that the collateralisation is under a perfect CSA agreement. The next definition was taken from \cite{ametrano2013everything}.
    \begin{defi}\textnormal{\textbf{(Perfect CSA)}} We define \emph{perfect CSA} an ideal collateral agreement with the following characteristics:\par
        \begin{itemize}[noitemsep]
    \item fully symmetric
    \item zero initial margin 
    \item cash collateral
    \item zero threshold
    \item zero minimum transfer amount
    \item continuous margination and instantaneous margination rate.
    \end{itemize}
    \label{defPerfectCSA}
    \end{defi}
    Let us explain briefly each of the characteristics of a perfect CSA contract. In a CSA agreement, conditions are not necessarily equal for the two counterparties, however if we assume that the CSA is fully symmetric hence both parties have the same conditions and indeed both have to post collateral. The initial margin is the amount of money that the counterparties have to post when closing the deal. A zero initial margin implies that neither of the counterparties have to post an initial margin. The threshold is a level of exposure below which collateral will not be called. The threshold therefore represents an amount of undercollateralized exposure. If the exposure is above the threshold, only the incremental exposure will be collateralized. In the case of threshold equal to zero then every movement in mark-to-market should be collateralised. The minimum transfer amount defines the minimum amount of collateral that can be called for at a time. Collateral cannot be transferred in blocks that are smaller than the minimum transfer amount, but with this minimum transfer amount equal to zero then every update of mark-to-market has to be transferred. Finally the margination frequency refers to the periodic timescale which may be called and returned. Intraday margination is common for plain vanilla products such as repos and for derivatives cleared via central counterparties. For the perfect CSA, the continuous margination assumption is in practice operationally impossible, however this assumption will help us for the mathematical model. In the next subsection, we present the valuation framework for pricing any derivative with a given collateral account, this model is useful only when the derivative currency coincides with the collateral currency.
    
    \subsection{Valuation Framework in a Single Currency}
    Let $X$ be a derivative (of some underlying product) with maturity at time $T$ and with $(h(t))_{t \geq 0}$ price process. Let $V(t)$ be a stochastic self-financing collateral account under a perfect CSA agreement for a position of derivatives $X$\footnote{We will assume the conditions under a perfect CSA (see definition \ref{defPerfectCSA}). Nevertheless the model is built thinking the derivative cab be partially collateralised. The idea behind this assumption is to show the difference of valuation between collateralised and non-collateralised.}. Assume that the derivative payoff and the collateral account are posted in the same currency, later we will present the case of different currencies. Then the stochastic process of the collateral account is given by
    \begin{equation}\label{eq:dV}
        dV(s)=\alpha \left[ r(s)-c(s) \right] V(s) ds + a(s)dh(s),
    \end{equation}
    where, $\alpha$ is the percentage of collateralization on the derivative $X$, $r(s)$ and $c(s)$ are the funding rate and the collateral rate at time $s$, respectively, $h(s)$ is the value of the derivative at time $s$ which matures at time $T$ with cash flow $h(T)$, and $a(s)$ represents the number of positions of the derivative at time $s$. Equation \eqref{eq:dV} can be interpreted as follows: the change of the value of the collateral account depends on the interest differential earned on the partial posted collateral over time as well as the change in the value of the $a(s)$ underlying derivatives. Note that equation \eqref{eq:dV} has the form $dQ(t)=(\mu(t) - r_f(t)) Q(t) dt+\sigma_Q(t) dW(t)$ which is similar to the dynamic for pricing future contracts of stocks in a Black-Scholes setting. To solve \eqref{eq:dV} we have to multiply the equation by $\exp({\int_s^T \alpha y(\eta) d\eta})$, where $y(s)=r(s)-c(s)$, to get
    \begin{equation}\label{eq:expdV}
        e^{\int_s^T \alpha y(\eta) d\eta} dV(s)=e^{\int_s^T \alpha y(\eta) d\eta} \alpha y(s) V(s) ds + e^{\int_s^T \alpha y(\eta) d\eta} a(s)dh(s).
    \end{equation}
    Integrating \eqref{eq:expdV} over $[t,T]$ give us
    \begin{equation}\label{eq:intexpdV1}
        \int_t^T e^{\int_s^T \alpha y(\eta) d\eta} dV(s)= \int_t^T e^{\int_s^T \alpha y(\eta) d\eta} \alpha y(s) V(s) ds + \int_t^T e^{\int_s^T \alpha y(\eta) d\eta} a(s)dh(s).
    \end{equation}
    Let us define $u=e^{\int_s^T \alpha y(\eta) d\eta}$ and $dv=dV(s)$, using integration by parts formula we obtain
    \begin{align}
        \int udv &= uv - \int vdu\\
        \int_t^T e^{\int_s^T \alpha y(\eta) d\eta} dV(s) &= e^{\int_s^T \alpha y(\eta) d\eta} V(s)\Big|_t^T + \int_t^T V(s)e^{\int_s^T \alpha y(\eta) d\eta} \alpha y(s) ds. \label{intByParts}
    \end{align}
    Then using \eqref{eq:intexpdV1} and \eqref{intByParts} we have
    \begin{equation}\label{eq:intexpdV2}
        V(T)=e^{\int_t^T \alpha y(\eta) d\eta}V(t) + \int_t^T e^{\int_s^T \alpha y(\eta) d\eta} a(s)dh(s).
    \end{equation}
    As in \cite{fujii2010note}, we adopt the trading strategy specified by  
    \begin{equation}\label{tradingstrat}
    \left\{
    \begin{aligned}
        V(t)&=h(t)\\
        a(s)&=e^{\int_t^s \alpha y(\eta) d\eta}
    \end{aligned}
    \right.
	\end{equation}

    If we substitute the trading strategy \eqref{tradingstrat} in \eqref{eq:intexpdV2} we get
    \begin{align}
        V(T)&=e^{\int_t^T \alpha y(\eta) d\eta}V(t) + \int_t^T e^{\int_s^T \alpha y(\eta) d\eta} a(s)dh(s)
 \nonumber \\
            &=e^{\int_t^T \alpha y(\eta) d\eta}h(t) + \int_t^T e^{\int_s^T \alpha y(\eta) d\eta + \int_t^s \alpha y(\eta) d\eta} dh(s)
 \nonumber \\
            &=e^{\int_t^T \alpha y(\eta) d\eta}h(t) + e^{\int_t^T \alpha y(\eta) d\eta} \int_t^T dh(s)
 \nonumber \\
            &=e^{\int_t^T \alpha y(\eta) d\eta}h(t)+e^{\int_t^T \alpha y(\eta) d\eta}h(T)-e^{\int_t^T \alpha y(\eta) d\eta}h(t)
 \nonumber \\
            &=e^{\int_t^T \alpha y(\eta) d\eta}h(T)
    \end{align}
    Now we can calculate the present value $h(t)$ of the derivative $X$ considering the collateral account $V$ associated to it. Using the risk-neutral measure $\mathbb{Q}$ with numéraire $B(T)=\exp({\int_t^T r(s)ds})$ (money-market account) yields
    \begin{align}
        h(t) &= B(t) \mathbb{E}^\mathbb{Q}_t \left[ \frac{V(T)}{B(T)} \right]
 \nonumber \\
             &= \mathbb{E}^\mathbb{Q}_t \left[ \frac{e^{\int_t^T \alpha y(s) ds}h(T)}{e^{\int_t^T r(s)ds}} \right]
 \nonumber \\
             &= \mathbb{E}^\mathbb{Q}_t \left[ e^{-\int_t^T r(s) - \alpha (r(s) - c(s)) ds} h(T) \right]
 \nonumber \\
             &= \mathbb{E}^\mathbb{Q}_t \left[ e^{-\int_t^T (1 - \alpha)r(s) + \alpha c(s) ds} h(T) \right]. \label{eq5.9}
    \end{align}
    Therefore, if $X$ is fully collateralized, i.e. $\alpha=1$, then
    \begin{equation}\label{fullycoll}
        h(t)= \mathbb{E}^\mathbb{Q}_t \left[ e^{-\int_t^T c(s) ds} h(T) \right],\\
    \end{equation}
    which implies that collateralized cash flows must be discounted considering the collateral rate. Unlike this case, if $X$ is not collateralized, i.e. $\alpha=0$, then
        \begin{equation}
        h(t)= \mathbb{E}^\mathbb{Q}_t \left[ e^{-\int_t^T r(s) ds} h(T) \right],\\
    \end{equation}
    which implies that uncollateralized cash flows must be discounted considering the funding rate.
    \begin{defi}\label{Def:FullyCollZeroCoupon}
        We denote the price of a $T$-maturity fully-collateralized zero coupon bond by
        \begin{equation}
            P^{c}(t,T)=\mathbb{E}^\mathbb{Q}_t \left[ e^{-\int_t^T c(s) ds} \right].
        \end{equation}
    \end{defi}
    We will use this definition to introduce the forward measure under collateralization.
    \subsubsection{Forward Measure}
    Before we can go on, let us present in an informal way the concept of Radon-Nikod\'ym derivative and two important results that are used to build the collateralized forward measure. Consider a general finite probability space $(\Omega,\mathfrak{F},\mathbb{P})$. Suppose that on this space we have another probability measure $\mathbb{Q}$. Let us assume that $\mathbb{P}>0$ and $\mathbb{Q}>0$ for every $\omega \in \Omega$, so we can define
    \begin{equation}
        Z(\omega)=\frac{\mathbb{Q}(\omega)}{\mathbb{P}(\omega)}.
    \end{equation}
    Since $Z>0$ for all $\omega \in \Omega$ and $Z$ is a random variable we can compute the expected value of $Z$ under measure $\mathbb{P}$
    \begin{equation}
        \mathbb{E}^\mathbb{P}(Z)=\sum_{\omega \in \Omega} Z(\omega)\mathbb{P}(\omega)=\sum_{\omega \in \Omega} \frac{\mathbb{Q}(\omega)}{\mathbb{P}(\omega)} \mathbb{P}(\omega) = \sum_{\omega \in \Omega} \mathbb{Q}(\omega)=1.
    \end{equation}
    Now for any random variable $Y$,
    \begin{equation}
        \mathbb{E}^\mathbb{P}(ZY)=\sum_{\omega \in \Omega} Z(\omega)Y(\omega)\mathbb{P}(\omega)=\sum_{\omega \in \Omega} \frac{\mathbb{Q}(\omega)}{\mathbb{P}(\omega)} Y(\omega) \mathbb{P}(\omega) = \sum_{\omega \in \Omega} Y(\omega) \mathbb{Q}(\omega)=\mathbb{E}^\mathbb{Q}(Y).
    \end{equation}
    In this case, the random variable $Z$ is called the Radon-Nikod\'ym derivative of $\mathbb{Q}$ with respect to $\mathbb{P}$. The name of $Z$ is a consequence of its definition in the continuous case, since 
    \begin{align*}
        & Z(\omega):=\frac{d\mathbb{Q}}{d\mathbb{P}}\\
        \Rightarrow \hspace{5mm} & d\mathbb{Q}=Z(\omega)d\mathbb{P}\\
        \Rightarrow \hspace{5mm} & Y(\omega)d\mathbb{Q}=Y(\omega)Z(\omega)d\mathbb{P}\\
        \Rightarrow \hspace{5mm} & \int_{\Omega} Y(\omega) d\mathbb{Q} = \int_{\Omega}Y(\omega)Z(\omega)d\mathbb{P}\\
        \Rightarrow \hspace{5mm} & \mathbb{E}^\mathbb{Q}(Y)=\mathbb{E}^\mathbb{P}(YZ).
    \end{align*}
    Now we are able to present a formal definition of the Radon-Nikod\'ym derivative. 
    \begin{teo}\label{teozeta}
    Given $\mathbb{P}$ and $\mathbb{Q}$ equivalent probability measures and a time horizon $T$, we can define a random variable $\frac{d\mathbb{Q}}{d\mathbb{P}}$ defined on $\mathbb{P}$-possible paths, taking positive real values, such that
    \begin{enumerate}[label=(\roman*)]
    \item $\mathbb{E}^\mathbb{Q}(X(T))=\mathbb{E}^\mathbb{P} \bigg( \dfrac{d\mathbb{Q}}{d\mathbb{P}} X(T) \bigg)$, \hspace{3mm} for all random variables $X(T)$;
    \item $\mathbb{E}^\mathbb{Q}(X(T) | \mathfrak{F}(t) )=\zeta^{-1}(t) \mathbb{E}^\mathbb{P} (\zeta(T) X(T) | \mathfrak{F}(t)),$ \hspace{5mm} $t \leq T,$
    \end{enumerate}
    where $\zeta(t)$ is the process $\mathbb{E}^\mathbb{P} ( \frac{d\mathbb{Q}}{d\mathbb{P}} | \mathfrak{F}(t) )$.
    \end{teo}
    \begin{teo}\label{ChangeMeasureTheo}
        Consider two numéraires $N(t)$ and $M(t)$, inducing equivalent martingale measures $\mathbb{Q}^N$ and $\mathbb{Q}^M$, respectively. If the market is complete, then the density of Radon-Nikodým derivative relating the two measures is uniquely given by
        \begin{equation}
            \zeta(t)=\mathbb{E}^{\mathbb{Q}^N} \bigg( \frac{d\mathbb{Q}^M}{d\mathbb{Q}^N} \bigg| \mathfrak{F}(t) \bigg) = \frac{M(t)/M(0)}{N(t)/N(0)}.
        \end{equation}
    \end{teo}
    Then let us define the forward measure $\mathbb{T}^c$ associated with the numéraire $P^{c}(t,T)$. Using theorem \ref{ChangeMeasureTheo} and the risk neutral measure $\mathbb{Q}$ with the numéraire $C(t):=$ $\exp(-\int_t^T c(s)ds)$ we have that
    \begin{equation}
        \zeta(t)=\mathbb{E}^{\mathbb{Q}}_t \bigg( \frac{d\mathbb{T}^c}{d\mathbb{Q}} \bigg)=\frac{P^c(t,T)/P^c(0,T)}{C(t)/C(0)}=\frac{P^c(t,T)C(0)}{P^c(0,T)C(t)}.
    \end{equation}
    Using equation \eqref{fullycoll} and theorem \ref{teozeta} we have
    \begin{align*}
        h(t)&=\mathbb{E}^\mathbb{Q}_t \left[ e^{-\int_t^T c(s) ds} h(T) \right]\\
            &=\mathbb{E}_t^\mathbb{Q}\left[ C(t)h(T) \right]\\
            &=\mathbb{E}_t^\mathbb{Q}\left[ P^c(t,T)\frac{\zeta(T)}{\zeta(t)} h(T) \right]\\
            &=\zeta^{-1}(t) \mathbb{E}_t^\mathbb{Q}\left[ \zeta(T) P^c(t,T)h(T) \right]\\
            &=\mathbb{E}_t^{\mathbb{T}^c}\left[ P^c(t,T) h(T) \right]\\
            &=P^c(t,T) \mathbb{E}_t^{\mathbb{T}^c}\left[ h(T) \right].
    \end{align*}
    This bring us to the following definition.
    \begin{defi}\label{Def:FullyColl} The price of a fully-collateralized derivative at time $t$ which matures at time $T$ with payoff $h(T)$ is given by,
    \begin{equation}\label{eq:PriceCollDer1Curr}
        h(t)=P^c(t,T) \mathbb{E}_t^{\mathbb{T}^c}\left[ h(T) \right],
    \end{equation}
    where $P^c(t,T)=\mathbb{E}^\mathbb{Q}_t \left[ \exp({-\int_t^T c(s) ds}) \right]$ is the $T$-maturity fully-collateralized zero coupon bond and $\mathbb{T}^c$ is the forward measure associated with the numéraire $P^c(t,T)$.
    \end{defi}
    
    Repeating the previous arguments that yield in definitions \ref{Def:FullyCollZeroCoupon} and \ref{Def:FullyColl} we are able to define the zero coupon bond and forward measure when the derivative $X$ is partially-collateralized.
    
    \begin{defi}\label{Def:AlphaColl}
        We denote the price of a $T$-maturity $\alpha$-collateralized zero coupon bond by
        \begin{equation}
            P^{\alpha c}(t,T)=\mathbb{E}^\mathbb{Q}_t \left[ e^{-\int_t^T (1 - \alpha)r(s) + \alpha c(s) ds} \right].
        \end{equation}
        where $r(s)$ denotes the funding rate, $c(s)$ the collateral rate and $\alpha \in [0,1]$.
    \end{defi}
    
    \begin{teo}\label{Teo:AlphaColl} The price of an $\alpha$-collateralized derivative at time $t$, which matures at time $T$ with payoff $h(T)$, is given by
    \begin{equation}\label{eq:PriceAlphaCollDer1Curr}
        h(t)=P^{\alpha c}(t,T) \mathbb{E}_t^{\mathbb{T}^{\alpha c}}\left[ h(T) \right],
    \end{equation}
    where $P^{\alpha c}(t,T)$ is the $T$-maturity $\alpha$-collateralized zero coupon bond and $\mathbb{T}^{\alpha c}$ is the forward measure associated with the numéraire $P^{\alpha c}(t,T)$.
    \end{teo}
    \begin{proof}
    See appendix \ref{App:Appendix0}.
    \end{proof}
    
    Note that when $\alpha=1$ we obtain the result of definition \ref{Def:FullyColl}. In contrast, when $\alpha=0$ we get that $P^{\alpha c}(t,T)=\mathbb{E}^\mathbb{Q}_t \big[ \exp({-\int_t^T r(s)ds}) \big]=\mathbb{E}^\mathbb{Q}_t \big[ (B(T))^{-1} \big]$, but $B(T)$ is the numéraire associated to the risk measure $\mathbb{Q}$, hence $\mathbb{T}^{\alpha c} \sim \mathbb{Q}$ and $P^{\alpha c}(t,T)=(B(T))^{-1}$.
    
    \subsubsection{Calibration of the USD discount curve}\label{USDDiscountCurveSection}
    Up to now we have defined fully-collateralized discount factors $P^c(t,T)$, but we have not described how to obtain them from the market. Under the assumption that the collateral rates are equal to the overnight rates, we can build the collateral discount curve using the Overnight Index Swaps (OISs) that are liquid market instruments. OIS cash flows are linked to overnight rates since the floating leg of the swap compounds them. It is important to point out that typically these swaps are fully-collateralized and the collateral currency coincides with the currency of the overnight rate.\par\smallskip 
    It is important to point out that, regardless the creditworthiness of the counterparty, the correct discount curve in a collateralized world is always the collateral rate (in this case the OISs implied curve) since the risk-neutral pricing framework, presented in the previous section, is always based on discounting with the closest \emph{risk-free} rate. However, in practice, we have to consider the counterparty characteristics for a correct valuation of the derivative. To do that we have to price every derivative in a \emph{risk-free} environment and then the price has to be adjusted considering every x-valuation adjustments (XVA: CVA, DVA, FVA, KVA) rather than simply modifying the discount curve according to the credit quality of the counterparty. In other words, the price of a derivative may be seen as the following equation:
    \begin{equation}
    	\textnormal{Price}_\textnormal{Ctpy} :=  \textnormal{Price}_\textnormal{Risk free} + \underbrace{ \textnormal{CVA}_\textnormal{Ctpy} + \textnormal{FVA}_\textnormal{Ctpy} + \textnormal{KVA}_\textnormal{Ctpy} }_{\text{Valuation Adjustments}}.
    \end{equation}
    As we state on the introduction, in this work we will only focus on the \emph{risk-free} price, i.e. the price when the counterparty is under a perfect CSA and all valuation adjustments are negligible.\par\smallskip
    The overnight rate in the USD market is the Federal Funds Rate (Fed Fund Rate or FF Rate) that is published every business day with a publication lag of 1 day\footnote{Publication lag is the number of days between the start date of the period and the rate publication. A lag of 0 means on the start date, a lag of 1 means on the period end date.}, therefore OISs based on Fed Funds rates are fully collateralized in USD. Therefore, since the estimation curve (overnight rates) and the discount curve (collateral rates) are the same curve, with a simple bootstrapping we could construct the complete overnight rate curve. OISs are defined for a great variety of currencies, in this work we only focus on USD swaps market\footnote{The overnight rates could be either overnight (ON) loans or tomorrow/next (TN) loans. Main currencies overnight rates are: USD (Fed Fund), EUR (EONIA), GBP (SONIA), CHF (TOIS), JPY (TONAR), CAD (CORRA), HKD (HONIX) (see \cite{henrard2012interest}).}.\par
    \begin{rem}
    The existence of an overnight rate does not necessarily means the existence of an OIS market. For example, in the MXN interest rates market the overnight rate is called Tasa de Fondeo Bancario\footnote{The Tasa de Fondeo Bancario is defined every business day with a publication lag of 1 day.} and is computed/published by Banco de M\'exico. However, OISs using Tasa de Fondeo Bancario do not exist or at least are not quoted in the market.
    \end{rem}
	\begin{assum}
    For overnight curve calibration we assume that OISs are quoted and have liquidity in the market for maturities up to 50 years. We do not use quotes of Federal Fund Swaps (FFSs) (see section \ref{OISsFFSs}). Note that if FFSs are used, the discounting, overnight forward and LIBOR 3m forward curves need to be calibrated simultaneously using a dual bootstrapping.
    \end{assum}
   An OIS is a swap that exchanges a fixed rate coupon for a daily compounded overnight rate coupon, where the dates of the two coupon payments typically coincide. Hence, the floating payment between two dates $t$ and $T$ is given by
    \begin{equation}\label{OISfedfundpayment}
        \prod_{i=0}^{K-1} \left( 1 + \alpha(t_i,t_{i+1}) \textnormal{\textbf{FF}}(t_i,t_{i+1}) \right) - 1,
    \end{equation}
    \noindent \makebox[2.8cm][l]{$K$:} number of all the business days in the time interval $[t,T]$\par
     \noindent \makebox[2.8cm][l]{$\{t_i\}_{i=0}^K$:} all the business days in the accrual period $[t,T]$ with $t_0=t$ and $t_K=T$\par
     \noindent \makebox[2.8cm][l]{$\textnormal{\textbf{FF}}(t_i,t_{i+1})$:} represents the overnight interest rate for the period $(t_i,t_{i+1})$\par
    \noindent \makebox[2.8cm][l]{$\alpha(t_i,t_{i+1})$:} denotes the year fraction between $t_i$ and $t_{i+1}$ according to a market convention.\par\medskip
    Let us calculate the net present value of a payer OIS contract. Assume that at time $t$ we enter into a payer OIS with $N$ coupons, payments dates at $T_1<T_2< \dots <T_N$ and $T_0-t$ days of spot lag. Suppose that we pay a fixed rate $k$ and receive a floating rate (daily compounded overnight rate). Using equation \eqref{OISfedfundpayment} the $i$th floating leg rate $\textnormal{\textbf{FFComp}}(T_{i-1},T_i)$ could be write as
    \begin{equation}\label{FFCompDef}
        \textnormal{\textbf{FFComp}}(T_{i-1},T_i)=\frac{1}{\sum_{j=0}^{K_i-1} \alpha(t_j,t_{j+1})} \bigg (\prod_{j=0}^{K_i-1} \big( 1 + \alpha(t_j,t_{j+1}) \textnormal{\textbf{FF}}(t_j,t_{j+1}) \big) - 1 \bigg),
    \end{equation}
    \noindent \makebox[2.8cm][l]{$K_i$:} number of all the business days in the time interval $[T_{i-1},T_i]$\par
     \noindent \makebox[2.8cm][l]{$\{t_j\}_{j=0}^{K_i}$:} all business day in the accrual period $[T_{i-1},T_i]$ with $t_0=T_{i-1}$ and $t_{K_i}=T_i$\par
     \noindent \makebox[2.8cm][l]{$\textnormal{\textbf{FF}}(t_i,t_{i+1})$:} represents the overnight interest rate for the period $(t_j,t_{j+1})$\par
    \noindent \makebox[2.8cm][l]{$\alpha(t_j,t_{j+1})$:} denotes the year fraction between $t_j$ and $t_{j+1}$ according to a market convention.\par\medskip 
    It is easy to see that $\sum_{j=0}^{K_i} \alpha(t_j,t_{j+1})=\alpha(T_{i-1},T_i)$. Hence, since we can define $\delta$ such that $e^{\delta t}=1+it$, then we have that
\begin{align}
    \prod_{j=0}^{K_i-1} \big( 1 + \alpha(t_j,t_{j+1}) \textnormal{\textbf{FF}}(t_j,t_{j+1}) \big) - 1 &=e^{\ln \big(\prod_{j=0}^{K_i-1} \big( 1 + \alpha(t_j,t_{j+1}) \textnormal{\textbf{FF}}(t_j,t_{j+1}) \big) \big)} - 1 \nonumber \\
    &=e^{\sum_{j=0}^{K_i-1} \ln \left( 1 + \alpha(t_j,t_{j+1}) \textnormal{\textbf{FF}}(t_j,t_{j+1})\right)} -1 \nonumber \\
    &=e^{\sum_{j=0}^{K_i-1} \alpha(t_j,t_{j+1}) \delta(t_j,t_{j+1}) }-1. \label{sumRiemann}
\end{align}
    Note that in the equation \eqref{sumRiemann} the term $\sum_{j=0}^{K_i-1} \alpha(t_j,t_{j+1}) \delta(t_j,t_{j+1})$ is a Riemann sum of function $\delta$ with partition $\mathcal{P}=\{[t_0,t_1],[t_1,t_2],\dots,[t_{K_i-1},t_{K_i}] \}$. Hence $$\sum_{j=0}^{K_i-1} \alpha(t_j,t_{j+1}) \delta(t_j,t_{j+1}) \approx \int_{T_{i-1}}^{T_{i}} \delta(s)ds$$ and we define $c(s):=\delta(s)$ as the collateral curve.
    Now, assuming that this contract if fully collateralized by \eqref{fullycoll} we have that the present value of a payer OIS contract is given by,
    \begin{align}
        \textnormal{PV}(t)&=\sum_{i=1}^N \tau(T_{i-1},T_i) \mathbb{E}_t \left[ e^{-\int_t^{T_i} c(s)ds} \big( \textnormal{\textbf{FFComp}}(T_{i-1},T_i) - k \big) \right] \nonumber \\
             &=\sum_{i=1}^N \tau(T_{i-1},T_i) \mathbb{E}_t \left[ e^{-\int_t^{T_i} c(s)ds} \bigg( \frac{1}{\sum_j \alpha(t_j,t_{j+1})} \bigg (\prod_{j=0}^{K_i-1} \big( 1 + \alpha(t_j,t_{j+1}) \textnormal{\textbf{FF}}(t_j,t_{j+1}) \big) - 1 \bigg) - k \bigg) \right]
 \nonumber \\
             &=\sum_{i=1}^N \tau(T_{i-1},T_i) \mathbb{E}_t \left[ e^{-\int_t^{T_i} c(s)ds} \bigg( \frac{1}{\tau(T_{i-1},T_i)} \bigg( e^{\int_{T_{i-1}}^{T_i} c(s)ds} - 1 \bigg) - k \bigg) \right]
 \nonumber \\
             &=\sum_{i=1}^N \mathbb{E}_t \left[ e^{-\int_t^{T_i} c(s)ds} \bigg( e^{\int_{T_{i-1}}^{T_i} c(s)ds} - 1 \bigg) \right] - k \sum_{i=1}^N \tau(T_{i-1},T_i) \mathbb{E}_t \left[ e^{-\int_t^{T_i} c(s)ds} \right]
 \nonumber \\
             &=\sum_{i=1}^N \mathbb{E}_t \left[ e^{-\int_t^{T_{i-1}} c(s)ds} - e^{\int_t^{T_i} c(s)ds} \right] - k \sum_{i=1}^N \tau(T_{i-1},T_i) \mathbb{E}_t \left[ e^{-\int_t^{T_i} c(s)ds} \right]
 \nonumber \\
             &=\sum_{i=1}^N (P^c(t,T_{i-1})-P^c(t,T_i)) - k \sum_{i=1}^N \tau(T_{i-1},T_i) P^c(t,T_i)
 \nonumber \\
             &=P^c(t,T_0)-P^c(t,T_N) - k \sum_{i=1}^N \tau(T_{i-1},T_i) P^c(t,T_i). \label{OIS_PV}
    \end{align}
    Therefore, if we assume that the fixed rate $k$ is a mid-market quote then by no-arbitrage arguments we have that the present value of the OIS is equal to zero. Setting equation \eqref{OIS_PV} equal to zero give us the following equation
    \begin{equation}\label{nCouponOISBootstrapping}
        P^c(t,T_N)=\frac{P^c(t,T_0)-k\sum_{i=1}^{N-1} \tau(T_{i-1},T_i) P^c(t,T_i)}{1+k\tau(T_{N-1},T_N)}.
    \end{equation}
    This equation is called the bootstrapping equation associated to the $N$-coupons OIS contract. It is easy to see that we have $N+1$ variables $P^c(t,T_0),P^c(t,T_1),\dots,P^c(t,T_N)$ and only one equation. Thus, we have to define more OIS contracts to obtain more bootstrapping equations and a method to solve this system of equations.\par \smallskip
    In the USD OIS market the swaps with maturities no longer than a year normally have one payment at maturity, while swaps with a maturity over a year normally have yearly payments. In the next table (Table \ref{OISQuotesBloomberg}), we present the most liquid tenors for OIS contracts, these swaps are used for the construction of the collateralized discount curve. 
    
    \begin{table}
    \scriptsize
    \centering
    \begin{tabular}{|c c c | c | c |}
      \hline
     \multicolumn{1}{|>{\centering\arraybackslash}m{10mm}}{\textbf{Tenor}} & 
     \multicolumn{1}{>{\centering\arraybackslash}m{10mm}}{\textbf{Rate}} &
     \multicolumn{1}{>{\centering\arraybackslash}m{10mm}|}{\textbf{Type}} & 
     \multicolumn{1}{>{\centering\arraybackslash}m{20mm}|}{\textbf{Number of Variables}} &
     \multicolumn{1}{>{\centering\arraybackslash}m{30mm}|}{\textbf{Variables}} \\ \hline
     
      ON & 0.1300 & Cash & 1 & $P^c(t,t+1)$ \\ 
      TN & 0.1300 & Cash & 1 & $P^c(t,t+2)$\\ \hline
      1W & 0.1340 & Swap & 1 & $P^c(t,T_{\textnormal{1w}})$ \\ 
      2W & 0.1338 & Swap & 1 & $P^c(t,T_{\textnormal{2w}})$ \\ 
      3W & 0.1340 & Swap & 1 & $P^c(t,T_{\textnormal{3w}})$ \\ 
      1M & 0.1340 & Swap & 1 & $P^c(t,T_{\textnormal{1m}})$ \\ 
      2M & 0.1420 & Swap & 1 & $P^c(t,T_{\textnormal{2m}})$ \\ 
      3M & 0.1469 & Swap & 1 & $P^c(t,T_{\textnormal{3m}})$ \\ 
      4M & 0.1760 & Swap & 1 & $P^c(t,T_{\textnormal{4m}})$ \\ 
      5M & 0.1990 & Swap & 1 & $P^c(t,T_{\textnormal{5m}})$ \\ 
      6M & 0.2190 & Swap & 1 & $P^c(t,T_{\textnormal{6m}})$ \\ 
      7M & 0.2460 & Swap & 1 & $P^c(t,T_{\textnormal{7m}})$ \\ 
      8M & 0.2740 & Swap & 1 & $P^c(t,T_{\textnormal{8m}})$ \\ 
      9M & 0.3000 & Swap & 1 & $P^c(t,T_{\textnormal{9m}})$ \\ 
      10M & 0.3275 & Swap & 1 & $P^c(t,T_{\textnormal{10m}})$ \\ 
      11M & 0.3560 & Swap & 1 & $P^c(t,T_{\textnormal{11m}})$ \\ 
      1Y & 0.3860 & Swap & 1 & $P^c(t,T_{\textnormal{1y}})$ \\ 
      18M$^{\ast}$ & 0.5795 & Swap & 2 & $P^c(t,T_{\textnormal{6m}}),P^c(t,T_{\textnormal{18m}})$\\ 
      2Y & 0.7850 & Swap & 2 & $P^c(t,T_{\textnormal{1y}}),P^c(t,T_{\textnormal{2y}})$\\ 
      3Y & 1.1500 & Swap & 3 & $P^c(t,T_{\textnormal{1y}}),\dots,P^c(t,T_{\textnormal{3y}})$\\ 
      4Y & 1.4460 & Swap & 4 & $P^c(t,T_{\textnormal{1y}}),\dots,P^c(t,T_{\textnormal{4y}})$\\ 
      5Y & 1.6870 & Swap & 5 & $P^c(t,T_{\textnormal{1y}}),\dots,P^c(t,T_{\textnormal{5y}})$\\ 
      6Y & 1.8790 & Swap & 6 & $P^c(t,T_{\textnormal{1y}}),\dots,P^c(t,T_{\textnormal{6y}})$\\ 
      7Y & 2.0350 & Swap & 7 & $P^c(t,T_{\textnormal{1y}}),\dots,P^c(t,T_{\textnormal{7y}})$\\ 
      8Y & 2.1535 & Swap & 8 & $P^c(t,T_{\textnormal{1y}}),\dots,P^c(t,T_{\textnormal{8y}})$\\ 
      9Y & 2.2530 & Swap & 9 & $P^c(t,T_{\textnormal{1y}}),\dots,P^c(t,T_{\textnormal{9y}})$\\ 
      10Y & 2.3320 & Swap & 10 & $P^c(t,T_{\textnormal{1y}}),\dots,P^c(t,T_{\textnormal{10y}})$\\ 
      12Y & 2.4625 & Swap & 12 & $P^c(t,T_{\textnormal{1y}}),\dots,P^c(t,T_{\textnormal{12y}})$\\ 
      15Y & 2.5815 & Swap & 15 & $P^c(t,T_{\textnormal{1y}}),\dots,P^c(t,T_{\textnormal{15y}})$\\ 
      20Y & 2.6950 & Swap & 20 & $P^c(t,T_{\textnormal{1y}}),\dots,P^c(t,T_{\textnormal{20y}})$\\ 
      25Y & 2.7470 & Swap & 25 & $P^c(t,T_{\textnormal{1y}}),\dots,P^c(t,T_{\textnormal{25y}})$\\ 
      30Y & 2.7720 & Swap & 30 & $P^c(t,T_{\textnormal{1y}}),\dots,P^c(t,T_{\textnormal{30y}})$\\ 
      40Y & 2.7790 & Swap & 40 & $P^c(t,T_{\textnormal{1y}}),\dots,P^c(t,T_{\textnormal{40y}})$\\ 
      50Y & 2.7651 & Swap & 50 & $P^c(t,T_{\textnormal{1y}}),\dots,P^c(t,T_{\textnormal{50y}})$\\ 
       \hline
    \end{tabular}
    \caption{Quoted USD OIS (\%) on May 29 2015 (Source: Bloomberg). $^{\ast}$The 18m OIS swap convention has an upfront short stub, i.e., each leg has two coupons: the first with an accrual period of 6M and the second with an accrual period of 12m (1m).}
    \label{OISQuotesBloomberg}
    \end{table}
    
    \paragraph{Bootstrapping}
    First, we see that using the overnight loan (ON) and tommorow/next (TN) loans we could calculate $P^c(t,T_0)$, since
    \begin{equation}
        P^c(t,T_0)=P^c(t,t+2)=P^c(t,t+1)P^c(t+1,t+2)=\frac{1}{(1+\tau(t,t+1)\textnormal{ON})} \cdot \frac{1}{(1+\tau(t+1,t+2)\textnormal{TN})}.
    \end{equation}
    Now, for the swaps with only one payment at the maturity the equation \eqref{nCouponOISBootstrapping} reduces to
    \begin{equation}\label{1CouponOISBootstrapping}
        P^c(t,T_X)=\frac{P^c(t,T_0)}{1+k_X\tau(T_0,T_X)},
    \end{equation}
    this expression is valid for maturities up to one year (see table \ref{OISQuotesBloomberg}), i.e. for $X=\textnormal{1w, 2w,}$ $\textnormal{3w, 1m,}$ $\dots,$ $\textnormal{11m, 1y}$. Then, for maturities $X=\textnormal{18m, 2y, 3y,}\dots,\textnormal{9y, 10y}$ we use equation \eqref{nCouponOISBootstrapping} and a forward substitution to obtain the values of $P^c(t,T_X)$ for all $X$. So far we have obtained discount factors without using the bootstrapping method. However it is easy to see that equation \eqref{nCouponOISBootstrapping} for the maturity of 12y has two unknown variables: $P^c(t,T_{\textnormal{11y}})$ and $P^c(t,T_{\textnormal{12y}})$. That leads us to a single equation with two variables, we use the bootstrapping method to resolve this issue. The idea of bootstrapping is to make an initial guess for the variable $P_0^c(t,T_{\textnormal{12y}})$\footnote{The subscript number 0 indicates that the number is an initial condition or an initial guess. If the subscript is a number $i$, then it would indicates the number of iteration.} and with an interpolation method calculate the value of $P_0^c(t,T_{\textnormal{11y}})$. It is important to mention that in the financial industry is commonly seen that the initial guess and interpolation is done in zero coupon rates. Therefore, for the case of maturity of 12y, first we make an initial guess for the zero coupon rate $R_0^c(t,T_{\textnormal{12y}})$ and then calculate $R_0^c(t,T_{\textnormal{11y}})$ by interpolation. We calculate the discount factors $P_0^c(t,T_{\textnormal{11y}})$ and $P_0^c(t,T_{\textnormal{12y}})$ using these zero rates. Then we assume that $P_0^c(t,T_{\textnormal{11y}})$ is ``correct''\footnote{We write ``correct'' since this value depends directly in the assumptions of the model, the instruments involved in it and naturally the interpolation. Recall that many choices of interpolation functions are possible and, according to the nature of the problem, we impose requirements such as: continuity, differentiability, twice differentiability, conditions at the boundary, and so on.} and we calculate the ``real'' value of $P_1^c(t,T_{\textnormal{12y}})$ using \eqref{nCouponOISBootstrapping}, i.e.,
    \begin{equation}\label{12CouponOISBootstrapping}
        P_1^c(t,T_{\textnormal{12y}})=\frac{P^c(t,T_0)-k_{\textnormal{12y}}\sum_{i=1}^{10} \tau(T_{(i-1){\textnormal{y}}},T_{i\textnormal{y}}) P^c(t,T_i)-k_{\textnormal{12y}} \tau(T_{{\textnormal{10y}}},T_{\textnormal{11y}}) P_0^c(t,T_{\textnormal{11y}}) }{1+k_{\textnormal{12y}}\tau(T_{\textnormal{11y}},T_{\textnormal{12y}})}.
    \end{equation}
    Next we calculate the new value of $R_1^c(t,T_{\textnormal{12y}})$ using $P_1^c(t,T_{\textnormal{12y}})$, also we calculate and $R_1^c(t,T_{\textnormal{11y}})$ again with the interpolation method. We iterate this procedure until it converges. For instance, the stop condition is given as
    \begin{equation}
         | P_{i+1}^c(t,T_{\textnormal{12y}})- P_{i}^c(t,T_{\textnormal{12y}}) | < \epsilon, \hspace{10mm} \epsilon > 0,
    \end{equation}
    Additionally, it is convenient to include a maximum number of iteration, i.e., iterate until one of the following conditions is met,
    \begin{equation}
         | P_{i+1}^c(t,T_{\textnormal{12y}})- P_{i}^c(t,T_{\textnormal{12y}}) | < \epsilon, \hspace{2mm} \epsilon > 0 \hspace{4mm} \textnormal{or} \hspace{4mm} i > N_{\max} > 0.
    \end{equation}
    This bootstrapping method is used similarly for all the OISs maturities up to 50y. In the appendix \ref{App:AppendixC} we present a pseudocode for the bootstrapping of the USD OIS Curve.\par
    \begin{rem}
    	Yield rates or zero rates in this work are treated as continuously compounded rates, i.e.
        \begin{equation}
        	P(t,T)=e^{-\tau(t,T)R(t,T)}.
        \end{equation}
    \end{rem}
            Using the data displayed in table \ref{OISQuotesBloomberg} we build in R the OIS yield curve and the discount curve. The resulting calibrated curves are presented in the next charts \ref{OvernightYldCurve:Fig} and \ref{OvernightDiscCurve:Fig}, where additionally are included the curves that we obtained using the software SuperDerivatives, see table \ref{OIS:SuperDerivatives}.
        \begin{figure}[H]\label{OISyieldcurvechart1}
        \centering
        \vspace{-2mm}
        \includegraphics[scale=0.58]{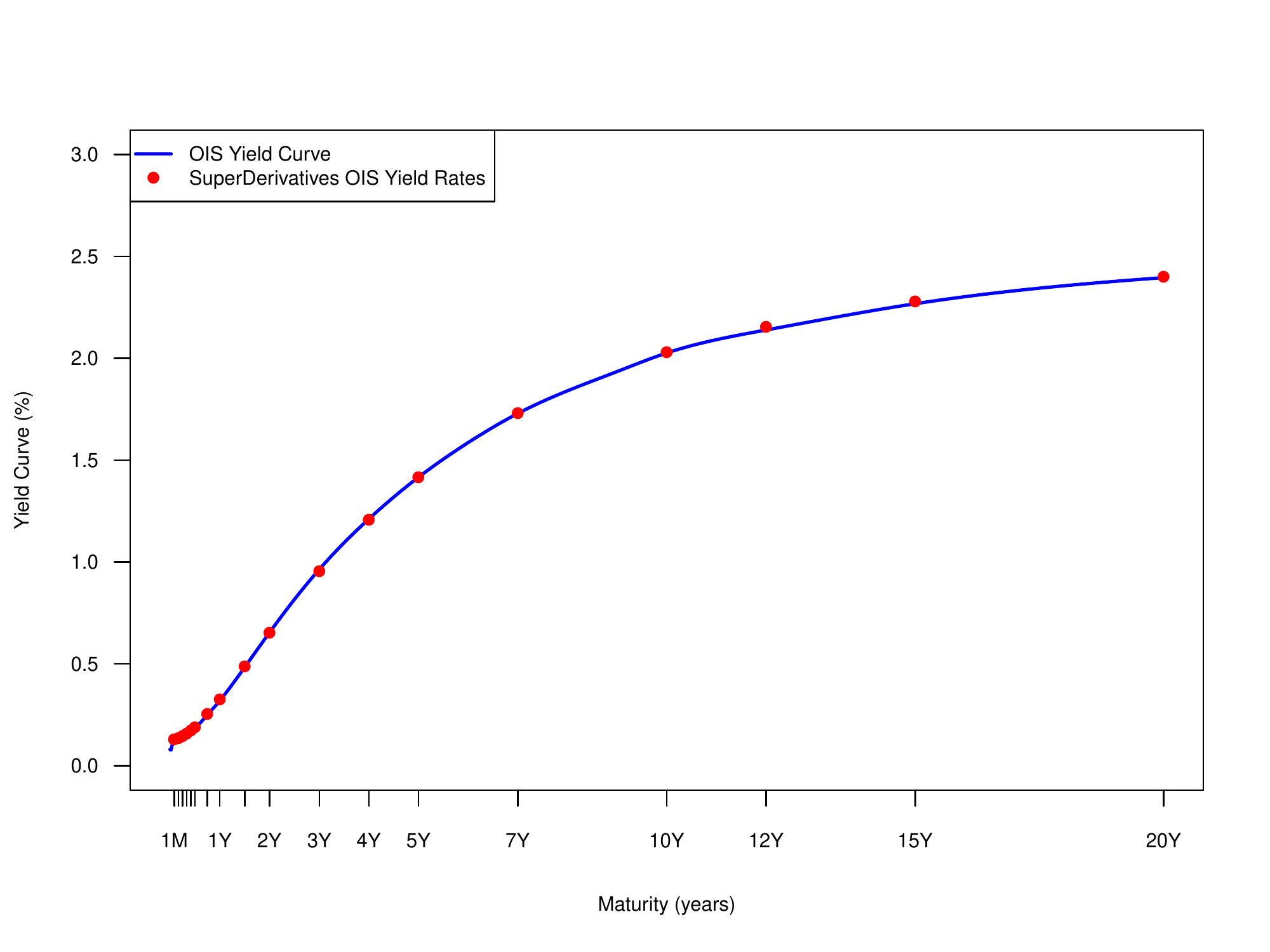}
        \caption[OIS Yield Curve]{OIS Yield Curve: $R(T)=R^c(t,T)$. The graph includes the yield rates of SuperDerivatives using the swap rates for the same date (May 29 2015).}
        \label{OvernightYldCurve:Fig}
        \vspace{-2mm}
        \end{figure}
        
        \begin{figure}[H]
        \centering
        \vspace{-2mm}
        \includegraphics[scale=0.58]{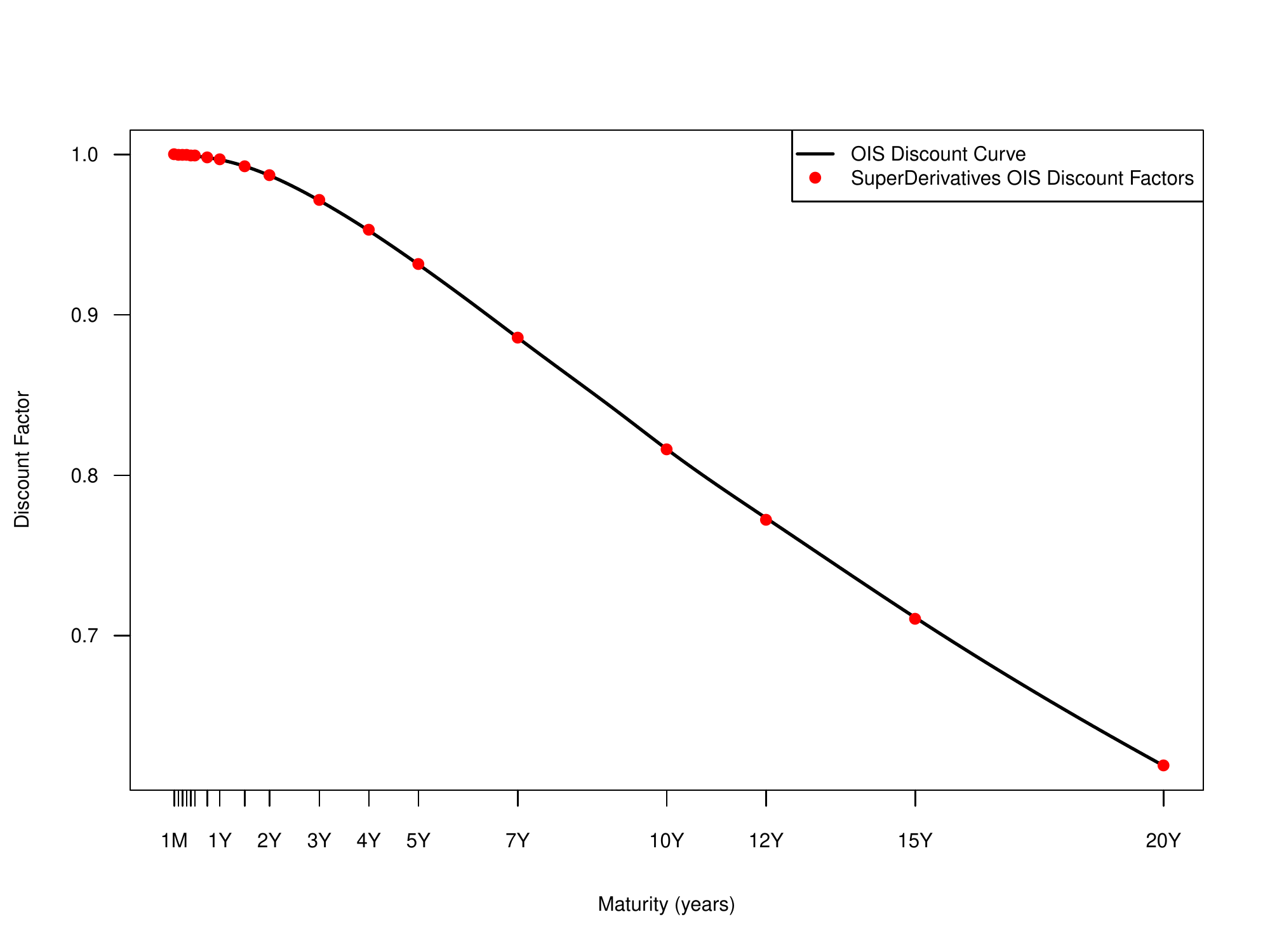}
        \caption[OIS Discount Curve]{OIS Discount Curve: $P(T)=P^c(t,T)$. The importance of this curve lies in that every USD dollar cash flow, inside a contract with CSA in USD, is discounted with it.}
        \label{OvernightDiscCurve:Fig}
        \vspace{-2mm}
        \end{figure}
  
    Recall the definition of instantaneous forward rate that is given by
    \begin{equation}
    \begin{aligned}
        f(t,T)&=-\frac{\partial}{\partial T}\big( \ln (P(t,T)) \big)\\
              &=- \lim_{\epsilon \rightarrow 0} \frac{\ln(P(t,T+\epsilon))-\ln(P(t,T))} {\epsilon}.\label{limit_instan_fwd}
    \end{aligned}
    \end{equation}
    Considering that the OIS yield curve $R^c(t,T))$ is given by a cubic splines function, then $R^c(t,T) \in \mathcal{C}^\infty$. Additionally, $P^c(t,T)$ is $\mathcal{C}^\infty$ function since $\mathcal{C}^\infty$ is closed under composition. Using the same argument $\ln(P^c(t,T))$ is $\mathcal{C}^\infty$ since $P^c(t,T)>0$ for all $T>t$. Therefore, the OIS instantaneous forward curve exists is defined by a piecewise function. Since the expression of this curve would be complicated, instead we calculate the daily forward curve $d(t,T)$. Using \eqref{limit_instan_fwd} we have that
    \begin{equation}
        d(t,T)=-\frac{\ln(P(t,T+1))-\ln(P(t,T))}{T+1-T}=\ln \bigg( \frac{P(t,T)}{P(t,T+1)} \bigg).
    \end{equation}
    In figure \ref{OvernightFwdRate:Fig} we present the daily forward overnight rate.
        \begin{figure}[H]
        \centering
        \vspace{-2mm}
        \includegraphics[scale=0.58]{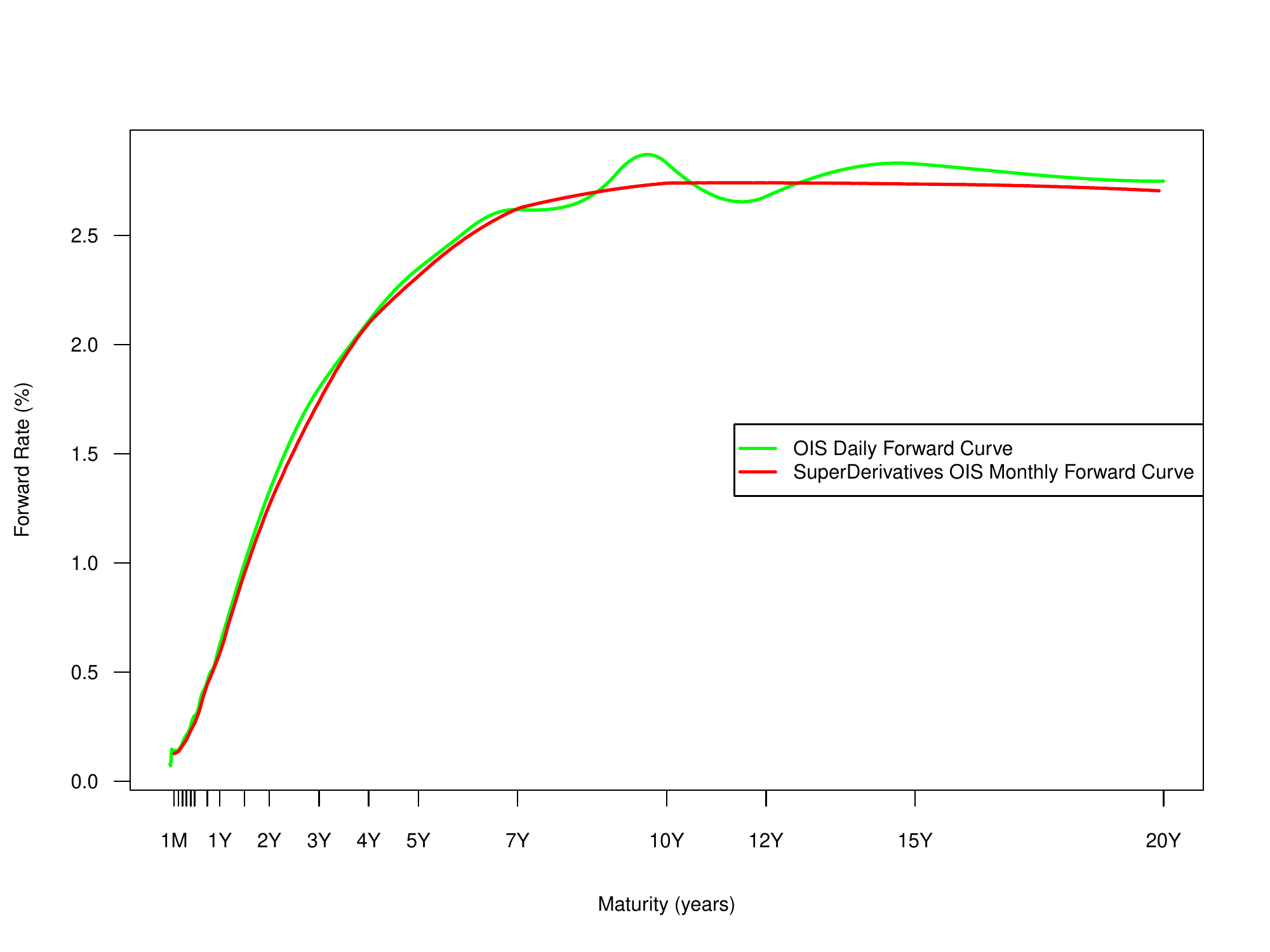}
        \caption[OIS 1d-Forward Curve]{OIS Daily Forward Curve. This curve is the implicit Fed Funds daily rate defined from the OIS discount curve. Since our model has to be consistent this forward rates guarantee that overnight index swaps have present value equal to zero.}
        \label{OvernightFwdRate:Fig}
        \vspace{-2mm}
        \end{figure}
    
    \paragraph{Forward Rates between FOMC Meetings Dates}
        The choice of which interpolation method we used will always be subjective and needs to be decided on a case by case basis. Indeed, the interpolation methods define the quality of the curve, particularly of the forward curve. Therefore for pricing IRSs or OISs we need a good fit and quality of the forward curve since with this curve we will project the future rates levels. In the case of overnight rates, we have that these rates are relatively constant between dates of \emph{Monetary Policy} meetings. For instance in the USD market, there are 8 meetings on each year and are performed by the Federal Open Market Committee (FOMC) in which they publish the target range (min-max) of the federal funds rate. Likewise, Banco de México has 8 meetings in a year in which they publish a target rate for the interbank overnight rate.\par \medskip
        Given the fact overnight rates, in particular the fed fund, are constant within a period of time some models are fitted to capture short-term monetary policy decisions. In \cite{clarkeconstructing}, Justin Clarke used an additive seasonality adjustment for building the short dated portion of the OIS curve. The idea is simple, he used market OIS quotes for pricing OISs with tenors equal the FOMC meeting dates and then keep constant these rates between meetings; the resulting forward curve is discontinuous in some short dated FOMC meetings. The longer term of the curve is calibrated using the simple bootstrapping presented above.\par \medskip
        There are other models that use Federal Fund Futures to capture the probability of a rate cut or hike, see \cite{robertson1997using}, \cite{kuttner2001monetary} and \cite{labuszewski2014fed}. However, for pricing purposes we could define a finite set of \emph{Monetary Policy} scenarios with a probability that corresponds to a economical view. In figure \ref{FOMCscenario}, we present an example of a scenario of FOMC rate hikes. This scenario was defined on May 29 2015 with five 2015-FOMC meetings to go. The idea of the model is to price an IRSs for each scenario and then calculated the expected value using the given probabilities.
        
        \begin{figure}[H]
        \centering
        \vspace{-2mm}
        \includegraphics[scale=0.58]{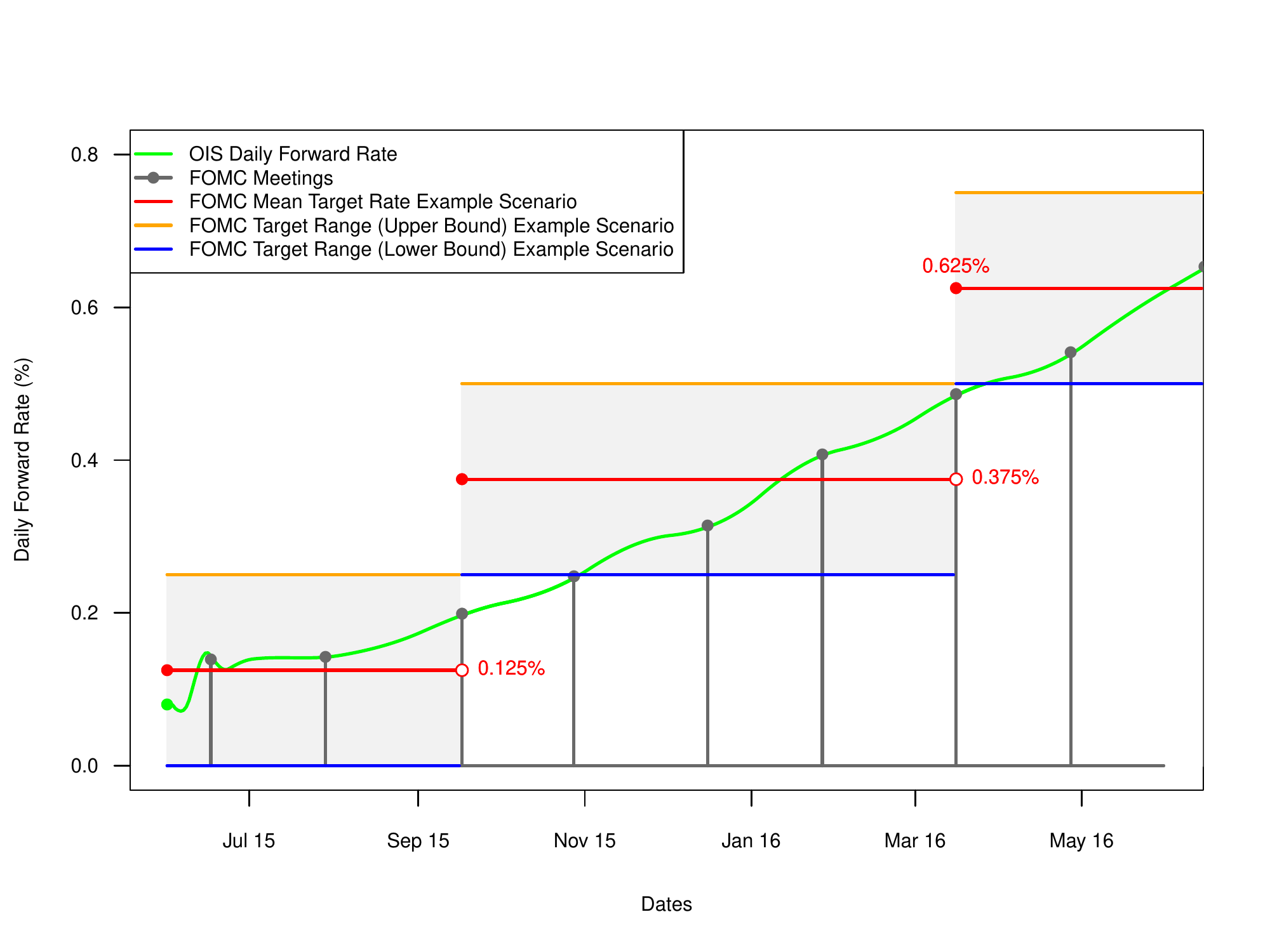}
        \caption[FOMC Meetings, OIS forward curve and FOMC monetary policy scenario]{FOMC Meetings and OIS forward curve and FOMC monetary policy scenario. In this chart we present the OIS forward curve calibrated when we perform a bootstrapping as in section. We also include a scenario of FOMC monetary policy decisions.}
        \label{FOMCscenario}
        \vspace{-2mm}
        \end{figure}
    
    \subsubsection{Calibration of the USD forward curve}\label{USDForwardCurveSection}
    In the previous section we explained the methodology used for the construction of the USD discounting curve. We remind that this curve is extremely important because every cash flow denominated in USD currency is discounted with this curve on the assumption that the CSA agreement is in USD currency. In this section we present the methodology for building index forward curves. We only describe how to estimate the forward curve for LIBOR 3m and LIBOR 1m. For the construction of the forward index LIBOR 3m curve we use the plain vanilla IRS and for the forward index LIBOR 1m we use plain vanilla IRSs for the short-term of the curve and TS for the longer part. In tables , we present the maturities and swap rates that we will use for the forward curve calibration.\par \medskip
    Let $\textnormal{PV}(t)$ be the present value of a payer IRS denominated in USD based on LIBOR 3m with maturity of $y$ years, hence
    \begin{equation}\label{payerLIBOR3M}
        \textnormal{PV}(t)= \sum_{i=1}^{Q_y} \alpha(t_{i-1},t_{i}) \mathbb{E}_t^{t_i}(\textnormal{\textbf{LIBOR3M}}(t_{i-1},t_{i})) P^{c}(t,\widetilde{t}_i) - k_y\sum_{j=1}^{S_y} \beta(s_{j-1},s_{j}) P^{c}(t,\widetilde{s}_j),
    \end{equation}
    where:\par \smallskip
    \noindent \makebox[4cm][l]{$k_y$:} fixed rate of the plain vanilla interest rate swap with maturity in $y$ years\par
    \noindent \makebox[4cm][l]{$Q_y$:} number of quarters in $y$ years \par
    \noindent \makebox[4cm][l]{$S_y$:} number of semesters in $y$ years \par
    \noindent \makebox[4cm][l]{$t_i,\widetilde{t}_i$:} coupon periods (start date, end date and payment date) for the floating leg\par 
    \noindent \makebox[4cm][l]{$s_j,\widetilde{s}_j$:} coupon periods (start date, end date and payment date) for the fixed leg\par 
    \noindent \makebox[4cm][l]{$\alpha(t_{i-1},t_i))$:} accrual factor of the $i$th coupon of the floating leg (ACT/360)\par
    \noindent \makebox[4cm][l]{$\beta(s_{j-1},s_j))$:} accrual factor of the $j$th coupon of the fixed leg (30/360)\par
    \noindent \makebox[4cm][l]{$P^{c}(t,x)$:} discount factors collateralized in USD\par
    \noindent \makebox[4cm][l]{$\mathbb{E}_t^{t_i}(\textnormal{\textbf{LIBOR3M}}(t_{i-1},t_{i}))$:} the LIBOR 3m forward rate of the $i$th coupon. \par \medskip
    This swap with the above characteristics is the \emph{most} plain vanilla IRS in the USD market. It exchanges LIBOR 3m payable quarterly versus a semiannual fixed rate coupon with a day count of 30/360 that is used primarily in government bonds. We say that it is the most plain vanilla since is the most standard and liquid swap in the market.\par\medskip
    Now, let us write $\mathbb{E}_t^{t_i}(\textnormal{\textbf{LIBOR3M}}(t_{i-1},t_{i}))$ from equation \eqref{payerLIBOR3M} in terms of a \emph{discount}\footnote{Note that the discount curve of LIBOR 3m is valid for discounting flows when the collateral rate is the LIBOR 3m index. It is important to remind that for uncollateralized deals  typically the discounting curve is the one based on LIBOR 3m i.e. $P^{3m}(t,x)$.} curve, i.e.
    \begin{equation}\label{LIBOR3MCoupon}
    	\mathbb{E}_t^{t_i}(\textnormal{\textbf{LIBOR3M}}(t_{i-1},t_{i})) = 
        \bigg( \frac{1}{\tau(t_{i-1},t_i)} \bigg( \frac{P^{3m}(t,t_{i-1})}{P^{3m}(t,t_{i})} - 1 \bigg) \bigg),
    \end{equation}
    where $\tau(t_{i-1},t_i)=\tau_i$ is the day count convention to determine the year fraction for discounting and curve building.\par \medskip
    Substituting equation \eqref{LIBOR3MCoupon} into \eqref{payerLIBOR3M} and solving for $P^{3m}(t,t_{Q_y})$ yields
    \begin{equation}\label{bootstrappingLIBOR3M}
    	P^{3m}(t,t_{Q_y})=\frac{P^{3m}(t,t_{Q_y-1})}{1+\frac{\tau_i \big( k_y\sum_{j=1}^{S_y} \beta_j P^{c}(t,\widetilde{s}_j) - \sum_{i=1}^{Q_y-1} \frac{\alpha_i}{\tau_i} \big( \frac{P^{3m}(t,t_{i-1})}{P^{3m}(t,t_{i})} - 1 \big) P^{c}(t,\widetilde{t}_i)\big)}{\alpha_i P^c(t,\widetilde{t}_{Q_y})}}
    \end{equation}
    This equation allows us to find the discount curve based on LIBOR 3m using a simple bootstrapping and an interpolation method. The equation \eqref{bootstrappingLIBOR3M} seems more complicated than the previous bootstrapping equations nevertheless the iterative process is the same as in OIS-USD curve or as in \emph{single-curve framework}.\par \medskip
    Once we have calibrated the forward LIBOR 3m curve we are able to build the LIBOR 1m forward curve. For the IRSs with maturities less than (or equal) 1 year, we could proceed for the curve construction exactly as in LIBOR 3m. Indeed, IRSs based on LIBOR 1m are quoted in the market for these maturities and the plain vanilla convention is: the fixed leg has annual payments (ACT/360); the floating leg has monthly payments with same day count convention. However for maturities greater than 1 year we have to use TSs quoted in the market that exchanges LIBOR 1m for LIBOR 3m. Let us present the bootstrapping equations for both cases.\par \medskip
    Suppose we enter into a payer short-term IRS (maturities of 2m, 3m, 4m, $\dots$, 12m) based on LIBOR 1m. For these swaps the fixed leg has only one coupon at maturity. The present value of the IRS is given by
    \begin{equation}\label{payerLIBOR1M}
        \textnormal{PV}(t)= \sum_{i=1}^{M} \alpha(t_{i-1},t_{i}) \mathbb{E}_t^{t_i}(\textnormal{\textbf{LIBOR1M}}(t_{i-1},t_{i})) P^{c}(t,\widetilde{t}_i) - k_M \beta(t_0,t_M) P^{c}(t,\widetilde{t}_M),
    \end{equation}
    where:\par \smallskip
    \noindent \makebox[4cm][l]{$k_M$:} fixed rate of the plain vanilla IRS with $M$ coupons\par
    \noindent \makebox[4cm][l]{$M$:} number of months (and coupons in floating leg) of the IRS \par
    \noindent \makebox[4cm][l]{$t_i,\widetilde{t}_i$:} coupon periods for the floating leg\par 
    \noindent \makebox[4cm][l]{$t_0,t_M$:} start date and end date of the IRS\par
    \noindent \makebox[4cm][l]{$\alpha(t_{i-1},t_i))$:} accrual factor of the $i$th coupon of the floating leg (ACT/360)\par
    \noindent \makebox[4cm][l]{$\beta(t_{0},t_M))$:} accrual factor of the fixed leg coupon (ACT/360)\par
    \noindent \makebox[4cm][l]{$P^{c}(t,x)$:} discount factors collateralized in USD\par
    \noindent \makebox[4cm][l]{$\mathbb{E}_t^{t_i}(\textnormal{\textbf{LIBOR1M}}(t_{i-1},t_{i}))$:} the LIBOR 1m forward rate of the $i$th coupon. \par \medskip
    Note that by setting equation \eqref{payerLIBOR1M} equal to zero and solving for $\mathbb{E}_t^{t_M}(\textnormal{\textbf{LIBOR1M}}(t_{M-1},t_{M}))$ yields,
    \begin{equation}\label{eq:libor1mfwdsubs}
    \mathbb{E}_t^{t_M}(\textnormal{\textbf{LIBOR1M}}(t_{M-1},t_{M}))=\frac{k_M \beta_{0,M} P^{c}(t,\widetilde{t}_M)-\sum_{i=1}^{M-1} \alpha_i \mathbb{E}_t^{t_i}(\textnormal{\textbf{LIBOR1M}}(t_{i-1},t_{i})) P^{c}(t,\widetilde{t}_i)}{\alpha_M P^{c}(t,\widetilde{t}_M)}
    \end{equation}
    Now, for maturity of 2m the forward rate of LIBOR 1m is easily determined since all the right side of equation is known. Taking advantage that the quotes of IRSs based on LIBOR 1m are available for every month then we could get all forward LIBOR 1m rates up to one year by forward substitution using equation \eqref{eq:libor1mfwdsubs}. Now for maturities greater than one year we use TSs quoted in the market. Recall that the most popular TSs in the USD market are: 1m vs 3m, 3m vs 6m and 3m vs 12m. All of them are traded against LIBOR 3m since it is the most liquid and traded IRSs, see table \ref{TSSuperDerivatives}. The TS spread should be positive, hence it is required to be added in the shorter tenor leg. Also, the payment frequency is determined by the longer tenor. In this work we only focus on the LIBOR 1m vs LIBOR 3m TS.\par\medskip
    
        \begin{table}
    \footnotesize
    \centering
    \begin{tabular}{|c c c c c|}
      \hline
     \multicolumn{1}{|>{\centering\arraybackslash}m{10mm}}{\textbf{Tenor}} & 
     \multicolumn{1}{>{\centering\arraybackslash}m{20mm}}{\textbf{End Date}} &
     \multicolumn{1}{>{\centering\arraybackslash}m{25mm}}{\textbf{1m vs 3m}} & 
     \multicolumn{1}{>{\centering\arraybackslash}m{25mm}}{\textbf{3m vs 6m}} &
     \multicolumn{1}{>{\centering\arraybackslash}m{25mm}|}{\textbf{3m vs 12m}}\\ \hline
     
    1Y & 2016-06-02 & -10.5000 & 8.3750 & 24.5000 \\
    2Y & 2017-06-02 & -12.1250 & 7.8750 & 23.2500 \\
    3Y & 2018-06-04 & -13.1250 & 7.8750 & 23.0000 \\
    4Y & 2019-06-03 & -13.5000 & 7.7500 & 23.0000 \\
    5Y & 2020-06-02 & -13.7500 & 7.7500 & 22.6250 \\
    6Y & 2021-06-02 & -13.8000 & 7.8000 & 22.5000 \\
    7Y & 2022-06-02 & -13.6250 & 7.8750 & 22.2500 \\
    8Y & 2023-06-02 & -13.3000 & 8.1000 & 21.8750 \\
    9Y & 2024-06-03 & -12.9000 & 8.4000 & 21.6250 \\
    10Y & 2025-06-02 & -12.7550 & 8.6250 & 21.3750 \\
    12Y & 2027-06-02 & -12.1000 & 9.1000 & 21.2500 \\
    15Y & 2030-06-03 & -11.6250 & 9.3750 & 21.5000 \\
    20Y & 2035-06-04 & -11.7500 & 9.5000 & 22.0000 \\
    25Y & 2040-06-04 & -11.7000 & 10.2500 & 21.7500 \\
    30Y & 2045-06-02 & -12.0000 & 9.7500 & 21.6250 \\
    40Y & 2055-06-02 & -12.6010 & 8.7490 & 21.3750 \\
    50Y & 2065-06-02 & -13-2020 & 7.7480 & 21.1240 \\
    60Y & 2075-06-03 & -13.8020 & 6.7460 & 20.8740 \\
    \hline
    \end{tabular}
    \caption{SuperDerivatives market data of the USD Tenor Swaps (1mv3m, 3mv6m, 3mv12m) (see Section \ref{USDForwardCurveSection}). Quotes are End of Day prices from May 29, 2015. The quotes were taken from \url{www.superderivatives.com} on June 21, 2015.}
    \label{TSSuperDerivatives}
    \end{table}
    
    Let $\textnormal{PV}(t)$ be the present value of a payer TS that exchanges LIBOR 1m with LIBOR 3m. Since the longer tenor is three months then both legs pay in a quarterly basis the coupons of the TS with the same business day calendar and the same day count convention ACT/360. Considering that the LIBOR 1m fixing is typically payable every month, we need to compound this monthly payments and pay them quarterly\footnote{There exist various types for compounding rates, in the plain vanilla TS the market convention is to use the compounding with simple spread i.e. compound the rates and then apply the spread. The other two commonly used types of spread are: 1) Compounding with spread in which the spread is applied to the rate and then we compound them; 2) Flat compounding in which we applied the spread before the compounding and also we consider the interests of the previous coupons while compounding. See \cite{ISDACompounding}.}. Hence, following equations \eqref{generalpayerTS}-\eqref{generalTSlegs} of section \ref{sec:IRproducts} we have that the present value is given by
     \begin{equation}\label{payerTS1mv3m}
        \begin{split}
        \textnormal{PV}(t)&= \sum_{i=1}^{M} \alpha(t_{i-1},t_{i}) \mathbb{E}_t^{t_i}(\textnormal{\textbf{LIBOR3M}}(t_{i-1},t_{i})) P^{c}(t,\widetilde{t}_i)\\
        &- \sum_{i=1}^{M} \alpha(t_{i-1},t_{i})\bigg[ \frac{1}{\sum_{j=1}^{N_i}\beta(s_{j-1},s_j)} \bigg( \prod_{j=1}^{N_i} \Big(1+\beta(s_{j-1},s_j)\mathbb{E}_t^{s_j}(\textnormal{\textbf{LIBOR1M}}(s_{j-1},s_{j}))\Big) \bigg)\\
        &+ B_M\bigg] P^{c}(t,\widetilde{t}_i)
        \end{split}
    \end{equation}
    where:\par \smallskip
    \noindent \makebox[4.9cm][l]{$B_M$:} TS spread\par
    \noindent \makebox[4.9cm][l]{$M$:} number of quarters of the TS\par
    \noindent \makebox[4.9cm][l]{$N_i$:} number of months in the $i$th quarter\par
    \noindent \makebox[4.9cm][l]{$t_i,\widetilde{t}_i$:} coupon periods for both legs\par 
    \noindent \makebox[4.9cm][l]{$t_0,t_M$:} start date and end date of the TS\par
    \noindent \makebox[4.9cm][l]{$\alpha(t_{i-1},t_i))$:} accrual factor of the $i$th coupon of the TS (ACT/360)\par
    \noindent \makebox[4.9cm][l]{$\beta(s_{j-1},s_j)$:} accrual factor of the $j$th monthly coupon (ACT/360)\par
    \noindent \makebox[4.9cm][l]{$P^{c}(t,x)$:} discount factors collateralized in USD\par
    \noindent \makebox[4.9cm][l]{$\mathbb{E}_t^{t_i}(\textnormal{\textbf{LIBOR3M}}(t_{i-1},t_{i}))$:} the LIBOR 3m forward rate of the $i$th quarter\par
    \noindent \makebox[4.9cm][l]{$\mathbb{E}_t^{s_j}(\textnormal{\textbf{LIBOR1M}}(s_{j-1},s_{j}))$:} the LIBOR 1m forward rate of the $j$th month.\par\medskip
Note that $\sum_{j=1}^{N_i}\beta(s_{j-1},s_j)=\alpha(t_{i-1},t_{i})$ since the monthly and quarterly calendars have the same conventions. Therefore if we assume that the tenor spread $B_M$ is a mid-market quote, then by no-arbitrage arguments we have that the present value of the TS is equal zero. Setting equation \eqref{payerTS1mv3m} equal to zero we get,
     \begin{equation}\label{payerTS1mv3mboot}
        \begin{split}
        \sum_{i=1}^{M}& \alpha_i \mathbb{E}_t^{t_i}(\textnormal{\textbf{LIBOR3M}}(t_{i-1},t_{i})) P^{c}(t,\widetilde{t}_i)\\
        &- \sum_{i=1}^{M} \alpha_i\bigg[ \frac{1}{\alpha_i} \bigg( \prod_{j=1}^{N_i} \Big(1+\beta_j\mathbb{E}_t^{s_j}(\textnormal{\textbf{LIBOR1M}}(s_{j-1},s_{j}))\Big) \bigg) + B_M\bigg] P^{c}(t,\widetilde{t}_i)=0\\
        \Longrightarrow \hspace{3mm} \sum_{i=1}^{M}& \alpha_i\bigg[ \mathbb{E}_t^{t_i}(\textnormal{\textbf{LIBOR3M}}(t_{i-1},t_{i}))\\
        &- \frac{1}{\alpha_i} \bigg( \prod_{j=1}^{N_i} \Big(1+\beta_j\mathbb{E}_t^{s_j}(\textnormal{\textbf{LIBOR1M}}(s_{j-1},s_{j}))\Big) \bigg) - B_M\bigg] P^{c}(t,\widetilde{t}_i)=0
        \end{split}
    \end{equation}
    This equation will help us to get the forward rates of the LIBOR 1m index. Indeed, the unknown values of equation \eqref{payerTS1mv3mboot} are only the forwards of LIBOR 1m index since, as we saw in the previous subsections, $P^c(t,x)$ values were calibrated using OISs and the forward rates of LIBOR 3m index were calibrated using the plain vanilla IRSs. To calibrate the forward rates of LIBOR 1m again we will use an interpolation method and the bootstrapping algorithm. It is important to remind that the interpolation method that we use in this work is natural cubic splines and is applied in the yield curves. Then it is convenient to write the unknown forward rates of LIBOR 1m in terms of a yield curve $R^{1m}$, so using equation \eqref{assumptionForwards} we have that
        \begin{align}\label{LIBOR1MCoupon}
    	\mathbb{E}_t^{t_i}(\textnormal{\textbf{LIBOR1M}}(t_{i-1},t_{i})) &= 
        \bigg( \frac{1}{\tau(t_{i-1},t_i)} \bigg( \frac{P^{1m}(t,t_{i-1})}{P^{1m}(t,t_{i})} - 1 \bigg) \bigg)\\
        &= \bigg( \frac{1}{\tau(t_{i-1},t_i)} \bigg( \frac{\exp(-\tau(t,t_i)R^{1m}(t,t_i))}{\exp(-\tau(t,t_{i-1})R^{1m}(t,t_{i-1}))} - 1 \bigg) \bigg)\label{LIBOR1MCoupon2}
    \end{align}
    In this case the bootstrapping is straightforward because in equation \eqref{payerTS1mv3mboot} we can substitute every forward rate in terms of the yield curve $R^{1m}(t,x)$ and in terms of the coefficients of the interpolation algorithm. Consequently we will get an equation with only one variable $R^{1m}(t,\widetilde{t}_M)$ that helped by the bootstrapping or an efficient root-finding method we get the yield curve of the LIBOR 1m index.\par\medskip
    
    For example, when $M=6$ i.e. the TS has a maturity of 18 months, then in equation \eqref{payerTS1mv3mboot} we have the following six unknown variables: $\mathbb{E}_t^{s_j}(\textnormal{\textbf{LIBOR1M}}(s_{j-1},s_{j}))$ when $s_{j} \in \{ t_\textnormal{13m}, t_\textnormal{14m}, t_\textnormal{15m}, t_\textnormal{16m}, t_\textnormal{17m}, t_\textnormal{18m}\}$. Based on the assumption that the yield curve $R^{1m}$ follows the natural cubic splines conditions (see appendix \ref{app:naturalcubic}) we have that:
        \begin{enumerate}
    	\item $R^{1m}(t,x)=a+b(x-t)+c(x-t)^2+d(x-t)^3$ with $t_\textnormal{12m} \leq x \leq t_\textnormal{18m}$ and $a,b,c,d \in \mathbb{R}$
        \item $R^{1m}(t,t_\textnormal{12m})=r_{t_\textnormal{12m}}$
        \item $R^{1m}(t,x)\in\mathcal{C}^2$ with $\frac{d^2R^{1m}}{dx^2}(t,t_\textnormal{12m})=0$ and $\frac{d^2R^{1m}}{dx^2}(t,t_\textnormal{18m})=0$.
    \end{enumerate}
    Note that the value of $r_{t_\textnormal{12m}}$ is known since the IRSs up to 1 year can be bootstrapped easily as we said earlier. So, if we define an initial value for $R^{1m}(t,t_\textnormal{18m})$, say $r_{0}$, then we are able to find the values of $a,b,c,d$ since we have the following system of equations:
    \begin{align}
    a+b(t_\textnormal{12m}-t)+c(t_\textnormal{12m}-t)^2+d(t_\textnormal{12m}-t)^3&=r_{t_\textnormal{12m}}\label{eq1:sys} \\
        a+b(t_\textnormal{18m}-t)+c(t_\textnormal{18m}-t)^2+d(t_\textnormal{18m}-t)^3&=r_{t_\textnormal{18m}}\label{eq2:sys} \\
        2c+6d(t_{\textnormal{12m}}-t) &=0\label{eq3:sys} \\ 
        2c+6d(t_{\textnormal{18m}}-t) &=0\label{eq4:sys}
    \end{align}
    The solution of this system of equations is easy to get since equations \eqref{eq3:sys} and \eqref{eq4:sys} defines the values of $c$ and $d$. By substituting $c$ and $d$ in \eqref{eq1:sys} and \eqref{eq2:sys} we obtain the values of $a$ and $b$. With this values we are able to get $R^{1m}(t,x)$ for all $x \in (t,t_\textnormal{18m}]$. However, as we make an initial guess for $R^{1m}(t,t_\textnormal{18m})$ then we cannot guarantee that equation \eqref{payerTS1mv3mboot} holds, since we are not sure that the implied forwards 
$\mathbb{E}_t^{s_j}(\textnormal{\textbf{LIBOR1M}}(s_{j-1},s_{j}))$ when $s_{j} \in \{ t_\textnormal{13m}, t_\textnormal{14m}, t_\textnormal{15m}, t_\textnormal{16m}, t_\textnormal{17m}, t_\textnormal{18m}\}$ from the yield curve $R^{1m}(t,x)$ function, hold the conditions and assumptions of our calibration model. So the process that we have to follow to get a solution is the same that we applied for the calibration of the OIS-USD curve. Indeed, we will assume that the implied forward rates up to $t_\textnormal{17m}$ are the \emph{correct} forward rates, so we have just to find the value of $\mathbb{E}_t^{t_\textnormal{18m}}(\textnormal{\textbf{LIBOR1M}}(t_\textnormal{17m},t_\textnormal{18m}))$ that makes equation \eqref{payerTS1mv3mboot} equal to zero. Once we get this value (with the help of an equation root finding method, see \cite{burden2010numerical}) we are able to calculate the new value of $R^{1m}(t,t_\textnormal{18m})$ using equation \eqref{LIBOR1MCoupon2}. Using this zero rate we set a new system of equations with the idea of getting new values of $a,b,c,d$. Finally we calculate the new implied forward rates and we iterate this method until we converge for a solution. Again, as in the method for the calibration of the OIS-USD discount curve, we have to define a rule of convergence that not allow us to iterate 
indefinitely. So we include the next condition to the iterative process:
    \begin{equation}
         | R^{1m}_{i}(t,t_{\textnormal{18m}})- R^{1m}_{i+1}(t,t_{\textnormal{18m}}) | < \varepsilon, \hspace{2mm} \varepsilon > 0 \hspace{4mm} \textnormal{or} \hspace{4mm} i > N_{\max} > 0,
    \end{equation}
    The first condition guarantees us that the change of the zero rate $R^{1m}(t,t_{\textnormal{18m}})$ between two iterations is sufficiently small, whereas the second condition assures us to iterate at the most $N_{\max}$ times.\par\medskip
    Once that we have presented an example for the bootstrapping when the TS has $M=6$ coupons or a maturity length of 18 months, we have to introduce an efficient method for the calibration of the curve considering all the TSs maturities. The process is the same, however it operates along all TSs maturities with the idea of make more efficiently the number of iterations, without mentioning an efficient method for the solution of the coefficients of the cubic splines among tenors. Let us present the steps to follow for the calibration of the LIBOR 1m curve:
    \begin{enumerate}
    	\item Find the zero rates $R^{1m}(t,x)$ for the LIBOR 1m swap market (2 months up to 1 year) 
        \item Guess the initial values for $\{R^{1m}_0(t,t_N)\}$ where $t_N$ are all the maturities of the TS quotes
        \item With an interpolation method (natural cubic splines) calculate $R^{1m}_0(t,t_i)$ for all $i$th coupon date (every month) and get the implied forward LIBOR 1m rates $\{ \mathbb{E}_t^{t_i}(\textnormal{\textbf{LIBOR1M}}(t_{i-1},t_i)) \}_0$
        \item Insert these forward rates into the equation \eqref{payerTS1mv3mboot} and solve for $\{R^{1m}(t,t_N)\}$
        \item We take these new zero rates $\{R^{1m}_1(t,t_N)\}$ and again apply the interpolation method and calculate the new implied forwards rates
        \item Repeat Steps 3, 4 and 5 until the following condition is met: \begin{equation}
        \sum_{j=1}^N | R^{1m}_{i}(t,t_{j})- R^{1m}_{i+1}(t,t_{j}) | < \varepsilon, \hspace{2mm} \varepsilon > 0 \hspace{4mm} \textnormal{or} \hspace{4mm} i > N_{\max} > 0,
    \end{equation} 
    \end{enumerate}
    For the interpolation method we use the algorithm presented in appendix \ref{app:naturalcubic}.
     \begin{figure}[H]
        \centering
        \vspace{-2mm}
        \includegraphics[scale=0.58]{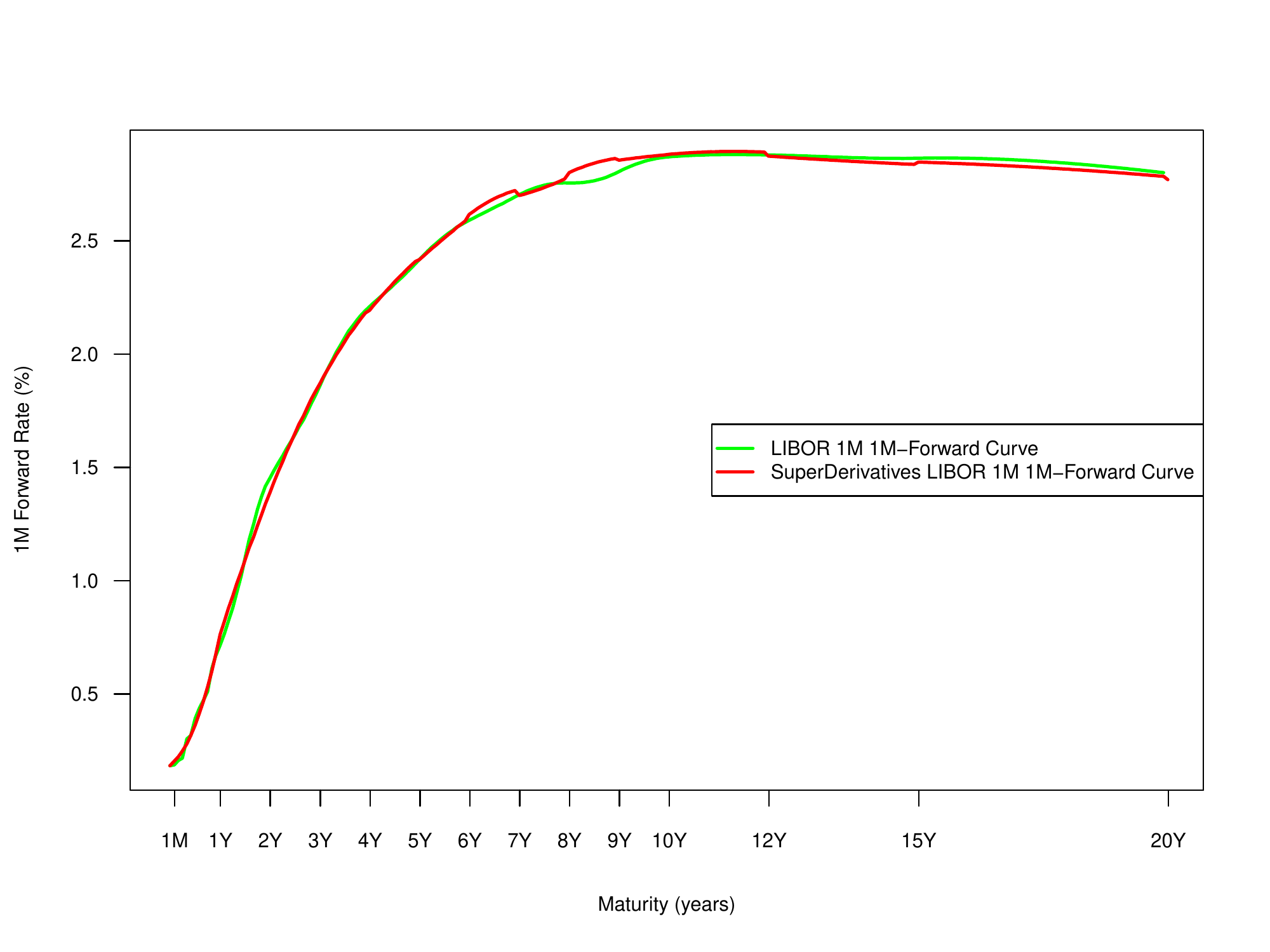}
         \caption[LIBOR 1m 1m-Forward Curve]{LIBOR 1m Forward Curve: $f(t,x)=
\mathbb{E}_t^{x}(\textnormal{\textbf{LIBOR1M}}(x,x+\textnormal{1m}))$. This curve give us the forward LIBOR 1m rate for any given day $x$ that is effective for the time interval $[x,x+\textnormal{1m}]$.}
        \vspace{-2mm}
     \end{figure}

        \begin{figure}[H]
        \centering
        \vspace{-2mm}
        \includegraphics[scale=0.58]{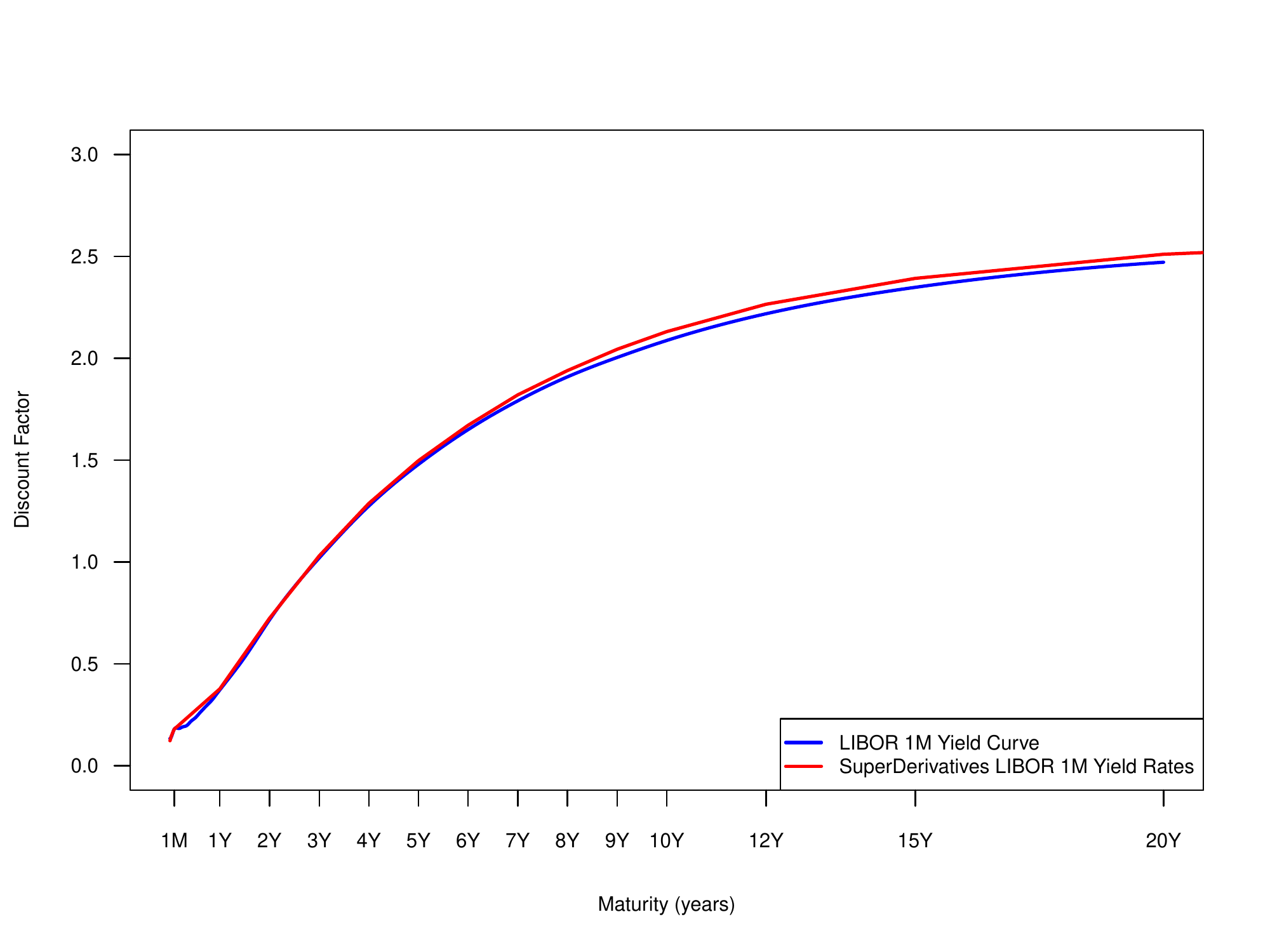}
        \caption[LIBOR 1m Yield Curve]{LIBOR 1m Yield Curve: $R^{1m}(T)=R^{1m}(t,T)$. This curve is used for the bootstrapping of the LIBOR 1m forward rates. $R^{1m}$ curve is built as a piecewise-defined function with the natural cubic splines interpolation method. Recall that the yields curves help us to obtain the forward rates based on an Ibor index rate.}
        \vspace{-2mm}
        \end{figure}
     
        \begin{figure}[H]
        \centering
        \vspace{-2mm}
        \includegraphics[scale=0.58]{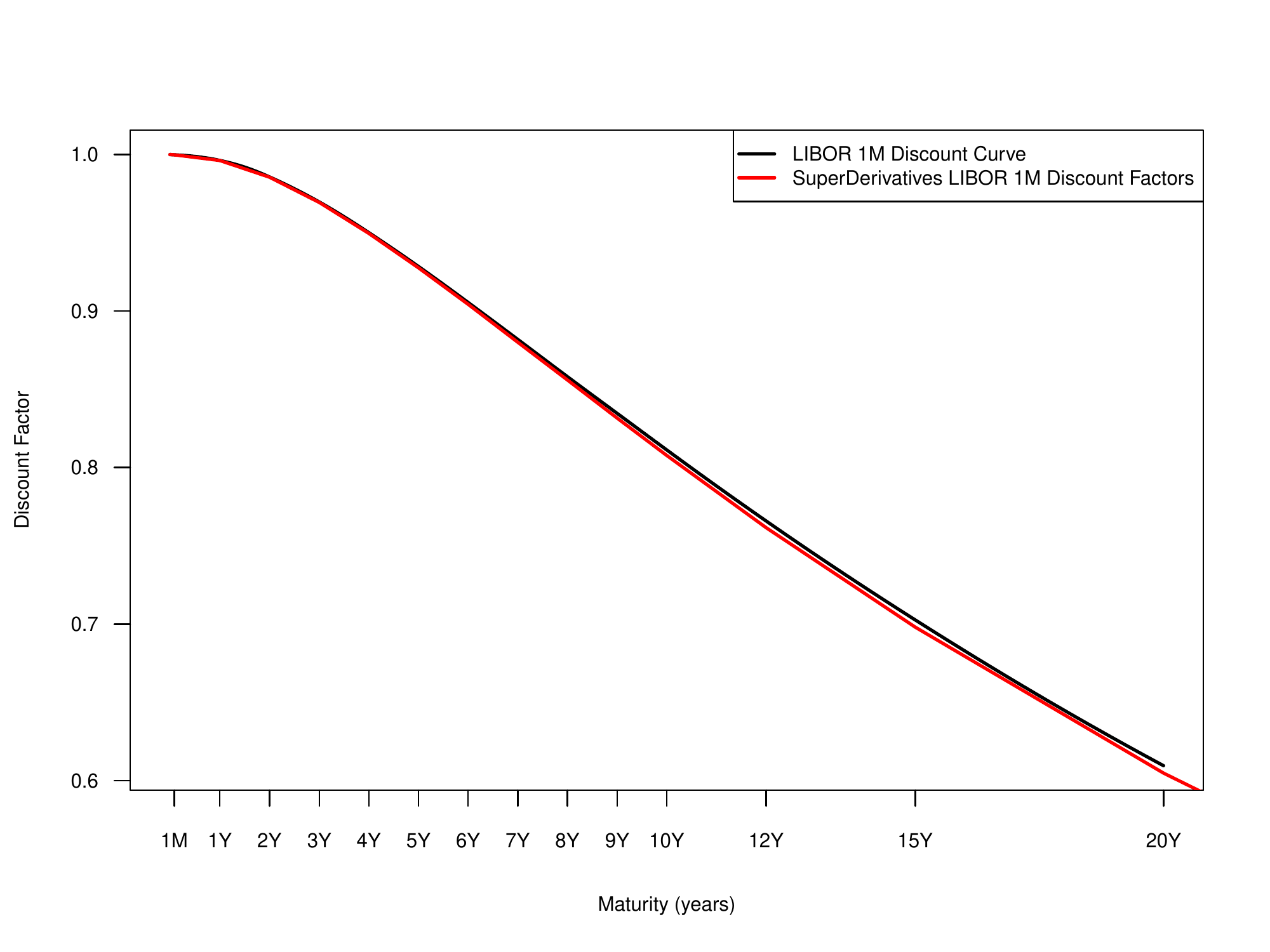}
        \caption[LIBOR 1m Discount Curve]{LIBOR 1m Discount Curve: $P^{1m}(T)=P^{1m}(t,T)$. In this work this curve is not used for discounting cash flows, however this curve helps us to calculate the 1m-forward and the instantaneous forward curve for LIBOR 1m index as in equation \eqref{LIBOR1MCoupon}.}
        \vspace{-2mm}
        \end{figure}

    \newpage
    
    \subsection{Valuation Framework in Multiple Currencies}
    We saw before that the collateral determines the discount rate, i.e. a USD swap which is collateralized (cash) in USD is priced using the OIS (Fed Funds) curve. Now we have to answer the following question: what if the swap is collateralized in EUR? And in JPY?\par\smallskip
    Let $V^{(j)}$ be the value collateral account for a derivative $X$ in terms of currency $(j)$. The stochastic process of the collateral account is given by
    \begin{equation}\label{eq:dV(j)}
        dV^{(j)}(s)=\big( r^{(j)}(s)-c^{(j)}(s) \big) V^{(j)}(s) ds + a(s)d\big( h^{(i)}(s) \cdot f^{i \rightarrow j}(s) \big),
    \end{equation}
    where, $r^{(j)}(s)$ and $c^{(j)}(s)$ are the risk-free rate and the collateral rate at time $s$ of currency $(j)$, respectively, $h^{(i)}(s)$ is the value of the derivative $X$ in terms of currency $(i)$ that matures at time $T$ with cash flow $h^{(i)}(T)$, $f^{i \rightarrow j}(s)$ is the foreign exchange (FX) rate at time $s$ representing the price of the unit amount of currency $(i)$ in terms of currency $(j)$. Finally, $a(s)$ represents the number of positions of the derivative at time $s$. To solve \eqref{eq:dV(j)} we have to multiply the equation by $e^{\int_s^T y^{(j)}(\eta) d\eta}$, where $y^{(j)}(s)=r^{(j)}(s)-c^{(j)}(s)$, and we get
    \begin{equation}\label{eq:expdV(j)}
        e^{\int_s^T y^{(j)}(\eta) d\eta} dV^{(j)}(s)=e^{\int_s^T y^{(j)}(\eta) d\eta} y^{(j)}(s) V^{(j)}(s) ds + e^{\int_s^T y^{(j)}(\eta) d\eta} a(s)d \big( h^{(i)}(s) \cdot f^{i \rightarrow j}(s) \big).
    \end{equation}
    Then, by integrating \eqref{eq:expdV(j)} we have
    \begin{equation}\label{eq:intexpdV(j)}
        \int_t^T e^{\int_s^T y^{(j)}(\eta) d\eta} dV^{(j)}(s)= \int_t^T e^{\int_s^T y^{(j)}(\eta) d\eta} y^{(j)}(s) V^{(j)}(s) ds + \int_t^T e^{\int_s^T y^{(j)}(\eta) d\eta} a(s)d \big( h^{(i)}(s) \cdot f^{i \rightarrow j}(s) \big).
    \end{equation}
    Let define $u=e^{\int_s^T y^{(j)}(\eta) d\eta}$ and $dv=dV^{(j)}(s)$, using integration by parts formula we obtain
    \begin{align}
        \int udv &= uv - \int vdu\\
        \int_t^T e^{\int_s^T y^{(j)}(\eta) d\eta} dV^{(j)}(s) &= e^{\int_s^T y^{(j)}(\eta) d\eta} V^{(j)}(s)\Big|_t^T + \int_t^T V^{(j)}(s)e^{\int_s^T y^{(j)}(\eta) d\eta} y^{(j)}(s) ds. \label{intByParts(j)}
    \end{align}
    Then, using \eqref{eq:intexpdV(j)} and \eqref{intByParts(j)} we get
    \begin{equation}\label{eq:intexpdV(j)2}
        V^{(j)}(T)=e^{\int_t^T y^{(j)}(\eta) d\eta}V^{(j)}(t) + \int_t^T e^{\int_s^T y^{(j)}(\eta) d\eta} a(s)d \big( h^{(i)}(s) \cdot f^{i \rightarrow j}(s) \big).
    \end{equation}
    As in \cite{fujii2010note}, by adopting the trading strategy specified by
    \begin{equation}
    \left\{
    \begin{aligned}
        V^{(j)}(t)&=h^{(i)}(t)\cdot f^{i \rightarrow j}(t)\\
        a(s)&=e^{\int_t^s y^{(j)}(\eta) d\eta}
    \end{aligned}
    \right.
    \end{equation}
    By substituting the trading strategy in \eqref{eq:intexpdV(j)2} we obtain
    \begin{align}
  V^{(j)}(T)&=e^{\int_t^T y^{(j)}(\eta) d\eta}V^{(j)}(t) + \int_t^T e^{\int_s^T y^{(j)}(\eta) d\eta} a(s) d \big( h^{(i)}(s) \cdot f^{i \rightarrow j}(s) \big)
 \nonumber \\
            &=e^{\int_t^T y^{(j)}(\eta) d\eta}h^{(i)}(t)\cdot f^{i \rightarrow j}(t) + \int_t^T e^{\int_s^T y^{(j)}(\eta) d\eta + \int_t^s y^{(j)}(\eta) d\eta} d \big( h^{(i)}(s) \cdot f^{i \rightarrow j}(s) \big)
 \nonumber \\
            &=e^{\int_t^T y^{(j)}(\eta) d\eta}h^{(i)}(t)\cdot f^{i \rightarrow j}(t) + e^{\int_t^T y^{(j)}(\eta) d\eta} \int_t^T d \big( h^{(i)}(s) \cdot f^{i \rightarrow j}(s) \big)
 \nonumber \\
            &=e^{\int_t^T y^{(j)}(\eta) d\eta}h^{(i)}(t)\cdot f^{i \rightarrow j}(t) + e^{\int_t^T y^{(j)}(\eta) d\eta} \bigg( \big( h^{(i)}(T) \cdot f^{i \rightarrow j}(T) \big) - \big( h^{(i)}(t) \cdot f^{i \rightarrow j}(t) \big) \bigg)
 \nonumber \\
            &=e^{\int_t^T y^{(j)}(\eta) d\eta}\big( h^{(i)}(T) \cdot f^{i \rightarrow j}(T) \big).
    \end{align}
    Now we have that,
    \begin{equation}
        h^{(i)}(T)=V^{(j)}(T)f^{j \rightarrow i}(T)e^{-\int_t^T y^{(j)} (\eta)d\eta}=V^{(i)}(T)e^{-\int_t^T y^{(j)} (\eta)d\eta}.
    \end{equation}
    We are able to compute the present value of $h^{(i)}(t)$ of the derivative $X$ using the risk neutral measure $\mathbb{Q}_{(i)}$ associated to the numéraire (money-market account) $B^{(i)}(t)=e^{\int_t^T r^{(i)}(s)ds}$,i.e.,
    \begin{align}
        h^{(i)}(t)&=\mathbb{E}_t^{\mathbb{Q}_{(i)}} \left[ \frac{V^{(i)}(T)}{B^{(i)}(T)} \right]
 \nonumber \\
                  &=\mathbb{E}_t^{\mathbb{Q}_{(i)}} \left[ \big( e^{-\int_t^T r^{(i)}(s)ds} \big) \big( e^{\int_t^T y^{(j)}(s)ds} \big) h^{(i)}(T) \right]
    \end{align}
    If we define $y^{(i,j)}(s)=y^{(i)}(s)-y^{(j)}(s)$ then we could express $h^{(i)}(t)$ as follows
    \begin{align}
        h^{(i)}(t)&=\mathbb{E}_t^{\mathbb{Q}_{(i)}} \left[ \big( e^{-\int_t^T r^{(i)}(s)ds} \big) \big( e^{\int_t^T y^{(j)}(s)ds} \big) h^{(i)}(T) \right]
 \nonumber \\
                  &=\mathbb{E}_t^{\mathbb{Q}_{(i)}} \left[ \big( e^{-\int_t^T c^{(i)}(s)ds} \big) \right] \mathbb{E}_t^{\mathbb{T}^c_{(i)}} \left[ \big( e^{-\int_t^T y^{(i,j)}(s)ds} \big) h^{(i)}(T) \right]
 \nonumber \\
                  &=P_{(i)}^c(t,T) \mathbb{E}_t^{\mathbb{T}^c_{(i)}} \left[ \big( e^{-\int_t^T y^{(i,j)}(s)ds} \big) h^{(i)}(T) \right], \label{eq:PriceCollDer2Curr}
    \end{align}
    where $P_{(i)}^c(t,T)$ is the collateralized zero coupon bond of currency $(i)$ and $\mathbb{T}^c_{(i)}$ is the collateralized forward measure of the same currency where $P_{(i)}^c(t,T)$ is used as numéraire. Note that when $(i)=(j)$ then from \eqref{eq:PriceCollDer2Curr} we have that $h^{(i)}(t)=P_{(i)}^c(t,T) \mathbb{E}_t^{\mathbb{T}^c_{(i)}} \left[ h^{(i)}(T) \right]$ which coincides with equation \eqref{eq:PriceCollDer1Curr}.\par \medskip
    
    In the next subsection we will present the differences between pricing IRSs denominated in EUR and MXN when the collateral currency is USD. We will see that the existence of an OISs market eases the curve construction.
    \subsubsection{Case of EUR}\label{caseOfEUR}
    We have already built curves in USD, therefore the next step is to build curves in other currency ---namely EUR--- but keeping the same collateral currency (in this case USD). The EUR interest rate market instruments are mostly collateralized in EUR. In this market the overnight rate is called Eonia (acronym of Euro OverNight Index Average), and the OISs based on Eonia have the same characteristics of the OISs based on Fed Funds in the USD market.\par \medskip    
    Therefore with a list of EUR cash deposits and OISs swaps based on Eonia we are able to build and calibrate ---using a simple bootstrapping--- the EUR discount curve when the cash flows are collateralized in EUR. In fact, the methodology is identical to the one used in section \ref{USDDiscountCurveSection} to build USD discount curve. Once we have the discount curve (collateralized in EUR) we can use quotes from the EUR market based on the two dominant tenors: EURIBOR 3m and EURIBOR 6m; to build both forward curves. For EURIBOR 3m curve we use short-term interest rates (STIR) futures and tenor basis swaps (6m vs 3m), while for EURIBOR 6m we use forward rates agreements (FRAs) and vanilla IRSs.
    \begin{rem}
    As in USD market, EUR curves (Eonia (discount) EURIBOR 3m (forward) and EURIBOR 6m (forward)) can be build one by one with simple bootstrappings. The order is discounting, forward 6m and forward 3m.
    \end{rem}
    Now to build the discount curves of EUR when the collateral is USD or vice versa we need the XCSs market. Recall that a plain vanilla EURUSD XCS exchanges LIBOR 3m flat for EURIBOR 3m plus an additional spread, as is mention on section \ref{XCSs}. In the interest rate market there exists two types of XCSs: cnXCSs (constant notional XCSs) and mtmXCSs (mark-to-market XCSs)\footnote{Also classified as: resettable XCSs or non-resetable XCSs, since the \emph{reset} of the notional amount.}, see section \ref{XCSs}.  In a cnXCS the notionals of both legs are fixed using a FX rate agreed in the inception of the trade and kept constant until its maturity despite the FX rate moves. On the contrary, a mtmXCS resets the notional of one leg at the start of every coupon period (including the final notional exchange) while the other leg notional keeps constant. Even though the mtmXCS have better liquidity in almost every currency, in this work we will asume that the XCSs have constant notionals due to USDMXN XCSs are still quoted with constant notionals. Therefore consider a cnXCS of a pair of currencies $(i,j)$ and assume that the collateral is posted in currency $j$. Then the net present value of leg $j$ is calculated as: (suppose that the payment dates and coupon accrual periods are the same for both legs)\footnote{This assumption is, in fact, true for plain vanilla XCSs, i.e. the payment calendars for both legs coincide, either the day count conventions.}
    \begin{equation}
    \begin{split}
		\textnormal{\textbf{Leg}}_{(j)}(t)&=\textnormal{N}_{(j)} \Big[ -P^c_{(j)}(t,\widetilde{t}_0) + P^c_{(j)}(t,\widetilde{t}_N) \\
        & \hspace{5mm}+ \sum_{k=1}^N \beta(t_{k-1},t_{k}) P^c_{(j)}(t,\widetilde{t}_k) \mathbb{E}^{\mathbb{T}^c_{(j),k}}_t \big[ L^{(j)}(t_{k-1},t_k) \big] \Big].
    \end{split}
    \end{equation}
    where:\par
    \noindent \makebox[3.8cm][l]{$\textnormal{N}_{(j)}$:} notional of $j-$currency leg \par
    \noindent \makebox[3.8cm][l]{$N$:} number of coupons\par 
    \noindent \makebox[3.8cm][l]{$(t_{k-1},t_k)$:} period of $k$th coupon\par 
    \noindent \makebox[3.8cm][l]{$\widetilde{t}_k$:} time of payment of the $k$th coupon or notional exchange \par 
    \noindent \makebox[3.8cm][l]{$\beta(t_{k-1},t_k)$:} accrual factor of the $k$th coupon\par
    \noindent \makebox[3.8cm][l]{$\mathbb{E}^{\mathbb{T}^c_{(j),k}}_t \big[ L^{(j)}(t_{k-1},t_k) \big]$:} the $j-$forward reference rate of the $k$th coupon \par
    \noindent \makebox[3.8cm][l]{$P^c_{(j)}(t,\widetilde{t}_k)$:} $j-$currency discount factor collateralized in $j$ for time $\widetilde{t}_k$. \par \medskip
    \noindent Now, using equation \eqref{eq:PriceCollDer2Curr} we have that the net present value of leg $i$ is
    \begin{equation}\label{eq:leg(i)}
    \begin{split}
		\textnormal{\textbf{Leg}}_{(i)}(t)&=\textnormal{N}_{(i)} \Big[ -P^c_{(i)}(t,\widetilde{t}_0) \mathbb{E}^{\mathbb{T}^c_{(i),0}}_t \big[ e^{-\int_t^{\widetilde{t}_0} y_{(i,j)}(s)ds} \big] \\
        & \hspace{5mm}+ P^c_{(i)}(t,\widetilde{t}_N) \mathbb{E}^{\mathbb{T}^c_{(i),N}}_t \big[ e^{-\int_t^{\widetilde{t}_N} y_{(i,j)}(s)ds} \big] \\
        & \hspace{5mm}+ \sum_{k=1}^N \beta(t_{k-1},t_{k}) P^c_{(i)}(t,\widetilde{t}_k) \mathbb{E}^{\mathbb{T}^c_{(i),k}}_t \big[ \big( L^{(i)}(t_{k-1},t_k) + B_N\big) e^{-\int_t^{\widetilde{t}_k} y_{(i,j)}(s)ds} \big] \Big].
    \end{split}
    \end{equation}
    where:\par
    \noindent \makebox[5.7cm][l]{$\textnormal{N}_{(i)}$:} notional of $i-$currency leg \par
    \noindent \makebox[5.7cm][l]{$N$:} number of coupons\par
    \noindent \makebox[5.7cm][l]{$B_N$:} basis spread of the cnXCS with $N$ coupons\par 
    \noindent \makebox[5.7cm][l]{$(t_{k-1},t_k)$:} period of $k$th coupon\par 
    \noindent \makebox[5.7cm][l]{$\widetilde{t}_k$:} time of payment of the $k$th coupon or notional exchange \par 
    \noindent \makebox[5.7cm][l]{$\beta(t_{k-1},t_k)$:} accrual factor of the $k$th coupon\par
    \noindent \makebox[5.7cm][l]{$P^c_{(i)}(t,\widetilde{t}_k)$:} $i-$currency discount factor collateralized in $i$ for time $\widetilde{t}_k$ \par
    \noindent \makebox[5.7cm][l]{$\mathbb{E}^{\mathbb{T}^c_{(i),k}}_t \big[ L^{(i)}(t_{k-1},t_k) e^{-\int_t^{\widetilde{t}_k} y_{(i,j)}(s)ds} \big]$:} the $i-$forward reference rate of the $k$th coupon when the collateral currency is $j$. \par \medskip
     \noindent Consider that we enter in a payer cnXCS i.e. we pay the leg with the basis spread $B_N$, hence the net present value at time $t$ of the contract in terms of currency $j$ is given by
    \begin{equation}\label{cnXCS(i,j)}
    	\textnormal{\textbf{cnXCS}}_{\textnormal{Payer}}^{(i,j)}(t)=\textnormal{\textbf{Leg}}_{(j)}(t)-f^{i \rightarrow j} (t,t_0) \textnormal{\textbf{Leg}}_{(i)}(t),
    \end{equation}
	\noindent where $f^{i \rightarrow j} (t,t_0)$ is the spot FX rate.\par \medskip
    For the curve construction, we can assume that $\textnormal{N}_{(j)}=1$ and $\textnormal{N}_{(j)}=\textnormal{N}_{(i)}f^{i \rightarrow j} (t,t_0)$. Also, due to the difference between the trade date and the start date, i.e. $t_0-t$ is equal to 2 days and the cash deposits (overnight and tom-next) rates are nearly zero, then we can say that $P^c_{(j)}(t,\widetilde{t}_0)\approx 1$ and $P^c_{(i)}(t,\widetilde{t}_0)\approx 1$. Finally we will assume that $y_{(i,j)}(\cdot)$ is a deterministic function, therefore by setting equation \eqref{cnXCS(i,j)} equal to zero yields,
    \begin{equation}\label{eq:bootstrappingcnXCS}
	\begin{split}
    \hspace{2mm} & P^c_{(j)}(t,\widetilde{t}_N)-P^c_{(i)}(t,\widetilde{t}_N) f^{i \rightarrow j} (t,\widetilde{t}_0) P^y(t,\widetilde{t}_N) \\
    & \hspace{5mm} + \sum_{k=1}^N \beta(t_{k-1},t_k) P^c_{(j)}(t,\widetilde{t}_k) \mathbb{E}^{\mathbb{T}^c_{(j),N}}_t \big[ L^{(j)}(t_{k-1},t_k) \big] \\
    & \hspace{5mm} - f^{i \rightarrow j} (t,\widetilde{t}_0) \sum_{k=1}^N \beta(t_{k-1},t_k) P^c_{(i)}(t,\widetilde{t}_k) \Big[ \mathbb{E}^{\mathbb{T}^c_{(i),k}}_t \big[ L^{(i)}(t_{k-1},t_k) \big] + B_N \Big] P^y(t,\widetilde{t}_k) = 0
    \end{split}
	\end{equation}
    where $P^y(t,\widetilde{t}_k)=e^{-\int_t^{\widetilde{t}_k}y_{(i,j)}(s)ds}$ for all $k=0,1,\dots,N$. Note that when we assume that $y_{(i,j)}$ is a deterministic function then we are able to take out the term $\exp({-\int_t^{\widetilde{t}_k}y_{(i,j)}(s)ds})$ from the expectation of equation \eqref{eq:leg(i)}. Additionally, the terms
    \begin{equation}
     P^c_{(i)}(t,\widetilde{t}_k)P^y(t,\widetilde{t}_k)
    \end{equation}
    could be interpreted as the discount factors of currency $i$ fully collateralized in currency $j$.\par \medskip
    So, the only unknown factors of equation \eqref{eq:bootstrappingcnXCS} are the values of $P^y(t,\widetilde{t}_k)$. Hence using a bootstrapping with all the quotes of cnXCSs displayed in the market $\{B(i,j)_N\}$ and an interpolation method we can built the curve $\{P^y(t,T)\}$ and thus $\{y_{(i,j)}(T)\}$. We do not present results and an implementation of this bootstrapping since this work is focus on the MXN curve construction. For interested readers, see \cite{fujii2010collateral} for a deeper analysis of EURUSD and USDJPY XCSs.
    \subsubsection{Case of MXN}
    The main difference between EUR and MXN markets is the no existence of OISs in the latter. Besides that neither of the products IRSs nor XCSs are collateralized in MXN. Indeed they are collateralized in USD. Nevertheless we have to propose a \emph{multi-curve framework} for curve construction with the constraints of replicating the market quotes considering the main currency of the CSA agreements.\par \medskip 
    According to the \cite{mexder2014}, most of the banks in Mexico are funded with foreign capital and they manage their books and trading desks outside Mexico. It is estimated that around 80\% of the IRSs market operations are traded outside Mexico or by international (non-mexican) banks. For this reason IRSs and XCSs (USDMXN) based on TIIE 28d are executed under a CSA agreement with USD as collateral currency.\par \medskip
    As we will see in section \ref{sec:MXNIRSUnderDifColl}, neither the forward curve nor the discount curve could be obtained with a simple bootstrapping. In fact, we will present how to define a two-step bootstrapping to build the MXN forward curve (TIIE 28d as index rate) and the MXN discount curve collateralized in USD. In fact, the idea of this two-step bootstrapping is simple; we know that either the plain vanilla IRSs based on TIIE 28d and plain vanilla XCSs that exchange LIBOR 1m plus a spread for TIIE 28d have the same schedule i.e. coupons every 28 days with the following date convention and have the same maturities (84d, 168d, 252d and so on). However the spot lag is different since in IRSs it is of one day whereas in XCSs it is of two days; also the payments calendars are MX and US-MX, respectively. These differences forbid us to substitute the MXN floating leg of the XCS for the fixed leg of the IRS. Nevertheless if we assume that this difference can be neglected then we can replace MXN floating leg with a fixed rate leg, thus the discounting curve of MXN leg collateralized in USD could be obtained using a simple bootstrapping. Then it is immediately followed the calibration of the forward curve (TIIE 28d) from the IRSs since the discounting curve collateralized in USD is known. This calibration is again made by a simple bootstrapping. It is important to point out that in section \ref{sec:MXNIRSUnderDifColl} we will not assume that the spot lag and the calendars are equal in IRSs and XCSs, and, as a consequence, we will have to iterate the bootstrappings to converge for a solution that replicates the market swaps: IRSs and XCSs.\par \medskip
    In figure \ref{fig:flowchart} we present a flowchart with the differences between the EUR and MXN markets and how the USD market relates with them for the rate curves construction.
    
        \begin{sidewaysfigure}
    	\centering
        \begin{figure}[H]
          \centering
          \includegraphics[scale=0.75]{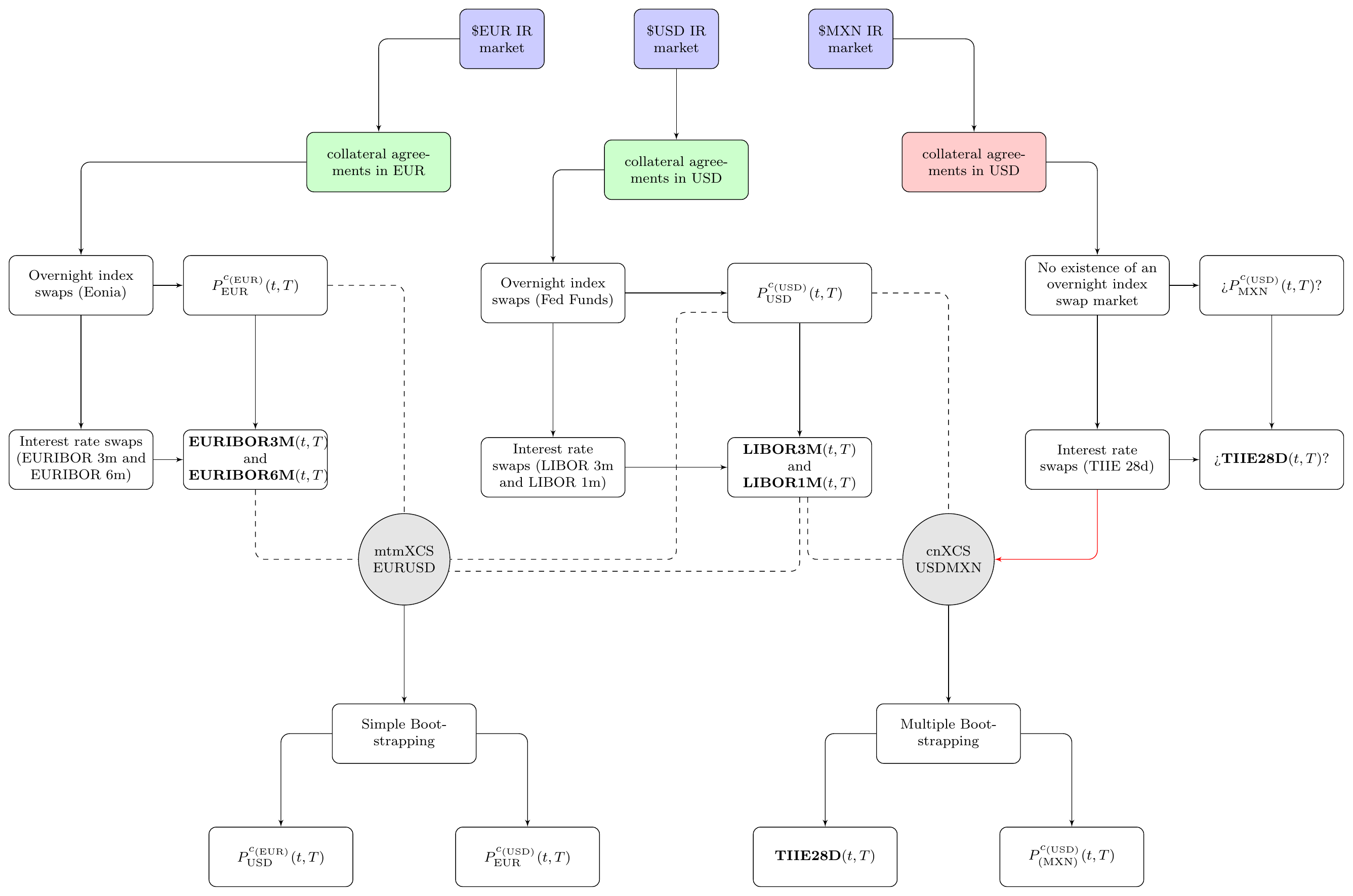}
        \end{figure}
    	\caption[Flowchart for curve construction: differences between EUR and MXN curves when the collateral currency is USD]{Flowchart for curve construction: In this flowchart we present the main differences between EUR and MXN curves. The no existence of MXN OISs market and the fact that the market IRSs based on TIIE 28d are fully collateralized in USD, incapacitated us for calibrating the MXN curves using simple bootstrappings}
    		\label{fig:flowchart}
		\end{sidewaysfigure}
    
    \newpage
    
    \section{Pricing MXN IRS Under Different Collateral Currencies}\label{sec:MXNIRSUnderDifColl}
    In this section we will present the methodology for the valuation of MXN IRSs in three different collateral currencies: USD, MXN and EUR. Also we include the methodology for pricing an IRS when it is uncollateralized i.e. without a collateral agreement. As we will see throughout this section, the calibration methods of each discount curve are different and follow distinct arguments for the construction. Indeed, when the collateral currency is USD we will use for valuation curves that are calibrated by multi-curve bootstrapping using quotes displayed in the market. Using the curves obtained in the USD case, the uncollateralized and MXN-collateral cases follow \emph{simple} arguments that are seen typically in the market, Finally when the collateral currency is EUR the curve calibration follows only arguments of a non-arbitrage market.
    \subsection{Calibration of the MXN Discount Curve Collateralized in USD and the MXN TIIE 28d Forward Curve}\label{sec:calibrationMXNCOLUSDandTIIE}
    The curves calibrated in this subsection correspond to the implied curves obtained from prices quoted in the market. Recall that the most traded interest rates derivatives in the MXN currency are: 1) plain vanilla IRSs based on TIIE 28d and 2) XCSs between USD and MXN currencies. It is important to point out that for these products, the prices displayed in trading screens of brokers and market-makers  typically correspond to prices under CSA agreements with cash collateral in USD currency. Before we present the methodology for the curve calibration let us present useful formulas and notation that we will use throughout the section. Let $\textnormal{\textbf{IRS}}_{\textnormal{Payer}}^{\textnormal{TIIE28D}}(t)$ be the present value of a plain vanilla payer $N$-coupon IRS based on TIIE 28d, then 
    \begin{equation}\label{payerIRSTIIE28D}
        \textnormal{\textbf{IRS}}_{\textnormal{Payer}}^{\textnormal{TIIE28D}}(t)=\textnormal{\textbf{FloatLeg}}(t)-\textnormal{\textbf{FixedLeg}}(t),
    \end{equation}
    with
    \begin{equation}\label{legsIRSTIIE}
    	\begin{aligned}
        \textnormal{\textbf{FloatLeg}}(t) &=\textnormal{N}_{\textnormal{MXN}} \sum_{i=1}^{N}\alpha(t_{i-1},t_{i}) \mathbb{E}_t^{\widetilde{t}_i}(\textnormal{\textbf{TIIE28D}}(t_{i-1},t_{i})) P^{c_{(\textnormal{USD})}}_{\textnormal{MXN }}(t,\widetilde{t}_i),\\
        \textnormal{\textbf{FixedLeg}}(t) &= \textnormal{N}_{\textnormal{MXN}} \cdot k \sum_{i=1}^{N} \alpha(t_{i-1},t_{i}) P^{c_{(\textnormal{USD})}}_{\textnormal{MXN}}(t,\widetilde{t}_i),
        \end{aligned}
    \end{equation}
    and where:\par \smallskip
    \noindent \makebox[3.8cm][l]{$\textnormal{N}_{\textnormal{MXN}}$:} notional of the IRS \par
    \noindent \makebox[3.8cm][l]{$k$:} fixed rate of the $N$-coupons IRS (swap rate)\par
    \noindent \makebox[3.8cm][l]{$N$:} number of coupons \par 
    \noindent \makebox[3.8cm][l]{$t_i$:} coupon periods of both legs \par 
    \noindent \makebox[3.8cm][l]{$\widetilde{t}_i$:} time of payment of the $i$th coupon \par 
    \noindent \makebox[3.8cm][l]{$\alpha(t_{i-1},t_i))$:} accrual factor of the $i$th coupon \par
    \noindent \makebox[3.8cm][l]{$\mathbb{E}_t^{\widetilde{t}_i}(\textnormal{\textbf{TIIE28D}}(t_{i-1},t_{i}))$:} the forward TIIE 28d rate of the $i$th coupon \par
    \noindent \makebox[3.8cm][l]{$P^{c_{(\textnormal{USD})}}_{\textnormal{MXN}}(t,\widetilde{t}_i)$:} MXN discount factor collateralized in USD for time $\widetilde{t}_i$. \par
    \begin{rem}\label{remark.calendar.IRS}
    In the MXN plain vanilla IRS we have that $\widetilde{t}_i=t_i$ for all $i=1,\dots,N$ since payment dates are equal to coupon end dates. The coupon dates $t_i$ are defined every 28 days with the following convention using a MX business days calendar. The spot lag (the difference between the start date and the trade date) is of one open day, i.e. the next business day after the trade date. 
    \end{rem}
    \begin{rem}
    Counterparties that have access to quotes or trading screens are typically market-makers, which do not have necessarily CSA agreements denominated in USD. Indeed, \cite{mexder2014} states that 80\% of the flows are traded outside Mexico or by foreign banks with non-local capital. Nevertheless these trading screen prices are references that are adjusted with a spread according to the collateral currency between the counterparties (valuation adjustment known as CollVA, see \cite{ruiz2015xva}).
    \end{rem}
    
    From equations \eqref{xcs:eq1}-\eqref{xcs:eq2}-\eqref{xcs:eq3} we have that the mark-to-market at time $t$ (in MXN currency) of a plain vanilla payer USDMXN cnXCS (that exchanges LIBOR 1m plus a spread for TIIE 28d) $\textnormal{\textbf{cnXCS}}_{\textnormal{Payer}}^{\textnormal{USDMXN}}(t)$ is given by
        \begin{equation}\label{payerXCSUSDMXN}
        \textnormal{\textbf{cnXCS}}_{\textnormal{Payer}}^{\textnormal{USDMXN}}(t)=\textnormal{\textbf{Leg}}_\textnormal{MXN}(t)-f^{\textnormal{USD} \rightarrow \textnormal{MXN}}(t)\textnormal{\textbf{Leg}}_\textnormal{USD}(t),
    \end{equation}
    with
    \begin{equation}\label{legXCS.MXN}
    	\begin{split}
        \textnormal{\textbf{Leg}}_\textnormal{MXN}(t) &= \textnormal{N}_{\textnormal{MXN}} \Big[ -P^{c_{(\textnormal{USD})}}_{\textnormal{MXN}}(t,\widetilde{s}_0) + P^{c_{(\textnormal{USD})}}_{\textnormal{MXN}}(t,\widetilde{s}_N) +\hspace{35mm} \\ & \hspace{5mm} \sum_{j=1}^{N}\beta(s_{j-1},s_{j}) \mathbb{E}_t^{\widetilde{s}_j}(\textnormal{\textbf{TIIE28D}}(s_{j-1},s_{j})) P^{c_{(\textnormal{USD})}}_{\textnormal{MXN}}(t,\widetilde{s}_j) \Big],
        \end{split}
    \end{equation}
    \begin{equation}\label{legXCS.USD}
    	\begin{split}
        \textnormal{\textbf{Leg}}_\textnormal{USD}(t) &= \textnormal{N}_{\textnormal{USD}} \Big[ -P^{c_{(\textnormal{USD})}}_{\textnormal{USD}}(t,\widetilde{s}_0) + P^{c_{(\textnormal{USD})}}_{\textnormal{USD}}(t,\widetilde{s}_N) + \\ & \hspace{5mm} \sum_{j=1}^{N}\beta(s_{j-1},s_{j}) \big( \mathbb{E}_t^{\widetilde{s}_j}(\textnormal{\textbf{LIBOR1M}}(s_{j-1},s_{j})) + B_N \big) P^{c_{(\textnormal{USD})}}_{\textnormal{USD}}(t,\widetilde{s}_j) \Big],
        \end{split}
    \end{equation}
    and $f^{\textnormal{USD} \rightarrow \textnormal{MXN}}(t)$ is the FX spot rate at time $t$. Additionally we have that,\par\smallskip
    \noindent \makebox[4.2cm][l]{$\textnormal{N}_{\textnormal{USD}}$:} notional of USD leg \par
    \noindent \makebox[4.2cm][l]{$\textnormal{N}_{\textnormal{MXN}}$:} notional of MXN leg \par
    \noindent \makebox[4.2cm][l]{$N$:} number of coupons\par 
    \noindent \makebox[4.2cm][l]{$B_N$:} basis spread of the $N$-coupons cnXCS\par
    \noindent \makebox[4.2cm][l]{$(s_{j-1},s_j)$:} period of $j$th coupon (in both legs) \par 
    \noindent \makebox[4.2cm][l]{$\widetilde{s}_j$:} time of payment of the $j$th coupon (in both legs) \par 
    \noindent \makebox[4.2cm][l]{$\beta(s_{j-1},s_j))$:} accrual factor of the $j$th coupon (in both legs)\par
    \noindent \makebox[4.2cm][l]{$\mathbb{E}_t^{\widetilde{s}_j}(\textnormal{\textbf{LIBOR1M}}(s_{j-1},s_{j}))$:} the forward LIBOR 1m rate of the $j$th coupon \par
    \noindent \makebox[4.2cm][l]{$\mathbb{E}_t^{\widetilde{s}_j}(\textnormal{\textbf{TIIE28D}}(s_{j-1},s_{j}))$:} the forward TIIE 28d rate of the $j$th coupon \par
    \noindent \makebox[4.2cm][l]{$P^{c_{(\textnormal{USD})}}_{\textnormal{USD}}(t,\widetilde{s}_j)$:} USD discount factor collateralized in USD for time $\widetilde{s}_j$ \par
    \noindent \makebox[4.2cm][l]{$P^{c_{(\textnormal{USD})}}_{\textnormal{MXN}}(t,\widetilde{s}_j)$:} MXN discount factor collateralized in USD for time $\widetilde{s}_j$. \par \medskip
    
    \noindent In the case of cnXCSs, no matter the number of coupons, we have that 
    \begin{equation}
    \textnormal{N}_{\textnormal{MXN}}=\textnormal{N}_{\textnormal{USD}} f_0^{\textnormal{USD} \rightarrow \textnormal{MXN}}, 
    \end{equation}
    where $f_0^{\textnormal{USD} \rightarrow \textnormal{MXN}}$ is the FX rate fixed by the two counterparties at the moment the deal is closed. Note that in general $f_0^{\textnormal{USD} \rightarrow \textnormal{MXN}} \neq f^{\textnormal{USD} \rightarrow \textnormal{MXN}}(t)$, since the first one is used to determine the notional of the MXN leg and the second is the FX spot rate used for the valuation of the mark-to-markets. However at the moment the trade is done counterparties agreed that the notionals are determines with the FX spot rate, i.e. $f_0^{\textnormal{USD} \rightarrow \textnormal{MXN}} = f^{\textnormal{USD} \rightarrow \textnormal{MXN}}(t)$.
    \begin{rem}\label{remark.calendar.XCS}
    In a plain vanilla cnXCS the payment dates are scheduled every 28 days using the following convention. However, in contrast with a plain vanilla IRS, these payment dates are determined using an US-MX business days calendar and the spot lag is of two open days, i.e. the second business day after the trade date. Additionally we have that $\widetilde{s}_j=s_j$ for all $j=1,\dots,N$.
    \end{rem}
    \begin{rem}\label{remark.fixingLIBOR.XCS}
    For a plain vanilla TS that exchanges LIBOR 1m for LIBOR 3m, the fixings of the LIBOR 1m leg are determined on a monthly basis and the accrual factors too. However on the plain vanilla USDMXN cnXCS, the applicable accrual factor is calculated considering coupons of 28 days.
    \end{rem}
    From the previous sections we know how to build the USD discount curve collateralized in USD and the LIBOR 1m forward curve. Hence the values of $P^{c_{(\textnormal{USD})}}_{\textnormal{USD}}(t,x)$ and $\mathbb{E}_t^{x}(\textnormal{\textbf{LIBOR1M}}(x,x+\textnormal{1m}))$ in equation \eqref{legXCS.USD} are known for all $x$. So the remaining curves that we have to calibrate are
    \begin{align}
		P_\textnormal{MXN}^{c_\textnormal{(USD)}}(t,x),& \hspace{4mm} \textnormal{and}\\
        \mathbb{E}_t^{x}(\textnormal{\textbf{TIIE28D}}(x,x+\textnormal{28d})),& \hspace{4mm} \textnormal{for all } x.
	\end{align}
    Fortunately, we have two swap curves as inputs (IRS and cnXCS market quotes, see table \ref{IRSXCSQuotesBloomberg}) and we have to solve two curves as outputs (discount of MXN collateralized in MXN and TIIE 28d index curve).
    \begin{table}
    \scriptsize
    \centering
    \begin{tabular}{| r c | c c | c c |}
      \hline
     \multicolumn{1}{|>{\centering\arraybackslash}m{10mm}}{\textbf{Tenor}} & 
     \multicolumn{1}{>{\centering\arraybackslash}m{20mm}|}{\textbf{Number of Coupons $(N)$}} & 
     \multicolumn{1}{>{\centering\arraybackslash}m{20mm}}{\textbf{Fixed Rate (\%)}} &
     \multicolumn{1}{>{\centering\arraybackslash}m{20mm}|}{\textbf{Type}} & 
     \multicolumn{1}{>{\centering\arraybackslash}m{20mm}}{\textbf{Spread (\%)}} &
     \multicolumn{1}{>{\centering\arraybackslash}m{20mm}|}{\textbf{Type}}\\ \hline
     
      84D & 3 & 3.3200 & IRS & 0.5400 & cnXCS \\ 
      168D & 6 & 3.4300 & IRS & 0.5900 & cnXCS \\ 
      252D & 9 & 3.5620 & IRS & 0.6400 & cnXCS \\ 
      364D & 13 & 3.7350 & IRS & 0.6800 & cnXCS \\ 
      728D & 26 & 4.2360 & IRS & 0.7200 & cnXCS \\ 
      1092D & 39 & 4.6710 & IRS & 0.8100 & cnXCS \\ 
      1456D & 52 & 5.0510 & IRS & 0.8800 & cnXCS \\ 
      1820D & 65 & 5.3610 & IRS & 0.9200 & cnXCS \\ 
      2548D & 91 & 5.8630 & IRS & 1.0050 & cnXCS \\ 
      3640D & 130 & 6.2380 & IRS & 1.0400 & cnXCS \\ 
      4368D & 156 & 6.4280 & IRS & 1.0400 & cnXCS \\ 
      5460D & 195 & 6.6320 & IRS & 1.0200 & cnXCS \\ 
      7280D & 260 & 6.8310 & IRS & 1.0250 & cnXCS \\ 
      10920D & 390 & 7.0210 & IRS & 1.0250 & cnXCS \\ 
       \hline
    \end{tabular}
    \caption{Quoted TIIE 28d IRSs and USDMXN cnXCSs on May 29 2015 (Source: Bloomberg).}
    \label{IRSXCSQuotesBloomberg}
    \end{table}
    Therefore, we get a kind of ``system of equations'' that may have a solution (not necessarily unique) that can be found easily since it is a $2 \times 2$ ``system of equations''\footnote{We are using \emph{quotation marks} because inputs and outputs do not define a real system of equations, we are abusing the language slightly.}. Before we start trying to solve this ``system of equations'' let us calculate the number of unknown variables. Note that the cnXCS with longest maturity (30 years) has 390 coupons $(=30 \textnormal{ years} \times 13 \textnormal{ coupons per year})$. For every MXN coupon we have two unknown variables: the discount factor and the forward index rate, hence we have 780 variables. However the TIIE 28d index rate for the first coupon is determined at time $t$ so it is not an unknown variable. In summary, we have 779 variables and 28 equations (14 IRS and 14 cnXCS from table \ref{IRSXCSQuotesBloomberg}). Now it is clearly that the system of equations have infinite solutions but we have find the more simple and adequate solution that is consistent with the term structure of the interest rate market. This solution is obtained through an interpolation method and a multiple bootstrapping that builds iteratively the two curves at the same time. To illustrate the idea behind this method let us suppose that our market only has one IRS and one cnXCS, both with maturity of 84 days. Also let us assume that $\textnormal{N}_\textnormal{MXN}=1$, hence the present value at time $t$ of the IRS and the cnXCS (both payers) are given by,
    \begin{equation}\label{payerIRSTIIE28D.84d}
    	\begin{split}
        \textnormal{\textbf{IRS}}_{\textnormal{Payer}}^{\textnormal{TIIE28D}}(t)=\sum_{i=1}^{3} &\alpha(t_{i-1},t_{i}) \mathbb{E}_t^{\widetilde{t}_i}(\textnormal{\textbf{TIIE28D}}(t_{i-1},t_{i})) P^{c_{(\textnormal{USD})}}_{\textnormal{MXN }}(t,\widetilde{t}_i)\\ &- k_\textnormal{84d} \sum_{i=1}^{3} \alpha(t_{i-1},t_{i}) P^{c_{(\textnormal{USD})}}_{\textnormal{MXN}}(t,\widetilde{t}_i),
        \end{split}
    \end{equation}
    and
    \begin{equation}\label{payerXCS.84d}
    	\begin{split}
        \textnormal{\textbf{cnXCS}}_{\textnormal{Payer}}^{\textnormal{USDMXN}}(t) &= \Big[ -P^{c_{(\textnormal{USD})}}_{\textnormal{MXN}}(t,\widetilde{s}_0) + P^{c_{(\textnormal{USD})}}_{\textnormal{MXN}}(t,\widetilde{s}_N)\hspace{35mm} \\ & \hspace{5mm} +\sum_{j=1}^{3}\beta(s_{j-1},s_{j}) \mathbb{E}_t^{\widetilde{s}_j}(\textnormal{\textbf{TIIE28D}}(s_{j-1},s_{j})) P^{c_{(\textnormal{USD})}}_{\textnormal{MXN}}(t,\widetilde{s}_j) \Big]\\
       & \hspace{5mm} - \frac{f^{\textnormal{USD} \rightarrow \textnormal{MXN}}(t)}{f_0^{\textnormal{USD} \rightarrow \textnormal{MXN}}} \Big[ -P^{c_{(\textnormal{USD})}}_{\textnormal{USD}}(t,\widetilde{s}_0) + P^{c_{(\textnormal{USD})}}_{\textnormal{USD}}(t,\widetilde{s}_N)\\ & \hspace{5mm} +\sum_{j=1}^{3}\beta(s_{j-1},s_{j}) \big( \mathbb{E}_t^{\widetilde{s}_j}(\textnormal{\textbf{LIBOR1M}}(s_{j-1},s_{j})) + B_\textnormal{84d} \big) P^{c_{(\textnormal{USD})}}_{\textnormal{USD}}(t,\widetilde{s}_j) \Big].
        \end{split}
    \end{equation}
    Assuming that the IRS and the cnXCS are mid market quotes we have that both equations are equal to zero, i.e.,
    \begin{align}
    	\textnormal{\textbf{IRS}}_{\textnormal{Payer}}^{\textnormal{TIIE28D}}(t) &=0\\
    	\textnormal{\textbf{cnXCS}}_{\textnormal{Payer}}^{\textnormal{USDMXN}}(t) &=0.
     \end{align}
    According to the remarks \ref{remark.calendar.IRS} and \ref{remark.calendar.XCS} we have that in general $t_i \neq s_j$, since the business days calendars are different and because the spot lag in the case of the cnXCS is one day greater. Nevertheless, for the calibration of the two curves we  will assume that this difference is sufficiently small to be negligible. Hence,
    \begin{equation}\label{negligible.diff.XCSandIRS}
    	\begin{split}
    	\sum_{i=1}^{3} \alpha(t_{i-1},t_{i}) & \mathbb{E}_t^{\widetilde{t}_i}(\textnormal{\textbf{TIIE28D}}(t_{i-1},t_{i})) P^{c_{(\textnormal{USD})}}_{\textnormal{MXN }}(t,\widetilde{t}_i) \approx\\
        & \sum_{j=1}^{3}\beta(s_{j-1},s_{j}) \mathbb{E}_t^{\widetilde{s}_j}(\textnormal{\textbf{TIIE28D}}(s_{j-1},s_{j})) P^{c_{(\textnormal{USD})}}_{\textnormal{MXN}}(t,\widetilde{s}_j),
        \end{split}
    \end{equation}
    and using equations \eqref{payerIRSTIIE28D.84d}-\eqref{negligible.diff.XCSandIRS} we get
        \begin{equation}\label{bootstrapping.XCSIRS.84d}
    	\begin{split}
        & \Big[ -P^{c_{(\textnormal{USD})}}_{\textnormal{MXN}}(t,\widetilde{s}_0) + P^{c_{(\textnormal{USD})}}_{\textnormal{MXN}}(t,\widetilde{s}_N)\hspace{35mm} \\ & \hspace{5mm} +k_\textnormal{84d} \sum_{j=1}^{3}\beta(s_{j-1},s_{j}) P^{c_{(\textnormal{USD})}}_{\textnormal{MXN}}(t,\widetilde{s}_j) \Big]\\
       & \hspace{5mm} - \frac{f^{\textnormal{USD} \rightarrow \textnormal{MXN}}(t)}{f_0^{\textnormal{USD} \rightarrow \textnormal{MXN}}} \Big[ -P^{c_{(\textnormal{USD})}}_{\textnormal{USD}}(t,\widetilde{s}_0) + P^{c_{(\textnormal{USD})}}_{\textnormal{USD}}(t,\widetilde{s}_N)\\ & \hspace{5mm} +\sum_{j=1}^{3}\beta(s_{j-1},s_{j}) \big( \mathbb{E}_t^{\widetilde{s}_j}(\textnormal{\textbf{LIBOR1M}}(s_{j-1},s_{j})) + B_\textnormal{84d} \big) P^{c_{(\textnormal{USD})}}_{\textnormal{USD}}(t,\widetilde{s}_j) \Big] =0
        \end{split}
    \end{equation}
    Note that in this equation we just have four unknown variables: 
    $$P^{c_{(\textnormal{USD})}}_{\textnormal{MXN}}(t,\widetilde{s}_0),
    P^{c_{(\textnormal{USD})}}_{\textnormal{MXN}}(t,\widetilde{s}_\textnormal{28d}),
    P^{c_{(\textnormal{USD})}}_{\textnormal{MXN}}(t,\widetilde{s}_\textnormal{56d}),
    P^{c_{(\textnormal{USD})}}_{\textnormal{MXN}}(t,\widetilde{s}_\textnormal{84d}).$$
    These four variables could be calculated using the short-term market, i.e. the depo and FX forwards markets. Let us present briefly how to perform this task. When we use the short-term markets, say FX Forwards, we are able to get the implied yield rates used for discounting MXN flows collateralized in USD from the forward points or outright rates. Indeed, we have that the USD/MXN outright FX rate at time $t$ and with delivery at time $T$ is given by the following formula:
    \begin{equation}\label{outrightFX}	f^{\textnormal{USD}\rightarrow\textnormal{MXN}}_T(t)=\frac{P^{c_(\textnormal{USD})}_\textnormal{USD}(t,T)}{P^{c_(\textnormal{USD})}_\textnormal{MXN}(t,T)} \cdot f^{\textnormal{USD}\rightarrow\textnormal{MXN}}(t),
    \end{equation}
    where $f^{\textnormal{USD}\rightarrow\textnormal{MXN}}(t)$ is the FX spot rate. This market is sufficiently liquid to get prices for many tenors, so we are able to get the following outright FX rates:
  $$f^{\textnormal{USD}\rightarrow\textnormal{MXN}}_{\widetilde{s}_0}(t),
    f^{\textnormal{USD}\rightarrow\textnormal{MXN}}_{\widetilde{s}_\textnormal{28d}}(t), 	 	 f^{\textnormal{USD}\rightarrow\textnormal{MXN}}_{\widetilde{s}_\textnormal{56d}}(t), 		f^{\textnormal{USD}\rightarrow\textnormal{MXN}}_{\widetilde{s}_\textnormal{84d}}(t).$$
    Hence, using equation \eqref{outrightFX} we can get the values of the discount factors $P^{c_{(\textnormal{USD})}}_\textnormal{MXN}(t,T)$ when $t \leq T \leq 1$. 
    One obstacle to proceeding with the short-term market method along the rest of the curve is that, even in major currencies (say G7 currencies\footnote{USD, CAD, GBP, EUR and JPY.}), FX forwards are only liquid for two or five years and quoted at most for ten years. For this reason we need long-dated market data (IRS and cnXCS) which are quoted up to 30 years. In this work we will only use IRSs and cnXCS for the curve calibration although \cite{mexder2014} suggests to use the short-term market for the calibration of the curves up to 1 year and the long-dated swaps for the rest of the curve.\par \medskip
    Before we present formally the algorithm of the multiple bootstrapping, let us continue with the curve calibration when the market only has one IRS and one XCS with maturity of 84 days (three coupons).\par \medskip
    We know that every discount factor has an associated yield rate that holds the following equation:
    \begin{equation}
    	P^{c_{(\textnormal{USD})}}_\textnormal{MXN}(t,x)=e^{-(x-t)R(t,x)}.
    \end{equation}
    Solving for $R(t,x)$ yields,
    \begin{equation}
    	R(t,x)=-\frac{\ln \left( P^{c_{(\textnormal{USD})}}_\textnormal{MXN}(t,x) \right)}{x-t}.
    \end{equation}
    In this method the value of $P^{c_{(\textnormal{USD})}}_\textnormal{MXN}(t,\widetilde{s}_0)$ is calculated by the short-term method and defines a yield rate $r_0$ associated to it. For the other three variables $P^{c_{(\textnormal{USD})}}_{\textnormal{MXN}}(t,\widetilde{s}_\textnormal{28d})$, $    P^{c_{(\textnormal{USD})}}_{\textnormal{MXN}}(t,\widetilde{s}_\textnormal{56d})$, $ P^{c_{(\textnormal{USD})}}_{\textnormal{MXN}}(t,\widetilde{s}_\textnormal{84d})$ we need to find a value of $R(t,\widetilde{s}_\textnormal{84d})$ that satisfy the following conditions:
    \begin{enumerate}
    	\item $R(t,x)=a+b(x-t)+c(x-t)^2+d(x-t)^3$ with $\widetilde{s}_0 \leq x \leq \widetilde{s}_\textnormal{84d}$ and $a,b,c,d \in \mathbb{R}$
        \item $R(t,\widetilde{s}_0)=r_0$
        \item $R(t,x)\in\mathcal{C}^2$ with $R''(t,\widetilde{s}_0)=0$ and $R''(t,\widetilde{s}_\textnormal{84d})=0$.
    \end{enumerate}
    These conditions corresponds to the natural cubic splines interpolation method. Note that the above conditions defines the next system of equations:
    \begin{align}
    	a+b(\widetilde{s}_0-t)+c(\widetilde{s}_0-t)^2+d(\widetilde{s}_0-t)^3 & =r_0\\
        2c+6d(\widetilde{s}_0-t) & = 0\\
        2c+6d(\widetilde{s}_\textnormal{84d}-t) & = 0.
    \end{align}
    So we get a system of equations with 4 variables and only 3 equations, hence we have two define a fourth equation with the intention to get a solution. So we say that
    \begin{align}
    	& R(t,\widetilde{s}_\textnormal{84d}):=k_\textnormal{84d}\\
        \Longrightarrow \hspace{7mm} & a+b(\widetilde{s}_\textnormal{84d}-t)+c(\widetilde{s}_\textnormal{84d}-t)^2+d(\widetilde{s}_\textnormal{84d}-t)^3 = k_\textnormal{84d}
    \end{align}
    In other words, we are claiming that the yield zero coupon rate at time $t$ with maturity in $\widetilde{s}_\textnormal{84d}$ is equal the swap rate for the same tenor. With this new equation the system has a unique solution given by the following vector $(a_0,b_0,c_0,d_0)$. We these coefficients we are now able to get the values of $ R_0(t,\widetilde{s}_0), R_0(t,\widetilde{s}_\textnormal{28d}), R_0(t,\widetilde{s}_\textnormal{56d}), R_0(t,\widetilde{s}_\textnormal{84d})$ and hence the values of $P^{c_{(\textnormal{USD})}}_{\textnormal{MXN},0}(t,\widetilde{s}_0)$, 
    $P^{c_{(\textnormal{USD})}}_{\textnormal{MXN},0}(t,\widetilde{s}_\textnormal{28d})$, 
    $P^{c_{(\textnormal{USD})}}_{\textnormal{MXN},0}(t,\widetilde{s}_\textnormal{56d})$,
    $P^{c_{(\textnormal{USD})}}_{\textnormal{MXN},0}(t,\widetilde{s}_\textnormal{84d})$. The subscript zero in $R_0$ and $P_0^c$ is because we want to emphasize that the values are initial values since we made an initial guess for $R_0(t,\widetilde{s}_\textnormal{84d})$. Note that we are also able to calculate the values of $ R_0(t,\widetilde{t}_0), R_0(t,\widetilde{t}_\textnormal{28d}), R_0(t,\widetilde{t}_\textnormal{56d}), R_0(t,\widetilde{t}_\textnormal{84d})$ and substitute them into equation \eqref{payerIRSTIIE28D.84d}. This give us the following equation
    \begin{equation}\label{MXN.IRS.TIIE.unknown}
    	\begin{split}
        \textnormal{\textbf{IRS}}_{\textnormal{Payer}}^{\textnormal{TIIE28D}}(t)=\sum_{i=1}^{3} &\alpha(t_{i-1},t_{i}) \mathbb{E}_t^{\widetilde{t}_i}(\textnormal{\textbf{TIIE28D}}(t_{i-1},t_{i})) P^{c_{(\textnormal{USD})}}_{\textnormal{MXN }}(t,\widetilde{t}_i)\\ &- k_\textnormal{84d} \sum_{i=1}^{3} \alpha(t_{i-1},t_{i}) P^{c_{(\textnormal{USD})}}_{\textnormal{MXN}}(t,\widetilde{t}_i).
        \end{split}
    \end{equation}
    Now the unknown variables in the equation \eqref{MXN.IRS.TIIE.unknown} are $\mathbb{E}_t^{\widetilde{t}_{\textnormal{56d}}}(\textnormal{\textbf{TIIE28D}}(t_{\textnormal{28d}},t_{\textnormal{56d}}))$ and $\mathbb{E}_t^{\widetilde{t}_{\textnormal{84d}}}(\textnormal{\textbf{TIIE28D}}(t_{\textnormal{56d}},t_{\textnormal{84d}}))$. To find the value of these two variables we are going to make the same assumptions done previously i.e. define $R^{\textnormal{TIIE}}(t,x)$ yield zero curve that replicates the TIIE 28d forward curve $\mathbb{E}_t^{x}(\textnormal{\textbf{TIIE28D}}(x,x+\textnormal{28d}))$. Likewise, we will assume that $R^{\textnormal{TIIE}}(t,x)$ is a piecewise-defined function with the natural cubic splines conditions. Therefore we have that
    \begin{enumerate}
    	\item $R^{\textnormal{TIIE}}(t,x)=e+f(x-t)+g(x-t)^2+h(x-t)^3$ with $t_0 \leq x \leq t_\textnormal{84d}$ and $e,f,g,h \in \mathbb{R}$,
        \item $R^{\textnormal{TIIE}}(t,y)=\textnormal{\textbf{TIIE28D}}(t)$ for all $y \in [t_0,t_\textnormal{28d}]$,
        \item $R^{\textnormal{TIIE}}(t,x)\in\mathcal{C}^2$ with $R''(t,t_\textnormal{28d})=0$ and $R''(t,t_\textnormal{84d})=0$.
    \end{enumerate}
    Condition 2 states that the function $R^{\textnormal{TIIE}}$ has a constant value in the interval $[t_0,t_\textnormal{28d}]$ and this value is equal to the TIIE 28d reference rate at time $t$, i.e. equal to the fixing rate published in the trade date. This assumption guarantees that the forward rate implied from the yield curve $R^{\textnormal{TIIE}}$ is the TIIE 28d fixing. Let us present the proof of this (straightforward) fact. We know that the TIIE 28d forward rates are given by the following equation
    \begin{equation}
    	\mathbb{E}_t^{T}(\textnormal{\textbf{TIIE28D}}(S,T))=-\frac{1}{\tau(S,T)} \ln \Bigg( \frac{P^\textnormal{TIIE}(t,S)}{P^\textnormal{TIIE}(t,T)}\Bigg)
    \end{equation}
    Now we have that $P^\textnormal{TIIE}(t,x)=e^{-(x-t)R^\textnormal{TIIE}(t,x)}$ and $R^\textnormal{TIIE}(t,S)=R^\textnormal{TIIE}(t,T)=\textnormal{\textbf{TIIE28D}}(t)$ hence
    \begin{align}
    	\mathbb{E}_t^{T}(\textnormal{\textbf{TIIE28D}}(S,T))&=-\frac{1}{T-S} \ln \Bigg( \frac{e^{-(S-t)\textnormal{\textbf{TIIE28D}}(t)}}{e^{-(T-t)\textnormal{\textbf{TIIE28D}}(t)}}\Bigg)\\
        &=-\frac{\ln(e^{-(T-S)\textnormal{\textbf{TIIE28D}}(t)})}{T-S}\\
        &=\textnormal{\textbf{TIIE28D}}(t).
    \end{align}
    Hence, condition 2 guarantees us that the curve replicates the known fixing at time $t$. Now if we write the system of equation that can be induced from conditions 1, 2 and 3 we obtain
    \begin{align}
    	e+f(t_0-t)+g(t_0-t)^2+h(t_0-t)^3 & =\textnormal{\textbf{TIIE28D}}(t)\\
        2g+6h(t_0-t) & = 0\\
        2g+6h(t_\textnormal{84d}-t) & = 0.
    \end{align}
    Again we have three equations and four variables $e,f,g,h$ so we have to define an equation that allow us to get a solution. So we will say that
    \begin{align}
    	& R^\textnormal{TIIE}(t,t_\textnormal{84d}):=\textnormal{\textbf{TIIE28D}}(t)+\varepsilon_0, \hspace{6mm} \varepsilon_0 \in \mathbb{R}\\
        \Longrightarrow \hspace{7mm} & e+f(t_\textnormal{84d}-t)+c(t_\textnormal{84d}-t)^2+d(t_\textnormal{84d}-t)^3 = \textnormal{\textbf{TIIE28D}}(t)+\varepsilon_0
    \end{align}
    The value of $\varepsilon_0$ is a parameter of the rate curve construction that will help us to converge rapidly to a solution. This parameter depends on the structure of the curve and in the monetary policy decisions that the central bank could take in the future. For example, if the market is pricing a rate hike in the following months then we have that $\varepsilon_0 > 0$. This parameter is only a variable that serves to ease the convergence of the algorithm. Once we have the coefficient vector $(e_0,f_0,g_0,h_0)$ we are able to calculate the forwards rates of TIIE 28d. Then we substitute these forward rates into equation \eqref{MXN.IRS.TIIE.unknown} and we generate a mark-to-market $\Pi_0^0$ for a plain vanilla IRS given by:
    \begin{equation}\label{PiIRS}
    	\begin{split}
        \Pi_0^0(R^\textnormal{TIIE}_0(t,\widetilde{s}_\textnormal{56d}), R^\textnormal{TIIE}_0(t,\widetilde{s}_\textnormal{84d}),\varepsilon_0)=\sum_{i=1}^{3} &\alpha(t_{i-1},t_{i}) \mathbb{E}_t^{\widetilde{t}_i}(\textnormal{\textbf{TIIE28D}}(t_{i-1},t_{i})) P^{c_{(\textnormal{USD})}}_{\textnormal{MXN }}(t,\widetilde{t}_i)\\ &- k_\textnormal{84d} \sum_{i=1}^{3} \alpha(t_{i-1},t_{i}) P^{c_{(\textnormal{USD})}}_{\textnormal{MXN}}(t,\widetilde{t}_i).
        \end{split}
    \end{equation}
    The value of $\Pi_0^0(R^\textnormal{TIIE}_0(t,\widetilde{s}_\textnormal{56d}), R^\textnormal{TIIE}_0(t,\widetilde{s}_\textnormal{84d}),\varepsilon_0)$ is not necessarily equal to zero, so the swap rate $k_\textnormal{84d}$ of the IRS is not at par. The idea of this algorithm is to make the mark-to-market $\Pi_0^0$ equal to zero, so we have to find the root of the equation \eqref{PiIRS}. In this work we will apply the bisection method for this task although there exists more efficient methods for finding roots such as Newton-Raphson method (see \cite{burden2010numerical}). Hence, we have to change the value of $\varepsilon_0$ and then calculate the values of $(e_0,f_0,g_0,h_0)$ until $\Pi_0^0 \approx 0$. Then we substitute the forward rates of TIIE 28d into the cnXCS equation \eqref{payerXCS.84d}. If $\textnormal{\textbf{cnXCS}}_{\textnormal{Payer}}^{\textnormal{USDMXN}}(t) =0$ then we are done with the iterations of the algorithm. However, typically before one iteration we do not have that $\textnormal{\textbf{cnXCS}}_{\textnormal{Payer}}^{\textnormal{USDMXN}}(t) =0$ therefore we have to proceed with more iterations. From equation \eqref{payerXCS.84d} we have to bootstrap the new coefficients $(a_1,b_1,c_1,d_1)$ to make it equal zero. Once we have achieved this task 
     \begin{equation}
    	\begin{split}
        \Pi_m(R_m(t,\widetilde{s}_0), R_m(t,&\widetilde{s}_\textnormal{28d}), R_m(t,\widetilde{s}_\textnormal{56d}), R_m(t,\widetilde{s}_\textnormal{84d})) = \Big[ -P^{c_{(\textnormal{USD})}}_{\textnormal{MXN},m}(t,\widetilde{s}_0) + P^{c_{(\textnormal{USD})}}_{\textnormal{MXN},m}(t,\widetilde{s}_N)\hspace{35mm} \\ & \hspace{2mm} +k_\textnormal{84d} \sum_{j=1}^{3}\beta(s_{j-1},s_{j}) P^{c_{(\textnormal{USD})}}_{\textnormal{MXN},m}(t,\widetilde{s}_j) \Big]\\
       & \hspace{2mm} - \frac{f^{\textnormal{USD} \rightarrow \textnormal{MXN}}(t)}{f_0^{\textnormal{USD} \rightarrow \textnormal{MXN}}} \Big[ -P^{c_{(\textnormal{USD})}}_{\textnormal{USD}}(t,\widetilde{s}_0) + P^{c_{(\textnormal{USD})}}_{\textnormal{USD}}(t,\widetilde{s}_N)\\ & \hspace{2mm} +\sum_{j=1}^{3}\beta(s_{j-1},s_{j}) \big( \mathbb{E}_t^{\widetilde{s}_j}(\textnormal{\textbf{LIBOR1M}}(s_{j-1},s_{j})) + B_\textnormal{84d} \big) P^{c_{(\textnormal{USD})}}_{\textnormal{USD}}(t,\widetilde{s}_j) \Big]
        \end{split}
    \end{equation}
    The value of $\Pi_0(R_0(t,\widetilde{s}_0), R_0(t,\widetilde{s}_\textnormal{28d}), R_0(t,\widetilde{s}_\textnormal{56d}), R_0(t,\widetilde{s}_\textnormal{84d}))$ is in general different to zero since we made an initial guess for the value of $R_0(t,\widetilde{s}_\textnormal{84d})$. The idea of the multiple bootstrapping is that have to iterate the values of $R_0(t,\widetilde{s}_\textnormal{84d})$ and make the IRS and cnXCS mark-to-markets equal to zero.\par \medskip
    The idea of this algorithm is to calibrate two curves by changing from an IRS to a cnXCS iteratively
    
    \begin{algorithm}[H]
     \SetAlgoLined
     \KwData{Tenors ($\overbar{X}=(\textnormal{84d,168d,}\dots\textnormal{,10920d)}$, IRS Rates ($\overbar{k_X}=(k_{\textnormal{84d}},\dots,k_{\textnormal{10920d}})$), and cnXCS Basis Spreads ($\overbar{B_X}=(B_{\textnormal{84d}},\dots,B_{\textnormal{10920d}})$)}
     \KwResult{$P_{\textnormal{MXN}}^{c_{\textnormal{(USD)}}}(t,x)$ for all $x\in [t,\textnormal{10920d}]$ and $\mathbb{E}_t^{x}(\textnormal{\textbf{TIIE28D}}(x,x+\textnormal{28d}))$ for all $x\in [t,\textnormal{10892d}]$}
     Define equations:\par $\Gamma_1:\overbar{\textnormal{\textbf{IRS}}}_{\textnormal{Payer}}^{\textnormal{TIIE28D}}(t,\overbar{X},P_{\textnormal{MXN}}^{c_{\textnormal{(USD)}}}(t,x),\mathbb{E}_t^{x}(\textnormal{\textbf{TIIE28D}}(x,x+\textnormal{28d})))\dots $ \eqref{payerIRSTIIE28D}\par
$\Gamma_2:\overbar{\textnormal{\textbf{cnXCS}}}_{\textnormal{Payer}}^{\textnormal{USDMXN}}(t,\overbar{X},P_{\textnormal{MXN}}^{c_{\textnormal{(USD)}}}(t,x),\mathbb{E}_t^{x}(\textnormal{\textbf{TIIE28D}}(x,x+\textnormal{28d})))\dots$ \eqref{payerXCSUSDMXN}\par
$\Gamma_2:\overbar{\textnormal{\textbf{cnXCS}}}_{\textnormal{Payer,MXN-FixedLeg}}^{\textnormal{USDMXN}}(t,\overbar{X},P_{\textnormal{MXN}}^{c_{\textnormal{(USD)}}}(t,x))\dots$ \eqref{bootstrapping.XCSIRS.84d}\par\medskip
$m=0$\;
1) Calculate $\{P_{\textnormal{MXN}}^{c_{\textnormal{(USD)}}}(t,x)
\}_m$ from $\Gamma_3$\;
2) Substitute $\{P_{\textnormal{MXN}}^{c_{\textnormal{(USD)}}}(t,x)
\}_m$ in $\Gamma_1$ and calculate $\{\mathbb{E}_t^{x}(\textnormal{\textbf{TIIE28D}}(x,x+\textnormal{28d})))\}_m$\;
3) Substitute $\{\mathbb{E}_t^{x}(\textnormal{\textbf{TIIE28D}}(x,x+\textnormal{28d})))\}_m$ in $\Gamma_2$ and calculate $\{P_{\textnormal{MXN}}^{c_{\textnormal{(USD)}}}(t,x)
\}_{m+1}$\;
4) Define $m:=m+1$ and repeat step 2 until convergence is met.
     \caption{Steps for the calibration of the MXN discount curve collateralized in USD and the index TIIE 28d forward rates.}
    \end{algorithm}
    
    	 \begin{figure}[H]
        \centering
        \vspace{-2mm}
        \includegraphics[scale=0.58]{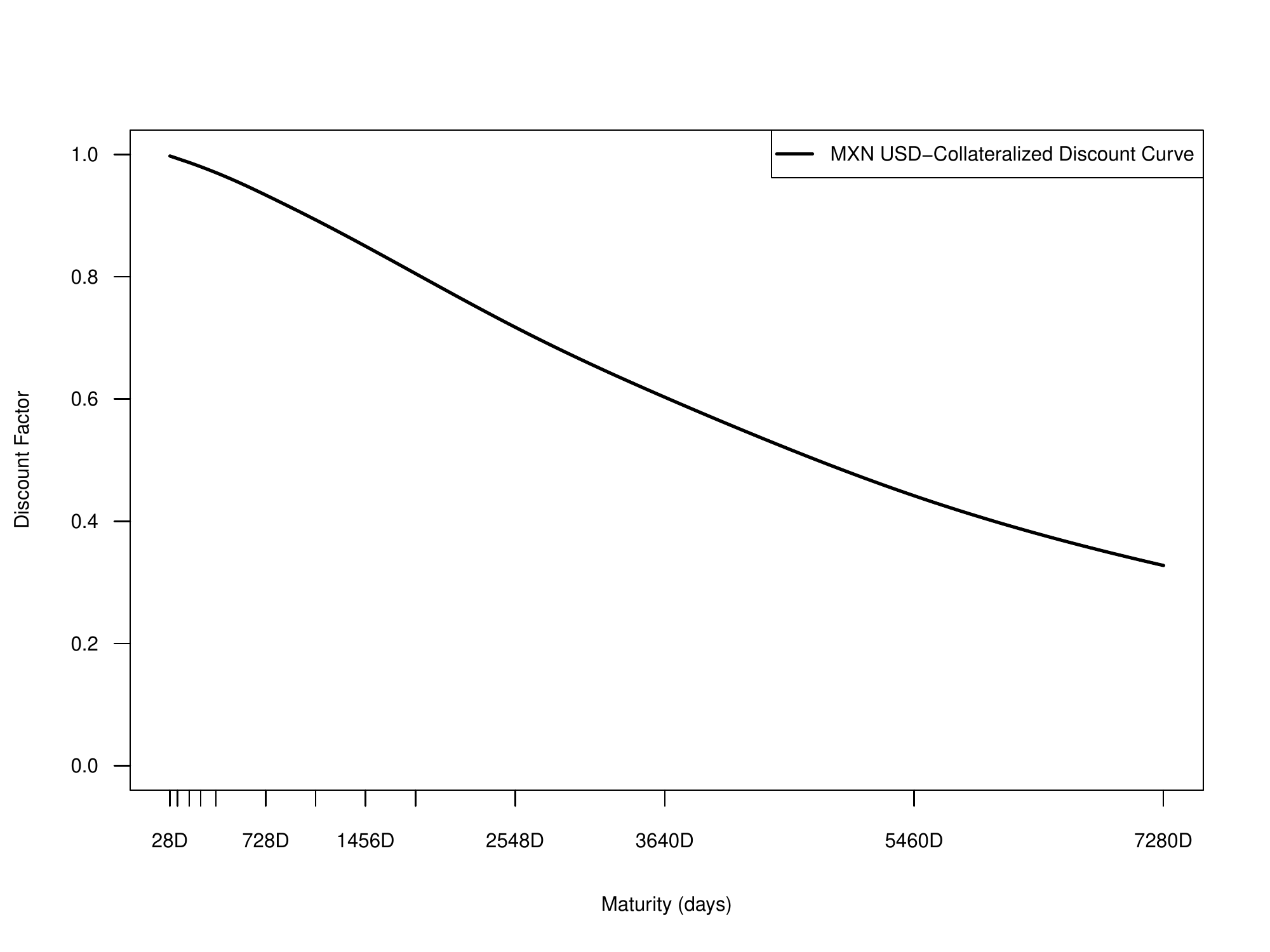}
        \caption[MXN discount curve collateralized in USD (\emph{multi-curve} framework)]{MXN Discount Curve Collateralized in USD: $P(T)=P^{c_{\textnormal{(USD)}}}_{\textnormal{MXN}}(t,x)$. The importance of this curve lies in that every MXN dollar cash flow, inside a contract with CSA in USD, is discounted with it.}
        \label{TIIE28D_YieldCurveSC}
        \vspace{-2mm}
     \end{figure}    

	\begin{figure}[H]
        \centering
        \vspace{-2mm}
        \includegraphics[scale=0.58]{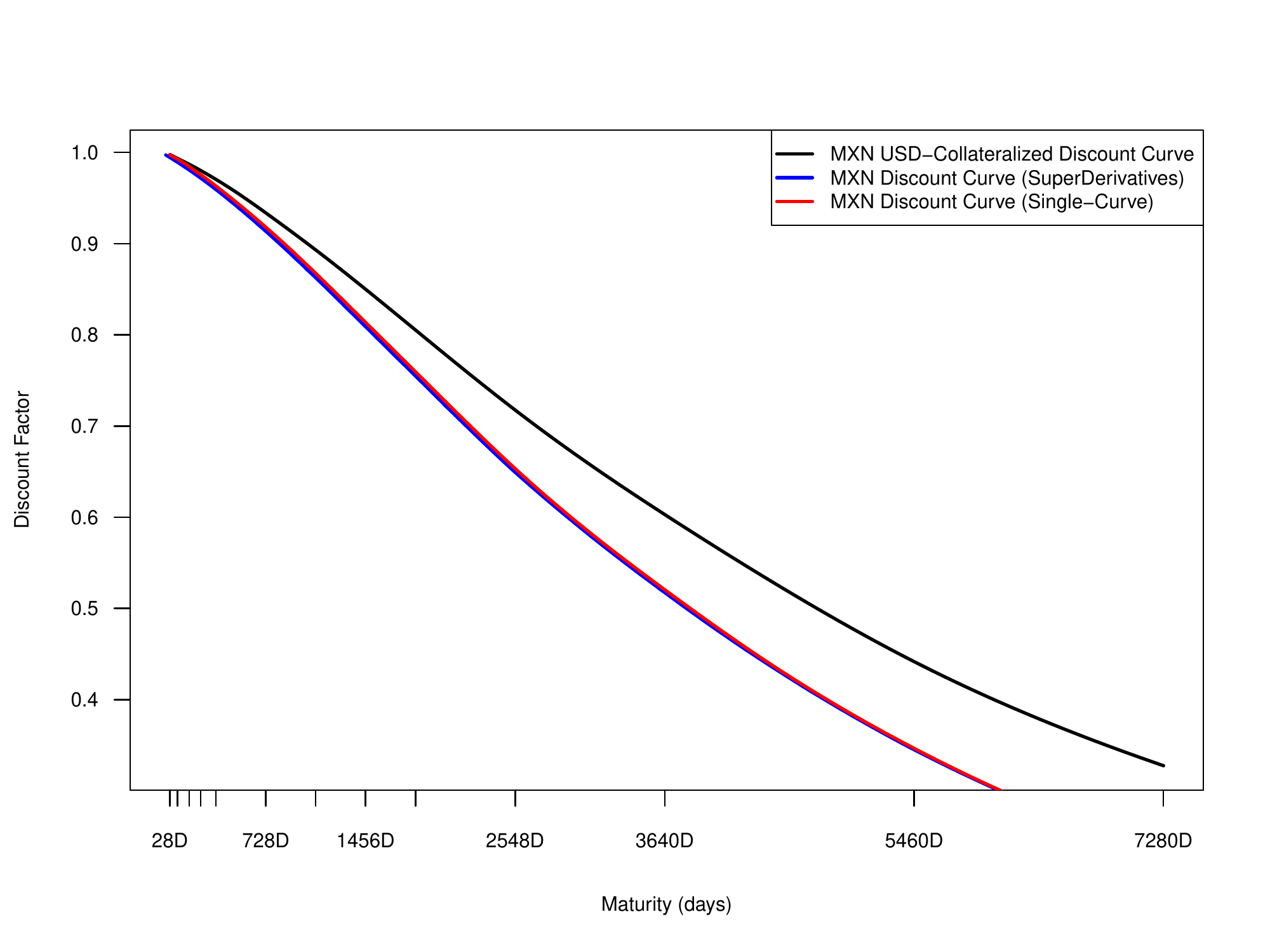}
        \caption[Comparison of MXN discount curves in a \emph{multi-curve} framework]{In this figure we present the MXN discount curves in a \emph{single-curve} framework, \emph{multi-curve} framework and the discount curve used in SuperDerivatives for discounting MXN cash flows.}
        \label{TIIE28D_DiscountCurveSC}
        \vspace{-2mm}
     \end{figure}
    
      \begin{figure}[H]
        \centering
        \vspace{-2mm}
        \includegraphics[scale=0.58]{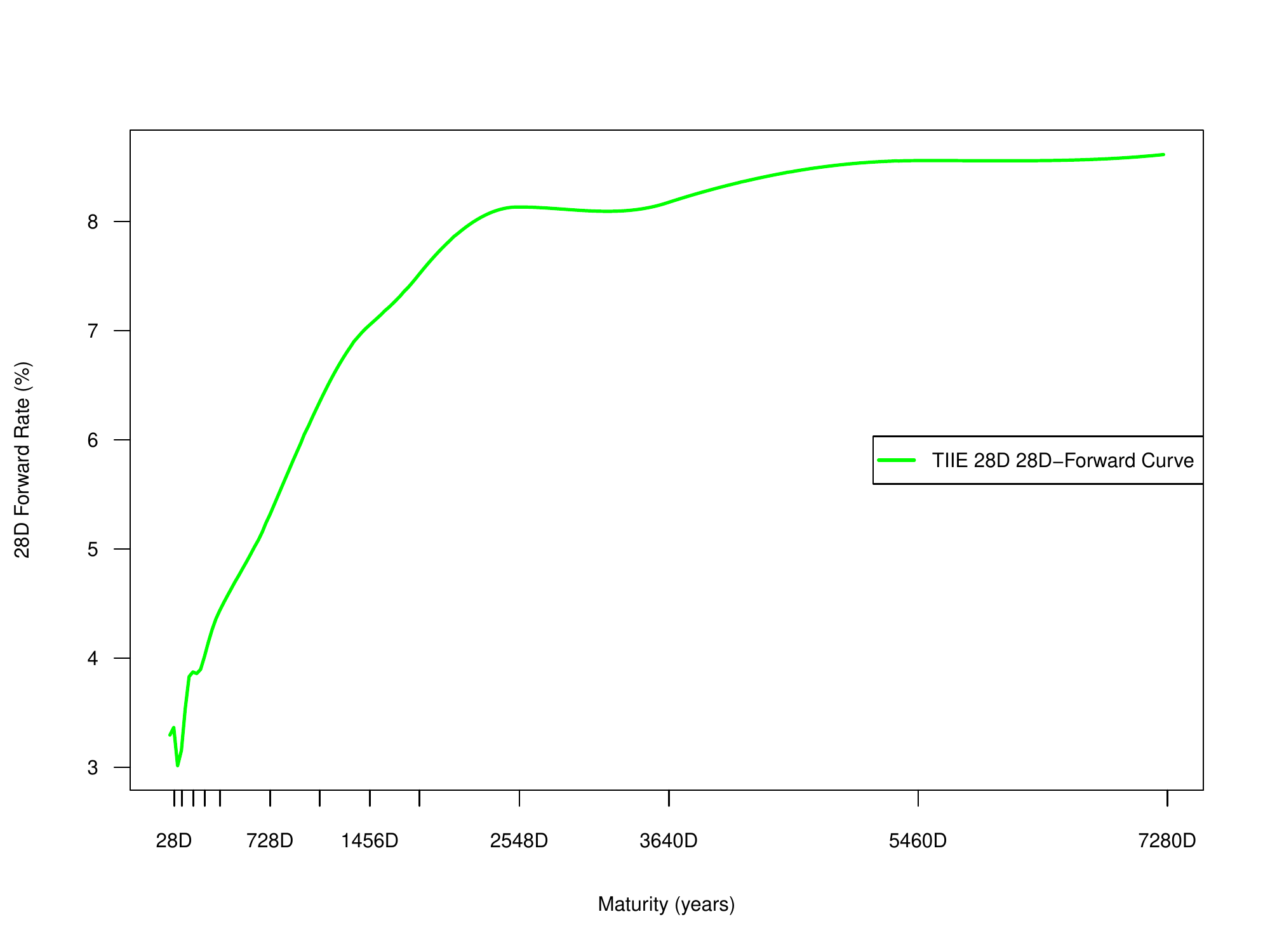}
        \caption[TIIE 28d forward curve in a \emph{multi-curve} framework]{TIIE 28d forward curve $\mathbb{E}_t^{x}(\textnormal{\textbf{TIIE28D}}(x,x+\textnormal{28d})))$ in a \emph{multi-curve} framework using natural cubic splines interpolation in the yield rates.}
        \label{TIIE28D_ForwardCurveSC}
        \vspace{-2mm}
     \end{figure}
     
     \begin{figure}[H]
        \centering
        \vspace{-2mm}
        \includegraphics[scale=0.58]{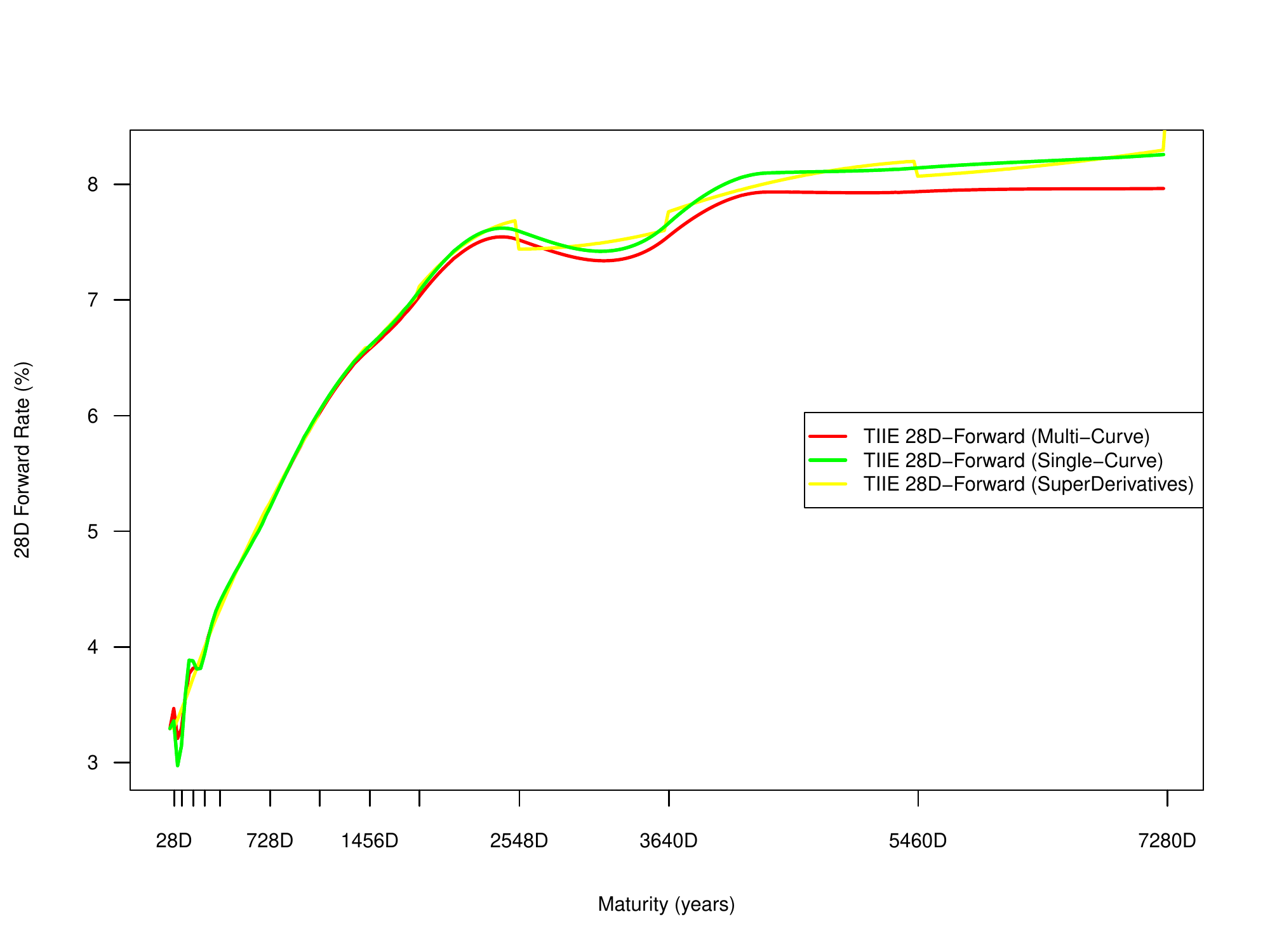}
        \caption[Comparison of TIIE 28d forward curve in a \emph{multi-curve} framework]{In this figure we present the TIIE 28d 28d-forward curves in a \emph{single-curve} framework, \emph{multi-curve} framework and the forward curve defined by SuperDerivatives for the same date (May 29, 2015).}
        \label{TIIE28D_ForwardCurveSC}
        \vspace{-2mm}
     \end{figure}
    
    \subsection{Calibration of the MXN Discount Curve Without Collateral}
    Uncollateralized or non collateral interest rates derivatives are also known as unsecured trades due to the absence of a CSA agreement or a clearing central  counterparty. The choice of which discount curve should be used for uncollateralized deals is a matter of debate among all the market participants. Indeed, since the crisis many derivatives dealers have made valuation adjustments (in particular FVAs) for uncollateralized transactions. This has an effect of increasing the discount rate to their average funding cost \cite{hull2014valuing}.\medskip
    
    In this work we present two alternatives for discounting flows in a uncollateralized world:
    
    \begin{enumerate}
    \item Use an Ibor-based discount curve. The idea is simple, use the same discount curve that was used in a pre-crisis world. Note that the usage of this curve is assuming that our funding rate is an Ibor rate (without any other cost).
    \item Use an internal discount rate. This curve is constructed internally and considers the funding costs, i.e. at what rates levels does the issuer of any derivative funds itself (lend and borrow money). This internal discount rate is the resulting discount curve of applying XVAs costs into the Ibor-based discounting curve.
    \end{enumerate}
    Therefore, in this work we will used the implied discount curve that is defined by the forward TIIE 28d curve.
    
    \subsection{Calibration of the MXN Discount Curve Collateralized in MXN}
    As we saw in the previous sections, the absence of an OIS MXN market limits us from building a MXN collateralized discount curve. Indeed, in the MXN overnight rates market the only available product is the overnight money market (borrow or lend at the overnight rate). However, if we are trying to build a MXN-collateralized discount curve then we may estimate the behavior of the daily overnight rate using the reference rate TIIE 28d. In figures \ref{TIIE-FONDEO:fig} and \ref{SPREAD-TIIE-FONDEO:fig} we could see the rate levels of both rates: TIIE 28d and Fondeo Bancario, during the period 2008-2015.
     \begin{figure}[H]
          \centering
          \vspace{-2mm}
          \includegraphics[scale=0.58]{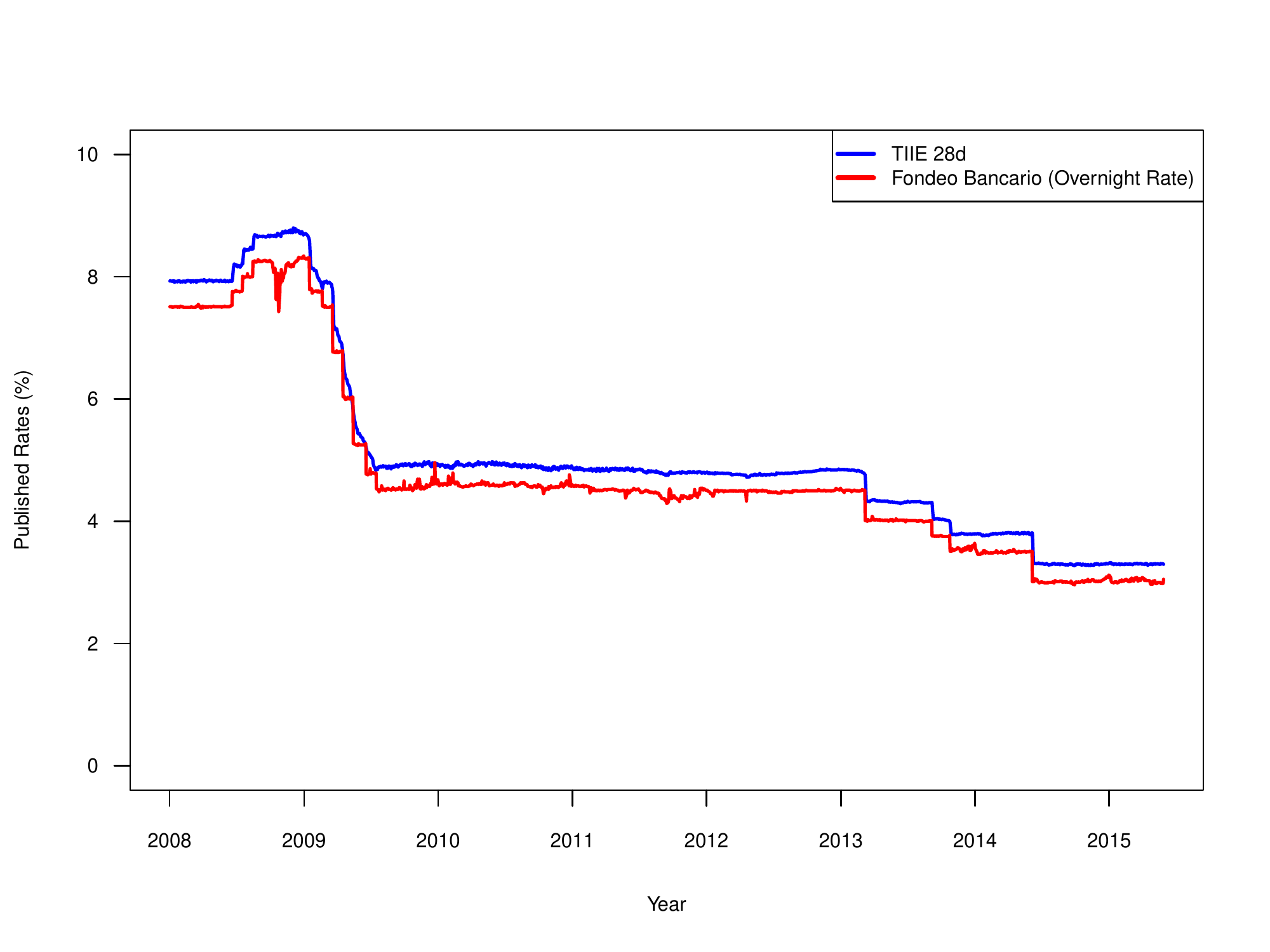}
          \caption[Overnight Rate (Fondeo Bancario) and Interbank Offered Rate (TIIE 28d) during the period 2008-2015.]{The overnight and the interbank reference rates behave similarly during the period of time 2008-2015. We could see that typically the overnight rate (Fondeo Bancario) is below the interbank offered rate (TIIE 28d), due to the lending period and the counterparty risk.}
          \label{TIIE-FONDEO:fig}
          \vspace{-2mm}
     \end{figure}
        
     \begin{figure}[H]
          \centering
          \vspace{-2mm}
          \includegraphics[scale=0.58]{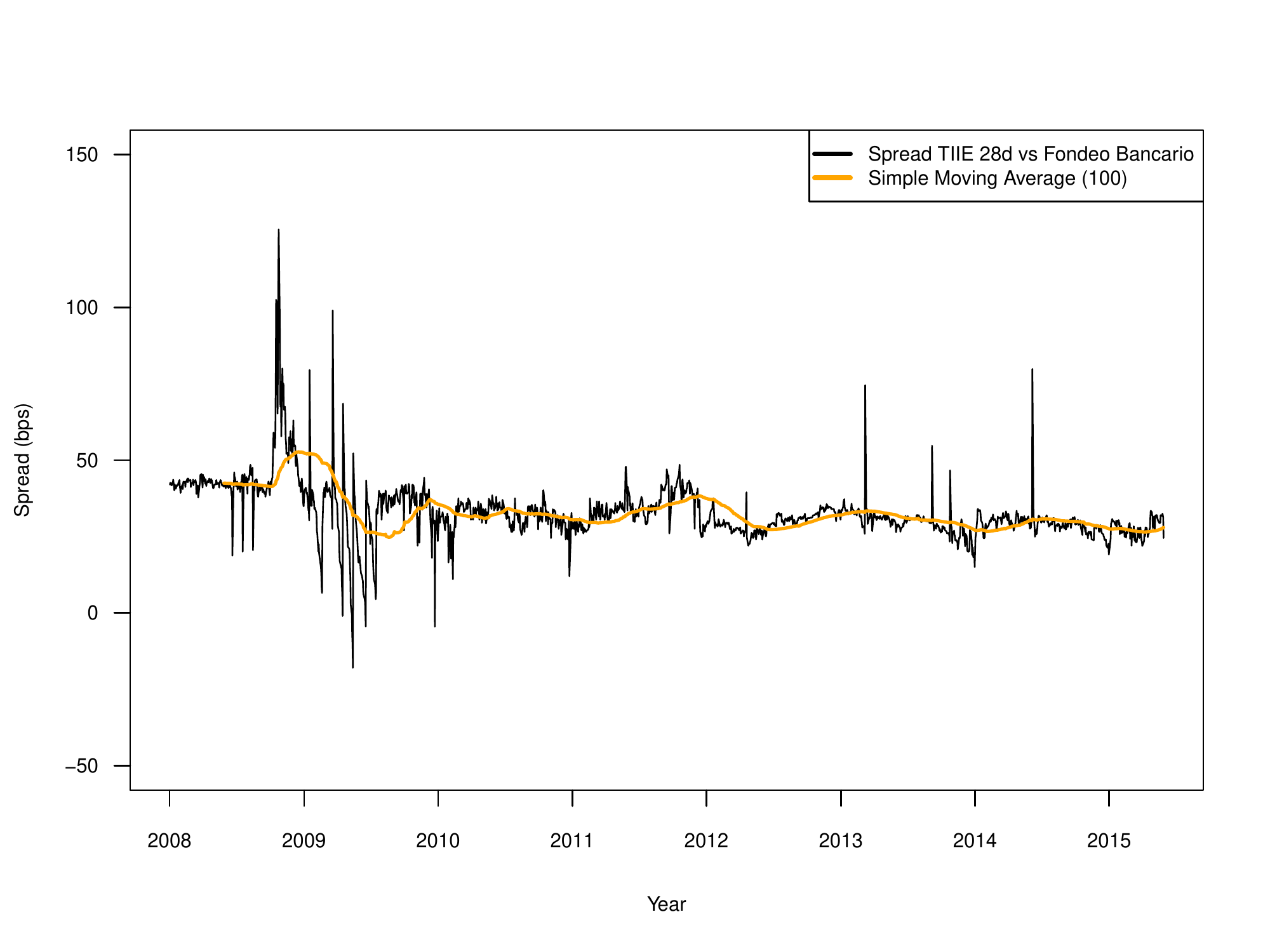}
          \caption[Spread between the Overnight Rate (Fondeo Bancario) and the Interbank Offered Rate (TIIE 28d) during the period 2008-2015.]{In this plot we could see the spread (bps) between the MXN Overnight Rate (Fondeo Bancario) and the MXN Interbank Offered Rate (TIIE 28d). The minimum value of the spread is -18bps, the maximum value is 125bps and the mean value is 32.92.}
          \label{SPREAD-TIIE-FONDEO:fig}
          \vspace{-2mm}
    \end{figure}
        
    \begin{figure}[H]
          \centering
          \vspace{-2mm}
          \includegraphics[scale=0.58]{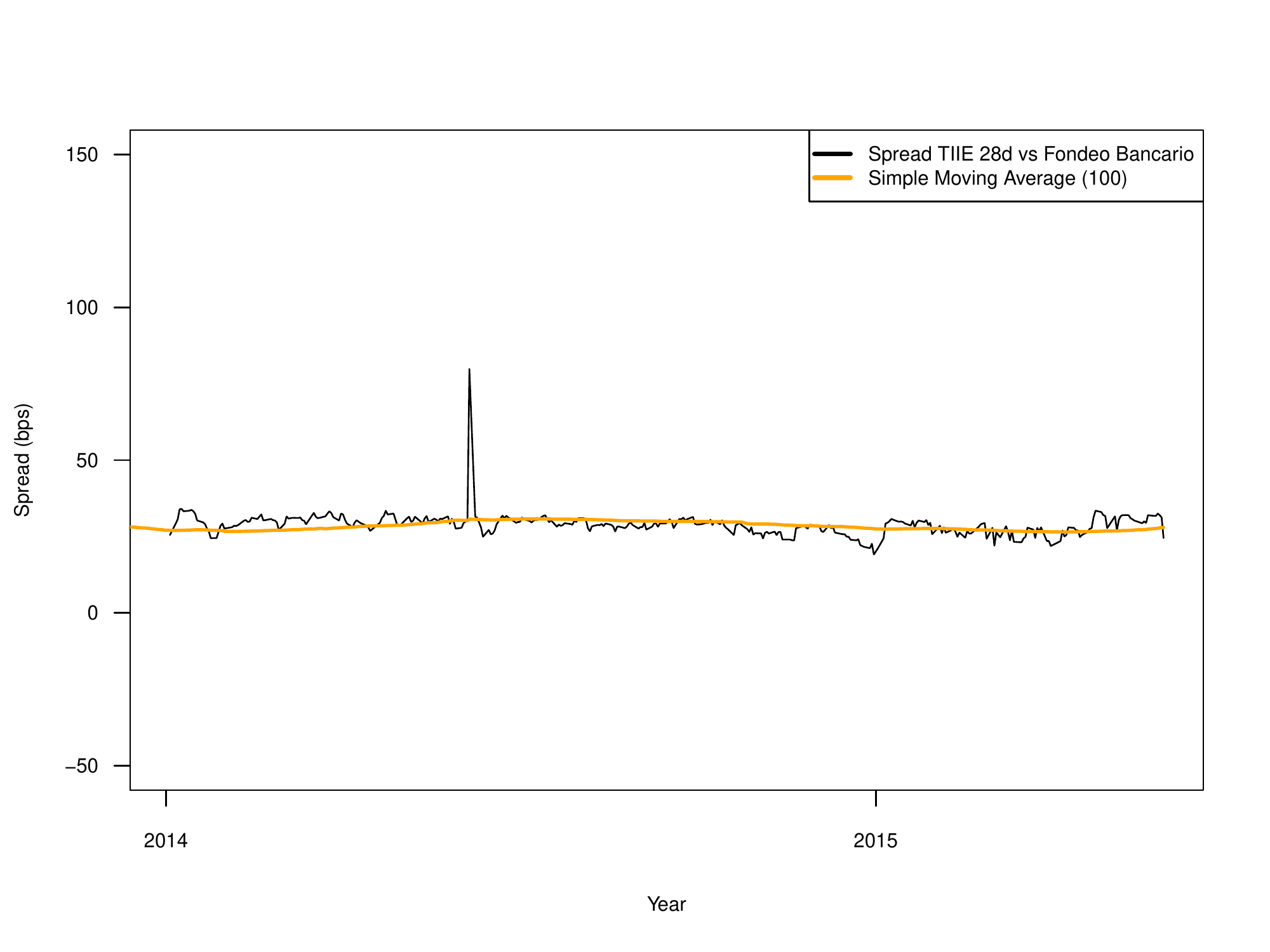}
          \caption[Spread between the Overnight Rate (Fondeo Bancario) and the Interbank Offered Rate (TIIE 28d) during the period 2014-2015.]{In this plot we present a zoom (2014-2015) of figure \ref{SPREAD-TIIE-FONDEO:fig}. The minimum value of the spread is 19bps, the maximum value is 34.05bps (excluding the value of 79.85bps) and the mean value is 28.70 (again excluding the value of 79.85bps).}
          \vspace{-2mm}
    \end{figure}
    Hence it is valid to assume that:
    \begin{equation}\label{collMXNassumption}
    	\textnormal{Overnight Rate in MXN} = \textnormal{TIIE 28d} - 29\textnormal{bps}
    \end{equation}
    Using the forward curve of TIIE 28d calibrated in section \ref{sec:calibrationMXNCOLUSDandTIIE} and equation \eqref{collMXNassumption} we can calculate the daily overnight rates (forward rates). Then it is quite simple to calculate the implied yield curve and hence we obtain the MXN-collateralized discount curve.
    \begin{rem}
    It is important to point out that the assumption of equation \eqref{collMXNassumption} could be dangerous since the spread between TIIE 28d and Fondeo Bancario can widen or reduce suddenly. Furthermore, we are unable to hedge this spread with interest rate products, hence pricing these MXN-collateralized IRSs with the assumption of constant difference between TIIE 28d and Fondeo Bancario is risky. 
    \end{rem}    
    \subsection{Calibration of the MXN Discount Curve Collateralized in EUR}
    The arguments for the construction of the MXN discount curve collateralized in EUR can be used for almost every currency in the world. Indeed, if we are able to build the discount curve for USD cash flows collateralized in a currency ABC, then by no-arbitrage arguments we could build the MXN discount curve collateralized in ABC.\par \smallskip
    Let us explain the arguments for the construction of the MXN discount curve collateralized in EUR. According to the flowchart in figure \ref{fig:flowchart}, we know that using the EURUSD mtmXCSs we get ---directly with a simple bootstrapping--- the curve $P^{c_{(\textnormal{EUR})}}_{\textnormal{USD}}(t,x)$, since the market quotes are collateralized in EUR. In fact, in section \ref{caseOfEUR} we exhibit the steps to follow for the calibration of this curve.
    \begin{prop}\label{prop:EURcoll}
    The discount curve for \textnormal{MXN} cash flows collateralized in \textnormal{EUR} is given by the following formula:
    \begin{equation}
    	P^{c_\textnormal{(EUR)}}_\textnormal{MXN}(t,x)=\frac{P^{c_\textnormal{(EUR)}}_\textnormal{USD}(t,x) P^{c_\textnormal{(USD)}}_\textnormal{MXN}(t,x)}{P^{c_\textnormal{(USD)}}_\textnormal{USD}(t,x)}.
    \end{equation}
    \end{prop}
    \begin{proof}
    Let $P^{c_\textnormal{(EUR)}}_\textnormal{MXN}(t,T)$ be the $T$-maturity MXN zero coupon bond fully-collateralized in EUR. In other words, $P^{c_\textnormal{(EUR)}}_\textnormal{MXN}(t,T)$ is the present value of MXN\$1 at time $T$. Then, if we buy the zero coupon bond we have the following cash flows:
      \begin{table}[H]
        \footnotesize
            \begin{center}
                \begin{tabular}{|c|c|}
                \hline
                \textbf{Time} & \textbf{Cash Flow (MXN)}\\
                \hline
				$t$ & $-P^{c_\textnormal{(EUR)}}_\textnormal{MXN}(t,T)$\\
                $T$ & $+1$\\
                \hline
                \end{tabular}
            \end{center}
        \end{table}
    Now we have to build a trading strategy that replicates the previous cash flows. Suppose that we buy USD\$$X$ units of a $T$-maturity USD zero coupon bond fully-collateralized in EUR. Hence we have the following cash flows:
     \begin{table}[H]
        \footnotesize
            \begin{center}
                \begin{tabular}{|c|c|}
                \hline
                \textbf{Time} & \textbf{Cash Flow (USD)}\\
                \hline
				$t$ & $-XP^{c_\textnormal{(EUR)}}_\textnormal{USD}(t,T)$\\
                $T$ & $+X$\\
                \hline
                \end{tabular}
            \end{center}
        \end{table}
    Note that these last cash flows are denominated in USD currency. Then we could enter into a FX Swap (mid market quote) that allow us to exchange the USD cash flow at time $T$ for cahs flows in MXN currency. Therefore, using the spot and forward rates we have the following MXN cash flows:
    \begin{table}[H]
        \footnotesize
            \begin{center}
                \begin{tabular}{|c|c|}
                \hline
                \textbf{Time} & \textbf{Cash Flow (MXN)}\\
                \hline
				$t$ & $-S_t X P^{c_\textnormal{(EUR)}}_\textnormal{USD}(t,T)$\\
                $T$ & $+S_T X$\\
                \hline
                \end{tabular}
            \end{center}
        \end{table}
    \noindent where $S_t$ is the spot FX rate and $S_T$ is the forward FX rate. For USD-collateralized FX forwards we know that,
    \begin{equation}
    	S_T=S_t\frac{P^{c_\textnormal{(USD)}}_\textnormal{USD}(t,T)}{P^{c_\textnormal{(USD)}}_\textnormal{MXN}(t,T)}.
    \end{equation}
    Now, let us define $X$ as
    $$X:=\frac{P^{c_\textnormal{(USD)}}_\textnormal{MXN}(t,T)}{S_t P^{c_\textnormal{(USD)}}_\textnormal{USD}(t,T)}.$$
    If we substitute $X$ into the MXN cash flows of the replication trading strategy we have that:
      \begin{table}[H]
        \footnotesize
            \begin{center}
                \begin{tabular}{|c|c|}
                \hline
                \textbf{Time} & \textbf{Cash Flow (MXN)}\\
                \hline
				$t$ & $-\frac{P^{c_\textnormal{(USD)}}_\textnormal{MXN}(t,T) P^{c_\textnormal{(EUR)}}_\textnormal{USD}(t,T)}{P^{c_\textnormal{(USD)}}_\textnormal{USD}(t,T)}$\\
                $T$ & $+1$\\
                \hline
                \end{tabular}
            \end{center}
        \end{table}
     This bring us a trading strategy that give us the same cash flow of MXN\$1 at maturity $T$ when the collateral currency is EUR. Hence, by no-arbitrage arguments,
    \begin{equation*}
    	P^{c_\textnormal{(EUR)}}_\textnormal{MXN}(t,T)=\frac{P^{c_\textnormal{(EUR)}}_\textnormal{USD}(t,T) P^{c_\textnormal{(USD)}}_\textnormal{MXN}(t,x)}{P^{c_\textnormal{(USD)}}_\textnormal{USD}(t,T)}.
    \end{equation*}
     \end{proof}
	As we mention earlier, this proposition can be reproduced for other currencies such as: CAD, JPY, GBP, CHF, BRL, etc. For the construction of the curve, We just need the discounting curve of USD cash flows collateralized in the other currency. The next proposition generalizes the previous case of EUR currency as collateral.
   \begin{prop}
    The discount curve for \textnormal{MXN} cash flows collateralized in \textnormal{ABC} is given by the following formula:
    \begin{equation}
    	P^{c_\textnormal{(ABC)}}_\textnormal{MXN}(t,x)=\frac{P^{c_\textnormal{(ABC)}}_\textnormal{USD}(t,x) P^{c_\textnormal{(USD)}}_\textnormal{MXN}(t,x)}{P^{c_\textnormal{(USD)}}_\textnormal{USD}(t,x)},
    \end{equation}
    where \textnormal{ABC} denotes any currency code \textnormal{ISO 4217}.
    \end{prop}
    \begin{proof}
    See the proof of proposition \ref{prop:EURcoll}.
    \end{proof}
    \newpage
    \section{Results}\label{sec:results}
    In this section we present the results of the curve calibration and an analysis of the factors that affect directly the curve calibration in any collateral currency. In the first subsection we present the par swap rates for plain vanilla IRSs through different collateral currencies: USD, MXN, EUR and without collateral. Then we will show how the level quotes of the cnXCSs impact on the discount curves construction and hence on the par swap rates.
    
    \subsection{Swap Rates through Different Collateral Currencies}
    In this subsection we show the differences (in basis points) that we obtained when changing the collateral currency of the IRSs based on TIIE 28d. In the table \ref{SwapsCollateral}, we compute the par swap rates (the rate that makes the mark-to-market equal to zero) for every tenor considering the following collateral currencies: USD, MXN and EUR. Also we include the par swap rates when the IRSs are uncollateralized. The inputs for the rate curves calibration were taken from tables \ref{OIS:Bloomberg}, \ref{LIBOR3M:Bloomberg}, \ref{LIBOR1M:Bloomberg}, \ref{TIIE28D:Bloomberg} and \ref{USDMXNXCS:Bloomberg}.

        \begin{table}[H]
        \footnotesize
            \begin{center}
                \begin{tabular}{|r|c|c|c|c|}
                \hline
                \textbf{Tenor} & \textbf{USD} & \textbf{No Coll} & \textbf{MXN} & \textbf{EUR}\\
                \hline
84d & 3.3200 & 3.3200 & 3.3200 & 3.3200\\
168d & 3.4300 & 3.4300 & 3.4300 & 3.4300\\
252d & 3.5620 & 3.5620 & 3.5620 & 3.5620\\
364d & 3.7350 & 3.7350 & 3.7350 & 3.7350\\
728d & 4.2360 & 4.2325 & 4.2350 & 4.2350\\
1092d & 4.6710 & 4.6650 & 4.6650 & 4.6675\\
1456d & 5.0510 & 5.0400 & 5.0425 & 5.0475\\
1820d & 5.3610 & 5.3425 & 5.3475 & 5.3550\\
2548d & 5.8630 & 5.8300 & 5.8375 & 5.8525\\
3640d & 6.2380 & 6.1900 & 6.2000 & 6.2250\\
4368d & 6.4280 & 6.3700 & 6.3800 & 6.4125\\
5460d & 6.6320 & 6.5575 & 6.5725 & 6.6150\\
7280d & 6.8310 & 6.7350 & 6.7550 & 6.8100\\
10920d & 7.0210 & 6.8900 & 6.9150 & 6.9975\\
                \hline
                \end{tabular}
            \end{center}
          \caption{Par Swap Rates of IRSs based on TIIE 28d with different collateral currencies.}
          \label{SwapsCollateral}
        \end{table}

         \begin{figure}[H]
          \centering
          \vspace{-2mm}
          \includegraphics[scale=0.55]{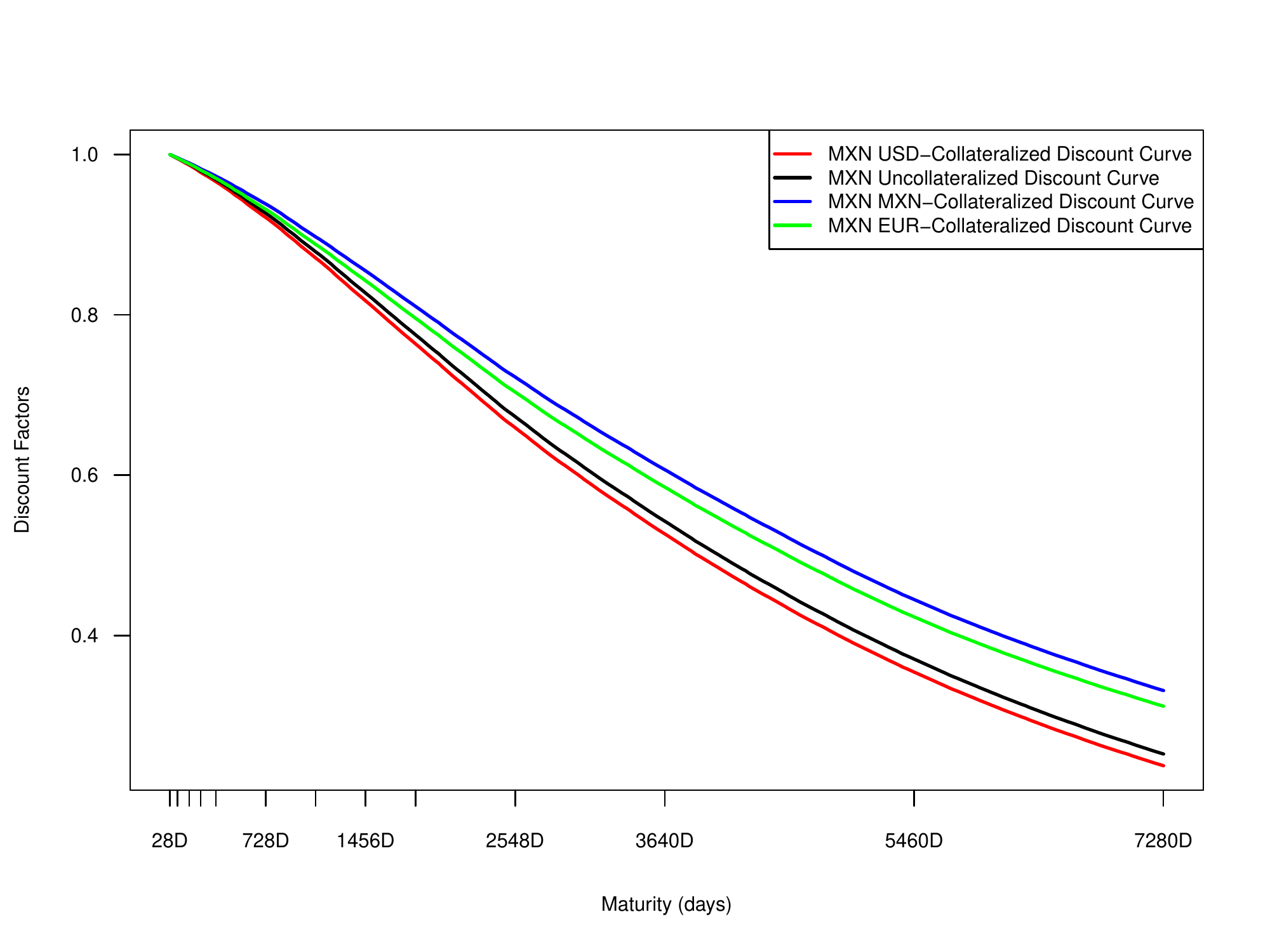}
          \caption[MXN Discount Curves in different collateral currencies]{MXN Discount Curves in different collateral currencies.}
          \vspace{-2mm}
        \end{figure}

        \begin{table}[H]
        \footnotesize
            \begin{center}
                \begin{tabular}{|r|c|c|c|}
                \hline
                \textbf{Tenor} & \textbf{No Coll} & \textbf{MXN} & \textbf{EUR}\\
                \hline
84d & 0.00 & 0.00 & 0.00\\
168d & 0.00 & 0.00 & 0.00\\
252d & 0.00 & 0.00 & 0.00\\
364d & 0.00 & 0.00 & 0.00\\
728d & -0.35 & -0.10 & -0.10\\
1092d & -0.60 & -0.60 & -0.35\\
1456d & -1.10 & -0.85 & -0.35\\
1820d & -1.85 & -1.35 & -0.60\\
2548d & -3.30 & -2.55 & -1.05\\
3640d & -4.80 & -3.80 & -1.30\\
4368d & -5.80 & -4.80 & -1.55\\
5460d & -7.45 & -5.95 & -1.70\\
7280d & -9.60 & -7.60 & -2.10\\
10920d & -13.10 & -10.60 & -2.35\\
                \hline
                \end{tabular}
            \end{center}
          \caption{Differences in basis points based on USD-collateralized par swap rates.}
          \label{SpreadsSwapsCollateral}
        \end{table}

    We can highlight that the greater the tenor is, the greater the difference in basis points is. It is important to point out that the only factor that affects the differences between the par swap rates across the different collateral currencies are the discount factors. However this discount factors are totally dependent of the cnXCSs and the USD swap market (OISs, IRSs based in LIBOR 3m and TS LIBOR 1m), recall that every swap, no matter which collateral currency we used, utilize the same TIIE 28d forward curve. In the case of EUR collateral currency the discount factors also are affected by the EUR swap market (OISs based on EONIA, mtmXCSs EURIBOR 3m vs LIBOR 3m, IRSs based on EURIBOR 6m and TSs EURIBOR 6m vs EURIBOR 3m). An analysis of the dependence of each curve can be studied through the sensitivities or deltas of each IRS. This work does not cover the calculation of sensitivities, however in the next subsection we will present an analysis of the effect of cnXCS on the curve construction is presented.

    
    \subsection{The Effect of Cross-Currency Swaps on Curve Construction}
    As we saw in the last subsection, the levels of the cnXCSs affect directly the differences between the swap rates through the different collateral currencies. So let us present an example of the effect of cnXCS on the curve construction. Using the end of day prices of IRSs based on TIIE 28d and USDMXN cnXCSs we will construct the par swap rates of the USD.collateralized swaps and the uncollateralized swaps. Additionally we will use a bunch of variations of the cnXCSs levels applying a factor $r\in (0,2]$, i.e. we will multiply the cnXCSs basis spreads by a factor $r$. In table we present the values of the cnXCS multiplied by the factor $r$.
       \FloatBarrier
        \begin{table}[H]
        	\footnotesize
            \begin{center}
                \begin{tabular}{|r|c|c|c|c|c|c|c|}
                \hline
                \textbf{Tenor} & $r=1.0$ & $r=0.1$ & $r=0.5$ & $r=0.7$ & $r=1.5$ & $r=2.0$\\
                \hline
84d & 0.5500 & 0.0550 & 0.2750 & 0.3850 & 0.8250 & 1.1000\\
168d & 0.6000 & 0.0600 & 0.3000 & 0.4200 & 0.9000 & 1.2000\\
252d & 0.6500 & 0.0650 & 0.3250 & 0.4550 & 0.9750 & 1.3000\\
364d & 0.6900 & 0.0690 & 0.3450 & 0.4830 & 1.0350 & 1.3800\\
728d & 0.8150 & 0.0815 & 0.4075 & 0.5705 & 1.2225 & 1.6300\\
1092d & 0.8900 & 0.0890 & 0.4450 & 0.6230 & 1.3350 & 1.7800\\
1456d & 0.9800 & 0.0980 & 0.4900 & 0.6860 & 1.4700 & 1.9600\\
1820d & 1.0200 & 0.1020 & 0.5100 & 0.7140 & 1.5300 & 2.0400\\
2548d & 1.0800 & 0.1080 & 0.5400 & 0.7560 & 1.6200 & 2.1600\\
3640d & 1.1000 & 0.1100 & 0.5500 & 0.7700 & 1.6500 & 2.2000\\
4368d & 1.1000 & 0.1100 & 0.5500 & 0.7700 & 1.6500 & 2.2000\\
5460d & 1.0800 & 0.1080 & 0.5400 & 0.7560 & 1.6200 & 2.1600\\
7280d & 1.0850 & 0.1085 & 0.5425 & 0.7595 & 1.6275 & 2.1700\\
10920d & 1.0850 & 0.1085 & 0.5425 & 0.7595 & 1.6275 & 2.1700\\
                \hline
                \end{tabular}
            \end{center}
          \caption{In this table we present the values of the cnXCS basis spreads that we use for the analysis if the effect .}
          \label{cnXCSbyr}
        \end{table}
    \FloatBarrier
         \begin{figure}[H]
          \centering
          \vspace{-2mm}
          \includegraphics[scale=0.55]{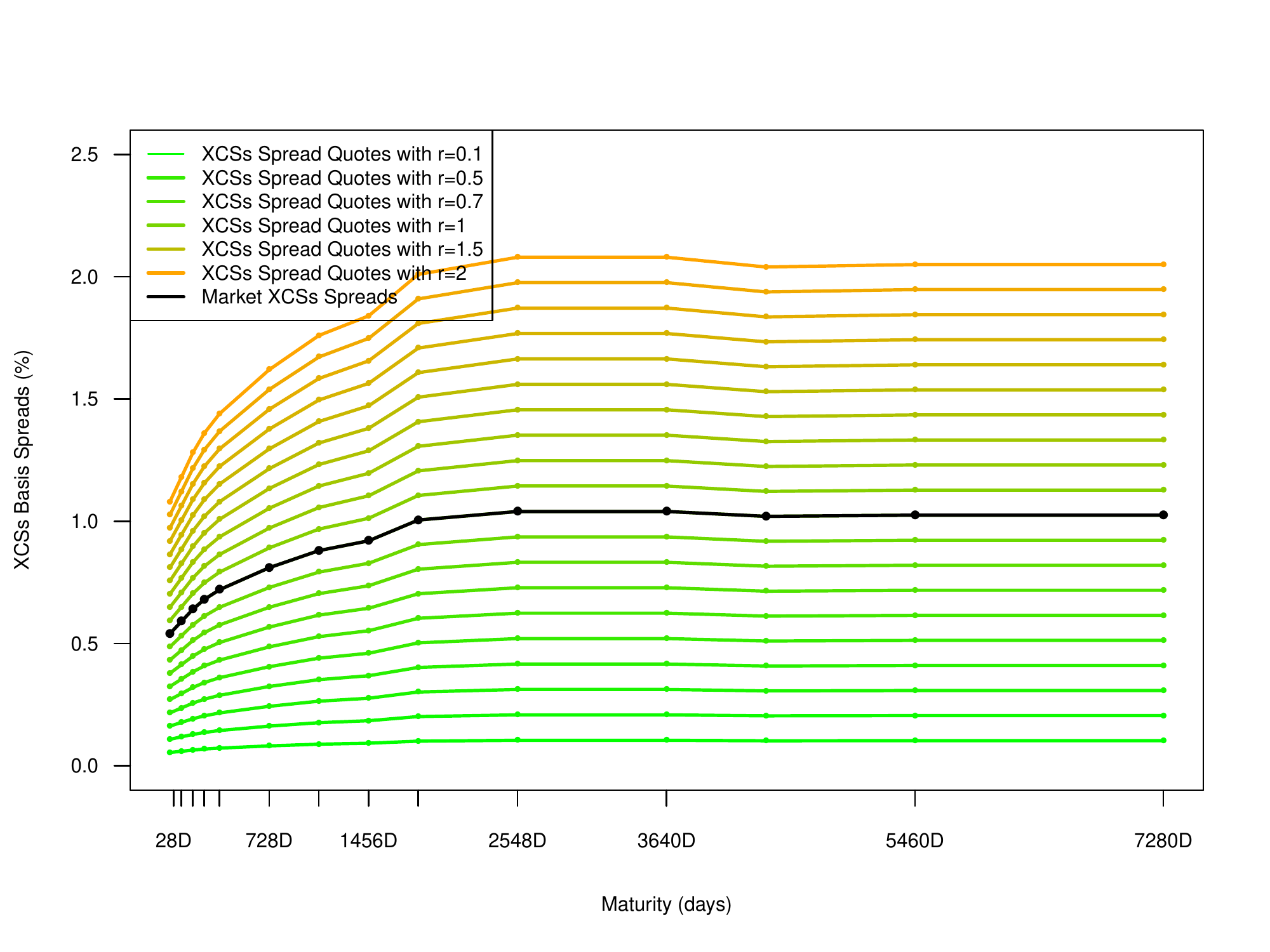}
          \caption[The Effect of XCSs in the Mexican Swap Market]{In this plot we present the levels used for the analysis of the effect of cnXCS. }
          \vspace{-2mm}
        \end{figure}
    
        \begin{figure}[H]
          \centering
          \vspace{-2mm}
          \includegraphics[scale=0.55]{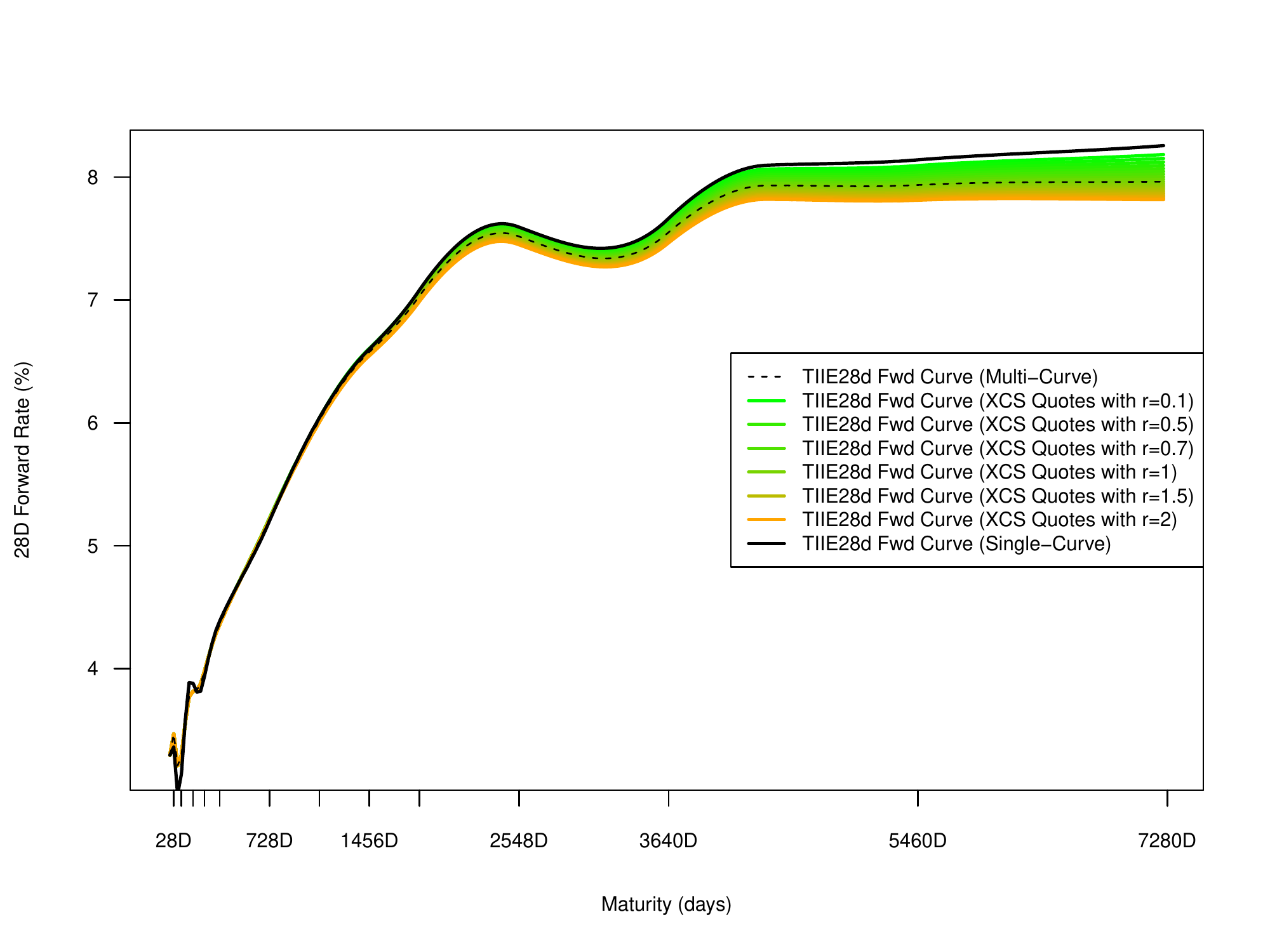}
          \caption[The Effect of XCSs in the TIIE 28d Forward Curve]{In this plot we could see the effect of the XCSs quotes level in the construction of the TIIE 28d Forward curve. It is easy to see that when the spread of the XCS is minimum the TIIE 28d forward curve obtained from the \emph{multi-curve framework} converges to the TIIE 28d forward curve obtained from the \emph{single-curve framework}. Additionally, as spreads of XCSs became bigger in scale, the TIIE 28d forward curve decreases, particularly in the long part of the curve.}
          \vspace{-2mm}
        \end{figure}
        
        \begin{figure}[H]
          \centering
          \vspace{-2mm}
          \includegraphics[scale=0.55]{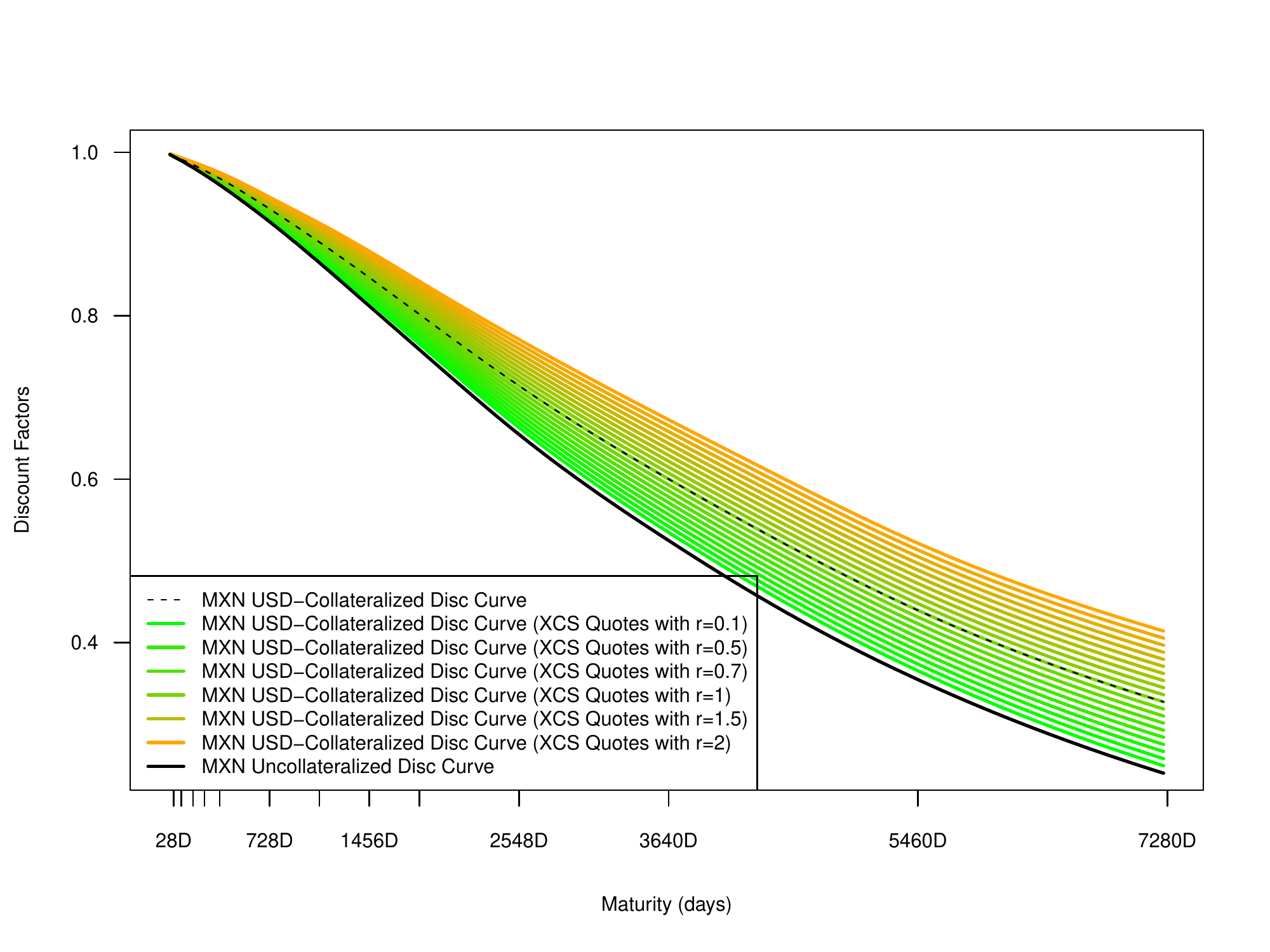}
          \caption[The Effect of XCSs in the Discount Curve]{In this plot we could see the effect of the XCSs quotes level in the construction of the discounting curves.}
          \vspace{-2mm}
        \end{figure}
    
         \begin{figure}[H]
          \centering
          \vspace{-2mm}
          \includegraphics[scale=0.55]{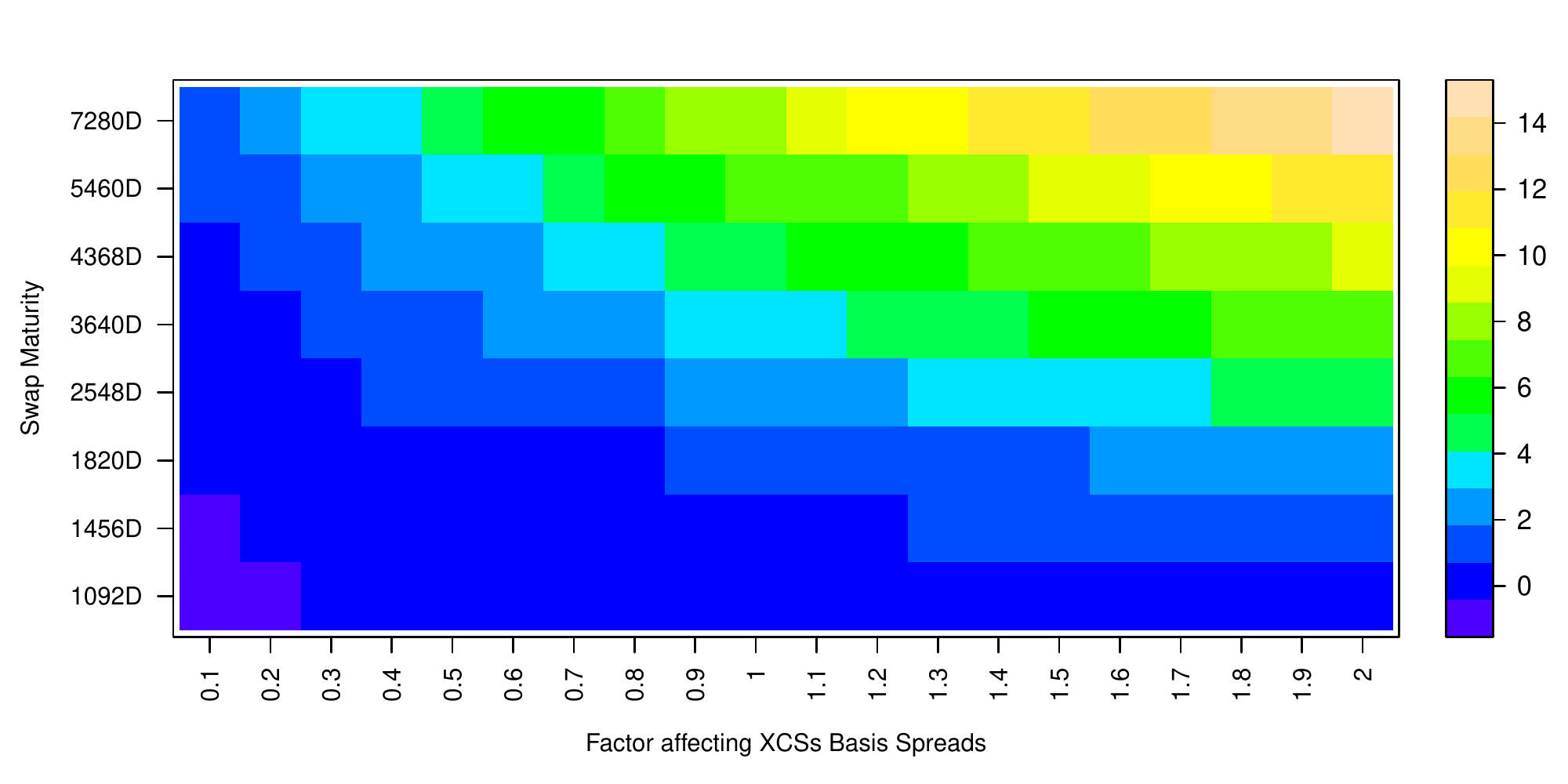}
          \caption[The Effect of cnXCSs: USD-collateralized vs uncollateralized swaps]{In this figure we present the effect of cnXCSs in the definition of the par swap rates when the IRS is uncollateralized. In $x$-axis, the factors that affect proportionally all the curve of cnXCS basis spreads are labeled, whereas $y$-axis displayed the maturities of the plain vanilla IRSs (we do not include maturities below 3 years since the effect of collateral is almost nil). We can conclude from this heatmap that, the greater the basis spread is, the greater the difference between USD-collateralized par swap rates and uncollateralized par swap rates is. Indeed, for $r=0.1$ we have that the greatest differences are for the tenors of 15 and 20 years, nevertheless this difference is of approximately 2 basis points. Now, when $r=2.0$, for the IRSs with maturity of 7 years have a discrepancy of 6 basis points with the USD-collateralized swap, while for an IRS of maturity of 20 years a difference is of about 14 basis points.}
          \vspace{-2mm}
        \end{figure}
        
        \begin{figure}[H]
          \centering
          \vspace{-2mm}
          \includegraphics[scale=0.55]{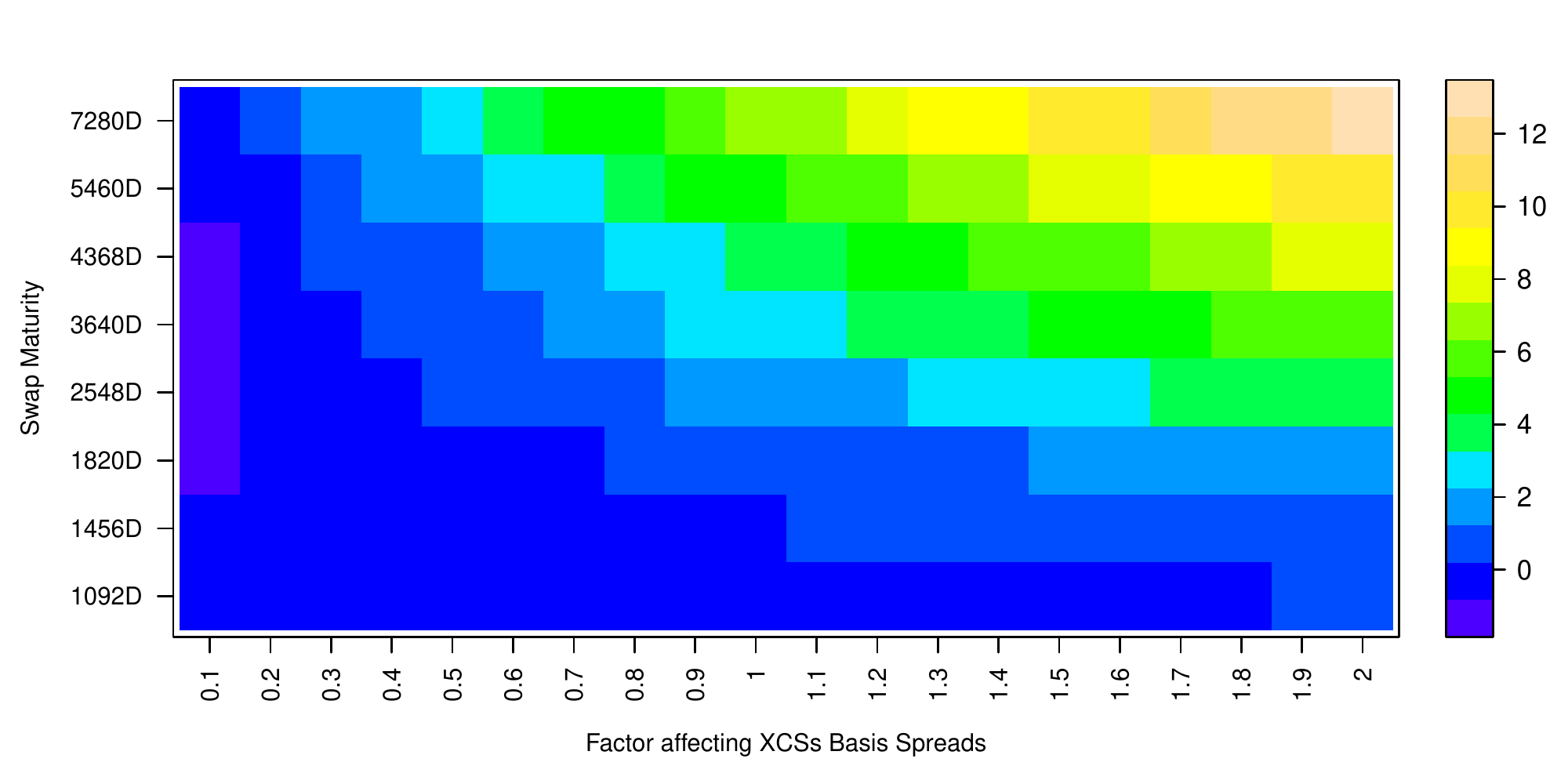}
          \caption[The Effect of cnXCSs: USD-collateralized vs MXN-collateralized swaps]{In this figure we present the effect of cnXCSs in the definition of the par swap rates when the IRS is collateralized in MXN. In $x$-axis, the factors that affect proportionally all the curve of cnXCS basis spreads are labeled, whereas $y$-axis displayed the maturities of the plain vanilla IRSs (we do not include maturities below 3 years since the effect of collateral is almost nil).}
          \vspace{-2mm}
        \end{figure}

    \newpage
    \section{Conclusions and Further Research}\label{sec:conclusions}
    As we have shown, the \emph{multi-curve} framework is totally different from the \emph{single-curve} framework since it requires more interest rate products for the curve construction. In this work, we have examined the assumptions of the \emph{multi-curve} framework in which a collateral account is included. The main idea of the \emph{multi-curve} framework consists on defining a coherent valuation framework that, with the absence of arbitrage, enable us to price interest rate product considering the currency of the collateral agreement of the counterparty.\par\medskip
    The purpose of this thesis was to apply the valuation framework in the Mexican swap market and to explore the differences in methodology for the curve calibration through distinct the collateral currencies.\par\medskip
    The contribution of this work compared to previous research, in particular \cite{mexder2014}, is that we explicitly presented the algorithms for the calibration of the discount and projection curves in the Mexican swap market. Moreover, we treated carefully the steps to follow in a simple and in a multiple bootstrappings, since we explained how to perform the interpolation of the curve through the iterative process. It is important to point out that the approach of this \emph{multi-curve} framework could be applied indistinctly for any emerging market currency. Indeed, throughout this work we can replace MXN for BRL, ZAR or RUB and the formulas and bootstrapping algorithms remain valid.\par\medskip
    It is concluded that, in contrast to a G7-currency market, the calibration of the \emph{multi-curve} framework in MXN currency is slightly difficult since it requires a multiple bootstrapping for the calibration of the forward and discount curves. Furthermore, we showed that the no-existence of an OIS market denominated in MXN currency exclude us to define the curves with a simple bootstrapping.\par\medskip
    
    
    \newpage
    \addcontentsline{toc}{section}{References}
    \bibliographystyle{apalike}
    \bibliography{reference}
    
    \newpage
    \appendix
    \addcontentsline{toc}{section}{Appendices}
    
    \section{Proofs }\label{App:Appendix0}
    \subsection{Proof of Theorem \ref{Teo:AlphaColl}}
    
    \begin{proof}
        Let us say that the forward measure $\mathbb{T}^{\alpha c}$ has associated the numéraire $P^{\alpha c}(t,T)$. This numéraire corresponds to the price of the partially collateralized ($\alpha$) zero coupon bond. Using theorem \ref{ChangeMeasureTheo} and the risk neutral measure $\mathbb{Q}$ with the numéraire $X(t):=\exp(-\int_t^T (1-\alpha)r(s) + \alpha c(s)ds)$ we have that
    \begin{equation}
        \zeta_\alpha(t)=\mathbb{E}^{\mathbb{Q}}_t \bigg( \frac{d\mathbb{T}^{\alpha c}}{d\mathbb{Q}} \bigg)=\frac{P^c(t,T)X(0)}{P^c(0,T)X(t)}.
    \end{equation}
    Using equation \eqref{eq5.9} and theorem \ref{teozeta} we get that
    \begin{align*}
        h(t)&=\mathbb{E}^\mathbb{Q}_t \left[ e^{-\int_t^T (1-\alpha)r(s) + \alpha c(s)ds} h(T) \right]\\
            &=\mathbb{E}_t^\mathbb{Q}\left[ X(t)h(T) \right]\\
            &=\mathbb{E}_t^\mathbb{Q}\left[ P^{\alpha c}(t,T)\frac{\zeta_\alpha(T)}{\zeta_\alpha(t)} h(T) \right]\\
            &=\zeta_{\alpha}^{-1}(t) \mathbb{E}_t^\mathbb{Q}\left[ \zeta_\alpha(T) P^{\alpha c}(t,T)h(T) \right]\\
            &=\mathbb{E}_t^{\mathbb{T}^{\alpha c}}\left[ P^{\alpha c}(t,T) h(T) \right]\\
            &=P^{\alpha c}(t,T) \mathbb{E}_t^{\mathbb{T}^{\alpha c}}\left[ h(T) \right].
    \end{align*}
    \end{proof}

    \newpage
    \section{Conventions \& Calendars}\label{App:AppendixA}
    In this appendix, we will define the conventions and calendars used within the work. For a more complete detailed information of the conventions and calendars throughout the different interest rates markets around world, see \cite{henrard2012interest}.
    \subsection{Day Count Conventions}
    \textbf{ACT/360:} The year fraction between two dates is computed by dividing the actual days between two dates by 360, i.e.,
    \begin{equation}
        \textnormal{Year Fraction }=\frac{\textnormal{Actual Days}}{360}.
    \end{equation}
    \textbf{ACT/365:} The year fraction between two dates is computed by dividing the actual days between two dates by 365, i.e.,
    \begin{equation}
        \textnormal{Year Fraction }=\frac{\textnormal{Actual Days}}{360}.
    \end{equation}
    \textbf{ACT/ACT (ISDA):} The year fraction is computed by making difference between days in a leap year and days in a non-leap year, i.e.,
    \begin{equation}
        \textnormal{Year Fraction }=\frac{\textnormal{Days in a leap year}}{366}+\frac{\textnormal{Days in a non-leap year}}{365}.
    \end{equation}
    \textbf{30/360:} The year fraction is computed by assuming that months have 30 days and years have 360, i.e.,
    \begin{equation}
        \textnormal{Year Fraction }=\frac{360\cdot(Y_2-Y_1)+30\cdot(M_2-M_1)+(D_2-D_1)}{360}.
    \end{equation}

    \subsection{Date Rolling Conventions}
    \noindent \textbf{Following:} If a payment date falls on a non-business day, the payment date is set to the next business day.\par \smallskip
    \noindent Example: Start date: 30-Oct-2014 with a period 1 month $\rightarrow$ End date: 01-Dec-2014\par \bigskip
    \noindent \textbf{Modified Following:} If a payment date falls on a non-business day, the payment date is set to the next business day with the exception that if the next business day falls in a new month then the payment date is set to the last business day.\par \smallskip
    \noindent Example: Start date: 30-Oct-2014 with a period 1 month $\rightarrow$ End date: 28-Nov-2014\par \smallskip
    \subsection{Calendars}
    In this work we used MX, US and UK calendars. The difference among them relies on the holidays asumed by each central bank or the responsible for determining the fixings of certain rates. Below you will find the holidays of each region.\par \bigskip
    \noindent \textbf{MX (Mexico):} New Year's Day, Constitution Day, Birth Anniversary of Benito Ju\'arez, Maundy Thursday, Good Friday, International Worker's Day, Mexican Independence Day, Day of the dead, Mexican Revolution Day, Bank Holiday and Christmas Day.\bigskip
    
    \noindent \textbf{US (United States of America):} New Year's Day, Martin Luther King, Jr. Day, Washington's Birthday, Good Friday, Memorial Day, United States of America Independence Day, Labor Day, Columbus Day, Veterans Day and Christmas Day.\bigskip
    
    \noindent \textbf{UK (United Kindom):}  New Year's Day, Good Friday, Easter Monday, Easter May Bank Holiday, Spring Bank Holiday, Summer Bank Holiday, Christmas Day and Boxing Day.\bigskip
    
    \noindent \textbf{TARGET (Trans-European Automated Real-time Gross settlement Express Transfer system):}  New Year's Day, Good Friday, Easter Monday, Labor Day, Christmas Day and Boxing Day.\bigskip
    
    It is important to mention that the MX Calendar is also known as Mexico City Calendar, whereas the US calendar is known as New York Calendar and the UK calendar is known as London  Calendar.
    
    \newpage
    \section{Methods of Interpolation}\label{App:AppendixB}
    Interpolation methods are essential for the construction of interest rates curves. Indeed, the bootstrapping techniques require an interpolation method. It is well-known that the simplest method for interpolating between two points is by connecting them through a straight line. This method can be applied to a variety of functions, such as the spot rate (zero curve), the discount factor curve, the forward curve, etc. In order to produce continuous forward rates, researchers (either in academia or industry) often apply cubic methods for interpolations. These cubic methods are fitted by cubic piecewise polynomials between knot (or nodal) points. The parameters of the polynomials can be computed by establishing a variety of conditions, such as continuity, differentiability, monotonicity, etc. For further references and a complete introduction of interpolation methods for interest rates curve see \cite{hagan2006interpolation}, \cite{hagan2008methods} and \cite{ron2000practical}. Specifically in \cite{hagan2006interpolation}, a wide array of possible interpolation techniques are discussed.
    
    \subsection{Linear Interpolation on Yield Curve}
    This technique assumes that the yield curve (zero rate curve) is linear between the nodal points and flat extrapolation outside nodal points. Consider $n+1$ knot points $(t_1, r_1), \dots, (t_{n+1}, r_{n+1})$. The linear interpolation function for $R(t)$ is defined as follows
    \begin{equation}\label{linearR}
        R(t)=\sum_{i=1}^{n} R_i(t) \mathbbm{1}_{ \{ t_{i} \leq t < t_{i+1} \} }(t), 
    \end{equation}
    where
    \begin{equation}
        R_i(t)=a_i+b_i(t-t_i), \hspace{2mm} a_i,b_i \in \mathbb{R}.
    \end{equation}
    $R_i(t)$ functions are called the linear piecewise functions. The piecewise functions require to to satisfy the following conditions:
    \begin{enumerate}
        \item $R(t_i)=r_i$ for all $i=1,\dots,n+1$ (function should returns all the inputs necessarily)
        \item $R \in \mathcal{C}^0(t_1,t_{n+1})$\footnote{$\mathcal{C}^0(a,b)=\{ f:(a,b)\rightarrow \mathbb{R} : f \text{ is continuous} \}$} (function should be continuous)
    \end{enumerate}
    The first condition may be satisfied by requering that $R_i(t_i)=r_i$ for all $i=1\dots,n$, therefore
    \begin{equation}
        R_i(t_i)=r_i \hspace{3mm} \forall i \hspace{1mm} \Longrightarrow \hspace{1mm} a_i:=r_i \hspace{3mm} \forall i.
    \end{equation}
    Then, condition of continuity can be achived by requiring that $R_i(t_{i+1})=r_{i+1}$ for all $i=1\dots,n$, hence
    \begin{align*}
        R_i(t_{i+1})=r_{i+1} \hspace{3mm} \forall i \hspace{1mm} \Longrightarrow & \hspace{1mm} r_i+b_i(t_{i+1}-t_{i})=r_{i+1} \hspace{3mm} \forall i, \\
                                    \Longleftrightarrow & \hspace{1mm} b_i:=\frac{r_{i+1}-r_{i}}{t_{i+1}-t_{i}} \hspace{3mm} \forall i.
    \end{align*}
    Consequently, the interpolation function is given by
    \begin{equation}
        R(t)=\sum_{i=1}^{n} \bigg( \frac{r_{i+1}-r_{i}}{t_{i+1}-t_{i}} (t-t_i) + r_i \bigg) \mathbbm{1}_{ \{ t_{i} \leq t < t_{i+1} \} }(t) +
        r_1 \mathbbm{1}_{ \{ t < t_1 \} }(t) + r_{n+1} \mathbbm{1}_{ \{ t > t_{n+1} \} }(t).
    \end{equation}
    \subsection{Linear Interpolation on Log Discount Factors}
    Similarily the technique previously mention, this method assumes that the logarithm of the discount curve is linear between the nodal points and flat extrapolation outside nodal points. Consider $n+1$ knot points $\{ (t_i, \log(P(t_i))) \}_{i=1}^{n+1}$. In this case, $P(t_1)$ corresponds to $P(t,t_1)$ following the notation used through the work, but to simplify the notation we will ignore the $t$. Following the same arguments mentioned in the linear yield curve technique, it can be shown that
    \begin{align*}
        \log (P(t))&=\sum_{i=1}^{n} \bigg( \frac{\log (P(t_{i+1}))-\log (P(t_{i}))}{t_{i+1}-t_{i}} (t-t_i) + \log (P(t_i)) \bigg) \mathbbm{1}_{ \{ t_{i} \leq t < t_{i+1} \} }(t)\\ &+
        \log (P(t_1)) \mathbbm{1}_{ \{ t < t_1 \} }(t) + \log (P(t_{n+1})) \mathbbm{1}_{ \{ t > t_{n+1} \} }(t).
    \end{align*}

    \subsection{Natural Cubic Splines Interpolation on Yield Curve}\label{app:naturalcubic}
    This technique assumes that the yield curve (zero rate curve) is a cubic spline . A cubic spline is a smooth polynomial function of degree three that is piecewise-defined. Consider $n+1$ knot points $(t_1, r_1), \dots, (t_{n+1}, r_{n+1})$. The cubic interpolation function for $R(t)$ is defined as follows
    \begin{equation}\label{cubicR}
        R(t)=\sum_{i=1}^{n} R_i(t) \mathbbm{1}_{ \{ t_{i} \leq t < t_{i+1} \} }(t), 
    \end{equation}
    where
    \begin{equation}
        R_i(t)=a_i+b_i(t-t_i)+c_i(t-t_i)^2+d_i(t-t_i)^3, \hspace{2mm} a_i,b_i,c_i,d_i \in \mathbb{R}.
    \end{equation}
    $R_i(t)$ functions are called the cubic piecewise functions.
    Using the notation of \cite{hagan2006interpolation} and \cite{du2011investigation} we define $h_i=t_{i+1}-t_i$ and $m_i=(a_{i+1}-a_i)/h_i$. Interpolation methods are constructed based on a criteria that satisfy some conditions (or constraints). These conditions allow us to build a unique and consistent piecewise continuous function. In the case of cubic interpolation, typically it is requiered to satisfy the following conditions:
    \begin{enumerate}
    \item $R(t_i)=r_i$ for all $i=1,\dots,n+1$ (function should returns all the inputs)
    \item $R \in \mathcal{C}^1(t_1,t_{n+1})$\footnote{$\mathcal{C}^n(a,b)=\{ f:(a,b)\rightarrow \mathbb{R} : f^{(n)}:(a,b)\rightarrow \mathbb{R} \text{ is continuous} \}$} (function should be continuous and differentiable)
    \end{enumerate}
    According to equation \eqref{cubicR}, the first condition is satisfied by requiring that 
    \begin{equation}\label{set1}
        R(t_i)=R_i(t_i)=r_i \hspace{1mm} \Longrightarrow \hspace{1mm} a_i:=r_i, \hspace{4mm} \text{for } i=1\dots,n
    \end{equation}
    Additionally, we need that the $n$th piecewise function $R_n(t)$ returns $r_{n+1}$ when $t=t_{n+1}$, i.e.,
    \begin{equation}\label{set2}
        a_{n}+b_{n}h_{n}+c_{n}h^2_{n}+d_{n}h^3_{n}=r_{n+1}:=a_{n+1}.
    \end{equation}
    Condition of continuity in the yield curve can be achieved by requering that
    \begin{equation}
        \lim_{x \rightarrow t_i^{-}} R(x) = \lim_{x \rightarrow t_i^{+}} R(x)
    \end{equation}
    This condition could be met by forcing that $R_{i}(t_{i+1})=R_{i+1}(t_{i+1})$, i.e.,
    \begin{equation}\label{set3}
        a_i+b_i h_i+c_i h_i^2+d_i h_i^3=a_{i+1}, \hspace{4mm} \text{for } i=1\dots,n-1.
    \end{equation}
    Then, condition of differentiability can be achieved by requiring that $R_{i}'(t_{i+1})=R_{i+1}'(t_{i+1})$, i.e.,
    \begin{equation}\label{set4}
        b_{i}+2c_ih_i+3d_ih_i^2=b_{i+1}, \hspace{4mm} \text{for } i=1\dots,n-1.
    \end{equation}
    Note that from the set of equations \eqref{set1}, \eqref{set2}, \eqref{set3} and \eqref{set4} we defined a system of $3n-1$\footnote{$n+1+(n-1)+(n-1)=3n-1$} equations with $4n$ unknown variables. Hence, we need an extra $n+1$ constraints to solve this system. The natural cubic splines method allows us to define natural boundary conditions to solve the system of equations. The conditions are as following:
    \begin{enumerate}
    \item $R \in \mathcal{C}^2(t_1,t_{n+1})$ (function should be twice differentiable)
    \item $R''(t_1)=0$ and $R''(t_{n+1})=0$ (second derivative is zero at the two extreme breaks)
    \end{enumerate}
    Condition of smoothness (twice differentiable) can be achieved by requiring that $R_{i}''(t_{i+1})=R_{i+1}''(t_{i+1})$ where $R_i''(t)=2c_i+6d_i(t-t_i)$, hence
    \begin{equation}\label{smoothness_equations}
        c_i+3d_ih_i=c_{i+1}, \hspace{4mm} \text{for } i=1\dots,n-1.
    \end{equation}
    Then, it is easy to see that the condition of smoothness on the extremes give us the following equations
    \begin{align}
        c_1&=0\\
        c_n+3d_nh_n&=0 \label{second_deriv_zero}
    \end{align}
    Note that the equation \eqref{smoothness_equations} claims that $c_n+3d_nh_n=c_{n+1}$ and \eqref{second_deriv_zero} that $c_n+3d_nh_n=0$, therefore $c_{n+1}=0$. By imposing these natural conditions, the system now has $4n$ equations with $4n$ unknown variables. This system has a unique solution if and only if the coefficient matrix of the system is nonsingular.
    \begin{defi}
    A matrix $A \in \mathbb{R}^{n \times n}$ is said to be diagonally dominant when
    \begin{equation}\label{defi_diag_domin}
        |a_{ii}|\geq \sum_{\substack{j=1 \\ j\neq i}}^n |a_{ij}|,
    \end{equation}
    holds for each $i=1,\dots,n$.
    \end{defi}
    A diagonally dominant matrix is said to be strictly diagonally dominant when the inequality in \eqref{defi_diag_domin} is strict for each $n$.
    \begin{teo}\label{teo_diag_domin}
    A strictly diagonally dominant matrix $A$ is nonsingular.
    \end{teo}
    \begin{proof}
    See \cite{burden2010numerical} page 412.
    \end{proof}
    Consequently, to prove that the system has a unique solution, is enough to show that it can be represent with a strictly diagonally dominant coeficients matrix. Indeed, solving for $d_i$ in equation \eqref{smoothness_equations} give us,
    \begin{equation}\label{smoothness_equations_mod}
        d_i=\frac{c_{i+1}-c{i}}{3h_i},\hspace{4mm} \text{for } i=1\dots,n-1.
    \end{equation}
    Substituing these values in \eqref{set3} and \eqref{set4} give us,
    \begin{align}
        a_{i+1}=a_i+b_ih_i+\frac{h_i^2}{3}(2c_i+c_{i+1}), \hspace{4mm} &\text{for } i=1\dots,n-1; \label{newEqA}\\
        b_{i+1}=b_{i}+h_i(c_i+c_{i+1}), \hspace{4mm} &\text{for } i=1\dots,n-1.\label{newEqB1}
    \end{align}
    Now from equation \eqref{newEqA} if we solve for $b_i$ we obtain
    \begin{equation}
        b_i=\frac{(a_{i+1}-a_i)}{h_i}-\frac{h_i}{3}(2c_i+c_{i+1}), \hspace{4mm} \text{for } i=1\dots,n-1. \label{newEqB2}
    \end{equation}
    With a reduction of the index $i$ we get
    \begin{equation}
        b_{i-1}=\frac{(a_{i}-a_{i-1})}{h_{i-1}}-\frac{h_{i-1}}{3}(2c_{i-1}+c_{i}), \hspace{4mm} \text{for } i=2\dots,n. \label{newEqB3}
    \end{equation}
    Thus, when substituing \eqref{newEqB2} and \eqref{newEqB3} in equation \eqref{newEqB1} we obtain
    \begin{equation}
        \frac{3}{h_i}(a_{i+1}-a_{i})-\frac{3}{h_{i-1}}(a_{i}-a_{i-1})=h_ic_{i+1}+2(h_i+h_{i-1})c_i+h_{i-1}c_{i-1}, \hspace{2mm} \text{for } i=2\dots,n. \label{system_eqs}
    \end{equation}
    Recall that $\{a_i\}_{i=1}^{n+1}$ and $\{h_i\}_{i=1}^{n}$ are given, therefore the system of equations \eqref{system_eqs} involves only the $\{c_i\}_{i=1}^{n}$ as unknowns. Once the values of $\{c_i\}_{i=1}^{n}$ are determined, it is a simple task to find the values of $\{c_i\}_{i=1}^{n}$ and $\{d_i\}_{i=1}^{n}$ from equations \eqref{newEqB2} and \eqref{smoothness_equations_mod}, respectively. Using the two equations $c_1=0$ and $c_{n+1}=0$ together with the equations in \eqref{system_eqs} we may produce a linear system described by the following vector equation
    \begin{equation}
        \begin{bmatrix}
            {1}      & {0}          & {0}             & {\cdots} & {\cdots}          & { 0 }   \\
            {h_1}    & {2(h_1+h_2)} & {h_2}           & {0}      & {\cdots}          & {\vdots}\\
            { 0 }    & {h_2}        & {2(h_2+h_3)}    & {h_3}    & {\ddots}          & {\vdots}\\
            {\vdots} & {\ddots}    & {\ddots}        & {\ddots} & {\ddots}           & {0}\\
            {\vdots} & {   }        & {\ddots}        &{h_{n-1}} & {2(h_{n-1}+h_n)}  & {h_n}\\
            { 0 }    & {\cdots}     & {\cdots}        & {0}      & {0}               & { 1 }\\
        \end{bmatrix}
        \begin{bmatrix}
            {c_1}  \\
            {c_2}  \\
            {c_3}  \\
            \vdots   \\
            {c_{n}}   \\
            {c_{n+1}}  \\
        \end{bmatrix}
        =
        \begin{bmatrix}
            {0}  \\
            {\frac{3}{h_2}(a_{3}-a_{2})-\frac{3}{h_{1}}(a_{2}-a_{1})}  \\
            {\frac{3}{h_3}(a_{4}-a_{3})-\frac{3}{h_{2}}(a_{3}-a_{2})}  \\
            \vdots   \\
            {\frac{3}{h_n}(a_{n+1}-a_{n})-\frac{3}{h_{n-1}}(a_{n}-a_{n-1})}   \\
            {0}  \\
        \end{bmatrix}
    \end{equation}
    It is easy to see that the coefficient matrix is stricly diagonally dominant, using theorem \ref{teo_diag_domin} we conclude that the system is solvable and in particular the solution is unique. Now we include an algorithm in R to solve the natural cubic splines numerically..\par \medskip
    \noindent \textbf{Natural Cubic Splines Solution Algorithm (in R)}
    \begin{lstlisting}[language=R]
CubicInterpolation<-function(x,y){
    n=length(x) # number of knot points
          
    if (length(y) == n){
        alpha=rep(0,n)
        h=rep(0,n)
        L=rep(0,n)
        mu=rep(0,n)
        Z=rep(0,n)
        Coef=matrix(0,ncol=4,nrow=n-1)
        # Coef: [ ai | bi | ci | di ]
                
        for (i in 1:(n-1)){
            h[i]=x[i+1]-x[i]
            Coef[i,1]=y[i]
        }
              
        for (i in 2:(n-1)){
            alpha[i] = (3/h[i]*(y[i+1]-y[i])) - (3/h[i-1]*(y[i] - y[i-1]))
        }
              
        L[1]=1
        mu[1]=0
        Z[1]=0
              
        for (i in 2:(n-1)){
            L[i]=2*(x[i+1]-x[i-1]) - h[i-1]*mu[i-1]
            mu[i]=h[i]/L[i]
            Z[i]=(alpha[i]-h[i-1]*Z[i-1])/L[i]
        }
              
        L[n]=1
        Z[n]=0
                
        for (i in seq(from=n-1,to=1,by=-1)){
            if (i == n-1){
                Coef[i,3]=Z[i]
                Coef[i,2]=(y[i+1]-y[i])/h[i] - h[i]*2*Coef[i,3]/3
                Coef[i,4]=(-Coef[i,3])/(3*h[i])
            } else{
                Coef[i,3]=Z[i]- mu[i]*Coef[i+1,3]
                Coef[i,2]=(y[i+1]-y[i])/h[i]-h[i]*(Coef[i+1,3]+2*Coef[i,3])/3
                Coef[i,4]=(Coef[i+1,3]-Coef[i,3])/(3*h[i])
            }
        }
        return(Coef)
    } else{
        print("Error5")
    }
}
    \end{lstlisting}
    
    \newpage
    \section{Example of a Bootstrapping Algorithm (OIS Curve)} \label{App:AppendixC}
    
    \begin{algorithm}[H]
     \SetAlgoLined
     \KwData{Swap Curve: Tenors ($\overbar{X}=(\textnormal{ON,TN,1W,2W,}\dots\textnormal{,50Y)}$ and Swap Rates ($\overbar{k_X}=(k_{\textnormal{ON}},\dots,k_{\textnormal{50Y}})$)}
     \KwResult{$P^c(t,T)$ for all $t+1 \leq T \leq 50$}
     Calculate $P^c(t,t+1)$ and $P^c(t,t+2)$ using ON and TN rates\;
     $P^c(t,T_0)=P^c(t,t+2)$\;
     \For{$\overbar{X}(i)=\textnormal{1W}$ \textnormal{to} $\overbar{X}(i)=\textnormal{1Y}$}{
        $P^c(t,T_{\overbar{X}(i)})=P^c(t,T_0)/(1+\overbar{k_X}(i) \tau(T_0,T_{\overbar{X}(i)}))$\;
     }
     \For{$\overbar{X}(i)=\textnormal{18M}$ \textnormal{to} $\overbar{X}(i)=\textnormal{10Y}$}{
        $N_i=$\textnormal{ number of coupons}\;
        $P^c(t,T_{\overbar{X}(i)})=\frac{P^c(t,T_0)-\overbar{k_X}(i)\sum_{j=1}^{N_i-1} \tau(T_{(j-1)^{(\overbar{X}(i))}},T_{(j)^{(\overbar{X}(i))}}) P^c(t,T_{(j)^{(\overbar{X}(i))}})} {1+\overbar{k_X}(i)\tau(T_{(N_i-1)^{(\overbar{X}(i))}},T_{(N_i)^{(\overbar{X}(i))}})}$\;
     }
     \For{$\overbar{X}(i)=\textnormal{12Y}$ \textnormal{to} $\overbar{X}(i)=\textnormal{50Y}$}{
        $P^c(t,T_{\overbar{X}(i)})=P_0(T_{\overbar{X}(i)})$ (initial guess)\;
        Define unknown variables: $\overbar{S}=\{S_m : S_m \textnormal{ is a payment date and } \overbar{X}(i-1) < S_m < \overbar{X}(i) \}$\;
        $P_1(T_{\overbar{X}(i)})=P_0(T_{\overbar{X}(i)})+0.1$\;
        \While{$|P_1(T_{\overbar{X}(i)})-P_0(T_{\overbar{X}(i)})|<10^{-9}$}{
            $P_1(T_{\overbar{X}(i)})=P_0(T_{\overbar{X}(i)})$\;
            Calculate $R^c(t,T)=\frac{-\ln(P^c(t,T))}{\tau(t,T)}$ for all $T \in \{ T_0,T_{\overbar{X}(i)} \} $\;
            Interpolate $R^c(t,T)$\;
            Calculate $R^c(t,S_m)$ for all $S_m \in \overbar{S}$\;
            Calculate $P^c(t,S_m)$ for all $S_m \in \overbar{S}$\;
            $P_0(T_{\overbar{X}(i)})=\frac{P^c(t,T_0)-\overbar{k_X}(i)\sum_{j=1}^{N_i-1} \tau(T_{(j-1)^{(\overbar{X}(i))}},T_{(j)^{(\overbar{X}(i))}}) P^c(t,T_{(j)^{(\overbar{X}(i))}})} {1+\overbar{k_X}(i)\tau(T_{(N_i-1)^{(\overbar{X}(i))}},T_{(N_i)^{(\overbar{X}(i))}})}$\;
            $P^c(t,T_{\overbar{X}(i)})=P_0(T_{\overbar{X}(i)})$
        }
     }
     \caption{USD OIS Bootstrapping}
    \end{algorithm}
    
    \newpage
    \section{Bloomberg \& SuperDerivatives Data} \label{App:AppendixD}
    \begin{table}[h]
    \footnotesize
    \centering
    \begin{tabular}{|c c c c c|}
      \hline
     \multicolumn{1}{|>{\centering\arraybackslash}m{10mm}}{\textbf{Maturity}} & 
     \multicolumn{1}{>{\centering\arraybackslash}m{10mm}}{\textbf{Date}} &
     \multicolumn{1}{>{\centering\arraybackslash}m{25mm}}{\textbf{Swap Rate}} & 
     \multicolumn{1}{>{\centering\arraybackslash}m{25mm}}{\textbf{Discount Factor}} &
     \multicolumn{1}{>{\centering\arraybackslash}m{25mm}|}{\textbf{Continuously Compounded Zero Rate}}\\ \hline
     
    1m & 2015-07-02 & 0.12800 & 0.99988 & 0.12923 \\
    2m & 2015-08-03 & 0.13200 & 0.99976 & 0.13329 \\
    3m & 2015-09-02 & 0.14200 & 0.99962 & 0.14316 \\
    4m & 2015-10-02 & 0.15400 & 0.99946 & 0.15512 \\
    5m & 2015-11-02 & 0.17000 & 0.99926 & 0.17110 \\
    6m & 2015-12-02 & 0.18600 & 0.99904 & 0.18714 \\
    9m & 2016-03-02 & 0.25100 & 0.99808 & 0.25239 \\
    1y & 2016-06-02 & 0.32300 & 0.99671 & 0.32477 \\
    18m$^{\ast}$ & 2016-12-02 & 0.48300 & 0.99266 & 0.48632 \\
    2y & 2017-06-02 & 0.64700 & 0.98696 & 0.65204 \\
    3y & 2018-06-04 & 0.94258 & 0.97167 & 0.95201 \\
    4y & 2019-06-03 & 1.19070 & 0.95274 & 1.20536 \\
    5y & 2020-06-02 & 1.39274 & 0.93157 & 1.41304 \\
    7y & 2022-06-02 & 1.69758 & 0.88570 & 1.72982 \\
    10y & 2025-06-02 & 1.98045 & 0.81607 & 2.02866 \\
    12y & 2027-06-02 & 2.09662 & 0.77200 & 2.15300 \\
    15y & 2030-06-03 & 2.21173 & 0.71025 & 2.27716 \\
    20y & 2035-06-04 & 2.32231 & 0.61866 & 2.39735 \\
    25y & 2040-06-04 & 2.37697 & 0.54082 & 2.45516 \\
    30y & 2045-06-02 & 2.40250 & 0.47497 & 2.47897 \\
    40y & 2055-06-02 & 2.41850 & 0.36979 & 2.48468 \\
    50y & 2065-06-02 & 2.40200 & 0.29453 & 2.44244 \\
    \hline
    \end{tabular}
    \caption{SuperDerivatives market data of the USD OIS Swap Curve (see Section \ref{USDDiscountCurveSection}). Quotes are End of Day prices from May 29, 2015. The quotes were taken from \url{www.superderivatives.com} on June 21, 2015. $^{\ast}$The 18m OIS swap convention has an upfront short stub, i.e., each leg has two coupons: the first with an accrual period of 6m and the second with an accrual period of 12m (1y).}
    \label{OIS:SuperDerivatives}
    \end{table}
    
    \begin{sidewaystable}
    \footnotesize
    \centering
    \begin{tabular}{|l c c c c c c c c c c r|}
      \hline
     \multicolumn{1}{|>{\centering\arraybackslash}m{10mm}}{\textbf{Index}} & 
     \multicolumn{1}{>{\centering\arraybackslash}m{10mm}}{\textbf{Spot Lag}} &
     \multicolumn{1}{>{\centering\arraybackslash}m{10mm}}{\textbf{Tenor}} & 
     \multicolumn{1}{>{\centering\arraybackslash}m{10mm}}{\textbf{Rate}} &
     \multicolumn{1}{>{\centering\arraybackslash}m{10mm}}{\textbf{Type}} &
     \multicolumn{1}{>{\centering\arraybackslash}m{10mm}}{\textbf{Floating Period}} &
     \multicolumn{1}{>{\centering\arraybackslash}m{15mm}}{\textbf{Floating Day Rolling}} &
     \multicolumn{1}{>{\centering\arraybackslash}m{10mm}}{\textbf{Floating Day Count}} &
     \multicolumn{1}{>{\centering\arraybackslash}m{10mm}}{\textbf{Fixed Period}} &
     \multicolumn{1}{>{\centering\arraybackslash}m{15mm}}{\textbf{Fixed Day Rolling}} &
     \multicolumn{1}{>{\centering\arraybackslash}m{10mm}}{\textbf{Fixed Day Count}} &
     \multicolumn{1}{>{\centering\arraybackslash}m{15mm}|}{\textbf{Bloomberg Ticker}}\\ \hline
     
    FEDL01 & 0d & 1d & 0.08000 & Cash & - & Following & ACT/360 & - & - & - & FEDL01 Index \\
    FEDL11 & 1d & 1d & 0.08000 & Cash & - & Following & ACT/360 & - & - & - & FEDL01 Index \\ \hline
    USSO1Z & 2d & 1w & 0.12000 & OI Swap & 1y & Mod Fol & ACT/360 & 1y & Mod Fol & ACT/360 & USSO1Z Curncy \\
    USSO2Z & 2d & 2w & 0.12380 & OI Swap & 1y & Mod Fol & ACT/360 & 1y & Mod Fol & ACT/360 & USSO2Z Curncy \\
    USSO3Z & 2d & 3w & 0.12600 & OI Swap & 1y & Mod Fol & ACT/360 & 1y & Mod Fol & ACT/360 & USSO3Z Curncy \\
    USSOA & 2d & 1m & 0.12800 & OI Swap & 1y & Mod Fol & ACT/360 & 1y & Mod Fol & ACT/360 & USSOA Curncy \\
    USSOB & 2d & 2m & 0.13400 & OI Swap & 1y & Mod Fol & ACT/360 & 1y & Mod Fol & ACT/360 & USSOB Curncy \\
    USSOC & 2d & 3m & 0.14100 & OI Swap & 1y & Mod Fol & ACT/360 & 1y & Mod Fol & ACT/360 & USSOC Curncy \\
    USSOD & 2d & 4m & 0.15400 & OI Swap & 1y & Mod Fol & ACT/360 & 1y & Mod Fol & ACT/360 & USSOD Curncy \\
    USSOE  & 2d & 5m & 0.16900 & OI Swap & 1y & Mod Fol & ACT/360 & 1y & Mod Fol & ACT/360 & USSOE  Curncy \\
    USSOF & 2d & 6m & 0.18730 & OI Swap & 1y & Mod Fol & ACT/360 & 1y & Mod Fol & ACT/360 & USSOF Curncy \\
    USSOG & 2d & 7m & 0.20700 & OI Swap & 1y & Mod Fol & ACT/360 & 1y & Mod Fol & ACT/360 & USSOG Curncy \\
    USSOH & 2d & 8m & 0.22800 & OI Swap & 1y & Mod Fol & ACT/360 & 1y & Mod Fol & ACT/360 & USSOH Curncy \\
    USSOI & 2d & 9m & 0.24900 & OI Swap & 1y & Mod Fol & ACT/360 & 1y & Mod Fol & ACT/360 & USSOI Curncy \\
    USSOJ & 2d & 10m & 0.27400 & OI Swap & 1y & Mod Fol & ACT/360 & 1y & Mod Fol & ACT/360 & USSOJ Curncy \\
    USSOK & 2d & 11m & 0.29450 & OI Swap & 1y & Mod Fol & ACT/360 & 1y & Mod Fol & ACT/360 & USSOK Curncy \\
    USSO1 & 2d & 1y & 0.31900 & OI Swap & 1y & Mod Fol & ACT/360 & 1y & Mod Fol & ACT/360 & USSO1 Curncy \\
    USSO1F & 2d & 18m & 0.48100 & OI Swap & 1y & Mod Fol & ACT/360 & 1y & Mod Fol & ACT/360 & USSO1F Curncy \\
    USSO2 & 2d & 2y & 0.64900 & OI Swap & 1y & Mod Fol & ACT/360 & 1y & Mod Fol & ACT/360 & USSO2 Curncy \\
    USSO3 & 2d & 3y & 0.95500 & OI Swap & 1y & Mod Fol & ACT/360 & 1y & Mod Fol & ACT/360 & USSO3 Curncy \\
    USSO4 & 2d & 4y & 1.19680 & OI Swap & 1y & Mod Fol & ACT/360 & 1y & Mod Fol & ACT/360 & USSO4 Curncy \\
    USSO5 & 2d & 5y & 1.39600 & OI Swap & 1y & Mod Fol & ACT/360 & 1y & Mod Fol & ACT/360 & USSO5 Curncy \\
    USSO6 & 2d & 6y & 1.56000 & OI Swap & 1y & Mod Fol & ACT/360 & 1y & Mod Fol & ACT/360 & USSO6 Curncy \\
    USSO7 & 2d & 7y & 1.69700 & OI Swap & 1y & Mod Fol & ACT/360 & 1y & Mod Fol & ACT/360 & USSO7 Curncy \\
    USSO8 & 2d & 8y & 1.80300 & OI Swap & 1y & Mod Fol & ACT/360 & 1y & Mod Fol & ACT/360 & USSO8 Curncy \\
    USSO9 & 2d & 9y & 1.89300 & OI Swap & 1y & Mod Fol & ACT/360 & 1y & Mod Fol & ACT/360 & USSO9 Curncy \\
    USSO10 & 2d & 10y & 1.97800 & OI Swap & 1y & Mod Fol & ACT/360 & 1y & Mod Fol & ACT/360 & USSO10 Curncy \\
    USSO12 & 2d & 12y & 2.08400 & OI Swap & 1y & Mod Fol & ACT/360 & 1y & Mod Fol & ACT/360 & USSO12 Curncy \\
    USSO15 & 2d & 15y & 2.20300 & OI Swap & 1y & Mod Fol & ACT/360 & 1y & Mod Fol & ACT/360 & USSO15 Curncy \\
    USSO20 & 2d & 20y & 2.32000 & OI Swap & 1y & Mod Fol & ACT/360 & 1y & Mod Fol & ACT/360 & USSO20 Curncy \\
    USSO25 & 2d & 25y & 2.38100 & OI Swap & 1y & Mod Fol & ACT/360 & 1y & Mod Fol & ACT/360 & USSO25 Curncy \\
    USSO30 & 2d & 30y & 2.40800 & OI Swap & 1y & Mod Fol & ACT/360 & 1y & Mod Fol & ACT/360 & USSO30 Curncy \\
    USSO40 & 2d & 40y & 2.42800 & OI Swap & 1y & Mod Fol & ACT/360 & 1y & Mod Fol & ACT/360 & USSO40 Curncy \\
    USSO50 & 2d & 50y & 2.41800 & OI Swap & 1y & Mod Fol & ACT/360 & 1y & Mod Fol & ACT/360 & USSO50 Curncy \\ \hline
    \end{tabular}
    \caption{In this table, we present the Bloomberg market data used for the construction of the USD discount curve (see Section \ref{USDDiscountCurveSection}). Quotes are End of Day prices from May 29, 2015. The quotes were taken from a Bloomberg Terminal on June 21, 2015. $^{\ast}$The 18m OIS swap convention has an upfront short stub, i.e., each leg has two coupons: the first with an accrual period of 6m and the second with an accrual period of 12m (1y).}
    \end{sidewaystable}
    
    \FloatBarrier
    \begin{table}
    \footnotesize
    \centering
    \begin{tabular}{|c c c c c c|}
      \hline
     \multicolumn{1}{|>{\centering\arraybackslash}m{10mm}}{\textbf{Maturity}} & 
     \multicolumn{1}{>{\centering\arraybackslash}m{10mm}}{\textbf{Date}} &
     \multicolumn{1}{>{\centering\arraybackslash}m{15mm}}{\textbf{Rate}} & 
     \multicolumn{1}{>{\centering\arraybackslash}m{15mm}}{\textbf{Price}} & 
     \multicolumn{1}{>{\centering\arraybackslash}m{20mm}}{\textbf{Discount Factor}} &
     \multicolumn{1}{>{\centering\arraybackslash}m{30mm}|}{\textbf{Continuously Compounded Zero Rate}}\\ \hline
     
O/N & 2015-06-01 & 0.12100 & - & 0.99999 & 0.12268 \\
T/N & 2015-06-02 & 0.12100 & - & 0.99999 & 0.12268 \\
1w & 2015-06-09 & 0.15025 & - & 0.99996 & 0.14155 \\ \hline
3m & 2015-09-02 & 0.28375 & - & 0.99926 & 0.28072 \\ \hline
Jun-15 & 2015-09-17 & 0.29250 & 99.70750 & 0.99915 & 0.28126 \\ 
Sep-15 & 2015-12-16 & 0.40000 & 99.60000 & 0.99814 & 0.33740 \\
Dec-15 & 2016-03-16 & 0.57000 & 99.43000 & 0.99671 & 0.41216 \\
Mar-16 & 2016-06-16 & 0.74500 & 99.25500 & 0.99482 & 0.49404 \\
Jun-16 & 2016-09-15 & 0.95500 & 99.04500 & 0.99242 & 0.58474 \\
Sep-16 & 2016-12-21 & 1.17000 & 98.83000 & 0.98931 & 0.68568 \\
Dec-16 & 2017-03-21 & 1.37500 & 98.62500 & 0.98594 & 0.78050 \\
Mar-17 & 2017-06-15 & 1.53000 & 98.47000 & 0.98236 & 0.86837 \\ \hline
1y & 2016-06-02 & 0.48050 & - & 0.99510 & 0.48418 \\
2y & 2017-06-02 & 0.84842 & - & 0.98290 & 0.85636 \\
3y & 2018-06-04 & 1.16383 & - & 0.96555 & 1.16106 \\
4y & 2019-06-03 & 1.42070 & - & 0.94449 & 1.42179 \\
5y & 2020-06-02 & 1.62899 & - & 0.92130 & 1.63408 \\
6y & 2021-06-02 & 1.79750 & - & 0.89692 & 1.80811 \\
7y & 2022-06-02 & 1.93882 & - & 0.87180 & 1.95539 \\
8y & 2023-06-02 & 2.04908 & - & 0.84702 & 2.07106 \\
9y & 2024-06-03 & 2.14464 & - & 0.82208 & 2.17160 \\
10y & 2025-06-02 & 2.22420 & - & 0.79765 & 2.25651 \\
12y & 2027-06-02 & 2.34238 & - & 0.75091 & 2.38340 \\
15y & 2030-06-03 & 2.45548 & - & 0.68634 & 2.50513 \\
20y & 2035-06-04 & 2.56606 & - & 0.59100 & 2.62570 \\
25y & 2040-06-04 & 2.61595 & - & 0.51161 & 2.67695 \\
30y & 2045-06-02 & 2.64250 & - & 0.44415 & 2.70236 \\
40y & 2055-06-02 & 2.65850 & - & 0.33828 & 2.70714 \\
50y & 2065-06-02 & 2.64200 & - & 0.26376 & 2.66299 \\
60y & 2075-06-03 & 2.63134 & - & 0.20565 & 2.63355 \\
    \hline
    \end{tabular}
    \caption{SuperDerivatives market data of the USD LIBOR 3m Swap Curve (see Section \ref{USDForwardCurveSection}). Quotes are End of Day prices from May 29, 2015. The quotes were taken from \url{www.superderivatives.com} on June 21, 2015.}
    \label{OIS:Bloomberg}
    \end{table}
    \FloatBarrier
    
    \begin{table}
    \footnotesize
    \centering
    \begin{tabular}{|c c c c c|}
      \hline
     \multicolumn{1}{|>{\centering\arraybackslash}m{10mm}}{\textbf{Maturity}} & 
     \multicolumn{1}{>{\centering\arraybackslash}m{10mm}}{\textbf{Date}} &
     \multicolumn{1}{>{\centering\arraybackslash}m{20mm}}{\textbf{Swap Rate}} & 
     \multicolumn{1}{>{\centering\arraybackslash}m{20mm}}{\textbf{Discount Factor}} &
     \multicolumn{1}{>{\centering\arraybackslash}m{30mm}|}{\textbf{Continuously Compounded Zero Rate}}\\ \hline
     
    O/N & 2015-06-01 & 0.12100 & 0.99999 & 0.12268 \\
    T/N & 2015-06-02 & 0.12100 & 0.99999 & 0.12268 \\ \hline
    1w & 2015-06-09 & 0.15025 & 0.99996 & 0.14155 \\
    1m & 2015-07-02 & 0.18400 & 0.99983 & 0.17903 \\
    1y & 2016-06-02 & 0.37550 & 0.99619 & 0.37706 \\
    2y & 2017-06-02 & 0.72717 & 0.98554 & 0.72318 \\
    3y & 2018-06-04 & 1.03258 & 0.96936 & 1.03078 \\
    4y & 2019-06-03 & 1.28570 & 0.94960 & 1.28766 \\
    5y & 2020-06-02 & 1.49149 & 0.92763 & 1.49747 \\
    6y & 2021-06-02 & 1.65950 & 0.90435 & 1.67101 \\
    7y & 2022-06-02 & 1.80258 & 0.88011 & 1.82019 \\
    8y & 2023-06-02 & 1.91608 & 0.85601 & 1.93940 \\
    9y & 2024-06-03 & 2.01564 & 0.83157 & 2.04436 \\
    10y & 2025-06-02 & 2.09670 & 0.80775 & 2.13088 \\
    12y & 2027-06-02 & 2.22138 & 0.76167 & 2.26509 \\
    15y & 2030-06-03 & 2.33923 & 0.69809 & 2.39210 \\
    20y & 2035-06-04 & 2.44856 & 0.60480 & 2.51052 \\
    25y & 2040-06-04 & 2.49895 & 0.52652 & 2.56218 \\
    30y & 2045-06-02 & 2.52250 & 0.46037 & 2.58294 \\
    40y & 2055-06-02 & 2.53249 & 0.35621 & 2.57809 \\
    50y & 2065-06-02 & 2.50999 & 0.28290 & 2.52299 \\
    60y & 2075-06-03 & 2.49331 & 0.22523 & 2.48214 \\
    \hline
    \end{tabular}
    \caption{SuperDerivatives market data of the USD LIBOR 1m Swap Curve (see Section \ref{USDForwardCurveSection}). Quotes are End of Day prices from May 29, 2015. The quotes were taken from \url{www.superderivatives.com} on June 21, 2015.}
    \end{table}
    \FloatBarrier   
    
    \begin{sidewaystable}
    \footnotesize
    \centering
    \begin{tabular}{|l c c c c c c c c c c r|}
      \hline
     \multicolumn{1}{|>{\centering\arraybackslash}m{10mm}}{\textbf{Index}} & 
     \multicolumn{1}{>{\centering\arraybackslash}m{10mm}}{\textbf{Spot Lag}} &
     \multicolumn{1}{>{\centering\arraybackslash}m{10mm}}{\textbf{Tenor}} & 
     \multicolumn{1}{>{\centering\arraybackslash}m{10mm}}{\textbf{Rate}} &
     \multicolumn{1}{>{\centering\arraybackslash}m{10mm}}{\textbf{Type}} &
     \multicolumn{1}{>{\centering\arraybackslash}m{10mm}}{\textbf{Floating Period}} &
     \multicolumn{1}{>{\centering\arraybackslash}m{15mm}}{\textbf{Floating Day Rolling}} &
     \multicolumn{1}{>{\centering\arraybackslash}m{10mm}}{\textbf{Floating Day Count}} &
     \multicolumn{1}{>{\centering\arraybackslash}m{10mm}}{\textbf{Fixed Period}} &
     \multicolumn{1}{>{\centering\arraybackslash}m{15mm}}{\textbf{Fixed Day Rolling}} &
     \multicolumn{1}{>{\centering\arraybackslash}m{10mm}}{\textbf{Fixed Day Count}} &
     \multicolumn{1}{>{\centering\arraybackslash}m{15mm}|}{\textbf{Bloomberg Ticker}}\\ \hline
     
    US00O/N & 0d & 1d & 0.12100 & Cash & - & Following & ACT/360 & - & - & - & US00O/N Index \\
    USDR2T  & 1d & 1d & 0.17000 & Cash & - & Following & ACT/360 & - & - & - & USDR2T Index \\ \hline
    US0003M & 2d & 1m & 0.28375 & Cash & 3m & Mod Fol & ACT/360 & - & - & - & US0003M Index \\ \hline
    USSWF & 2d & 6m & 0.33700 & Swap & 3m & Mod Fol & ACT/360 & 6m & Mod Fol & 30/360 & USSWF Curncy \\
    USSWI & 2d & 9m & 0.40630 & Swap & 3m & Mod Fol & ACT/360 & 6m & Mod Fol & 30/360 & USSWI Curncy \\
    USSW1 & 2d & 1y & 0.48730 & Swap & 3m & Mod Fol & ACT/360 & 6m & Mod Fol & 30/360 & USSW1 Curncy \\
    USSW1C & 2d & 15m & 0.57620 & Swap & 3m & Mod Fol & ACT/360 & 6m & Mod Fol & 30/360 & USSW1C Curncy \\
    USSW1F & 2d & 18m & 0.66570 & Swap & 3m & Mod Fol & ACT/360 & 6m & Mod Fol & 30/360 & USSW1F Curncy \\
    USSW1I & 2d & 21m & 0.75710 & Swap & 3m & Mod Fol & ACT/360 & 6m & Mod Fol & 30/360 & USSW1I Curncy \\
    USSW2 & 2d & 2y & 0.84950 & Swap & 3m & Mod Fol & ACT/360 & 6m & Mod Fol & 30/360 & USSW2 Curncy \\
    USSW3 & 2d & 3y & 1.17050 & Swap & 3m & Mod Fol & ACT/360 & 6m & Mod Fol & 30/360 & USSW3 Curncy \\
    USSW4 & 2d & 4y & 1.42800 & Swap & 3m & Mod Fol & ACT/360 & 6m & Mod Fol & 30/360 & USSW4 Curncy \\
    USSW5 & 2d & 5y & 1.63350 & Swap & 3m & Mod Fol & ACT/360 & 6m & Mod Fol & 30/360 & USSW5 Curncy \\
    USSW6 & 2d & 6y & 1.80150 & Swap & 3m & Mod Fol & ACT/360 & 6m & Mod Fol & 30/360 & USSW6 Curncy \\
    USSW7 & 2d & 7y & 1.93750 & Swap & 3m & Mod Fol & ACT/360 & 6m & Mod Fol & 30/360 & USSW7 Curncy \\
    USSW8 & 2d & 8y & 2.04880 & Swap & 3m & Mod Fol & ACT/360 & 6m & Mod Fol & 30/360 & USSW8 Curncy \\
    USSW9 & 2d & 9y & 2.13850 & Swap & 3m & Mod Fol & ACT/360 & 6m & Mod Fol & 30/360 & USSW9 Curncy \\
    USSW10 & 2d & 10y & 2.21450 & Swap & 3m & Mod Fol & ACT/360 & 6m & Mod Fol & 30/360 & USSW10 Curncy \\
    USSW12 & 2d & 12y & 2.33350 & Swap & 3m & Mod Fol & ACT/360 & 6m & Mod Fol & 30/360 & USSW12 Curncy \\
    USSW15 & 2d & 15y & 2.45250 & Swap & 3m & Mod Fol & ACT/360 & 6m & Mod Fol & 30/360 & USSW15 Curncy \\
    USSW20 & 2d & 20y & 2.56800 & Swap & 3m & Mod Fol & ACT/360 & 6m & Mod Fol & 30/360 & USSW20 Curncy \\
    USSW25 & 2d & 25y & 2.62150 & Swap & 3m & Mod Fol & ACT/360 & 6m & Mod Fol & 30/360 & USSW25 Curncy \\
    USSW30 & 2d & 30y & 2.64950 & Swap & 3m & Mod Fol & ACT/360 & 6m & Mod Fol & 30/360 & USSW30 Curncy \\
    USSW40 & 2d & 40y & 2.66700 & Swap & 3m & Mod Fol & ACT/360 & 6m & Mod Fol & 30/360 & USSW40 Curncy \\
    USSW50 & 2d & 50y & 2.64950 & Swap & 3m & Mod Fol & ACT/360 & 6m & Mod Fol & 30/360 & USSW50 Curncy \\ \hline

    \end{tabular}
    \caption{In this table, we present the Bloomberg market data used for the construction of the LIBOR 3m forward curve (see Section \ref{USDForwardCurveSection}). Quotes are End of Day prices from May 29, 2015. The quotes were taken from a Bloomberg Terminal on June 21, 2015.}
    \label{LIBOR3M:Bloomberg}
    \end{sidewaystable}

    \begin{sidewaystable}
    \scriptsize
    \centering
    \begin{tabular}{|c c c c c c c c c c c r|}
      \hline
     \multicolumn{1}{|>{\centering\arraybackslash}m{10mm}}{\textbf{Index}} & 
     \multicolumn{1}{>{\centering\arraybackslash}m{10mm}}{\textbf{Spot Lag}} &
     \multicolumn{1}{>{\centering\arraybackslash}m{10mm}}{\textbf{Tenor}} & 
     \multicolumn{1}{>{\centering\arraybackslash}m{10mm}}{\textbf{Rate}} &
     \multicolumn{1}{>{\centering\arraybackslash}m{10mm}}{\textbf{Type}} &
     \multicolumn{1}{>{\centering\arraybackslash}m{10mm}}{\textbf{Leg 1 Period}} &
     \multicolumn{1}{>{\centering\arraybackslash}m{15mm}}{\textbf{Leg 1 Day Rolling}} &
     \multicolumn{1}{>{\centering\arraybackslash}m{10mm}}{\textbf{Leg 1 Day Count}} &
     \multicolumn{1}{>{\centering\arraybackslash}m{10mm}}{\textbf{Leg 2 Period}} &
     \multicolumn{1}{>{\centering\arraybackslash}m{15mm}}{\textbf{Leg 2 Day Rolling}} &
     \multicolumn{1}{>{\centering\arraybackslash}m{10mm}}{\textbf{Leg 2 Day Count}} &
     \multicolumn{1}{>{\centering\arraybackslash}m{12mm}|}{\textbf{Bloomberg Ticker}}\\ \hline
     
    US00O/N & 0d & 1d & 0.12100 & Cash & - & Following & ACT/360 & - & - & - & US00O/N Index \\
    USDR2T  & 1d & 1d & 0.17000 & Cash & - & Following & ACT/360 & - & - & - & USDR2T Index \\ \hline
    US0001M & 2d & 1m & 0.18400 & Swap & 1m & Mod Fol & ACT/360 & - & - & - & US0001M Index \\ \hline
    USSWBV1 & 2d & 2m & 0.18600 & Swap & 1m & Mod Fol & ACT/360 & 1y (Fixed) & Mod Fol & ACT/360 & USSWBV1 Curncy \\
    USSWCV1 & 2d & 3m & 0.19300 & Swap & 1m & Mod Fol & ACT/360 & 1y (Fixed) & Mod Fol & ACT/360 & USSWCV1 Curncy \\
    USSWDV1 & 2d & 4m & 0.19900 & Swap & 1m & Mod Fol & ACT/360 & 1y (Fixed) & Mod Fol & ACT/360 & USSWDV1 Curncy \\
    USSWEV1 & 2d & 5m & 0.22000 & Swap & 1m & Mod Fol & ACT/360 & 1y (Fixed) & Mod Fol & ACT/360 & USSWEV1 Curncy \\
    USSWFV1 & 2d & 6m & 0.23600 & Swap & 1m & Mod Fol & ACT/360 & 1y (Fixed) & Mod Fol & ACT/360 & USSWFV1 Curncy \\
    USSWGV1 & 2d & 7m & 0.26000 & Swap & 1m & Mod Fol & ACT/360 & 1y (Fixed) & Mod Fol & ACT/360 & USSWGV1 Curncy \\
    USSWHV1 & 2d & 8m & 0.28100 & Swap & 1m & Mod Fol & ACT/360 & 1y (Fixed) & Mod Fol & ACT/360 & USSWHV1 Curncy \\
    USSWIV1 & 2d & 9m & 0.30130 & Swap & 1m & Mod Fol & ACT/360 & 1y (Fixed) & Mod Fol & ACT/360 & USSWIV1 Curncy \\
    USSWJV1 & 2d & 10m & 0.32400 & Swap & 1m & Mod Fol & ACT/360 & 1y (Fixed) & Mod Fol & ACT/360 & USSWJV1 Curncy \\
    USSWKV1 & 2d & 11m & 0.34900 & Swap & 1m & Mod Fol & ACT/360 & 1y (Fixed) & Mod Fol & ACT/360 & USSWKV1 Curncy \\ \hline
    USBA1 & 2d & 1y & 0.10369 & Tenor Swap & 3m & Mod Fol & ACT/360 & 3m (Reset 1m) & Mod Fol & ACT/360 & USBA1 Curncy \\
    USBA1F & 2d & 18m & 0.11156 & Tenor Swap & 3m & Mod Fol & ACT/360 & 3m (Reset 1m) & Mod Fol & ACT/360 & USBA1F Curncy \\
    USBA2 & 2d & 2y & 0.12011 & Tenor Swap & 3m & Mod Fol & ACT/360 & 3m (Reset 1m) & Mod Fol & ACT/360 & USBA2 Curncy \\
    USBA3 & 2d & 3y & 0.13154 & Tenor Swap & 3m & Mod Fol & ACT/360 & 3m (Reset 1m) & Mod Fol & ACT/360 & USBA3 Curncy \\
    USBA4 & 2d & 4y & 0.13563 & Tenor Swap & 3m & Mod Fol & ACT/360 & 3m (Reset 1m) & Mod Fol & ACT/360 & USBA4 Curncy \\
    USBA5 & 2d & 5y & 0.13805 & Tenor Swap & 3m & Mod Fol & ACT/360 & 3m (Reset 1m) & Mod Fol & ACT/360 & USBA5 Curncy \\
    USBA6 & 2d & 6y & 0.13811 & Tenor Swap & 3m & Mod Fol & ACT/360 & 3m (Reset 1m) & Mod Fol & ACT/360 & USBA6 Curncy \\
    USBA7 & 2d & 7y & 0.13625 & Tenor Swap & 3m & Mod Fol & ACT/360 & 3m (Reset 1m) & Mod Fol & ACT/360 & USBA7 Curncy \\
    USBA8 & 2d & 8y & 0.13337 & Tenor Swap & 3m & Mod Fol & ACT/360 & 3m (Reset 1m) & Mod Fol & ACT/360 & USBA8 Curncy \\
    USBA9 & 2d & 9y & 0.12987 & Tenor Swap & 3m & Mod Fol & ACT/360 & 3m (Reset 1m) & Mod Fol & ACT/360 & USBA9 Curncy \\
    USBA10 & 2d & 10y & 0.12688 & Tenor Swap & 3m & Mod Fol & ACT/360 & 3m (Reset 1m) & Mod Fol & ACT/360 & USBA10 Curncy \\
    USBA12 & 2d & 12y & 0.12177 & Tenor Swap & 3m & Mod Fol & ACT/360 & 3m (Reset 1m) & Mod Fol & ACT/360 & USBA12 Curncy \\
    USBA15 & 2d & 15y & 0.11802 & Tenor Swap & 3m & Mod Fol & ACT/360 & 3m (Reset 1m) & Mod Fol & ACT/360 & USBA15 Curncy \\
    USBA20 & 2d & 20y & 0.11801 & Tenor Swap & 3m & Mod Fol & ACT/360 & 3m (Reset 1m) & Mod Fol & ACT/360 & USBA20 Curncy \\
    USBA25 & 2d & 25y & 0.11875 & Tenor Swap & 3m & Mod Fol & ACT/360 & 3m (Reset 1m) & Mod Fol & ACT/360 & USBA25 Curncy \\
    USBA30 & 2d & 30y & 0.11938 & Tenor Swap & 3m & Mod Fol & ACT/360 & 3m (Reset 1m) & Mod Fol & ACT/360 & USBA30 Curncy \\
\hline

    \end{tabular}
    \caption{In this table, we present the Bloomberg market data used for the construction of the LIBOR 1m forward curve (see Section \ref{USDForwardCurveSection}). Recall that plain vanilla TSs (LIBOR 1m vs LIBOR 3m) have a payment frequency of 3 months (quarterly) for both legs. However, in the LIBOR 1m leg, the fixing resets on a monthly basis and then is compounded to define the quarterly payment. Quotes are End of Day prices from May 29, 2015. The quotes were taken from a Bloomberg Terminal on June 21, 2015.}
    \label{LIBOR1M:Bloomberg}
    \end{sidewaystable}

    \begin{table}[h]
    \footnotesize
    \centering
    \begin{tabular}{|c c c|}
      \hline
     \multicolumn{1}{|>{\centering\arraybackslash}m{10mm}}{\textbf{Maturity}} & 
     \multicolumn{1}{>{\centering\arraybackslash}m{10mm}}{\textbf{Date}} &
     \multicolumn{1}{>{\centering\arraybackslash}m{30mm}|}{\textbf{Basis Spread (\%)}}\\ \hline
    84d    & 2015-08-24 & 0.9015\\
    168d   & 2015-11-17 & 0.8530\\
    252d   & 2016-02-08 & 0.9030\\
    364d   & 2016-05-31 & 0.6900\\
    728d   & 2017-05-30 & 0.9135\\
    1092d  & 2018-05-29 & 1.0160\\
    1456d  & 2019-05-28 & 1.0520\\
    1820d  & 2020-05-26 & 1.1070\\
    2548d  & 2022-05-23 & 1.0800\\
    3640d  & 2025-05-19 & 1.1000\\
    5460d  & 2030-05-13 & 1.0800\\
    7280d  & 2035-05-07 & 1.0800\\
    10920d & 2045-04-24 & 1.0800\\
    \hline
    \end{tabular}
    \caption{SuperDerivatives market data of the USDMXN Cross-Currency Swaps (Constant Notional). Quotes are End of Day prices from May 29, 2015. The quotes were taken from \url{www.superderivatives.com} on June 21, 2015.}
    \end{table}
    
        \begin{table}[h]
    \footnotesize
    \centering
    \begin{tabular}{|c c c c c|}
      \hline
     \multicolumn{1}{|>{\centering\arraybackslash}m{10mm}}{\textbf{Maturity}} & 
     \multicolumn{1}{>{\centering\arraybackslash}m{10mm}}{\textbf{Date}} &
     \multicolumn{1}{>{\centering\arraybackslash}m{20mm}}{\textbf{Swap Rate}} & 
     \multicolumn{1}{>{\centering\arraybackslash}m{20mm}}{\textbf{Discount Factor}} &
     \multicolumn{1}{>{\centering\arraybackslash}m{30mm}|}{\textbf{Continuously Compounded Zero Rate}}\\ \hline
ON & 2015-06-01 & 2.6095 & 0.99978 & 2.64546\\
1d & 2015-06-29 & 3.3000 & 0.99722 & 3.27418\\ \hline
84d & 2015-08-24 & 3.3250 & 0.99207 & 3.34202\\
168d & 2015-11-17 & 3.4350 & 0.98381 & 3.46406\\
252d & 2016-02-08 & 3.5600 & 0.97520 & 3.59466\\
364d & 2016-05-30 & 3.7300 & 0.96280 & 3.77051\\
728d & 2017-05-29 & 4.2350 & 0.91758 & 4.29479\\
1092d & 2018-05-28 & 4.6700 & 0.86708 & 4.75416\\
1456d & 2019-05-27 & 5.0500 & 0.81347 & 5.16457\\
1820d & 2020-05-25 & 5.3600 & 0.75953 & 5.50727\\
2548d & 2022-05-23 & 5.8650 & 0.65342 & 6.08870\\
3640d & 2025-05-19 & 6.2400 & 0.52088 & 6.53483\\
5460d & 2030-05-13 & 6.6300 & 0.34754 & 7.06126\\
7280d & 2035-05-07 & 6.8250 & 0.23028 & 7.35954\\
10920d & 2045-04-24 & 7.0550 & 0.09549 & 7.84857\\
    \hline
    \end{tabular}
    \caption{SuperDerivatives market data of the TIIE 28d Swap Curve. Quotes are End of Day prices from May 29, 2015. The quotes were taken from \url{www.superderivatives.com} on June 21, 2015.}
    \end{table}

    \begin{sidewaystable}
    \scriptsize
    \centering
    \begin{tabular}{|l c c c c c c c c c c r|}
      \hline
     \multicolumn{1}{|>{\centering\arraybackslash}m{10mm}}{\textbf{Index}} & 
     \multicolumn{1}{>{\centering\arraybackslash}m{10mm}}{\textbf{Spot Lag}} &
     \multicolumn{1}{>{\centering\arraybackslash}m{10mm}}{\textbf{Tenor}} & 
     \multicolumn{1}{>{\centering\arraybackslash}m{10mm}}{\textbf{Rate}} &
     \multicolumn{1}{>{\centering\arraybackslash}m{10mm}}{\textbf{Type}} &
     \multicolumn{1}{>{\centering\arraybackslash}m{10mm}}{\textbf{Floating Period}} &
     \multicolumn{1}{>{\centering\arraybackslash}m{15mm}}{\textbf{Floating Day Rolling}} &
     \multicolumn{1}{>{\centering\arraybackslash}m{10mm}}{\textbf{Floating Day Count}} &
     \multicolumn{1}{>{\centering\arraybackslash}m{10mm}}{\textbf{Fixed Period}} &
     \multicolumn{1}{>{\centering\arraybackslash}m{15mm}}{\textbf{Fixed Day Rolling}} &
     \multicolumn{1}{>{\centering\arraybackslash}m{10mm}}{\textbf{Fixed Day Count}} &
     \multicolumn{1}{>{\centering\arraybackslash}m{15mm}|}{\textbf{Bloomberg Ticker}}\\ \hline
     
    MXONBR & 0d & 1d & 3.05000 & Cash & - & Following & ACT/360 & - & - & - & MXONBR Index \\
    MXTNBR  & 1d & 1d & 3.05000 & Cash & - & Following & ACT/360 & - & - & - & MXTNBR Index \\ 
    MXIBTIIE  & 1d & 28d & 3.29500 & Cash & - & Following & ACT/360 & - & - & - & MXIBTIIE Index \\ \hline
    MPSWC & 1d & 84d & 3.32000 & Swap & 28d & Following & ACT/360 & 28d & Following & ACT/360 & MPSWC Curncy \\
    MPSWF & 1d & 168d & 3.43000 & Swap & 28d & Following & ACT/360 & 28d & Following & ACT/360 & MPSWF Curncy \\
    MPSWI & 1d & 252d & 3.56200 & Swap & 28d & Following & ACT/360 & 28d & Following & ACT/360 & MPSWI Curncy \\
    MPSW1A & 1d & 364d & 3.73500 & Swap & 28d & Following & ACT/360 & 28d & Following & ACT/360 & MPSW1A Curncy \\
    MPSW2B & 1d & 728d & 4.23600 & Swap & 28d & Following & ACT/360 & 28d & Following & ACT/360 & MPSW2B Curncy \\
    MPSW3C & 1d & 1092d & 4.67100 & Swap & 28d & Following & ACT/360 & 28d & Following & ACT/360 & MPSW3C Curncy \\
    MPSW4D & 1d & 1456d & 5.05100 & Swap & 28d & Following & ACT/360 & 28d & Following & ACT/360 & MPSW4D Curncy \\
    MPSW5E & 1d & 1820d & 5.36100 & Swap & 28d & Following & ACT/360 & 28d & Following & ACT/360 & MPSW5E Curncy \\
    MPSW7G & 1d & 2548d & 5.86300 & Swap & 28d & Following & ACT/360 & 28d & Following & ACT/360 & MPSW7G Curncy \\
    MPSW10K & 1d & 3640d & 6.23800 & Swap & 28d & Following & ACT/360 & 28d & Following & ACT/360 & MPSW10K Curncy \\
    MPSW156M & 1d & 4368d & 6.42800 & Swap & 28d & Following & ACT/360 & 28d & Following & ACT/360 & MPSW156M Curncy \\
    MPSW16C & 1d & 5460d & 6.63200 & Swap & 28d & Following & ACT/360 & 28d & Following & ACT/360 & MPSW16C Curncy \\
    MPSW21H & 1d & 7280d & 6.83100 & Swap & 28d & Following & ACT/360 & 28d & Following & ACT/360 & MPSW21H Curncy \\
    MPSW32F & 1d & 10920d & 7.02100 & Swap & 28d & Following & ACT/360 & 28d & Following & ACT/360 & MPSW32F Curncy \\ \hline

    \end{tabular}
    \caption{In this table, we present the Bloomberg market data of the TIIE 28d IRSs used for the construction of the MXN-discount curve (collateralized in USD) and the TIIE 28d forward curve. Quotes are End of Day prices from May 29, 2015. The quotes were taken from a Bloomberg Terminal on June 21, 2015.}
    \label{TIIE28D:Bloomberg}
    \end{sidewaystable}
    
    \begin{sidewaystable}
    \footnotesize
    \centering
    \begin{tabular}{|l c c c c c c c c c c r|}
      \hline
     \multicolumn{1}{|>{\centering\arraybackslash}m{10mm}}{\textbf{Index}} & 
     \multicolumn{1}{>{\centering\arraybackslash}m{10mm}}{\textbf{Spot Lag}} &
     \multicolumn{1}{>{\centering\arraybackslash}m{10mm}}{\textbf{Tenor}} & 
     \multicolumn{1}{>{\centering\arraybackslash}m{10mm}}{\textbf{Rate}} &
     \multicolumn{1}{>{\centering\arraybackslash}m{10mm}}{\textbf{Type}} &
     \multicolumn{1}{>{\centering\arraybackslash}m{10mm}}{\textbf{Leg 1 Period}} &
     \multicolumn{1}{>{\centering\arraybackslash}m{15mm}}{\textbf{Leg 1 Day Rolling}} &
     \multicolumn{1}{>{\centering\arraybackslash}m{10mm}}{\textbf{Leg 1 Day Count}} &
     \multicolumn{1}{>{\centering\arraybackslash}m{10mm}}{\textbf{Leg 2 Period}} &
     \multicolumn{1}{>{\centering\arraybackslash}m{15mm}}{\textbf{Leg 2 Day Rolling}} &
     \multicolumn{1}{>{\centering\arraybackslash}m{10mm}}{\textbf{Leg 2 Day Count}} &
     \multicolumn{1}{>{\centering\arraybackslash}m{15mm}|}{\textbf{Bloomberg Ticker}}\\ \hline
     
    MPBSC & 2d & 84d & 0.54000 & Swap & 28d & Following & ACT/360 & 28d & Following & ACT/360 & MPBSC Curncy \\
    MPBSF & 2d & 168d & 0.59000 & Swap & 28d & Following & ACT/360 & 28d & Following & ACT/360 & MPBSF Curncy \\
    MPBSI & 2d & 252d & 0.64000 & Swap & 28d & Following & ACT/360 & 28d & Following & ACT/360 & MPBSI Curncy \\
    MPBS1A & 2d & 364d & 0.68000 & Swap & 28d & Following & ACT/360 & 28d & Following & ACT/360 & MPBS1A Curncy \\
    MPBS2B & 2d & 728d & 0.72000 & Swap & 28d & Following & ACT/360 & 28d & Following & ACT/360 & MPBS2B Curncy \\
    MPBS3C & 2d & 1092d & 0.81000 & Swap & 28d & Following & ACT/360 & 28d & Following & ACT/360 & MPBS3C Curncy \\
    MPBS4D & 2d & 1456d & 0.88000 & Swap & 28d & Following & ACT/360 & 28d & Following & ACT/360 & MPBS4D Curncy \\
    MPBS5E & 2d & 1820d & 0.92000 & Swap & 28d & Following & ACT/360 & 28d & Following & ACT/360 & MPBS5E Curncy \\
    MPBS7G & 2d & 2548d & 1.00500 & Swap & 28d & Following & ACT/360 & 28d & Following & ACT/360 & MPBS7G Curncy \\
    MPBS10J & 2d & 3640d & 1.04000 & Swap & 28d & Following & ACT/360 & 28d & Following & ACT/360 & MPBS10J Curncy \\
    MPBS13 & 2d & 4368d & 1.04000 & Swap & 28d & Following & ACT/360 & 28d & Following & ACT/360 & MPBS13 Curncy \\
    MPBS16C & 2d & 5460d & 1.02000 & Swap & 28d & Following & ACT/360 & 28d & Following & ACT/360 & MPBS16C Curncy \\
    MPBS21H & 2d & 7280d & 1.02500 & Swap & 28d & Following & ACT/360 & 28d & Following & ACT/360 & MPBS21H Curncy \\
    MPBS32F & 2d & 10920d & 1.02500 & Swap & 28d & Following & ACT/360 & 28d & Following & ACT/360 & MPBS32F Curncy \\ \hline

    \end{tabular}
    \caption{In this table, we present the Bloomberg market data of the USDMXM cnCXS used for the construction of the MXN-discount curve (collateralized in USD) and the TIIE 28d forward curve. Quotes are End of Day prices from May 29, 2015. The quotes were taken from a Bloomberg Terminal on June 21, 2015.}
    \label{USDMXNXCS:Bloomberg}
    \end{sidewaystable}

\end{document}